\documentclass[dsingle]{Dissertate}

\usepackage{epigraph,tabularx,enumitem,fourier-orns}
\usepackage{tikz,caption,subcaption}
\usetikzlibrary{arrows.meta}

\usepackage{singer-macros}

\usepackage{bm}
\usepackage{graphicx}
\usepackage{color,xcolor}
\usepackage{amsmath,amsfonts, amssymb}

\usepackage{amsthm,amsmath,thmtools}

\PassOptionsToPackage{hyphens}{url}
\usepackage[capitalise,nameinlink]{cleveref}

\usepackage[style=alphabetic, backend=biber, minalphanames=3, maxalphanames=4, maxbibnames=99]{biblatex}

\crefformat{section}{#2\S#1#3}
\crefformat{subsection}{#2\S#1#3}
\crefformat{subsubsection}{#2\S#1#3}

\declaretheoremstyle[bodyfont=\it,qed=\qedsymbol]{noproofstyle}

\numberwithin{equation}{chapter}

\declaretheorem[numberlike=equation]{observation}

\declaretheorem[name=Observation,numbered=no]{observation*}
\declaretheorem[numberlike=equation]{remark}
\declaretheorem[name=Remark,numbered=no]{remark*}

\declaretheorem[numberlike=equation]{fact}

\declaretheorem[numberlike=equation]{theorem}

\declaretheorem[name=Theorem,numbered=no]{theorem*}

\declaretheorem[numberlike=equation]{lemma}
\declaretheorem[name=Lemma,numbered=no]{lemma*}

\declaretheorem[numberlike=equation]{corollary}
\declaretheorem[name=Corollary,numbered=no]{corollary*}

\declaretheorem[numberlike=equation]{proposition}
\declaretheorem[name=Proposition,numbered=no]{proposition*}

\declaretheorem[numberlike=equation]{claim}
\declaretheorem[name=Claim,numbered=no]{claim*}

\declaretheorem[name=Conjecture,numbered=no]{conjecture*}

\declaretheorem[name=Question,numbered=no]{question*}

\declaretheorem[name=Open Problem]{openproblem}

\declaretheoremstyle[bodyfont=\it]{defstyle} 

\declaretheorem[numberlike=equation,style=defstyle]{definition}
\declaretheorem[unnumbered,name=Definition,style=defstyle]{definition*}

\declaretheorem[numberlike=equation,style=defstyle]{example}
\declaretheorem[unnumbered,name=Example,style=defstyle]{example*}

\declaretheorem[unnumbered,name=Notation=defstyle]{notation*}

\declaretheorem[numberlike=equation,style=defstyle]{construction}
\declaretheorem[unnumbered,name=Construction,style=defstyle]{construction*}

\declaretheoremstyle[]{rmkstyle} 

\newcommand{\m}[1][]{\textsf{Max-}#1}
\newcommand{\mcsp}[1][]{\textsf{Max-CSP}(#1)}
\newcommand{\mbcsp}[1][]{\textsf{Max-}\overline{\mbox{\textsf{B}}}\textsf{CSP}(#1)}

\newcommand{\bgd}{(\beta,\gamma)\text{-}}
\newcommand{\ebgd}{(\beta+\epsilon,\gamma-\epsilon)\text{-}}
\newcommand{\mF}{\mcsp[\mathcal{F}]}
\newcommand{\mf}{\mcsp[\{f\}]}
\newcommand{\mbf}{\mbcsp[f]}

\newcommand{\val}{\textsf{val}}

\newcommand{\cut}{\textsf{Cut}}
\newcommand{\mcut}{\m[\cut]}
\newcommand{\dcut}{\textsf{DiCut}}
\newcommand{\mdcut}{\m[\dcut]}
\newcommand{\twoand}{\textsf{2AND}}
\newcommand{\mtwoand}{\m[\twoand]}
\newcommand{\threeand}{\textsf{3AND}}
\newcommand{\mthreeand}{\m[\threeand]}
\newcommand{\kand}{k\textsf{AND}}
\newcommand{\mkand}{\m[\kand]}

\newcommand{\twosat}{\textsf{2SAT}}
\newcommand{\mtwosat}{\m[\twosat]}
\newcommand{\threesat}{\textsf{3SAT}}
\newcommand{\mthreesat}{\m[\threesat]}
\newcommand{\ksat}{k\textsf{SAT}}
\newcommand{\mksat}{\m[\ksat]}

\newcommand{\ug}{\textsf{UG}}
\newcommand{\mug}{\m[\ug]}
\newcommand{\mas}{\textsf{MAS}}

\newcommand{\bpd}{\textsf{BPD}}
\newcommand{\seqbpd}{\textsf{sBPD}}
\newcommand{\seqibpd}{\textsf{siBPD}}
\newcommand{\seqirsd}{\textsf{siRSD}}
\newcommand{\rmd}{\textsf{RMD}}
\newcommand{\seqrmd}{\textsf{sRMD}}
\newcommand{\pllrmd}{\textsf{pRMD}}

\newcommand{\Bern}{\CB}
\newcommand{\Unif}{\CU}
\newcommand{\Matchings}{\CM}
\newcommand{\Graphs}{\CG}

\newcommand{\yes}{\textbf{Yes}}
\newcommand{\no}{\textbf{No}}
\newcommand{\y}{\textbf{Y}}
\newcommand{\n}{\textbf{N}}

\newcommand{\veca}{\mathbf{a}}
\newcommand{\vecb}{\mathbf{b}}
\newcommand{\vecc}{\mathbf{c}}
\newcommand{\vece}{\mathbf{e}}

\newcommand{\vecj}{\mathbf{j}}

\newcommand{\vecs}{\mathbf{s}}
\newcommand{\vecu}{\mathbf{u}}
\newcommand{\vecv}{\mathbf{v}}

\newcommand{\vecx}{\mathbf{x}}
\newcommand{\vecy}{\mathbf{y}}
\newcommand{\vecz}{\mathbf{z}}

\newcommand{\vecsigma}{\boldsymbol{\sigma}}

\newcommand{\vecmu}{\boldsymbol{\mu}}

\newcommand{\vecpi}{\boldsymbol{\pi}}
\newcommand{\vecomega}{\boldsymbol{\omega}}

\newcommand{\veczero}{\mathbf{0}}
\newcommand{\vecone}{\mathbf{1}}

\newcommand{\tv}{\mathrm{tv}}

\newcommand{\Alice}{\textsf{Alice}}
\newcommand{\Bob}{\textsf{Bob}}
\newcommand{\Carol}{\textsf{Carol}}
\newcommand{\Player}{\textsf{Player}}
\newcommand{\Alg}{\textsf{Alg}}
\newcommand{\R}{\textsf{R}}

\newcommand{\bias}{\textsf{bias}}
\newcommand{\fold}{\textsf{fold}}
\newcommand{\supp}{\textsf{supp}}

\newcommand{\ord}{\mathbf{ord}}

\newcommand{\Test}{\textsf{Test}}
\newcommand{\Prot}{\textsf{Prot}}

\newcommand{\Th}{\textsf{Th}}

\newcommand{\sym}{\mathfrak{S}}

\newcommand{\mocsp}[1][]{\textsf{Max-OCSP}(#1)}
\newcommand{\mPi}{\mocsp[\Pi]}
\newcommand{\ordval}{\textsf{ordval}}

\newcommand{\Btwn}{\textsf{Btwn}}
\newcommand{\mbtwn}{\m[\Btwn]}

\newcommand{\coarsen}{{\downarrow}}
\newcommand{\refine}{{\uparrow}}

\newcommand{\Sym}{\textsf{Sym}}

\renewcommand{\hat}{\widehat}
\renewcommand{\tilde}{\widetilde}
\allowdisplaybreaks

\titleformat{\part}{\filcenter}{}{8pt}{ {\fontsize{100}{110}\selectfont \color{chaptergrey} \thepart} \vspace{1in} \\ \Huge\bfseries\scshape}

\begin{document}

% the front matter

% Some details about the dissertation.
\title{On streaming approximation algorithms for constraint satisfaction problems}
\author{Noah Singer}

%If you have one advisor
\advisor{Madhu Sudan}

\committeeInternalOne{Salil Vadhan}
% \committeeInternalTwo{Person Inside Two}

%If you are coadvised
% \coadvisorOne{Delightful Researcher}
% \coadvisorTwo{Equally D. Researcher}
% \committeeInternal{Person Inside}

% Everyone has an External committee member
% \committeeExternal{Person Outside}

% ... about the degree.
\degree{Bachelor of Arts}
\field{Computer Science and Mathematics}
\degreeyear{2022}
\degreeterm{Spring}
\degreemonth{March}
\department{Computer Science}

% ... about the candidate's previous degrees.

\maketitle
% \copyrightpage
% \frontmatter
\setstretch{\dnormalspacing}
\abstractpage
\tableofcontents
% %\authorlist
% \listoffigures
\setlength{\epigraphwidth}{0.8\textwidth}
\dedicationpage
\acknowledgments

% \doublespacing

% include each chapter...
% \setcounter{chapter}{-1}  % start chapter numbering at 0

\chapter{Introduction}\label{chap:introduction}

\newthought{This thesis sits at the intersection} of two broad subfields of computer science: \emph{combinatorial optimization} and \emph{big data}. The former is an umbrella term for computational problems whose goal is to find the ``best'' solution among a finite, ``structured'' set of solutions, including tasks such as routing, packing, scheduling, and resource allocation. The latter encompasses a similar breadth of computational settings involving ``massive'' amounts of input data which necessitate ``highly efficient'' resource usage (quantified in terms of memory, time, energy, etc.), leading to e.g. online, distributed, parallel, and sublinear-time algorithms.

In this thesis, more specifically, we consider a particular class of combinatorial optimization problems called \emph{constraint satisfaction problems (CSPs)}, and a particular class of algorithms for big data called \emph{streaming algorithms}. Roughly, the goal of a CSP is to find a ``global solution'' satisfying as many ``local constraints'' as possible. More precisely, fix a finite set $\Sigma$, called the \emph{alphabet}. A \emph{$k$-ary predicate} is a function $\Sigma^k \to \{0,1\}$. A set $\CF$ of $k$-ary predicates defines a CSP denoted $\mF$. An \emph{instance} of $\mF$ is defined by $n$ \emph{variables}, each of which can be \emph{assigned} to a value drawn from the alphabet $\Sigma$, and $m$ \emph{constraints}, each of which applies a predicate from $\CF$ to some subset of $k$ variables. The goal of the $\mF$ problem is to find an assignment of variables to values satisfying as many constraints as possible. (See \cref{sec:csps} below for a formal definition.) For instance, consider the \emph{maximum cut} problem ($\mcut$), arguably the simplest CSP, which is defined over the Boolean alphabet $\Sigma = \{0,1\}$ by the binary predicate $\cut(a,b) = a \oplus b$ (i.e., $\cut(a,b)=1$ iff $a \neq b$). Thus in a $\mcut$ instance, we have $n$ variables which can be assigned to one of two values ($0$ or $1$), and $m$ constraints, each of which says ``variable $i$ and variable $j$ should have different values''. $\mcut$ has practical applications in e.g. circuit design and statistical physics \cite{BGJR88}; a toy application is splitting children into two groups on a field trip so as to minimize conflict, given a list of pairs of children who dislike each other. The CSP framework includes many other problems which are widely studied both in theory and in practice, such as $\mksat$ and $\m[q\textsf{Cut}]$.

Now, suppose that we want to solve a CSP such as $\mcut$ --- that is, find a good assignment, or at least understand whether good assignments exist --- on instances which are ``big'' in the following sense: The constraints are generated in some sequence, and there are too many of them to store. For example, we could imagine a setting in which many clients transmit many constraints to a server, which tries to satisfy as many of the constraints as possible. The theoretical model of \emph{streaming algorithms} attempts to capture these challenges: An algorithm is presented with a sequence of inputs, has limited memory space, and can only access the inputs in a single sequential pass.\footnote{I.e., the algorithm lacks \emph{random access} to the inputs --- it only sees the first input $C_1$, then $C_2$, etc., and cannot access an input $C_j$ before or after its position in the sequence. Of course, it can choose to \emph{store} an input $C_j$ once it's seen it, but its storage space is very limited.}\footnote{This particular streaming model can be relaxed in various ways, such as allowing multiple passes over the stream, or randomly ordering the stream's contents. See \cref{sec:diff-streaming-models}.} (See \cref{sec:streaming} for a formal definition of the model.) Streaming algorithms were introduced by Alon, Matias, and Szegedy~\cite{AMS99}, and model practical settings such as real-time analysis of network traffic and scientific data. For more on streaming algorithms, see the surveys \cite{Mut05,Cha20}.

Concretely, a ``streaming CSP algorithm'' is presented with the constraints of an input instance $\Psi$ and is tasked with \emph{estimating} the \emph{value} of $\Psi$, denoted $\val_\Psi$, which is the maximum fraction of constraints satisfiable by any assignment to the variables.\footnote{We don't typically require that the algorithm actually \emph{output} a good assignment, since even writing down such an assignment may take too much space.} See \cref{fig:streaming-mcut} on the next page for a visual representation. To be precise, we use the following standard notion of ``estimation'' for CSP values: For $\alpha \in [0,1]$, we say $\tilde{v}$ \emph{$\alpha$-approximates} $\val_\Psi$ if $\alpha \val_\Psi \leq \tilde{v} \leq \val_\Psi$. In other words, $\tilde{v}$ is an underestimate for $\val_\Psi$, but not by a factor smaller than $\alpha$.

\def\hexa{0.5}
\def\hexb{0.866}
\def\hexc{1}
\def\hexs{3.5}
\def\hext{1}
\def\hexu{1.4}

\begin{figure}
\centering
\begin{tikzpicture}[vertex/.style={fill=black}, nxt/.style={-{Triangle[width=15pt,length=8pt]}, line width=8pt,draw=gray}, edge/.style={line width=1.5pt,draw=black!20!blue}]

\foreach \i in {0,...,4} {
    \draw[vertex] (\hexs*\i+\hext*\hexc,0) circle (3pt);
    \draw[vertex] (\hexs*\i+\hext*\hexa,\hext*\hexb) circle (3pt);
    \draw[vertex] (\hexs*\i+\hext*\hexa,-\hext*\hexb) circle (3pt);
    \draw[vertex] (\hexs*\i-\hext*\hexc,0) circle (3pt);
    \draw[vertex] (\hexs*\i-\hext*\hexa,\hext*\hexb) circle (3pt);
    \draw[vertex] (\hexs*\i-\hext*\hexa,-\hext*\hexb) circle (3pt);
}

\foreach \i in {0,...,3} {
    \draw[nxt] (\hexs*\i+\hext*\hexc*\hexu,0) to (\hexs*\i+\hexs-\hext*\hexc*\hexu,0);
}

\filldraw[white!50!green] (\hexs,-1.5) -- (\hexs-0.75,-2.5) -- (\hexs+0.75,-2.5);
\node[black,align=center] at (\hexs,-2.25) {$\Alg$};

\draw[edge] (\hexs*0-\hext*\hexc,0) to (\hexs*0+\hext*\hexc,0);
\draw[edge] (\hexs*1+\hext*\hexc,0) to (\hexs*1+\hext*\hexa,-\hext*\hexb);
\draw[edge] (\hexs*2+\hext*\hexa,-\hext*\hexb) to (\hexs*2-\hext*\hexa,\hext*\hexb);
\draw[edge] (\hexs*3-\hext*\hexa,-\hext*\hexb) to (\hexs*3+\hext*\hexa,\hext*\hexb);
\draw[edge] (\hexs*4-\hext*\hexa,\hext*\hexb) to (\hexs*4-\hext*\hexa,-\hext*\hexb);
\end{tikzpicture}

\caption{A visual representation of an instance of $\mcut$ on $n=6$ variables (``vertices'') with $m=5$ constraints (``edges''). The streaming algorithm $\Alg$ makes a single linear pass through the list of constraints, and tries to decide whether it's possible to find a partition of the vertices which is crossed by most of the edges.}
\label{fig:streaming-mcut}
\end{figure}
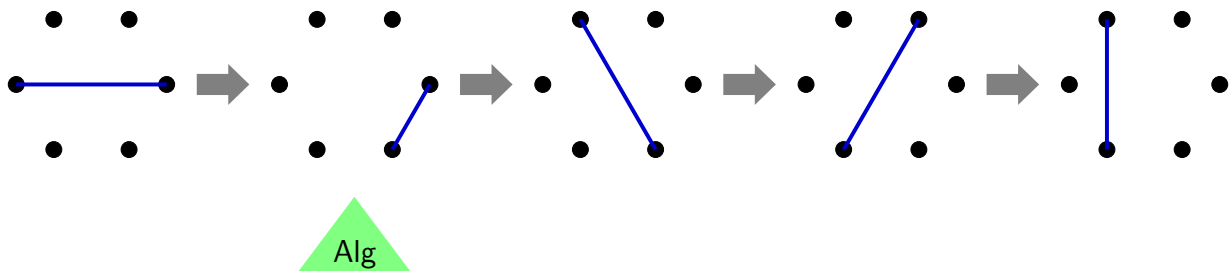

Every $\mcut$ instance has value at least $\frac12$; indeed, a random assignment satisfies half the constraints in expectation. Thus, the estimate $\frac12$ is always a $\frac12$-approximation to $\val_\Psi$. Conversely, instances of $\mcut$ on many variables with many uniformly random constraints have values arbitrarily close to $\frac12$ (see \cref{sec:bpd}). Thus, $\frac12$ is the infimum of $\val_\Psi$ over all $\mcut$ instances $\Psi$, so it is the best possible ``trivial'' (i.e., input-independent) estimate. A \emph{nontrivial} approximation to $\mcut$ is therefore a $(\frac12+\epsilon)$-approximation for some $\epsilon > 0$. For other CSPs we similarly term an $\alpha$-approximation nontrivial if $\alpha$ exceeds the infimum of all instances' values \cite{CGSV21-finite}. If a CSP cannot be nontrivially approximated, it is \emph{approximation-resistant}.

The central question is now:

\begin{center}
    \emph{For which families of predicates, which desired approximation ratios $\alpha$, and which classes of streaming algorithms can we prove positive (a.k.a. algorithmic) or negative (a.k.a. hardness) results for $\mF$?}
\end{center}

Strikingly, until 2015 no research explored the intersection of CSPs and streaming algorithms, though alone each area had been explored extensively. But since then, based on open questions posed at the 2011 Bertorino Workshop on Sublinear Algorithms \cite[Question 10]{IMNO11}, there has been a significant line of research on streaming algorithms for CSPs \cite{KK15,KKS15,GVV17,KKSV17,KK19,CGV20,CGSV21-boolean,CGSV21-finite,SSV21,BHP+22,CGS+22} which has both resolved a good number of the initial questions and advanced a number of new and interesting ones.

\section{Summary of prior work}

Kogan and Krauthgamer~\cite{KK15} and Kapralov, Khanna, and Sudan~\cite{KKS15} studied streaming approximations for $\mcut$; the latter proved that for every $\epsilon > 0$, streaming $(\frac12+\epsilon)$-approximations to $\mcut$ on instances with $n$ variables require $\Omega(\sqrt n)$ space (see \cref{thm:mcut-hardness} below).\footnote{Note that the implicit constant in the $\Omega(\cdot)$ notation can depend on the desired constant in the approximation factor, i.e., on $\epsilon$. To simplify the language, we refer to this as a ``$\sqrt n$-space inapproximability result'', but we are carefully to make the quantifiers explicit in theorem statements.} Their result actually holds in a stronger model than we've described so far, in which the constraints of an instance $\Psi$ are ``randomly ordered''. Guruswami, Velingker, and Velusamy~\cite{GVV17} gave an $O(\log n)$-space streaming $(\frac25-\epsilon)$-approximation algorithm for the \emph{maximum directed cut} ($\mdcut$) problem, and also showed, via reduction from $\mcut$, that $(\frac12+\epsilon)$-approximations require $\Omega(\sqrt n)$ space. Chou, Golovnev, and Velusamy~\cite{CGV20} closed this gap, by showing that for every $\epsilon > 0$, $\mdcut$ can be $(\frac49-\epsilon)$-approximated in $O(\log n)$ space, but $(\frac49+\epsilon)$-approximations require $\Omega(\sqrt n)$ space (\cref{thm:mdcut-characterization} below). \cite{CGV20} also analyzes the approximability of several other problems, including $\mksat$.

Building on \cite{CGV20}, for every $\mF$ problem, Chou, Golovnev, Sudan, and Velusamy~\cite{CGSV21-finite} proved a \emph{dichotomy theorem} for \emph{sketching algorithms}, which are ``composable'' streaming algorithms (see \cref{sec:streaming} for the definition). Their theorem says that for every family of predicates $\CF$, there exists some \emph{sketching approximability threshold} $\alpha(\CF) \in [0,1]$ such that for every $\epsilon > 0$, $\mF$ can be $(\alpha(\CF)-\epsilon)$-approximated in $O(\polylog n)$ space, but $(\alpha(\CF)+\epsilon)$-approximations require $\Omega(\sqrt n)$ space. In an earlier paper \cite{CGSV21-boolean}, they prove the same result in the special case of so-called ``Boolean CSPs with negations'': For a Boolean predicate $f : \{0,1\}^k\to\{0,1\}$, consider the family $\CF_{\neg f}$ of $2^k$ predicates corresponding to negating subsets of $f$'s inputs, which define a CSP $\mbf \eqdef \mcsp[\CF_{\neg f}]$. \cite{CGSV21-boolean} provides a sketching dichotomy theorem for $\mbf$, yielding a sketching approximability threshold $\alpha(f)\eqdef \alpha(\CF_{\neg f})$. However, neither \cite{CGSV21-boolean,CGSV21-finite} provide explicit procedures for calculating these thresholds; indeed, it is not \emph{a priori} clear that they even have closed-form expressions. On the other hand, \cite{CGSV21-boolean,CGSV21-finite} include a number of other results, in particular conditions under which the dichotomy's lower bound extends to \emph{streaming} algorithms. See \cref{sec:cgsv} below for formal descriptions of the \cite{CGSV21-boolean} results.

Another line of work \cite{KKSV17,KK19,CGS+22} extends some of the above inapproximability results to the setting of \emph{linear}-space algorithms. In particular, Kapralov and Krachun~\cite{KK19} show that streaming $(\frac12+\epsilon)$-approximations for $\mcut$ require $\Omega(n)$ space, and Chou, Golovnev, Sudan, Velingker, and Velusamy~\cite{CGS+22} extend this tight inapproximability to certain other ``linear'' CSPs (see \cref{thm:cgsvv} in \cref{sec:cgsvv} below).

\section{Contributions and outline}

\cref{chap:prelims} contains formal definitions for CSPs and streaming algorithms, along with miscellaneous preliminary material.

\paragraph{Expositions of prior work.} In \cref{part:prior-results} of this thesis, we present several of the foundational works on streaming algorithms for CSPs. The aim here is to make these results accessible to a general theory audience and take advantage of newer perspectives from later papers which, in some cases, have substantially simplified earlier constructions.

\cref{chap:mcut} contains the first self-contained writeup of the result, implied by \cite{KKS15}, that nontrivially approximating $\mcut$ requires $\Omega(\sqrt n)$ streaming space. Recall, \cite{KKS15} includes a stronger version of this statement for ``randomly ordered'' streams of constraints; the proof of this stronger statement requires some additional steps which are somewhat tangential to the fundamental question of $\mcut$'s approximability. Our proof combines and simplifies pieces of \cite{GKK+08, KKS15, CGSV21-boolean}. We then discuss various ways in which $\mcut$'s hardness has been strengthened in \cite{KKS15,KK19}.

In \cref{chap:mdcut}, we present \cite{CGV20}'s proof that $\mdcut$'s streaming approximability threshold (for $\log n$ vs. $\sqrt n$ space) is $\frac49$. Our exposition differs from the original in \cite{CGV20} in two important ways. Firstly, we describe a substantially simpler algorithm, which is based on observations from our joint work \cite{BHP+22} (which we later explore in \cref{chap:sym-bool}). Secondly, we emphasize structural similarities between the $(\frac49-\epsilon)$-approximation algorithm and the hardness proof for $(\frac49+\epsilon)$-approximation, namely the use of so-called \emph{template distributions}, which later become the basis for the dichotomy theorems for \emph{all} CSPs of
\cite{CGSV21-boolean,CGSV21-finite}.

We turn to these dichotomy theorems, as well as the \cite{CGS+22} linear-space lower bounds, in \cref{chap:framework-papers}. We give high-level surveys of these works, which are less technical than our discussions in the previous two chapters.\footnote{In \cref{chap:mcut,chap:mdcut} we sometimes omit concentration bounds in proofs in order to focus on more important quantitative aspects; these places are carefully noted.} However, we do take care when stating results which we'll require in our own work, presented in \cref{part:contributions}.

\paragraph{Ordering CSPs.} In \cref{chap:ocsps} (which begins \cref{part:contributions}), we present the main result of our joint work with Sudan and Velusamy~\cite{SSV21}, published in APPROX'21. This result (\cref{thm:ocsp-hardness} below) states that for a certain class of ``CSP-like'' problems called \emph{ordering constraint satisfaction problems (OCSPs)}, all nontrivial streaming approximations require $\Omega(n)$ space. Recall that in a CSP, the solution space is the set of assignments from variables to alphabet values; in an OCSP, the solution space is instead the set of permutations on the variables (see \cref{sec:ocsps} below for details), and thus, OCSPs are good models for scheduling problems.

Our result in \cite{SSV21} is ``triply optimal'': It rules out \emph{all} nontrivial approximations for \emph{all} OCSPs, and the $\Omega(n)$-space bound is \emph{tight} (up to polylogarithmic factors, see \cref{rem:sparsifier} below). Previous works \cite{GVV17,GT19} studied the \emph{maximum acyclic subgraph ($\mas$)} problem --- a simple OCSP defined by the predicate ``variable $i$ is before variable $j$'' --- and proved that some approximations require $\Omega(\sqrt n)$ space; thus, our result in \cite{SSV21} is an improvement in all three ``parameters''.

Additionally, our inapproximability proof for OCSPs relies on linear-space inapproximability results for (non-ordering) CSPs from \cite{CGS+22}, which we'll have described earlier in \cref{sec:cgsvv}. Given this context, in \cref{chap:ocsps} we develop a more modular version of the proof than appeared originally in \cite{SSV21}.

\paragraph{Symmetric Boolean CSPs.} In \cref{chap:sym-bool}, we present our joint work with Boyland, Hwang, Prasad, and Velusamy~\cite{BHP+22} which investigates streaming algorithms for specific types of CSPs over the Boolean alphabet $\{0,1\}$, namely those with \emph{symmetric} predicates (i.e., predicates which depend only on the number of 1's, a.k.a. Hamming weight, of their inputs). These CSPs are an interesting ``lens'' through which we examine several questions left open by the work of Chou \emph{et al.}~\cite{CGSV21-boolean,CGSV21-finite}.

Our main goal in this work is to discover and exploit ``structural'' properties of the \cite{CGSV21-boolean} dichotomy theorem in order to give explicit expressions for the sketching approximability threshold $\alpha(f)$ for several classes of predicates $f$. For instance, letting $\alpha'_k \eqdef 2^{-(k-1)}(1-k^{-2})^{(k-1)/2}$, we show that $\alpha(\kand) = \alpha'_k$ for odd $k$; $\alpha(\kand)=2\alpha'_{k+1}$ for even $k$ (\cref{thm:kand-approximability} below); and $\alpha(\Th^{k-1}_k) = \frac{k}2\alpha'_{k-1}$ for odd $k$ (\cref{thm:k-1-k-approximability} below). We also resolve the thresholds for fifteen other specific functions (\cref{sec:other-analysis} below).

We also present two other results based on our study of symmetric Boolean predicates. Firstly, for \emph{threshold} predicates (which equal $1$ iff their input's Hamming weight is at least $t$ for some constant $t$, a.k.a. monotone symmetric predicates), we develop substantially simpler sketching algorithms which also achieve the optimal approximation thresholds given by \cite{CGSV21-boolean} (see \cref{thm:thresh-bias-alg} below). Secondly, we show that the criteria in \cite{CGSV21-boolean} for proving \emph{streaming} hardness results are ``incomplete'', in the sense that they cannot establish a sharp approximability threshold for $\mthreeand$.

Most proofs and discussions in \cref{chap:sym-bool} are reproduced with few changes from our paper \cite{BHP+22}.

\paragraph{Open questions.} Finally, in \cref{chap:conclusions}, we collect several interesting open questions and directions for further investigation into streaming algorithms for CSPs. Some were posed already in prior work, while others arose in our work with collaborators and appear for the first time here.

\section{Important themes and motivations}

What are the theoretical implications of all this work on streaming algorithms for approximating CSPs? In contrast to streaming algorithms, there is an extensive theory of CSP approximability for \emph{classical} algorithms, where the performance requirement is only running in $\poly(n)$ time; see e.g. the survey \cite{Tre04}. However, this theory assumes complexity-theoretic conjectures such as $\P\neq\NP$. By contrast, as we'll see, the hardness results for streaming CSPs discussed in this thesis are all \emph{unconditional}, i.e., they do not rely on any unproven conjectures!

In some sense, the fact that we have unconditional hardness results may more of a ``bug'' than a ``feature'' of the streaming model. Indeed, almost all useful algorithms from the world of classical CSPs seemingly cannot be implemented as streaming algorithms; thus, it's arguably unsurprising that it is feasible to prove hardness results against the remaining algorithms (though when we instead manage to develop streaming algorithms, it's quite exciting!).

In this section, we argue that the streaming CSP theory has more to offer than the technical statements ``in it of themselves''. In particular, we highlight three themes which seem relevant to the broader areas of CSPs or streaming models, and which will be helpful to keep in mind for the remainder of the thesis. This section may be safely skipped (and hopefully revisited!) for readers unfamiliar with the preliminaries discussed in \cref{chap:prelims}.

\subsection{Parallels to classical hardness-of-approximation}

Many of the first classical hardness-of-approximation results for CSPs, and other combinatorial optimization problems, were proven using the machinery of \emph{probabilistically checkable proofs (PCPs)} \cite{AS98,ALM+98}. Indeed, the classical approximability of many problems have been tightly resolved using PCP techniques, such as the CSP $\mthreesat$ \cite{Has01} as well as non-CSPs including set cover \cite{Fei98}, chromatic number \cite{FK98}, and max-clique \cite{Has99}. PCP-based hardness results typically only assume the $\P \neq \NP$ conjecture. Later, Khot~\cite{Kho02} introduced a new complexity-theoretic hypothesis (stronger than $\P \neq \NP$) called the \emph{unique games conjecture (UGC)}, with an aim towards further understanding the classical approximability of combinatorial optimization problems.\footnote{The UGC roughly posits optimal hardness for the CSP $\mug$ which allows binary predicates with the property fixing one variable uniquely specifies the other variable, over an arbitrarily large alphabet $\Sigma$. In other words, $\mug = \mcsp[\{f_{\vecpi} : \vecpi \text{ is a bijection on }\Sigma\}]$, where $f_{\vecpi} : \Sigma^2 \to \{0,1\}$ is defined by $f_{\vecpi}(a,b)=1$ iff $b=\vecpi(a)$.} The UGC is now known to imply tight approximability results for CSPs including $\mcut$ \cite{KKMO07}, $\mtwosat$ \cite{Aus07}, and $\mkand$ (up to a constant factor) \cite{ST09,CMM09} as well as for other combinatorial optimization problems including OCSPs \cite{GHM+11}, $\mcut$ on instances of value $\frac12+\epsilon$ \cite{KO09}, and vertex cover \cite{KR08}.

Consider two ``worlds'' for a programmer hoping to approximately solve combinatorial optimization problems: World U, where the UGC is true and the programmer can employ polynomial-time algorithms, and World S, where the programmer is restricted to polylogarithmic-space streaming algorithms. There are surprising parallels between what we know about approximability in these two worlds:

\begin{itemize}
    \item The Chou \emph{et al.}~\cite{CGSV21-finite} dichotomy theorem for $\sqrt n$-space sketching in World S, presented in \cref{sec:cgsv}, is analagous to Raghavendra's dichotomy theorem~\cite{Rag08} in World U. The former shows that in World S, bias-based linear sketching algorithms \`a la \cite{GVV17} are optimal for every CSP; the latter shows that in World U, the optimal algorithms are \emph{semidefinite programming (SDP)} ``relax and round'' algorithms \`a la \cite{GW95} for $\mcut$ (see \cite{MM17} for a recent survey).\footnote{There are a number of other ``classical dichotomy theorems'' for CSPs, specifically concerning \emph{exact computation} (i.e., deciding whether $\val_\Psi = 1$) \cite{Sch78,FV98,Bul17,Zhu20} and so-called \emph{coarse approximation} (see \cite{KSTW01} and the book \cite{CKS01}).}
    \item In World S, \cite{CGSV21-finite} also shows that every CSP satisfying a natural property called ``supporting one-wise independence'' is streaming approximation-resistant in $\sqrt n$ space (see \cref{ex:one-wise-indep} below). Austrin and Mossel~\cite{AM09} proved an analogous result in World U for CSPs satisfying a higher-order condition called ``supporting \emph{two}-wise independence''.
    \item For ordering CSPs in World S, the linear-space approximation-resistance result from our joint work~\cite{SSV21}, presented in \cref{chap:ocsps}, is analogous to a theorem of Guruswami, H{\aa}stad, Manokaran, Raghavendra, and Charikar~\cite{GHM+11} which states that OCSPs are approximation-resistant in World U. As we'll see in \cref{chap:ocsps}, there are striking similarities between the proof methods for these two results --- both rely on reducing from inapproximability results for ``coarse'' CSP variants of OCSPs.
    \item $\mkand$ is the most ``approximable'' $\mbf$ problem (see \cref{rem:kand-approx} below), intuitively because its constraints are the most ``informative'' --- every constraint tells us \emph{exactly} what its variables need to be assigned to. Thus, it is fortunate that we can resolve its approximability in both worlds: In World S, our joint work~\cite{BHP+22}, presented in \cref{chap:sym-bool}, shows that the $\sqrt n$-space sketching approximability of the $\mkand$ problem is $(2-o(1)) 2^{-k}$ (\cref{thm:kand-approximability} below). In World U, $\mkand$ is $\Theta(k2^{-k})$-approximable \cite{ST09,CMM09}.\footnote{Indeed, \cite{ST09} show even that $O(k2^{-k})$-approximation is $\NP$-hard, though their result is stronger assuming the UGC.} Moreover, with regard to the algorithmic aspects discussed in the first bullet point, for $\mkand$, our optimal sketching algorithms from \cite{BHP+22} and the optimal SDP rounding algorithms from \cite{Has05,CMM09} have a surprisingly similar structure: Find a good ``guess'' assignment $\vecx \in \BZ_2^n$ and then randomly perturb each of its bits independently with some small constant probability. However, our algorithm chooses $\vecx$ purely combinatorially --- based on whether a variable occurs more often positively or negatively --- while the \cite{Has05,CMM09} algorithms produce it by randomly rounding an SDP.
\end{itemize}

These parallels may be evidence that there is some truth to World U's hardness results, at least for weak classes of algorithms.

\subsection{Random instances of CSPs (and average-case hardness)}

In the classical setting, there has also been a great deal of interest in algorithms for ``random instances'' of CSPs; in a typical example, constraints are sampled by uniformly sampling variables and predicates. Feige's well-known \emph{random $\threesat$ hypothesis}~\cite{Fei02} states that classical algorithms cannot ``refute'' random instances of $\mthreesat$ for any constant ratio $\frac{m}n$ of constraints to variables; we very roughly define ``refuting'' as ``distinguishing from perfectly satisfiable instances.'' Feige's and related conjectures have numerous applications in hardness-of-approximation, cryptography, and learning (see \cite[\S1.2]{KMOW17} for a review). On the other hand, there has been significant algorithmic progress for a wide variety of CSPs in the setting where the constraint-to-variable ratio is larger \cite{BKPS98,FGK05,CGL07,FO07,AOW15,BM16,RRS17,GKM22}.

In contrast, in the setting of streaming algorithms, all our lower bounds come from random instances! Specifically, we'll show that streaming algorithms cannot distinguish between ``$\no$'' instances which are sampled fully randomly and ``$\yes$'' instances sampled randomly \emph{conditioned on having high value on a random ``planted'' assignment.} We'll explore this paradigm in the simple case of $\mcut$ in \cref{chap:mcut}.

Both the algorithmic results in the classical setting, and the hardness results in the streaming setting, rely on careful combinatorial analyses of random graphs. Again, streaming hardness may also be heuristic evidence for the truth of e.g. Feige's random $\threesat$ hypothesis, at least for weak classes of algorithms.

\subsection{Subtleties of streaming models}\label{sec:diff-streaming-models}

An instance of a binary CSP can be viewed as a (directed) graph in which edges are labeled with predicates; more broadly, $k$-ary CSPs are so-called ``$k$-uniform hypergraphs'', or \emph{$k$-hypergraphs} for short, with predicates labeling the hyperedges. Thus, the streaming approximation algorithms for CSPs which we've discussed fall under the purview of algorithms for ``graph streams'', which are problems in which the input is a stream of labeled (hyper)edges in a (hyper)graph (see e.g. the survey \cite{McG14}).

There are a number of interesting variations on the basic graph streaming model of ``small space, one pass''. In this subsection, we focus specifically on two ways in which the model can be weakened, namely randomly ordering the input or allowing multiple passes, and one way it can be strengthened, namely requiring the algorithms to be \emph{composable} (resulting in so-called ``sketching algorithms''). For each of these, we'll cite known separations for various approximate combinatorial optimization problems. 

\paragraph{Input ordering.} What happens when we require that a streaming algorithm succeeds only on randomly ordered input streams (with high probability), instead of on \emph{all} input orderings? Intuitively, this may lessen the burden on the algorithm, because since the algorithm is very ``forgetful'' about its input in the long term, a short sequence of bad inputs may cause catastrophic failures; in randomly ordered streams, such sequences may be less likely. This phenomenon has been widely explored throughout the literature on ``big data'' algorithms; see \S1 of the survey \cite{GS21} for an instructive exposition in the simple case of online algorithms for calculating the maximum element in a list. In the more immediate setting of graph streaming, Peng and Sohler~\cite{PS18} established provable separations between the ``random-order'' and ``adversarial-rder'' settings for graph problems including approximately counting components and approximating minimum spanning tree weight. See \cref{sec:mcut-input-ordering,sec:conc-rand-lspace} for discussions in the setting of streaming algorithms for CSPs.

\paragraph{Sketching vs. streaming.} \emph{Sketching} algorithms are special streaming algorithms which roughly have the following ``composability'' property: We can choose to split the input stream into pieces, and the algorithm has to ``support'' being run independently on each piece and then combining the results; See \cref{sec:streaming} for a formal definition. Kapralov, Kallaugher, and Price~\cite{KKP18,KP20} proved separations between sketching and streaming algorithms for the problem approximate triangle counting in graphs.

\paragraph{Multiple passes.} Finally, we can consider allowing the streaming algorithm to make multiple passes over the input data. Proving lower bounds against streaming algorithms even in two passes is typically very difficult, outside of highly structured contexts such as following paths in graphs (an example of so-called ``pointer chasing'' problems, see e.g. \cite{GO16}). Provable separations between single- and multi-pass streaming algorithms are known for estimating matrix norms \cite{BCK+18,BKKS20} and approximately counting subgraphs (see \cite{BC17}). Recently, Assadi, Kol, Saxena, and Yu~\cite{AKSY20} and Assadi and N~\cite{AN21} proved multipass lower bounds for approximating (variants of) cycle counting in graphs, which rule out \emph{some} nontrivial approximations for CSPs (though we seem quite far from even ruling out all nontrivial approximations for $\mcut$ in two passes).

\vspace{0.1in}

Furthermore, as we'll see in \cref{sec:lb-basics}, lower bounds for streaming problems are typically proven by reducing from \emph{communication problems}. In these problems, several players each get a ``chunk'' of the input stream, and in the corresponding \emph{communication-to-streaming reduction}, each player will run the streaming algorithm on their input chunk and then pass the current state of the streaming algorithm onto the next player. Thus, small-space streaming algorithms make for efficient communication protocols (and our goal is then to prove lower bounds against efficient communication protocols).

Details in the definitions of these communication problems crucially affect the streaming models we can prove lower bounds against. We give some high-level intuition for why this is the case. For starters, proving lower bounds against random-ordering algorithms generally requires \emph{symmetry} between the players: They should receive inputs drawn from the same sources, and behave the same way on these inputs. Otherwise, they'll be constructing a stream in the reduction which is far from randomly-ordered. On the other hand, to prove lower bounds against sketching algorithms, it suffices to consider communication games in which the players communicate in \emph{parallel} instead of sequentially, because the composability property implies that each player can independently run the algorithm on their own ``chunk''. Finally, the difficulty in proving multi-pass lower bounds is that the communication problem has multiple ``rounds'' (corresponding to each pass) in which each player gets to see their chunk again. All of these subtleties mean that the communication problems arising from the study of streaming approximations for CSPs are arguably quite interesting even apart from applications to streaming lower bounds. Finally, another motivation is that these communication problems have compelling connections to Fourier analysis and random graph combinatorics.
\chapter{Preliminaries}\label{chap:prelims}

\newthought{We begin with tools and definitions} which are necessary background for the remainder of the thesis. This chapter reviews topics including CSPs (\cref{sec:csps}), streaming and sketching algorithms (\cref{sec:streaming}), random (hyper)graphs (\cref{sec:hypergraphs}), Fourier analysis over $\BZ_q^n$ (\cref{sec:fourier}), and, importantly, the use of one-way communication lower bounds and Fourier analysis to prove streaming lower bounds (\cref{sec:lb-basics}).

We let $[n]$ denote the set of natural numbers $\{1,\ldots,n\}$ and $\BZ_q$ the integers modulo $q$.\footnote{The main difference between $[q]$ and $\BZ_q$ is the implied addition operation.} We typically use bold to denote vectors but not their components, e.g., $\vecb=(b_1,\ldots,b_k)$. Sequences of vectors are indexed with parentheses (e.g., $\vecb(\ell) = (b(\ell)_1,\ldots,b(\ell)_k)$). For a vector $\vecb \in \BZ_q^n$, we define its \emph{Hamming weight} (a.k.a. \emph{$0$-norm}) $\|\vecb\|_0 \eqdef |\{i \in [n]: b_i \neq 0\}|$ as its number of nonzero entries. For a finite alphabet $\Sigma$, we let $\Sigma^k,\Sigma^{\leq k},$ and $\Sigma^*$ denote strings over $\Sigma$ of length $k$, length at most $k$, and arbitrary length, respectively. Let $\1_S$ denote the indicator variable/function for an event $S$, and given any finite set $\Omega$, let $\Delta(\Omega)$ denote the space of probability distributions over $\Omega$.

We typically write distributions and some sets (such as function families) using calligraphic letters. For instance, $\Unif_S$ denotes the uniform distribution over $S$, $\Bern_p$ the Bernoulli distribution taking value $0$ with probability $p$ and $1$ with probability $1-p$, and and $\CP_X$ the probability distribution for a (discrete) random variable $X$. We use functional notation for distribution probabilities, i.e., $\CD(\omega) \eqdef \Pr_{\omega' \sim \CD}[\omega=\omega']$. The \emph{support} of a distribution $\CD \in \Delta(\Omega)$ is the set $\supp(\CD) \eqdef \{\omega \in \Omega: \CD(\omega) \neq 0\}$. We similarly define $\supp(f) \eqdef \{\omega \in \Omega: f(\omega) \neq 0\}$ for functions $f : \Omega \to \BC$ and $\supp(\vecv) \eqdef \{i \in [n] : v_i \neq 0\}$ for vectors $\vecv \in \BZ_q^n$ (so that $\|\vecv\|_0 = |\supp(\vecv)|$).

\section{Constraint satisfaction problems}\label{sec:csps}

A \emph{constraint satisfaction problem} $\mF$ is defined by $2 \leq q,k \in \BN$ and a labeled set of \emph{predicates} $\CF = \{f_b : \BZ_q^k \to \{0,1\}\}_{b\in B_{\CF}}$ for some finite set of labels $B_{\CF}$.\footnote{Properly speaking, the CSP is defined only by the family of predicates; we include labels for notational convenience.} A \emph{constraint} $C$ on $n \in \BN$ \emph{variables} is given by a triple $(b,\vecj,w)$ consisting of a label $b \in B_\CF$, a $k$-tuple $\vecj = (j_1,\ldots,j_k)\in[n]^k$ of distinct indices, and a weight $w \geq 0$. An \emph{instance} $\Psi$ of $\mF$ consists of a list of $m$ constraints $(C_\ell = (b(\ell),\vecj(\ell),w(\ell)))_{\ell\in[m]}$. We sometimes omit weights for the constraints, in which case we take them to be identically $1$; similarly, if $|\CF|=1$, we omit labels for the constraints.

The \emph{union} of two instances $\Psi_1$ and $\Psi_2$ of $\mF$, denoted $\Psi_1 \cup \Psi_2$, is the instance given by concatenating the lists of constraints for $\Psi_1$ and $\Psi_2$.

For an \emph{assignment} $\vecx \in \BZ_q^n$, let $\vecx|_\vecj \eqdef (x_{j_1},\ldots,x_{j_k}) \in \BZ_q^k$ denote $\vecx$'s ``restriction'' to the indices $\vecj$. An assignment $\vecx$ \emph{satisfies} $C=(b,\vecj,w)$ iff $f_b(\vecx|_{\vecj})=1$. The \emph{value} of an assignment $\vecx \in \BZ_q^n$ on an instance $\Psi$, denoted $\val_\Psi(\vecx)$, is the (fractional) weight of constraints satisfied by $\vecx$, i.e., \[ \val_\Psi(\vecx) \eqdef \frac1{W_\Psi} \sum_{\ell=1}^m w(\ell) f_{b(\ell)}(\vecx|_{\vecj(\ell)}) \] where $W_\Psi \eqdef \sum_{\ell=1}^m w(\ell)$ is the total weight in $\Psi$. Finally, the \emph{value} of $\Psi$, denoted $\val_\Psi$, is the maximum value of any assignment, i.e., \[ \val_\Psi \eqdef \max_{\vecx \in \BZ_q^n}\left( \val_\Psi(\vecx)\right). \] Computationally, the goal of the $\mF$ problem is to ``approximate'' $\val_\Psi$. More precisely, for $\alpha \in (0,1]$, we say a randomized algorithm $\Alg$ \emph{$\alpha$-approximates $\mF$} if $\alpha \, \val_\Psi \leq \Alg(\Psi) \leq \val_\Psi$ with probability at least, say, $\frac23$ over the choice of randomness. For $\beta < \gamma \in[0,1]$, we also consider the closely-related $\bgd\mF$ \emph{gap problem}, the goal of which is to distinguish between the cases $\val_\Psi \leq \beta$ and $\val_\Psi \geq \gamma$, again with probability at least $\frac23$ over the choice of randomness.\footnote{One direction of this ``close relationship'' is that if $\Alg$ $\alpha$-approximates $\mF$ and $\frac{\beta}{\gamma} < \alpha$, then $\Alg$ also solves the $\bgd\mF$ problem. For the other direction, see the proof sketch of \cref{cor:cgsv-bool-approx} below.}

\subsection*{Approximation resistance}

For a CSP $\mF$, define \[ \rho(\CF) \eqdef \inf_\Psi \val_\Psi. \] $\rho(\CF)$ has the following explicit formula:

\begin{proposition}[{\cite[Proposition 2.12]{CGSV21-finite}}]
    For every $\CF$, \[ \rho(\CF) = \min_{\CD \in \Delta(\CF)} \max_{\CD' \in \Delta(\BZ_q)} \left(\E_{f\sim\CD,\veca\sim(\CD')^k}[f(\veca)]\right). \]
\end{proposition}

In the prototypical case $|\CF|=1$, $\rho(\CF)$ captures the maximum value of any probabilistic assignment to $f$ which is \emph{symmetric} in the sense that every variable is assigned values from the same distribution independently.

By definition, $\mF$ has a $\rho(\CF)$-approximation given by simply outputting $\rho(\CF)$ on every input; we call this the \emph{trivial approximation}. We say $\mF$ is \emph{approximation-resistant} (for a certain class $\CS$ of algorithms) if for every $\epsilon > 0$, no algorithm in $\CS$ can $(\rho(F)+\epsilon)$-approximate $\mF$. Otherwise, we say $\mF$ is \emph{approximable} (for $\CS$).

\subsection*{CSPs of interest}

Specific CSPs which we study in this thesis include the following. In the case $k=q=2$, we let $\cut(x_1,x_2) \eqdef x_1+x_2= \1_{x_1\neq x_2}$, and we consider the problem $\mcut\eqdef\mcsp[\{\cut\}]$. Similarly, we let $\dcut(x_1,x_2) \eqdef x_1(x_2+1) = \1_{x_1=1,x_2=0}$, and we consider the problem $\mdcut\eqdef\mcsp[\{\dcut\}]$.

In the case $q=2$, for a predicate $f : \BZ_2^k \to \{0,1\}$, we define the problem $\mbcsp[f]\eqdef\mcsp[\{f_\vecb : \BZ_2^k \to \{0,1\}\}_{\vecb \in \BZ_2^k}]$ where $f_\vecb(\vecx)=f(\vecb+\vecx)$; the predicates of this CSP correspond to ``$f$ with negations''. For instance, for $k=2$  we let $\twoand(x_1,x_2) \eqdef x_1x_2 = \1_{x_1=x_2=1}$. Then $\mtwoand\eqdef\mbcsp[\twoand]$ contains the four predicates $\twoand_{b_1,b_2}(x_1,x_2)=(x_1+b_1)(x_2+b_2)$ for $b_1,b_2\in\BZ_2$. (Note that $\twoand_{0,0} = \twoand$ and $\twoand_{0,1} = \dcut$.) More generally we define $\kand(x_1,\ldots,x_k)=\prod_{i=1}^k x_i$ and consider $\mkand$.

The reader can check that the trivial approximation ratios $\rho(\CF)$ for $\mcut,\mdcut$, and $\mbcsp[f]$ are $\frac12,\frac14,$ and $\rho(f) \eqdef \E_{\veca\sim\CU_{\BZ_2^k}}[f(\veca)]$, respectively.

\section{Streaming and sketching algorithms}\label{sec:streaming}

\newcommand{\NextState}{\textsf{NextState}}
\newcommand{\Output}{\textsf{Output}}
\newcommand{\FinalState}{\textsf{FinalState}}
\newcommand{\Compose}{\textsf{Compose}}

For predicate families $\CF$, we consider algorithms which attempt to solve the approximation problem $\mF$ or the distinguishing problem $\bgd\mF$ in the \emph{$s(n)$-space streaming setting}, where $s(n)$ is typically small (e.g., $\polylog(n)$). First, we give an informal definition. On input $\Psi$ with $n$ variables, a streaming algorithm is limited to $s(n)$ space and can only access the constraints in $\Psi$ via a single pass through some ordering of $\Psi$'s constraints; this ordering can be chosen either \emph{adversarially} or (uniformly) \emph{randomly}. (When not specified, we assume the input is ordered adversarially.) On the other hand, the algorithm can use randomness and has no time or uniformity restrictions. We also consider a subclass of streaming algorithms called \emph{sketching} algorithms, which have the property that the algorithm can be run independently on two halves of the input stream and the resulting states can be composed. A sketching algorithm is \emph{linear} if the algorithm's state encodes an element of a vector space and composition corresponds to vector addition.

To be (somewhat) more formal, we define streaming and sketching algorithms as follows. Let $\Sigma$ denote the \emph{input space} of the stream (e.g., constraints of a $\mF$ instance on $n$ variables). A \emph{deterministic space-$s$ streaming algorithm} $\Alg$ is specified by a pair of functions $\NextState : \{0,1\}^s \times \Sigma \to \{0,1\}^s$ and $\Output : \{0,1\}^s \to \{0,1\}^*$. For an input stream $\vecsigma = (\sigma_1,\ldots,\sigma_{m}) \in \Sigma^*$, we define $\FinalState(\vecsigma) \in \{0,1\}^s$ as the result of initializing the state $S \gets 0^s$ and iterating $S \gets \NextState(S,\sigma_\ell)$ for $\ell \in [m]$; then $\Alg$ outputs $\Output(\FinalState(\vecsigma))$. Moreover, $\Alg$ is a \emph{sketching algorithm} if there exists another function $\Compose : \{0,1\}^s \times \{0,1\}^s \to \{0,1\}^s$ such that for every two input streams $\vecsigma,\vecsigma'\in\Sigma^*$, we have \[ \FinalState(\vecsigma\vecsigma') = \Compose(\FinalState(\vecsigma),\FinalState(\vecsigma')), \] where $\vecsigma\vecsigma'$ denotes concatenation. \emph{Randomized} streaming and sketching algorithms are distributions over streaming and sketching algorithms, respectively, which succeed with at least $\frac23$ probability.\footnote{Technical note: Since we allow repeated stream elements (and in particular, repeated constraints in $\mF$ instances), we have to pick some \emph{a priori} bound on stream lengths in order to get $\polylog(n)$-space algorithms. Throughout the paper, we assume instances contain at most $O(n^c)$ constraints for some (large) fixed $c < \infty$. Moreover, in order to store constraint weights in the algorithms' states, we assume that they are integers and are bounded by $O(n^c)$ in absolute value. We generally omit these details throughout the paper for ease of presentation.}

One particular sketching algorithm of interest is the following classic algorithm for sketching $1$-norms, which we use as a black box in later chapters:

\begin{theorem}[{\cite{Ind06,KNW10}}]\label{thm:l1-sketching}
For every $\epsilon>0$ and $c < \infty$, there exists an $O(\log n/\epsilon^2)$-space randomized sketching algorithm for the following problem: The input is an (adversarially ordered) stream $\vecsigma$ of updates from the set $\Sigma = [n] \times \{-O(n^c),\ldots,O(n^c)\}$, and the goal is to estimate the $1$-norm of the vector $\vecx \in \BN^n$ defined by $x_i = \sum_{(i,v) \in \vecsigma} v$, up to a multiplicative factor of $1\pm\epsilon$.
\end{theorem}

\section{Hypergraphs}\label{sec:hypergraphs}

Let $2 \leq k,n \in \BN$. A \emph{$k$-hyperedge} on $[n]$ is simply a $k$-tuple $\vece=(e_1,\ldots,e_k) \in [n]^k$ of distinct indices, and a \emph{$k$-hypergraph} (a.k.a. ``$k$-uniform hypergraph'') $G$ on $[n]$ is a sequence $(\vece(1),\ldots,\vece(m))$ of (not necessarily distinct) $k$-hyperedges. We assume $k=2$ when $k$ is omitted, and in this case, we drop the prefix ``hyper''. Given an instance $\Psi$ of a $k$-ary CSP $\mF$ with constraints $(\vecb(\ell),\vecj(\ell),w(\ell))_{\ell\in[m]}$, we can define the \emph{constraint (hyper)graph} $G(\Psi)$ of $\Psi$ as the $k$-hypergraph with hyperedges $(\vecj(\ell))_{\ell \in [m]}$. Note that when $|\CF|=1$ (as is the case for e.g., $\mcut$ and $\mdcut$) and we restrict our attention to unweighted instances, $\Psi$ and $G(\Psi)$ carry the exact same data.

To a $k$-hypergraph $G$ with $m$ edges, we associate an \emph{adjacency matrix} $M \in \{0,1\}^{km \times n}$, whose $(\ell k + j,i)$-th entry is $1$ iff $e(\ell)_j = i$ (for $\ell \in [m],j\in[k],i\in[n]$). Since they encode the same information, we will often treat adjacency matrices and $k$-hypergraphs as interchangeable. Importantly, we will often consider products of $M$ and $M^\top$ with vectors over $\BZ_q$. Given $\vecv \in \BZ_q^n$ and $\vecs = M\vecv$, let $\vecs(\ell) = (s_{(k-1)\ell+1},\ldots,s_{k\ell})$ denote the $\ell$-th block of $k$ coordinates in $\vecs$; then $\vecs(\ell) = \vecv|_{\vece(\ell)}$. Thus, we can view $\vecv$ as a ``$\BZ_q$-labeling of vertices'', and $\vecs$ as the corresponding ``$\BZ_q$-labeling of hyperedge-vertex incidences'', where each hyperedge $\vece=(e_1,\ldots,e_k)$ determines $k$ unique incidences $(\vece,e_1),\ldots,(\vece,e_k)$. Conversely, given $\vecs\in\BZ_q^m$, if $\vecv =M^\top \vecs$, for each $i \in[n]$, we have \[ v_i = \sum_{\ell\in[m], j\in[k]} \1_{\vece(\ell)_j=i}\, \vecs(\ell)_j. \] We again view $\vecs$ as labeling hyperedge-vertex incidences; $\vecv$ then describes the sums of $\vecs$-labels over the hyperedges incident at each vertex. Also, we will sometimes consider ``folded'' variants of the adjacency matrix $M$ which ``compress'' each block of $k$ columns (corresponding to a $k$-hyperedge) into fewer columns, e.g., by summing them into a single column, and these will have corresponding interpretations for $M\vecv$ and $M^\top \vecs$.

For $\alpha \in (0,1),n\in\BN$, let $\Graphs_{k,\alpha}(n)$ denote the uniform distribution over $k$-hypergraphs on $[n]$ with $\alpha n$ hyperedges. A $k$-hypergraph $G$ is a \emph{$k$-hypermatching} if no vertex is shared by any hyperedge, i.e., if $v \in e(\ell), e(\ell')$ then $\ell=\ell'$; equivalently, the adjacency matrix $M$ contains at most a single $1$ in each row. We refer to a $k$-hypermatching $G$ with $\alpha n$ edges as \emph{$\alpha$-partial}. We let $\Matchings_{k,\alpha}(n)$ denote the uniform distribution over $k$-hypermatchings on $[n]$ with $\alpha n$ hyperedges (for $\alpha \in (0,1),n\in\BN$).

\section{Fourier analysis}\label{sec:fourier}

Let $q \geq 2 \in \BN$, and let $\omega \eqdef e^{2\pi i/q}$ denote a (fixed primitive) $q$-th root of unity. Here, we summarize relevant aspects of Fourier analysis over $\BZ_q^n$; see e.g. \cite[\S8]{OD14} for details.\footnote{\cite{OD14} uses a different normalization for norms and inner products, essentially because it considers expectations instead of sums over inputs.} Given a function $f : \BZ_q^n \to \BC$ and $\vecs \in \BZ_q^n$, we define the \emph{Fourier coefficient} \[ \hat{f}(\vecs) \eqdef \sum_{\vecx \in \BZ_q^n} \omega^{-\vecs \cdot \vecx} f(\vecx) \] where $\cdot$ denotes the inner product over $\BZ_q$. For $p \in (0,\infty)$, we define $f$'s \emph{$p$-norm} \[ \|f\|_p \eqdef \left(\sum_{\vecx \in \BZ_q^n} |f(\vecx)|^p\right)^{1/p}. \] We also define $f$'s $0$-norm \[ \|f\|_0 \eqdef \sum_{\vecx \in \BZ_q^n} \1_{f(\vecx)\neq0} \] (a.k.a. the size of its support and the Hamming weight of its ``truth table''). Also, for $\ell \in \{0\} \cup [n]$, we define the \emph{level-$\ell$ Fourier ($2$-)weight} as \[ \W^{\ell}[f] \eqdef \sum_{\vecs\in\BZ_q^n : \|\vecs\|_0 = \ell} |\hat{f}(\vecs)|^2. \] These weights are closely connected to $f$'s $2$-norm:

\begin{proposition}[Parseval's identity]\label{prop:parseval}
For every $q,n \in \BN$ and $f : \BZ_q^n \to \BC$, we have \[ \|f\|_2^2 = q^n \sum_{\ell=0}^n \W^\ell[f]. \]
\end{proposition}

Moreover, let $\BD \eqdef \{w \in \BC : |w|\leq 1\}$ denote the (closed) unit disk in the complex plane. The following lemma bounding the low-level Fourier weights for functions mapping into $\BD$ is derived from hypercontractivity theorems in \cite{CGS+22}:

\begin{lemma}[{\cite[Lemma 2.11]{CGS+22}}]\label{lemma:low-fourier-bound}
There exists $\zeta > 0$ such that the following holds. Let $q \geq 2,n \in \BN$ and consider any function $f : \BZ_q^n \to \BD$. If for $c \in \BN$, $\|f\|_0 \geq q^{n-c}$, then for every $\ell \in \{1,\ldots,4c\}$, we have \[ \frac{q^{2n}}{\|f\|_0^2} \W^{\ell}[f] \leq \left(\frac{\zeta c}\ell\right)^\ell. \]
\end{lemma}

\section{Concentration inequalities}

We'll use the following concentration inequality for submartingales:

\begin{lemma}[{{\cite[Lemma 2.5]{KK19}}}]\label{lemma:azuma}
Let $X_1,\ldots,X_m$ be (not necessarily independent) $\{0,1\}$-valued random variables, such that for some $p \in (0,1)$, $\E[X_\ell\mid X_0,\ldots,X_{\ell-1}] \leq p$ for every $\ell \in [m]$. Let $\mu = pm$. Then for every $\eta > 0$, \[ \Pr\left[\sum_{\ell=1}^m X_\ell \geq \mu + \eta\right] \leq e^{-\eta^2/(2(\mu+\eta)))}. \] 
\end{lemma}

\section{Advantage and total variation distance}

Let $\CY,\CN \in \Delta(\Omega)$, and consider a \emph{test function} $\Test : \Omega \to \{0,1\}$ which attempts to distinguish between $\CY$ and $\CN$ by outputting $1$ more often on inputs sampled from $\CY$ than those sampled from $\CN$ (or vice versa). The \emph{advantage} of $f$ measures its success at this distinguishing task: \[ \adv_{\CY,\CN}(\Test) \eqdef \left\lvert\E\left[\Test(\CY)\right]-\E[\Test(\CN)]\right\rvert \in [0,1]. \] The \emph{total variation distance} between two distributions $\CY, \CN$ is the maximum advantage any test $f$ achieves in distinguishing $\CY$ and $\CN$. The optimal $f$ is the so-called ``maximum likelihood estimator'' which, on input $\omega\in\Omega$, outputs $\1_{\CY(\omega) \geq \CN(\omega)}$. Thus, \[ \|\CY-\CN\|_{\tv} \eqdef \max_{\Test : \Omega \to \{0,1\}} (\adv_{\CY,\CN}(\Test)) = \frac12\sum_{\omega \in \Omega} \left\lvert \CY(\omega)- \CN(\omega)\right\rvert. \] Also, for two random variables $Y$ and $N$, we use $\|Y-N\|_{\tv}$ as shorthand for $\|\CP_Y-\CP_N\|_{\tv}$ (recall that e.g. $\CP_Y$ denotes the distribution of $Y$).

The total variation distance satisfies two important inequalities for our purposes:

\begin{lemma}[Triangle inequality]\label{lemma:rv-triangle}
Let $\CY,\CN,\CZ \in \Delta(\Omega)$. Then \[ \|\CY-\CN\|_{\tv} \geq \|\CY-\CZ\|_{\tv} - \|\CZ-\CN\|_{\tv}. \]
\end{lemma}

\begin{lemma}[Data processing inequality]\label{lemma:data-processing}
Let $Y,N$ be random variables with sample space $\Omega$, and let $Z$ be a random variable with sample space $\Omega'$ which is independent of $Y$ and $N$. If $g:\Omega\times\Omega'\to\Omega''$ is any function, then \[ \|Y-N\|_{\tv} \geq \|g(Y,Z) - g(N,Z)\|_{\tv}. \]
\end{lemma}

Intuitively, \cref{lemma:data-processing} says that to distinguish the distributions of two random variables $Y$ and $N$, it is not helpful to perform any additional (possibly random) transformations first.

\section{Lower bound basics}\label{sec:lb-basics}

Finally, we consider several specific types of tests for distinguishing distributions over particular kinds of sets. These notions will be crucial for the proofs of lower bounds against streaming approximations for CSPs.

Firstly, let $\Sigma$ be a finite input space and consider the case $\Omega = \Sigma^*$. Given a pair of distributions $\CY,\CN \in \Delta(\Omega)$, we can view a deterministic streaming algorithm $\Alg$ as a test for distinguishing $\CY$ from $\CN$. This perspective lets us rule out algorithms for $\bgd\mF$ (and by extension $(\frac{\beta}{\gamma}+\epsilon)$-approximations to $\mF$) by constructing \emph{indistinguishable} $\CY$ and $\CN$ distributions:

\begin{proposition}[Minimax lemma \cite{Yao77}, informal statement for $\mF$]\label{prop:yao}
    Consider a CSP $\mF$, and let $\CS$ denote a ``class'' of randomized algorithms (e.g., $O(\sqrt n)$-space streaming algorithms with adversarial input ordering). Let $\beta < \gamma \in [0,1]$, and suppose that $\CY$ and $\CN$ are distributions over $\mF$ instances such that \[ \Pr_{\Psi \sim \CN}[\val_\Psi \geq \beta] \leq 0.01 \text{ and } \Pr_{\Psi \sim \CY}[\val_\Psi \leq \gamma] \leq 0.01. \] Then if there exists $\Alg \in \CS$ solving the $\bgd\mF$ problem, there is a \emph{deterministic} algorithm in $\CS$ distinguishing $\CY$ and $\CN$ with advantage at least $\frac16$.
\end{proposition}

Now consider the case of a product set $\Omega = \Omega_1 \times \cdots \times \Omega_T$. A set of functions $\Prot_t : \{0,1\}^s \times \Omega_t \to \{0,1\}^s$ for $t \in [T]$ defines a \emph{space-$s$ communication protocol} $\Prot : \Omega \to \{0,1\}$ in the following way. Given input $\vecomega = (\omega_1,\ldots,\omega_T) \in \Omega$, set $S \gets 0^s$ and iteratively apply $S \gets \Prot_t(S, \omega_t)$ for $t \in [T]$; finally, output $S$. (We assume that $\Prot_T$'s codomain is $\{0,1\}$.) $\Prot$ is a special type of test for distinguishing distributions $\CY,\CN \in \Delta(\Omega)$. We can also interpret such a protocol as a strategy in the following \emph{one-way communication game} (or \emph{problem}) with players $\Player_1,\ldots,\Player_T$:

\begin{itemize}
    \item We sample $\vecomega = (\omega_1,\ldots,\omega_T)$ either from $\CY$ (the \emph{$\yes$ case}) or $\CN$ (the \emph{$\no$ case}). $\Player_t$ receives the input $\omega_t$.
    \item $\Player_1$ sends a message, based on their input $\omega_1$, to $\Player_2$. For $t \in \{2,\ldots,T-1\}$, $\Player_t$ sends a message, based on $\Player_{t-1}$'s message and their own input $\omega_t$, to $\Player_{t+1}$. $\Player_T$ decides, based on $\Player_{T-1}$'s message and their own input $\omega_T$, whether to output $1$ or $0$.
    \item The players' collective goal is to maximize their advantage in distinguishing the $\yes$ and $\no$ cases.
\end{itemize}

This type of game can be used to model the flow of information during the execution of a streaming algorithm. The intuitive picture is that we can think of a streaming algorithm on a stream of length $m$ as a protocol for an $m$-player one-way communication game, where $\Player_t$ gets the $t$-th element of the stream, and each player transmits the state of the streaming algorithm onto the next player. To prove lower bounds for such a protocol, it suffices to prove lower bounds in the ``coarser'' game with only a constant number $T = O(1)$ of players, each of which gets a ``chunk'' of, say, $m/T$ stream elements. This corresponds to relaxing the definition of the streaming model to only require that the state is succinct in $T$ ``bottleneck'' locations along the stream; thus, to prove streaming lower bounds, we are proving the sufficient condition that at these bottlenecks, the algorithm's state cannot capture enough information about the elements it's already seen in the stream. Through this ``reduction'', lower bounds for a streaming problem can follow from lower bounds for an appropriately defined communication game. (See \cref{sec:bpd} for a more concrete description in the particular case of $\mcut$.)

We now make this \emph{communication-to-streaming (C2S)} reduction precise in a more convenient and general formulation where each player's input is not necessarily a chunk of constraints; rather, each player constructs constraints from their input according to some pre-defined ``reduction functions''. Suppose $\Sigma$ is a finite input space, while $\Omega = \Omega_1 \times \cdots \times \Omega_T$ is still a product space. Given a distribution $\CD \in \Delta(\Omega)$ and reduction functions $\R_t : \Omega_t \to \Sigma^*$ for $t \in [T]$, define $(R_1,\ldots,R_T) \circ \CD$ as the distribution over $\Sigma^*$ given by sampling $(\omega_1,\ldots,\omega_T) \sim \CD$ and outputting the concatenation $\R_1(\omega_1) \cdots \R_T(\omega_T)$. 

\begin{lemma}[Communication-to-streaming reduction]\label{lemma:comm-to-strm}
Let $\Omega = \Omega_1 \times \cdots \times \Omega_t$ and $\Sigma$ be finite sets. Let $\CY,\CN \in \Delta(\Omega)$ and $\R_t : \Omega_t \to \Sigma^*$ for $t \in [T]$. If there exists a deterministic space-$s$ streaming algorithm $\Alg$ for distinguishing $(\R_1,\ldots,\R_T) \circ \CY$ from $(\R_1,\ldots,\R_T) \circ \CN$ with advantage $\delta$, then there exists a space-$s$ communication protocol $\Prot$ for distinguishing $\CY$ from $\CN$ with advantage $\delta$.
\end{lemma}

\begin{proof}
Let $\Alg$ be given by $(\NextState,\Output)$. Consider the protocol $\Prot$ in which $\Player_1$, on input $\omega_1$, sets $S \gets \NextState(\R_1(\omega_1),0^s)$ and sends $S$ to $\Player_2$. Now for $t \in \{2,\ldots,T-1\}$, $\Player_t$ receives $S$ from $\Player_{t-1}$, sets $S \gets \NextState(\R_t(\omega_t),S)$, and sends $S$ to $\Player_{t+1}$. Finally, $\Player_T$ outputs $\Output(\NextState(\R_T(\omega_T),S))$. By definition, when the players receive input $(\omega_1,\ldots,\omega_T)$, they are running $\Alg$ on the stream $\vecsigma = \R_1(\omega_1) \cdots \R_T(\omega_T)$. If the players' input comes from $\CY$, then $\vecsigma$ is distributed as $(\R_1,\ldots,\R_T) \circ \CY$, and similarly for $\CN$.
\end{proof}

In our setting, where the reduction produces CSP instances, we typically think of each reduction function as outputting ``subinstances'' for each player, whose union is the output instance $\Psi$.

The final special case of advantage we consider is distinguishing $\CD \in \Delta(\BZ_q^n)$ from the uniform distribution $\CU_{\BZ_q^n}$. Recalling that we view $\CD \in \Delta(\BZ_q^n)$ as a function $\BZ_q^n \to [0,1]$, we can consider the Fourier coefficients $\hat{\CD}(\vecs)$ for $\vecs \in \BZ_q^n$. The following simple but crucial lemma relates the distance-to-uniformity of $\CD$ with these coefficients:

\begin{lemma}\label{lemma:xor}
Let $\CD \in \Delta(\BZ_q^n)$ and let $\CU = \Unif_{\BZ_q^n}$. Then \[ \|\CD-\CU\|_\tv^2 \leq q^{2n} \sum_{\vecs\neq\veczero \in \BZ_q^n} |\hat{\CD}(\vecs)|^2. \]
\end{lemma}

\begin{proof}
We have $\hat{\CD}(\veczero) = \sum_{\vecz\in\BZ_q^{\alpha n}} \CD(\vecz) = 1$. Similarly, $\hat{\CU}(\veczero) = 1$, while for $\vecs \neq \veczero \in \BZ_q^{\alpha n}$, we have $\hat{\CU}(\vecs) = \frac1{q^{\alpha n}} \sum_{\vecz\in\BZ_2^{\alpha n}} (-1)^{\vecs \cdot \vecz} = 0$ by symmetry. Also by definition, $\|\CD-\CU\|_{\tv} = \frac12 \|\CD-\CU\|_1$, where $\CD-\CU : \BZ_q^n \to [-1,1]$ is the difference of the probability mass functions of $\CD$ and $\CU$.

Thus using Cauchy-Schwarz and Parseval's identify (\cref{prop:parseval}), we have \[ \|\CD-\CU\|_{\tv}^2 \leq q^n \|\CD-\CU\|_2^2 = q^{2n} \sum_{\vecs\in\BZ_q^n} |\hat{\CD}(\vecs)-\hat{\CU}(\vecs)|^2 = q^{2n} \sum_{\vecs\neq\veczero\in\BZ_q^n} |\hat{\CD}(\vecs)|^2, \] as desired.
\end{proof}

\cref{lemma:xor} is an example of a so-called ``\textsc{xor} lemma'' (see \cite[\S1]{Gol11}). In the $q=2$ case, for each $\vecs \in \BZ_q^n$, $\hat{\CD}(\vecs)$ is the advantage of the \emph{linear test} on $\CD$ which, given a sample $\vecz \sim \CD$, outputs $\sum_{i \in [n]:s_i = 1} z_i$. The lemma roughly says that if none of these tests work well, then $\CD$ is in fact close to uniform.

Together, \cref{prop:yao,lemma:comm-to-strm,lemma:xor} give us a ``roadmap'' for proving CSP streaming inapproximability results. Namely, we design a one-way communication game with the following two properties:

\begin{enumerate}
    \item The players can use reduction functions (\`a la \cref{lemma:comm-to-strm}) to produce CSP instances from their inputs with the property that there is a large gap between the instances' values in the $\yes$ and $\no$ cases (with high probability).
    \item The game's hardness itself can be proven using \cref{lemma:xor} and additional Fourier analysis and combinatorics.
\end{enumerate}

In the CSP context, this was first introduced by Kapralov, Khanna, and Sudan~\cite{KKS15} for $\mcut$. We turn to this proof in the next chapter.

\part{Prior results}\label{part:prior-results}

\newcommand{\GW}{\mathrm{GW}}

\chapter{$\mcut$ is approximation-resistant}\label{chap:mcut}

\epigraph{The problem is defined as follows: given a stream of edges of an $n$-node graph $G$, estimate the value of the maximum cut in $G$. \emph{Question:} Is there an algorithm with an approximation factor strictly better than 1/2 that uses $o(n)$ space?}{\cite[Question 10]{IMNO11}, attributed to Robert Krauthgamer}

\newthought{$\mcut$ was the first CSP} whose streaming approximability was tightly characterized. To be precise, Kapralov, Khanna, and Sudan~\cite{KKS15} proved the following theorem:

\begin{theorem}[{\cite{KKS15}}]\label{thm:mcut-hardness}
For every constant $\epsilon > 0$, any streaming algorithm which $(\frac12+\epsilon)$-approximates $\mcut$ requires $\Omega(\sqrt n)$ space.
\end{theorem}

This chapter is devoted to proving \cref{thm:mcut-hardness}. We remark also that in the classical setting, Goemans and Williamson~\cite{GW95} gave an algorithm based on SDP rounding which $\alpha_{\GW}$-approximates $\mcut$, where $\alpha_{\GW} = \min_{\theta \in [0,\Prot]} \frac{2\theta}{\Prot(1-\cos(\theta)} \approx 0.87856$;\footnote{Khot, Kindler, Mossel, and O'Donnell~\cite{KKMO07} showed that $(\alpha_{\GW}+\epsilon)$-approximations are UG-hard. Without the UGC, Trevisan \emph{et al.}~\cite{TSSW00} show that $(\frac{16}{17}+\epsilon)$-approximation is $\NP$-hard, but $\frac{16}{17}\approx 0.94118$.} thus, \cref{thm:mcut-hardness} shows that $\mcut$ is comparatively much \emph{less} approximable in the streaming setting relative to the classical setting.

Now, we begin with some intuition for why $\mcut$ should be hard to approximate with a small-space streaming algorithm. Consider a streaming algorithm solving $\mcut$ on an input instance $\Psi$. Suppose that we pause it halfway through the input stream, and at this point, the algorithm is fairly confident that $\val_\Psi$ is large and has a ``guess'' $\vecx \in \BZ_2^n$ for an assignment with high value. Then during the second half of the stream, the algorithm should be able to confirm that the constraints it sees are also (mostly) consistent with $\vecx$.

In order to prove streaming approximation-resistance for $\mcut$, we begin in \cref{sec:bpd} by defining a one-way communication problem which formalizes this difficulty, which we'll call \emph{Boolean partition detection} ($\bpd$),\footnote{The typical name in the literature is the \emph{Boolean hidden matching problem} (see e.g., \cite{KKS15}). In this thesis, however, we have to accommodate a variety of communication problems and so have chosen to adopt a more consistent naming scheme.} and we give a roadmap for how $\bpd$'s hardness implies $\mcut$'s hardness via the intermediate ``\emph{sequential Boolean partition detection} problem ($\seqbpd$)''. Next, in \cref{sec:bpd-hardness}, we describe the Fourier-analytic proof, originally due to Gavinsky \emph{et al.}~\cite{GKK+08}, that $\bpd$ is hard, and in \cref{sec:bpd-to-seqbpd}, we show how $\seqbpd$ reduces to $\bpd$ via the \emph{hybrid argument} of Kapralov, Khanna, and Sudan~\cite{KKS15}. Finally, in \cref{sec:mcut-discussion}, we make several comments on important features of the $\mcut$ lower bound which will remain important for the other CSPs considered in this thesis.

\section{Boolean partition detection problems}\label{sec:bpd}

Let $M\in \{0,1\}^{2\alpha n \times n}$ be an adjacency matrix for a graph on $n$ vertices and $\alpha n$ edges. Recall that in $M$, each edge corresponds to a $2 \times n$ block. We define a \emph{folded} variant of $M$, denoted $M^{\fold} \in \{0,1\}^{\alpha n \times n}$ by replacing each $2 \times n$ edge-block with the sum of its columns; thus, each column of $M^{\fold}$ corresponds to a single edge, and has $1$'s indicating the two vertices incident to that edge. Then the $\bpd$ problem is defined as follows.

\begin{definition}[$\bpd$]\label{def:bpd}
Let $\alpha \in (0,1)$ and $n \in \BN$. Then $\bpd_\alpha(n)$ is the following two-player one-way communication problem, with players $\Alice$ and $\Bob$:

\begin{itemize}[nosep]
    \item $\Alice$ receives a random vector $\vecx^* \sim \Unif_{\BZ_2^n}$.
    \item $\Bob$ receives an adjacency matrix $M \in \{0,1\}^{2\alpha n \times n}$ sampled from $\Matchings_\alpha(n)$, and a vector $\vecz \in \BZ_2^{\alpha n}$ labelling each edge of $M$ defined as follows:
    \begin{itemize}[nosep]
        \item $\yes$ case: $\vecz = (M^{\fold}) \vecx^*$.
        \item $\no$ case: $\vecz \sim \Unif_{\BZ_2^{\alpha n}}$.
    \end{itemize}
    \item $\Alice$ can send a message to $\Bob$, who must then decide whether they are in the $\yes$ or $\no$ case.
\end{itemize}
\end{definition}

We can view $\Alice$'s vector $\vecx^*$ as a partition of $M$'s vertices. In the $\yes$ case, $\Bob$'s vector $\vecz$ can be interpreted as follows: If $\vece(\ell) = (u,v)$ is the $\ell$-th edge of $M$, then $z_\ell = x^*_u + x^*_v$. Thus, $\vecz$ precisely encodes which edges in $M$ cross the partition $\vecx^*$. On the other hand, in the $\no$ case, $\vecz$ is uniformly random. Thus, $\Bob$'s goal is to decide whether his input $\vecz$ is consistent with partition $\vecx^*$ based on $\Alice$'s message.

In \cref{sec:bpd-hardness} below, we will prove that this task requires significant communication from $\Alice$ to $\Bob$:

\begin{theorem}[{\cite{GKK+08}}]\label{thm:bpd-hardness}
For every $\alpha, \delta \in (0,1)$, there exists $\tau > 0$ and $n_0 \in \BN$ such that for all $n \geq n_0$, any protocol for $\bpd_\alpha(n)$ achieving advantage at least $\delta$ requires $\tau \sqrt n$ communication.
\end{theorem}

While \cref{thm:bpd-hardness} captures the essential obstacle to computing $\mcut$ in the streaming setting, it is not alone sufficient to prove inapproximability. For this purpose, we want $\Alice$ and $\Bob$ to produce $\mcut$ instances using a streaming-to-communication reduction (see \cref{lemma:comm-to-strm}) which have a high value gap between the $\yes$ and $\no$ cases. Indeed, to rule out $\approx \frac12$-approximations, the $\yes$ instances should have value $\approx 1$ while the $\no$ instances should have value $\approx \frac12$. \emph{A priori}, we might hope to produce such instances via a direct reduction from $\bpd$ to $\mcut$. In the $\yes$ case of $\bpd$, suppose that for each edge which crosses the cut (i.e., those for which $z_\ell = 1$), $\Bob$ creates a corresponding $\mcut$ constraint; encouragingly, the resulting instance has value $1$! But unfortunately, the same is true in the $\no$ case, because every $\vecz$ is consistent with \emph{some} partition $\vecx'$ of $M$. For instance, for each $\vece(\ell) = (u,v)$, we could set $x'_u = 0$ and $x'_v = z_\ell$ (and assign all remaining $x$-values arbitrarily); in particular, since $M$ is a matching, none of these assignments will interfere with each other.\footnote{In other words, the graph corresponding to the $\mcut$ instance $\Bob$ creates will always be a matching, and matchings are always bipartite; thus $\Bob$'s instances always have value $1$.}

The issue, in brief, is that the underlying graph in $\bpd$ is too sparse to be of use in constructing low-value $\mcut$ instances. To remedy this, we introduce a \emph{sequential} variant of $\bpd$ which can give rise to much denser graphs:

\begin{definition}[$\seqbpd$]\label{def:seqbpd}
Let $\alpha \in (0,1)$ and $T, n \in \BN$. Then $\seqbpd_{\alpha,T}(n)$ is the following $(T+1)$-player one-way communication problem, with players $\Alice$ and $\Bob_1,\ldots,\Bob_T$:

\begin{itemize}
    \item $\Alice$ receives a random vector $\vecx^* \sim \Unif_{\BZ_2^n}$.
    \item Each $\Bob_t$ receives an adjacency matrix $M_t\in\{0,1\}^{2\alpha n \times n}$ sampled from $\Matchings_\alpha(n)$, and a vector $\vecz(t) \in \BZ_2^{\alpha n}$ labelling each edge of $M_t$ as follows:
    \begin{itemize}[nosep]
        \item $\yes$ case: $\vecz(t) = (M_t^{\fold}) \vecx^*$.
        \item $\no$ case: $\vecz(t) \sim \Unif_{\BZ_2^{\alpha n}}$.
    \end{itemize}
    \item $\Alice$ can send a message to $\Bob_1$; each $\Bob_t$ can send a message to $\Bob_{t+1}$; and at the end, $\Bob_T$ must decide whether they are in the $\yes$ or $\no$ case.
\end{itemize}
\end{definition}

$\seqbpd$ is a ``happy medium'' which allows us to effect reductions both \emph{from} $\bpd$ and \emph{to} $\mcut$. Indeed, we have:

\begin{lemma}\label{lemma:bpd-to-seqbpd}
Let $\alpha,\delta \in (0,1)$ and $T,n,s \in \BN$. Suppose there is a protocol for $\seqbpd_{\alpha,T}(n)$ achieving advantage $\delta$ using $s$ communication. Then there is a protocol for $\bpd_{\alpha}(n)$ achieving advantage at least $\frac{\delta}{T}$ using $s$ communication.
\end{lemma}

We prove \cref{lemma:bpd-to-seqbpd} in \cref{sec:bpd-to-seqbpd} below using the \emph{hybrid argument} of Kapralov, Khanna, and Sudan~\cite{KKS15}. We also have:

\begin{construction}[C2S reduction from $\seqbpd$ to $\mcut$]\label{cons:seqbpd-to-mcut}
$\Alice$'s reduction function, denoted $\R_0$, outputs no constraints. For each $t \in [T]$, $\Bob_t$'s reduction function $\R_t$ outputs an instance $\Psi_t$ as follows: For each $\vece(t,\ell) = (u,v)$ in $M_t$, $\Bob_t$ adds $\vece(t,\ell)$ to $\Psi_t$ iff $z(t)_\ell = 1$.
\end{construction}

The hard instances for $\mcut$ produced by \cref{cons:seqbpd-to-mcut} are represented pictorially in \cref{fig:mcut}.

\begin{figure}
\centering
\begin{subfigure}{0.4\textwidth}
\centering
\begin{tikzpicture}[vertex/.style={fill=black},block/.style={draw=black,fill=white!70!lightgray}, goodedge/.style={line width=1.5pt,draw=black!40!green}, badedge/.style={line width=1.5pt,draw=black!10!red}]

% block 0
\draw[block] (0,3) ellipse (0.75 and 3.5);
\draw[vertex] (0,0) circle (3pt);
\draw[vertex] (0,1) circle (3pt);
\draw[vertex] (0,2) circle (3pt);
\draw[vertex] (0,3) circle (3pt);
\draw[vertex] (0,4) circle (3pt);
\draw[vertex] (0,5) circle (3pt);
\draw[vertex] (0,6) circle (3pt);

% block 1
\draw[block] (3,3) ellipse (0.75 and 3.5);
\draw[vertex] (3,0) circle (3pt);
\draw[vertex] (3,1) circle (3pt);
\draw[vertex] (3,2) circle (3pt);
\draw[vertex] (3,3) circle (3pt);
\draw[vertex] (3,4) circle (3pt);
\draw[vertex] (3,5) circle (3pt);
\draw[vertex] (3,6) circle (3pt);

% crossing edges
\draw[goodedge] (0,0) to (3,4);
\draw[goodedge] (0,3) to (3,6);
\draw[goodedge] (0,5) to (3,6);
\draw[goodedge] (0,1) to (3,5);
\draw[goodedge] (0,4) to (3,4);
\draw[goodedge] (0,2) to (3,1);
\draw[goodedge] (0,5) to (3,1);

% bad edges
\draw[badedge] (0,0) to[bend left] (0,2);
\draw[badedge] (0,2) to[bend left] (0,6);
\draw[badedge] (0,1) to[bend left] (0,4);
\draw[badedge] (0,3) to[bend left] (0,5);
\draw[badedge] (3,2) to[bend right] (3,5);
\draw[badedge] (3,1) to[bend right] (3,6);
\draw[badedge] (3,0) to[bend right] (3,3);
\end{tikzpicture}
\caption{$\yes$ sample from $\seqbpd$.}
\label{fig:seqbpd-yes}
\end{subfigure}
\begin{subfigure}{0.4\textwidth}
\centering
\begin{tikzpicture}[vertex/.style={fill=black},block/.style={draw=black,fill=white!70!lightgray}, goodedge/.style={line width=1.5pt,draw=black!40!green}, badedge/.style={line width=1.5pt,draw=black!10!red}]

\draw[block] (0,3) ellipse (0.75 and 3.5);
\draw[vertex] (0,0) circle (3pt);
\draw[vertex] (0,1) circle (3pt);
\draw[vertex] (0,2) circle (3pt);
\draw[vertex] (0,3) circle (3pt);
\draw[vertex] (0,4) circle (3pt);
\draw[vertex] (0,5) circle (3pt);
\draw[vertex] (0,6) circle (3pt);

\draw[block] (3,3) ellipse (0.75 and 3.5);
\draw[vertex] (3,0) circle (3pt);
\draw[vertex] (3,1) circle (3pt);
\draw[vertex] (3,2) circle (3pt);
\draw[vertex] (3,3) circle (3pt);
\draw[vertex] (3,4) circle (3pt);
\draw[vertex] (3,5) circle (3pt);
\draw[vertex] (3,6) circle (3pt);

% crossing edges
\draw[goodedge] (0,0) to (3,4);
\draw[badedge] (0,3) to (3,6);
\draw[badedge] (0,5) to (3,6);
\draw[badedge] (0,1) to (3,5);
\draw[goodedge] (0,4) to (3,4);
\draw[badedge] (0,2) to (3,1);
\draw[badedge] (0,5) to (3,1);

% bad edges
\draw[goodedge] (0,0) to[bend left] (0,2);
\draw[goodedge] (0,2) to[bend left] (0,6);
\draw[badedge] (0,1) to[bend left] (0,4);
\draw[goodedge] (0,3) to[bend left] (0,5);
\draw[goodedge] (3,2) to[bend right] (3,5);
\draw[goodedge] (3,1) to[bend right] (3,6);
\draw[badedge] (3,0) to[bend right] (3,3);
\end{tikzpicture}
\caption{$\no$ sample from $\seqbpd$.}
\label{fig:seqbpd-no}
\end{subfigure}

\begin{subfigure}{0.4\textwidth}
\centering
\begin{tikzpicture}[vertex/.style={fill=black},block/.style={draw=black,fill=white!70!lightgray}, goodedge/.style={line width=1.5pt,draw=black!40!green}, badedge/.style={line width=1.5pt,draw=black!10!red}]

% block 0
\draw[block] (0,3) ellipse (0.75 and 3.5);
\draw[vertex] (0,0) circle (3pt);
\draw[vertex] (0,1) circle (3pt);
\draw[vertex] (0,2) circle (3pt);
\draw[vertex] (0,3) circle (3pt);
\draw[vertex] (0,4) circle (3pt);
\draw[vertex] (0,5) circle (3pt);
\draw[vertex] (0,6) circle (3pt);

% block 1
\draw[block] (3,3) ellipse (0.75 and 3.5);
\draw[vertex] (3,0) circle (3pt);
\draw[vertex] (3,1) circle (3pt);
\draw[vertex] (3,2) circle (3pt);
\draw[vertex] (3,3) circle (3pt);
\draw[vertex] (3,4) circle (3pt);
\draw[vertex] (3,5) circle (3pt);
\draw[vertex] (3,6) circle (3pt);

% crossing edges
\draw[goodedge] (0,0) to (3,4);
\draw[goodedge] (0,3) to (3,6);
\draw[goodedge] (0,5) to (3,6);
\draw[goodedge] (0,1) to (3,5);
\draw[goodedge] (0,4) to (3,4);
\draw[goodedge] (0,2) to (3,1);
\draw[goodedge] (0,5) to (3,1);
\end{tikzpicture}
\caption{$\yes$ instance of $\mcut$.}
\label{fig:mcut-yes}
\end{subfigure}
\begin{subfigure}{0.4\textwidth}
\centering
\begin{tikzpicture}[vertex/.style={fill=black},block/.style={draw=black,fill=white!70!lightgray}, goodedge/.style={line width=1.5pt,draw=black!40!green}, badedge/.style={line width=1.5pt,draw=black!10!red}]

\draw[block] (0,3) ellipse (0.75 and 3.5);
\draw[vertex] (0,0) circle (3pt);
\draw[vertex] (0,1) circle (3pt);
\draw[vertex] (0,2) circle (3pt);
\draw[vertex] (0,3) circle (3pt);
\draw[vertex] (0,4) circle (3pt);
\draw[vertex] (0,5) circle (3pt);
\draw[vertex] (0,6) circle (3pt);

\draw[block] (3,3) ellipse (0.75 and 3.5);
\draw[vertex] (3,0) circle (3pt);
\draw[vertex] (3,1) circle (3pt);
\draw[vertex] (3,2) circle (3pt);
\draw[vertex] (3,3) circle (3pt);
\draw[vertex] (3,4) circle (3pt);
\draw[vertex] (3,5) circle (3pt);
\draw[vertex] (3,6) circle (3pt);

\draw[goodedge] (0,0) to (3,4);
\draw[goodedge] (0,4) to (3,4);
\draw[goodedge] (0,0) to[bend left] (0,2);
\draw[goodedge] (0,2) to[bend left] (0,6);
\draw[goodedge] (0,3) to[bend left] (0,5);
\draw[goodedge] (3,2) to[bend right] (3,5);
\draw[goodedge] (3,1) to[bend right] (3,6);
\end{tikzpicture}
\caption{$\no$ instance of $\mcut$.}
\label{fig:mcut-no}
\end{subfigure}

\caption[]{Example hard instances for $\mcut$. \cref{fig:seqbpd-yes,fig:seqbpd-no} depict samples from the $\yes$ and $\no$ distributions of $\seqbpd$, respectively. Recall, in $\seqbpd$, $\Alice$ receives a hidden partition $\vecx^* \in \BZ_2^n$ and each $\Bob_t$ receive a matching $M_t$ along with a vector $\vecz(t)$ annotating $M_t$'s edges. In the $\yes$ case, $\vecz(t)$ marks the edges of $M_t$ which cross the partition $\vecx^*$, while in the $\no$ case, $\vecz(t)$ is uniformly random. The graphs in \cref{fig:seqbpd-yes,fig:seqbpd-no} represent the union of the matchings $M_1,\ldots,M_t$; $\vecx^*$ partitions the vertices into ``left'' ($0$) and ``right'' ($1$); and the edges' $z$-values are either ``green'' ($1$) or ``red'' ($0$). In our reduction from $\seqbpd$ to $\mcut$ (\cref{cons:seqbpd-to-mcut}), $\Alice$ adds no edges, and each $\Bob_t$ adds all edges with $z$-value ``green'' ($1$). In the $\yes$ case, the resulting graph is bipartite (\cref{fig:mcut-yes}) and thus has $\mcut$ value $1$, while in the $\no$ case, the graph is random (\cref{fig:mcut-no}) and has value $\approx \frac12$ with high probability (for sufficiently large $T$).
}
\label{fig:mcut}
\end{figure}

\begin{lemma}\label{lemma:seqbpd-to-mcut-analysis}
For all $\alpha \in (0,1)$ and $\epsilon \in (0,\frac12)$, there exist $T, n_0 \in \BN$ such that for every $n \geq n_0$, the following holds. Let $\CY$ and $\CN$ denote the $\yes$ and $\no$ distributions for $\seqbpd_{\alpha,T}(n)$, and let $(\R_0,\ldots,\R_T)$ be the reduction functions from \cref{cons:seqbpd-to-mcut}. Then \[ \Pr_{\Psi \sim (\R_0,\ldots,\R_T) \circ \CY}\left[\val_\Psi =1\right]=1 \text{ and } \Pr_{\Psi \sim (\R_0,\ldots,\R_T) \circ \CN}\left[\val_\Psi \geq \frac12+\epsilon\right]\leq \exp(-n). \]
\end{lemma}

Note that \cref{lemma:seqbpd-to-mcut-analysis} may force us to make $T$ very large; yet it is constant, in which case \cref{lemma:bpd-to-seqbpd} gives a small-but-constant advantage for $\bpd$, and fortunately, \cref{thm:bpd-hardness} rules out \emph{every} constant advantage for $\bpd$.

To conclude this section, we give proofs for \cref{thm:mcut-hardness} and \cref{lemma:seqbpd-to-mcut-analysis}.

\begin{proof}[Proof of \cref{thm:mcut-hardness}]
Consider any $\epsilon > 0$, let $\tau > 0$ be determined later, and let $\Alg$ be a randomized space-$s(n)$ streaming algorithm which $(\frac12+\epsilon)$-approximates $\mcut$. By \cref{lemma:seqbpd-to-mcut-analysis}, we can pick sufficiently large $T,n_0 \in \BN$ such that if we fix any $n \geq n_0$, we have \[ \Pr_{\Psi \sim (\R_0,\ldots,\R_T) \circ \CY}\left[\val_\Psi =1\right] =1 \text{ and } \Pr_{\Psi \sim (\R_0,\ldots,\R_T) \circ \CN}\left[\val_\Psi \geq \frac12+\frac{\epsilon}2\right]\leq \exp(-n) \] where $\CY,\CN$ are the $\yes$ and $\no$ distributions for $\seqbpd_{\alpha,T}(n)$ and $(\R_0,\ldots,\R_T)$ are as in \cref{cons:seqbpd-to-mcut}. Since $\Alg$ solves the $(1,1-\epsilon/2)\text{-}\mcut$ problem, by \cref{prop:yao}, there is a \emph{deterministic} space-$s(n)$ streaming algorithm which distinguishes $(\R_0,\ldots,\R_T) \circ \CY$ and $(\R_0,\ldots,\R_T) \circ \CN$ with advantage at least $\frac16$. By \cref{lemma:comm-to-strm}, there is a deterministic space-$s(n)$ communication protocol for $\seqbpd_{\alpha,T}(n)$ with advantage at least $\frac16$. By \cref{lemma:bpd-to-seqbpd}, there is a deterministic space-$s(n)$ communication protocol for $\bpd_\alpha(n)$ with advantage at least $\frac1{6T}$. Finally, by \cref{thm:bpd-hardness}, there is some $\tau > 0$ and $n_0' \in \BN$ such that further assuming $n \geq n_0'$, we can conclude $s(n) \geq \tau \sqrt{n}$, as desired.
\end{proof}

\begin{proof}[Proof sketch of \cref{lemma:seqbpd-to-mcut-analysis}]
Let $\Psi = \Psi_1\cup \cdots \cup \Psi_T$ be the instance created by the reduction. In the $\yes$ case, regardless of $T$, we always have $\val_\Psi(\vecx^*) = 1$, since every constraint $(u,v)$ in $\Psi$ is chosen such that $x^*_u + x^*_v = 1$.

For the $\no$ case, it is sufficient to show that for every fixed assignment $\vecx \in \BZ_2^n$,
\begin{equation}\label{eq:mcut-no-ub}
    \Pr\left[\val_\Psi(\vecx) \geq \frac12+\epsilon\right] \leq \exp(-\epsilon^2\alpha T n),
\end{equation}
since then we can take a union bound over $\vecx$ and set $T$ sufficiently large. In the ``nicer'' model where $\Psi$ has $\alpha T n$ constraints chosen uniformly at random, \cref{eq:mcut-no-ub} would follow immediately from the Chernoff bound, since $\vecx$ would satisfy each of the $\alpha T n$ constraints independently w.p. $\frac12$. Unfortunately, there are two issues:

\begin{enumerate}
    \item Since $\vecz(t)$ is uniformly random, $\Bob_t$ adds each edge in $M_t$ as a constraint in $\Psi_t$ only w.p. $\frac12$ (independently). Thus, the number of constraints in each sub-instance $\Psi_t$ is distributed binomially. In particular, \emph{the number of constraints in $\Psi$ is not constant.} 
    \item Each $M_t$ is a random \emph{matching}, so its edges are not independent. Thus, \emph{$\Psi$'s constraints are not independent}, although $\Psi_t$ and $\Psi_{t'}$ have independent constraints if $t \neq t'$.
\end{enumerate}

Issue (1) can be addressed by treating the number of constraints in each $\Psi_t$ as a random variable and conditioning. To be precise, we define $\beta_t \eqdef \frac{m(\Psi_t)}{n}$ for each $t \in [T]$ and $\beta \eqdef \frac{m(\Psi)}{Tn} = \frac1T\sum_{t=1}^T \beta_t$, we have $\E[\beta_1] = \cdots = \E[\beta_T] = \E[\beta] = \frac{\alpha}2$ and condition on fixed values $\beta_1,\ldots,\beta_T$. We can then treat the constraints of each $\Psi_t$ as the edges of a random matching drawn from $\Matchings_{\beta_t}(n)$.

Now, suppose we define random variables $\{X_{t,\ell}\}_{t \in [T],\ell\in[\beta_t n]}$, each of which is the indicator for the event that $\vecx$ satisfies the $\ell$-th constraint of $\Psi_t$. We have $\val_\Psi(\vecx) = \frac1{\beta T n} \sum_{t=1}^T \sum_{\ell=1}^{\beta_t n} X_{t,\ell}$. Because of Issue (2), we can't use the Chernoff bound on $\val_\Psi(\vecx)$, but we can use \cref{lemma:azuma}. For $t \neq t'$, $X_{t,\ell}$ and $X_{t',\ell'}$ will be independent, and even though $X_{t,\ell}$ is not independent of $X_{t,1},\ldots,X_{t,\ell-1}$, we have $\E[X_{t,\ell} \mid X_{t,1},\ldots,X_{t,\ell-1}] \leq \frac12$. Indeed, the $\ell$-th constraint in $\Psi_t$ is sampled uniformly from the set of constraints which do not share variables with the first $\ell-1$ constraints, and at most half of these are satisfied by $\vecx$.

(There's one other small issue: The probability bound \cref{lemma:azuma} gives us will be exponentially small in $\beta T n$, not $\alpha T n$. But by the Chernoff bound, we can assume WLOG, say, $\beta \geq \frac{\alpha}4$. This contributes an additional union bound term which is exponentially small in $\alpha n$.)
\end{proof}

\newcommand{\GOOD}{\textsf{GOOD}}

\section{$\bpd$ is hard: Proving \cref{thm:bpd-hardness}}\label{sec:bpd-hardness}

The goal of this section is to prove \cref{thm:bpd-hardness}, due to Gavinsky, Kempe, Kerenidis, Raz, and de Wolf~\cite{GKK+08}, which states that $\bpd$ requires significant communication.

Let $\CU \eqdef \Unif_{\BZ_2^{\alpha n}}$. To begin, suppose that $\Alice$, using a deterministic protocol, sends some fixed message $a \in \{0,1\}^s$ to $\Bob$, and let $A \subseteq \BZ_2^n$ be the set of $\vecx^*$'s consistent with this message. For each matching $M \in \{0,1\}^{2\alpha n \times n}$, we consider the conditional distribution of $\Bob$'s second input $\vecz \in \BZ_2^{\alpha n}$ in the $\yes$ case: \[ \CZ_{A,M}(\vecz) \eqdef \Pr_{\vecx^* \sim \Unif_A}[\vecz = (M^{\fold})\vecx^*]. \] We prove \cref{thm:bpd-hardness} by showing that if $A$ is sufficiently large (which will be the case w.h.p. when the communication $s$ is sufficiently small), then w.h.p. over $M$, the distribution $\CZ_{A,M}$ is statistically close to $\CU$, and so $\Bob$ cannot distinguish the $\yes$ and $\no$ cases. To achieve this, we rely crucially on the following ``reduction'':

\begin{lemma}[Fourier-analytic reduction]\label{lemma:bpd-fourier-reduce}
Let $A \subseteq \BZ_2^n$ and $\1_A : \BZ_2^n \to \{0,1\}$ be the indicator for $A$, and let $\alpha \in (0,1)$. Then \[ \E_{M\sim\Matchings_\alpha(n)}[\|\CZ_{A,M}-\CU\|_{\tv}^2] \leq \frac{2^{2n}}{|A|^2} \sum_{\ell \geq 2}^{2 \alpha n} h_\alpha(\ell,n) \W^{\ell}[\1_A] \] where for $\ell \in [n]$, \[ h_\alpha(\ell,n) \eqdef \max_{\vecv \in\BZ_2^n,  \|\vecv\|_0=\ell} \left(\Pr_{M \sim \Matchings_\alpha(n)} \left[\exists \vecs \neq\veczero \in \BZ_2^{\alpha n} \text{ s.t. } (M^{\fold})^\top \vecs = \vecv \right]\right). \]
\end{lemma}

To interpret the definition of $h_\alpha(\ell,n)$, we can view $\vecs$ as ``marking'' some edges of the matching $M$ with $1$'s; then the vector $(M^{\fold})^\top \vecs$ simply marks which vertices are incident to a marked edge.

To bound the sum from \cref{lemma:bpd-fourier-reduce}, we rely on two separate inequalities in the regimes of small and large $\ell$. In the small-$\ell$ regime, we apply \cref{lemma:low-fourier-bound}, and in the large-$\ell$ regime, we apply the following bound:

\begin{lemma}[Combinatorial bound on $h$]\label{lemma:bpd-combo-bound}
Let $h_\alpha(\ell,n)$ be defined as in \cref{lemma:bpd-fourier-reduce}. For every $\alpha \in (0,1)$, and for every $n, \ell \in \BN$ with even $\ell \leq n/2$, we have \[ h_\alpha(\ell,n) = \frac{\binom{\alpha n}{k/2}}{\binom{n}\ell} \leq \left(\frac{2\alpha e \ell}n\right)^{\ell/2}. \] For odd $\ell$, $h_\alpha(\ell,n)= 0$.
\end{lemma}

Before proving \cref{lemma:bpd-fourier-reduce,lemma:bpd-combo-bound}, let us show how they suffice to prove \cref{thm:bpd-hardness}.

\begin{proof}[Proof of \cref{thm:bpd-hardness}]
Suppose $\Alice$ and $\Bob$ use a one-way communication protocol $\Prot$ for $\bpd_\alpha$ which uses at most $s = \tau \sqrt n$ communication and achieves advantage $\delta$, where $\tau$ is a constant to be determined later. From $\Bob$'s perspective, $\Alice$'s message partitions the set of possible $\vecx^*$'s into sets $\{A_i \subseteq \BZ_2^n\}_{i\in[2^s]}$. Conditioned on a fixed set $A \subseteq \BZ_n^2$, $\Prot$ is distinguishing the distributions $\CZ_{A,M}$ and $\CU$ for random $M \sim \Matchings_\alpha(n)$, and thus it achieves advantage at most $\delta_A \eqdef \E_{M\sim\Matchings_\alpha(n)}[\|\CZ_{A,M} - \CU\|_{\tv}]$. Letting $\CA$ denote the distribution which samples each $A_i$ w.p. $|A_i|/2^n$, we have
\begin{equation}\label{eqn:bpd-cond-adv}
\delta \leq \E_{A \sim \CA}[\delta_A].
\end{equation}
Our goal is to find a contradiction to \cref{eqn:bpd-cond-adv} for a sufficiently small choice of $\tau$. We set $\tau = 2\tau'$, where $\tau'>0$ is to be determined later, and let $s' = \tau'\sqrt{n}$.

A ``typical'' $A \sim \CA$ is large, so to contradict \cref{eqn:bpd-cond-adv}, it is sufficient to show that $\delta_A$ is small for large $A$. Indeed, since $s' < s-\log(2/\delta)$ (for sufficiently large $n$), we have $\Pr_{A \sim \CA}[|A| \leq 2^{n-s'}] \leq \frac{\delta}2$, and it therefore suffices to prove that if $|A| \geq 2^{n-s'}$, then $\delta_A \leq \frac{\delta}2$.

Let $A \subseteq \BZ_2^n$ with $|A| \geq 2^{n-s'}$. By Jensen's inequality,
\begin{equation}\label{eqn:bpd-jensen}
    \delta_A \leq \sqrt{\E_{M\sim\Matchings_\alpha(n)}[\|\CZ_{A,M}-\CU\|_{\tv}^2]}.
\end{equation}

Now we apply \cref{lemma:bpd-fourier-reduce}:
\begin{align}
    \E_{M\sim\Matchings_\alpha(n)}[\|\CZ_{A,M}-\CU\|_{\tv}^2] &\leq \frac{2^{2n}}{|A|^2} \sum_{\ell=2}^{2\alpha n} h_\alpha(\ell,n) \W^{\ell}[\1_A]. \nonumber\\
    \intertext{We split the sum at $\ell=4s'$, using \cref{lemma:low-fourier-bound} for the first term and \cref{prop:parseval} for the second:}
    &= \frac{2^{2n}}{|A|^2} \sum_{\ell=2}^{4s'} h_\alpha(\ell,n) \W^\ell[\1_A] + \frac{2^{2n}}{|A|^2} \sum_{\ell=4s'}^{2\alpha n} h_\alpha(\ell,n) \W^\ell[\1_A] \nonumber\\
    &\leq \sum_{\ell=2}^{4s'} h_\alpha(\ell,n) \left(\frac{\zeta s'}{\ell}\right)^\ell + \frac{2^{2n}}{|A|^2} \max_{4s'\leq\ell\leq2\alpha n} h_\alpha(\ell,n). \label{eqn:bpd-split-form} \\
    \intertext{Applying \cref{lemma:bpd-combo-bound} and the inequality $|A| \geq 2^{n-s'}$:}
    &\leq \sum_{\text{even }\ell=2}^{4s'} \left(\frac{2\alpha e (\zeta s')^2}{\ell n}\right)^{\ell/2} +\left(\frac{16\alpha e s'}n\right)^{2s'}. \nonumber \\
    \intertext{Finally, we use the inequalities $\frac{s'}n \leq \frac{(s')^2}n = (\tau')^2$, $\ell \geq 2$, $2s' \geq 1$ and upper-bound with a geometric series:}
    &\leq \sum_{\text{even }\ell=2}^{4s'} \left(\tau' \zeta\sqrt{\alpha e}\right)^{\ell} + 16\alpha e (\tau')^2 \nonumber \\
    &\leq \sum_{\text{even }\ell=2}^{\infty} \left(\tau' \zeta\sqrt{\alpha e}\right)^{\ell} + 16 \alpha e(\tau')^2\nonumber\\
    &= \frac{\alpha e(\tau' \zeta)^2}{1-\alpha e(\tau' \zeta)^2} + 16\alpha e(\tau')^2 \nonumber.
\end{align}

Assuming WLOG $\alpha e(\tau' \zeta)^2 \leq \frac12$, \cref{eqn:bpd-jensen} then gives $\delta_A \leq \tau' \sqrt{4\zeta^2+8}$, yielding $\delta_A \leq \frac{\delta}2$ for a small enough choice of $\tau' = \Theta(\delta)$, as desired.
\end{proof}

\begin{remark}\label{rem:bpd-low-ell-terms}
In \cref{eqn:bpd-split-form}, the ``low-$\ell$ terms'' are qualitatively the most important; they are the site of ``balancing'' between powers of $n$ between the low-level Fourier weight bounds (\cref{lemma:low-fourier-bound}) and the random-graph analysis (\cref{lemma:bpd-combo-bound}). In particular, for $\ell \in \{2,\ldots,4s'\}$ we get terms of the form $h_\alpha(\ell,n) \left(\frac{\zeta s'}{\ell}\right)^\ell$, which are $\left(O\left(\frac{(s')^2}{\ell n}\right)\right)^{\ell /2}$ by \cref{lemma:bpd-combo-bound}. Even for e.g. $\ell=2$, this term is super-constant if $s = \omega(\sqrt n)$.
\end{remark}

We now prove \cref{lemma:bpd-combo-bound,lemma:bpd-fourier-reduce}.

\begin{proof}[Proof of \cref{lemma:bpd-fourier-reduce}]
Let $M \in \{0,1\}^{2\alpha n \times n}$ be a fixed matching. For fixed $\vecs \neq \veczero \in \{0,1\}^{\alpha n}$, we have
\begin{align*}
    \hat{\CZ_{A,M}}(\vecs) &= \frac1{2^{\alpha n}} \sum_{\vecz \in \BZ_2^{\alpha n}} (-1)^{-\vecs \cdot \vecz} \, \CZ_{A,M}(\vecz) \tag{definition of $\hat{\CZ_{A,M}}$} \\
    &= \frac1{2^{\alpha n}} \sum_{\vecz \in \BZ_2^{\alpha n}} (-1)^{\vecs \cdot \vecz} \left(\E_{\vecx^* \sim \Unif_A}\left[\1_{\vecz=(M^{\fold})\vecx^*}\right]\right) \tag{definition of $\CZ_{A,M}$} \\
    &= \frac1{2^{\alpha n}} \E_{\vecx^* \sim \Unif_A} \left[(-1)^{-\vecs \cdot ((M^{\fold})\vecx^*)}\right] \tag{linearity of expectation} \\
    &= \frac1{2^{\alpha n}} \E_{\vecx^* \sim \Unif_A} \left[(-1)^{-((M^\fold)^\top \vecs) \cdot \vecx^*}\right] \tag{adjointness} \\
    &= \frac{2^n}{2^{\alpha n}|A|} \hat{\1_A}((M^\fold)^\top \vecs) \tag{definition of $\hat{\1_A}$}.
\end{align*}
Combining with \cref{lemma:xor}, we get \[ \|\CZ_{A,M}-\CU\|_{\tv}^2 \leq \frac{2^{2n}}{|A|^2} \sum_{\vecs \neq \veczero \in \BZ_2^{\alpha n}} \hat{\1_A}((M^\fold)^\top \vecs). \] Finally, we observe that $(M^{\fold})^\top$ is an injective map, since there is at most a single $1$ in each row (because $M$ is a matching). Hence, taking expectation over $M$, we have \[ \E_M[\|\CZ_{A,M}-\CU\|_{\tv}^2] \leq \frac{2^{2n}}{|A|^2} \sum_{\vecv \neq \veczero \in \BZ_2^n} \hat{\1_A}(\vecv)\left(\Pr_M[\exists \vecs \neq \veczero \in \BZ_2^{\alpha n} \text{ s.t. }(M^\fold)^\top \vecs = \vecv]\right), \] proving the lemma.
\end{proof}

\begin{proof}[Proof of \cref{lemma:bpd-combo-bound}]

Suppose $\vecz \in \BZ_2^n$ and $M$ is an $\alpha$-partial matching on $[n]$. A vector $\vecs \neq \veczero \in \BZ_2^n$ such that $(M^{\fold})^\top \vecs = \vecv$ marks edges of $M$ such that $\vecv$ marks vertices incident to a marked edge. Thus, such a vector exists iff every pair of vertices in $\vecv$ is connected by an edge in $M$.

Under uniform relabeling of vertices, a fixed matching $M$ becomes a uniform matching. Thus, it is equivalent to fix $M$ to WLOG have edges $\{(1,2),(3,4),\ldots,(2 \alpha n -1,2 \alpha n)\}$ and let $\vecv$ be uniform among vectors in $\BZ_2^n$ with Hamming weight $\ell$. $\vecs$ exists iff $\vecv$ is supported entirely on $[2\alpha n]$ and whenever $\vecv$ is supported on $2i-1$ it is also supported on $2i$ and vice versa. There are $\binom{n}\ell$ total possibilities for $\vecv$, but only $\binom{\alpha n}{\ell/2}$ ways to pick $\vecv$'s support on odd vertices up to $2\alpha n -1$. Thus, \[ h_\alpha(\ell,n) = \frac{\binom{\alpha n}{\ell/2}}{\binom{n}\ell}, \] as desired.\footnote{The original proof of Gavinsky \emph{et al.}~\cite{GKK+08} used a different argument to arrive at the same answer. Consider a fixed vector $\vecv$, WLOG $1^\ell0^{n-\ell}$, and a random matching $M$. $\vecs$ exists iff $M$ is the disjoint union of a total matching on $[\ell]$ and a $(\alpha - \frac{\ell}{2n})$-partial matching on $[n]\setminus[\ell]$. Let $m_{\alpha,n}$ denote the number of $\alpha$-partial matchings on $n$ vertices. Then it can be shown that \[ m_{\alpha,n} = \frac{n!}{2^{\alpha n}(\alpha n)!(n-2\alpha n)!} \quad \text{and therefore} \quad \frac{m_{1,\ell} \cdot m_{\alpha-\frac{\ell}{2n},n-\ell}}{m_{\alpha,n}} = \frac{\binom{\alpha n}{\ell/2}}{\binom{n}\ell}. \]}

Finally, using the inequalities $\left(\frac{a}b\right)^b \leq \binom{a}b \leq \left(\frac{ea}b\right)^b$, we have \[ \frac{\binom{\alpha n}{\ell/2}}{\binom{n}\ell} \leq \frac{\left(\frac{e\alpha n}{\ell/2}\right)^{\ell/2}}{\left(\frac{n}{\ell}\right)^\ell} = (2\alpha e)^{\ell/2} \left(\frac{\ell}n\right)^{\ell/2}, \] as desired.
\end{proof}

\section{The hybrid argument: Proving \cref{lemma:bpd-to-seqbpd}}\label{sec:bpd-to-seqbpd}

To reduce $\seqbpd_{\alpha,T}(n)$ to $\bpd_\alpha(n)$ (and prove \cref{lemma:bpd-to-seqbpd}), we use a standard \emph{hybrid argument}, introduced in this context by Kapralov, Khanna, and Sudan~\cite{KKS15}. Intuitively, in $\seqbpd_{\alpha,T}(n)$, each $\Bob_t$ has to solve his own $\bpd_\alpha(n)$ instance (though he ``gets help'' from $\Bob_1,\ldots,\Bob_{t-1}$). Thus, in our proof, we use the triangle inequality to show that one of these $\Bob_t$'s must be ``doing a decent job'' at solving his $\bpd_\alpha(n)$ instance, and then we convert this to a general algorithm for $\bpd_\alpha(n)$ by simulating the ``help'' of $\Bob_1,\ldots,\Bob_{t-1}$.

\begin{proof}[Proof of \cref{lemma:bpd-to-seqbpd}]
Let $\Prot$ be a space-$s$ protocol for $\seqbpd_{\alpha,T}(n)$, given by message functions $\Prot_0,\ldots,\Prot_t$, such that $\Alice$'s message is determined by the function $\Prot_0$, which takes input $\vecx^* \in \BZ_2^n$, and $\Bob_t$'s by the function $\Prot_t$, which takes input $(m_{t-1},M_t,\vecz(t)) \in \{0,1\}^s \times \{0,1\}^{\alpha n \times n} \times \BZ_2^{\alpha n}$.

Now, we consider the ``coupled'' experiment where we sample $\Alice$'s input and then examine $\Bob$'s behavior in both the $\yes$ and $\no$ cases. Let $S^{\y}_0 = S^{\n}_0 \eqdef \Prot_0(\Unif_{\BZ_2^n})$ denote $\Alice$'s output (as a random variable). Then for $t \in [T]$, define \[ S^{\y}_t \eqdef \Prot_t(S^{\y}_{t-1},M_t,(M^{\fold})^\top\vecx^*) \text{ and } S^{\n}_t \eqdef \Prot_t(S^{\n}_{t-1},M_t,\Unif_{\BZ_2^{\alpha n}}) \in \Delta(\{0,1\}^s)\] as $\Bob_t$'s output message in the $\yes$ and $\no$ cases, respectively. Since $\Prot$ distinguishes the $\yes$ and $\no$ distributions with advantage $\delta$, we have \[ \|S^{\y}_T-S^{\n}_T\|_{\tv} \geq \delta. \] By the triangle inequality (\cref{lemma:rv-triangle}), there exists $t \in [T]$ such that
\begin{equation}\label{eqn:seqbpd-informative-index}
    \|S^{\y}_t-S^{\n}_t\|_{\tv} - \|S^{\y}_{t-1}-S^{\n}_{t-1}\|_{\tv} \geq \frac{\delta}T.
\end{equation}

Now, let $\tilde{S} \eqdef \Prot_t(S^{\y}_{t-1},M_t,\Unif_{\BZ_2^{\alpha n}})$, i.e., $\tilde{S}$ is $\Bob_t$'s output message in the following \emph{hybrid} experiment: $\Bob_1,\ldots,\Bob_{t-1}$ receive $\yes$ inputs, and $\Bob_t$ receives a $\no$ input. By the triangle inequality,
\begin{equation}\label{eqn:seqbpd-hybrid-triangle}
    \|S^{\y}_t - \tilde{S}\|_{\tv} \geq \|S^{\y}_t - S^{\n}_t\|_{\tv} - \|S^{\n}_t-\tilde{S}\|_{\tv}.
\end{equation}

Note that $S_t^{\n} = \Prot_t(S^{\n}_{t-1},M_t,\Unif_{\BZ_2^{\alpha n}}$ and $\tilde{S} = \Prot_t(S^{\y}_{t-1},M_t,\Unif_{\BZ_2^{\alpha n}}$. (I.e., in the two experiments, $\Bob_t$ receives an input sampled from the $\no$ distribution, while $\Bob_1,\ldots,\Bob_{t-1}$ receive inputs from the $\no$ and $\yes$ distributions, respectively.) In both cases, $\Bob_t$'s $\no$ input $(M_t,\Unif_{\BZ_2^{\alpha n}})$ is independent of both $S^{\n}_{t-1}$ and $S^{\y}_{t-1}$. Thus, by the data processing inequality (\cref{lemma:data-processing}), we have:
\begin{equation}\label{eqn:seqbpd-hybrid-data-processing}
    \|S^{\n}_t - \tilde{S}\|_{\tv} \leq \|S^{\y}_{t-1} - S^{\n}_{t-1}\|_{\tv}.
\end{equation}

Putting \cref{eqn:seqbpd-hybrid-data-processing,eqn:seqbpd-informative-index,eqn:seqbpd-hybrid-triangle} together gives \[ \|S^{\y}_t-\tilde{S}\|_{\tv} \geq \frac{\delta}T. \] But $\tilde{S} = \Prot_t(S^{\y}_{t-1},M_t,\Unif_{\BZ_2^{\alpha n}})$ and $S_t^{\y} = \Prot_t(S^{\y}_{t-1},M_t,(M_t^{\fold}) \vecx^*)$. (I.e., in the two experiments, $\Bob_1,\ldots,\Bob_{t-1}$ receive inputs sampled from the $\yes$ distribution, while $\Bob_t$ receives input from the $\no$ and $\yes$ distributions, respectively.) This yields an algorithm for $\bpd$ achieving advantage $\frac{\delta}T$: $\Alice$ can simulate $\yes$ inputs for $\Bob_1,\ldots,\Bob_{t-1}$, and then send $\Bob_{t-1}$'s message to $\Bob$, who can distinguish $S_t^{\y}$ and $\tilde{S}$ with advantage $\|S_t^{\y}-\tilde{S}\|_{\tv}$.\footnote{Explicitly, $\Bob$ should output $\yes$ or $\no$ based on whether $\Bob_t$'s output message has higher probability in $S_t^{\y}$ or $\tilde{S}$, respectively.}
\end{proof}

\section{Discussion}\label{sec:mcut-discussion}

We conclude this chapter by discussing some key features of the proof of \cref{thm:mcut-hardness} which will be relevant for the remainder of this thesis.

\subsection{Strengths of the model}

We proved \cref{thm:mcut-hardness} using a reduction from $\seqbpd$ to $\mcut$. The lower bound holds against ``streaming algorithms'', but what properties, exactly, of these algorithms do we require? We make no assumptions about their uniformity or time complexity. We do assume $O(\sqrt n)$ space, but only actually invoke this assumption when each player sends the state of the algorithm onto the next player. Moreover, the instances are \emph{constant-degree}, that is, each variable is involved in at most $O(1)$ constraints. Indeed, $\Psi$ is a union of $T=O(1)$ subinstances $\Psi_1 \cup \cdots \cup \Psi_T$, and each $\Psi_t$ corresponds to a matching, so each variable has degree at most $T$. Thus, the lower bounds actually hold in a stronger model where the streaming algorithm can process the input instance in $O(1)$ ``equally spaced chunks'' and the instance is promised to have constant degree.

\subsection{Weaknesses of the model: Input ordering}\label{sec:mcut-input-ordering}

Yet the lower bounds, and the techniques used to prove them so far, also have a number of weaknesses. Firstly, we focus on the assumption of adversarially-ordered input streams. The instances produced by the reduction (\cref{cons:seqbpd-to-mcut}) are not randomly ordered. Indeed, recall that in \cref{cons:seqbpd-to-mcut}, $\Alice$ adds no constraints, and each $\Bob_t$ adds a subinstance $\Psi_t$ corresponding to a random matching. Thus, the instance $\Psi = \Psi_1 \cup \cdots \cup \Psi_t$ has the property that in each chunk of $\alpha n$ constraints corresponding to $\Psi_t$, there are no repeated variables; this property is unlikely if we randomly reorder the constraints. Fortunately, Kapralov, Khanna, and Sudan~\cite{KKS15} were able to fix this issue by considering a variant of $\bpd$ based on Erd{\H o}s-R\'enyi graphs (i.e., $\Graphs_\alpha(n)$), instead of random matchings (i.e., $\Matchings_\alpha(n)$).\footnote{Note that graphs sampled from $\Graphs_\alpha(n)$ are simple, so the instances will still have the property that there are no repeated \emph{constraints} in each ``chunk'' of $\alpha n$ constraints. However, since there are only $O(1)$ chunks, this property remains likely when the constraints are randomly reordered; see \cite[Lemma 4.7]{KKS15}.} This change slightly complicates the proof of hardness; specifically, in the Fourier-analytic reduction (i.e., \cref{lemma:bpd-fourier-reduce}), $(M^\fold)^\top$ is no longer an injection, so $h_\alpha(\ell,n)$ must be redefined as \[ \max_{\vecv\in\BZ_2^n,  \|\vecv\|_0=\ell} \left(\E_{M\sim\Graphs_\alpha(n)}\left[\left|\left\{\vecs\in \BZ_2^{\alpha n}: \vecs\neq\veczero, (M^\fold)^\top \vecs = \vecv \right\}\right|\right]\right). \] Correspondingly, the bound on $h_\alpha(\ell,n)$ (cf. \cref{lemma:bpd-combo-bound}) becomes slightly more intricate, but the proof ultimately goes through.

\subsection{Weaknesses of the model: Space bound}\label{sec:mcut-linear-space}

There is also the question of extending the space bound for $\mcut$'s hardness beyond $o(\sqrt n)$. As discussed in \cref{rem:bpd-low-ell-terms} above, the \cite{GKK+08} proof of $\bpd$'s hardness (\cref{thm:bpd-hardness}) only works for $\sqrt n$-space protocols. But $\sqrt{n}$-dependence is not simply an artifact of this proof. Indeed, the \cref{thm:bpd-hardness} is tight in the following sense: For any $\alpha, \delta \in (0,1)$, there exists a protocol for $\bpd_\alpha(n)$ achieving advantage $\delta$ in $O(\sqrt n)$ communication. Indeed, $\Alice$ can just uniformly sample a set $S \subseteq [n]$ of size $\tilde{n}$ to be chosen later, and sends $x_s$ to $\Bob$ for each $s \in S$. Let $\tilde{m}$ denote the number of edges in $\Bob$'s matching between vertices in $S$. In the $\yes$ case, $\Bob$'s input $\vecz$ will always match the information he receives from $\Alice$ on each of these edges, while in the $\no$ case, they will all match only with probability $2^{-\tilde m}$. Moreover, $\E[\tilde m] \approx \frac{\alpha \tilde n^2}n$ by linearity of expectation, so (using concentration bounds, w.h.p.) $\tilde m$ is an arbitrarily large constant for arbitrarily large choice of $\tilde n = O(\sqrt n)$. (The fact that $\tilde n = O(\sqrt n)$ is the right number of vertices to sample in order to expect to see edges in the induced subgraph was termed an example of the ``birthday paradox'' by \cite{GKK+08}.) Therefore, $\Bob$ can distinguish the $\yes$ and $\no$ cases. Since this protocol also works for $\seqbpd$, it implies that better space lower bounds for $\mcut$ cannot rely on reductions from the $\seqbpd$ problem. 

To get around this issue, Kapralov, Khanna, Sudan, and Velingker~\cite{KKSV17} introduced an \emph{implicit} variant of $\seqbpd$, which we'll denote by $\seqibpd$. In the $\seqibpd$ problem, unlike $\seqbpd$, no party receives the hidden partition as input (i.e., there is no $\Alice$). Building on \cite{KKSV17}, Kapralov and Krachun~\cite{KK19} proved the following:

\begin{theorem}[\cite{KK19}]\label{thm:seqibpd-hardness}
For every $\alpha, \delta \in (0,1)$ and $T \in \BN$, there exists an $n_0 \in \BN$ such that for all $n \geq n_0$, any protocol for $\seqibpd_{\alpha,T}(n)$ achieving advantage at least $\delta$ requires $\Omega(n)$ communication.
\end{theorem}

The \cite{KK19} proof of \cref{thm:seqibpd-hardness} and its extensions in \cite{CGS+22} are very technically demanding; see the discussion at the end of \cref{sec:cgsvv} for some brief insights.

Recall that in \cref{cons:seqbpd-to-mcut}, $\Alice$ did not contribute any constraints, and so \cref{cons:seqbpd-to-mcut} might as well be a reduction from $\seqibpd$, instead of $\seqbpd$, to $\mcut$. Thus, \cref{thm:seqibpd-hardness} immediately implies an extension of $\mcut$'s hardness to the linear-space setting. That is, any $(\frac12+\epsilon)$-approximation to $\mcut$ requires $\Omega(n)$ space.

\begin{remark}\label{rem:sparsifier}
The $\Omega(n)$ bound is optimal up to logarithmic factors. Indeed, for every $\epsilon > 0$, we can pick a large constant $C$, and given an input instance $\Psi$, we can sample $Cn$ constraints to get a subinstance $\tilde{\Psi}$ of $\Psi$, and outputting $\val_{\tilde{\Psi}}$ will give a $(1-\epsilon)$-approximation to $\val_\Psi$ for sufficiently large $C$! This ``sparsification'' algorithm only requires $\tilde{O}(n)$ space, and the same technique yields arbitrarily-good approximations in $\tilde{O}(n)$ space for every CSP.
\end{remark}

\subsection{Weaknesses of the proof: Choosing the hybrid}\label{sec:mcut-hybrid-discussion}

Finally, we highlight a subtlety in the choice of the hybrid variable $\tilde{S}$ in the reduction from $\bpd$ to $\seqbpd$ (\cref{lemma:bpd-to-seqbpd}, see \cref{sec:bpd-to-seqbpd}). Recall that we applied the data processing inequality (\cref{lemma:data-processing}) to argue that $\Bob_t$ can't help distinguish $\tilde{S}$ from the $\no$ case $S^{\n}_t$ (see \cref{eqn:seqbpd-hybrid-data-processing}). But using this inequality relies on the fact that $\Bob_t$'s $\no$ input is independent of the inputs to $\Bob_1,\ldots,\Bob_{t-1}$ in both the $\yes$ and $\no$ cases. Thus, we couldn't have, for instance, defined $\tilde{S}$ by mixing $\no$ inputs for $\Bob_1,\ldots,\Bob_{t-1}$ and a $\yes$ input for $\Bob_t$. This same issue occurs in later works in streaming hybrid arguments for general CSPs \cite{CGSV21-boolean,CGSV21-finite}. In particular, in the appropriate generalizations of $\bpd$ and $\seqbpd$, $\Bob_t$ must have a uniformly distributed input, (typically) in the $\no$ case.

\chapter{$\mdcut$ is mildly approximable}\label{chap:mdcut}

\newcommand{\mus}{\mu_{\textsf{S}}}

\epigraph{[$\mcut$] raises the question whether streaming algorithms operating in small space can non-trivially approximate (i.e., beat the random assignment threshold) for \emph{some} CSP, or whether every CSP is approximation resistant in the streaming model.}{Guruswami \emph{et al.}~\cite{GVV17}}

\newthought{Unlike $\mcut$, $\mdcut$ is approximable} in the streaming setting. Indeed, Guruswami, Velingker, and Velusamy~\cite{GVV17} showed that $\mdcut$ can be $(\frac25-\epsilon)$-approximated by $O(\log n)$-space linear sketching algorithms, but not $(\frac12+\epsilon)$-approximated by $\sqrt n$-space streaming algorithms, for every $\epsilon > 0$. A tighter characterization was later given by Chou, Golovnev, and Velusamy~\cite{CGV20}:

\begin{theorem}[{\cite{CGV20}}]\label{thm:mdcut-characterization}
For every constant $\epsilon > 0$:
\begin{enumerate}[label={\roman*.},ref={\roman*}]
\item There is a $O(\log n)$-space linear sketching algorithm which $(\frac49-\epsilon)$-approximates $\mdcut$ (and even $\mtwoand$).\label{item:mdcut-algorithm}
\item Any streaming algorithm which $(\frac49+\epsilon)$-approximates $\mdcut$ requires $\Omega(\sqrt n)$ space.\label{item:mdcut-hardness}
\end{enumerate}
\end{theorem}

We remark that in the classical setting, several works \cite{GW95,FG95,Zwi00,MM01,LLZ02} have given algorithms achieving increasingly good approximation ratios for $\mtwoand$ and/or $\mdcut$. Most recently, Lewin, Livnat, and Zwick~\cite{LLZ02} presented and analyzed an algorithm for $\mtwoand$ which achieves an approximation ratio $\alpha_{\mathrm{LLZ}} \geq 0.87401$.\footnote{Austrin~\cite{Aus10} shows that $(\alpha_{\mathrm{Aus}}+\epsilon)$-approximating $\mtwoand$ is UG-hard, where $\alpha_{\mathrm{Aus}} \approx 0.87435$. Without the UGC, Trevisan \emph{et al.}~\cite{TSSW00} show that $(\frac{12}{13}+\epsilon)$-approximation is $\NP$-hard.} Thus, although $\mtwoand$ is nontrivially approximable in the streaming setting, its streaming approximability still falls far short of its classical approximability.

The goal of this chapter is to prove \cref{thm:mdcut-characterization}. We prove its two components separately --- addressing the algorithmic result, \cref{item:mdcut-algorithm}, in \cref{sec:mdcut-algorithm}, and the hardness result, \cref{item:mdcut-hardness}, in \cref{sec:mdcut-hardness}. In both cases, we highlight the crucial role played by certain information about CSP instances which we'll call \emph{template distributions}. Later, these will form the basis of the dichotomy theorems from \cite{CGSV21-boolean,CGSV21-finite} (see \cref{sec:cgsv} below). Finally, we conclude with more discussion in \cref{sec:mdcut-discussion}.

\section{Bias-based algorithms for $\mtwoand$}\label{sec:mdcut-algorithm}

The optimal $(\frac49-\epsilon)$-approximate sketching algorithm for $\mtwoand$ from \cite{CGV20} (i.e., \cref{item:mdcut-algorithm} of \cref{thm:mdcut-characterization}) is based on measuring a quantity called the \emph{bias} of the input instance $\Psi$. This quantity was introduced by Guruswami, Velingker, and Velusamy~\cite{GVV17} (who used it to achieve a weaker approximation factor of $\frac25$). In this section, we present a cleaner analysis due to subsequent simplifications in our joint work \cite{BHP+22}. The analysis in this latter work, which we present in \cref{sec:thresh-alg} below, generalizes our argument for $\mtwoand$ to $\mbcsp[f]$ for every ``threshold function'' $f$.

\subsection{Setup: Bias-based algorithms}

Throughout this section, we assume $\Psi$ has constraints $(\vecb(\ell),\vecj(\ell),w(\ell))_{\ell\in[m]}$. For each variable $i \in [n]$ we let
\begin{equation}\label{eqn:2and-bias-var}
    \bias_\Psi(i) \eqdef \sum_{\ell\in[m],t\in[2]:~j(\ell)_t=i} (-1)^{b(\ell)_t} w_\ell,
\end{equation}
and then we define
\begin{equation}\label{eqn:2and-bias}
    \bias_\Psi \eqdef \frac1{2W} \sum_{i\in[n]} |\bias_\Psi(i)|.
\end{equation}
(Note that $\bias_\Psi(i) \in [-W_\Psi,W_\Psi]$ where $W_\Psi$ is the total weight in $\Psi$, while $\bias_\Psi \in [0,1]$.)

For each variable $i$, $|\bias_\Psi(i)|$ measures the imbalance between $i$'s negated and non-negated appearances, and thus correlates with ``how easy $i$ is to assign''. For instance, if $\bias_\Psi(i) \gg 0$, then $x_i$ is rarely negated (i.e., the $b$-value is typically $0$), and so we should assign $x_i = 1$. Thus, we should expect to see some positive relationship between $\bias_\Psi$ and $\val_\Psi$. Indeed, we have:

\begin{lemma}[{\cite[Theorem 11]{GVV17}}]\label{lemma:mdcut-alg-lb}
Let $\Psi$ be a $\m[\twoand]$ instance. Then \[ \val_\Psi \geq \frac29 (1+\bias_\Psi). \]
\end{lemma}

\begin{lemma}[{\cite[Lemma 3.3]{CGV20}}]\label{lemma:mdcut-alg-ub}
Let $\Psi$ be a $\m[\twoand]$ instance. Then \[ \val_\Psi \leq \frac12 (1+\bias_\Psi). \]
\end{lemma}

Together \cref{lemma:mdcut-alg-lb,lemma:mdcut-alg-ub} imply that outputting $\frac29(1+\bias_\Psi)$ gives a $\frac49$-approximation to $\val_\Psi$. To implement measure $\bias_\Psi$ in the $O(\log n)$-space streaming setting algorithm, we observe that $\bias_\Psi$ is the $\ell_1$-norm of the vector $\bias(\Psi) \eqdef (\bias_1(\Psi),\ldots,\bias_n(\Psi))$. We can thus calculate it using $\ell_1$-sketching (\cref{thm:l1-sketching}) for $\bias(\Psi)$: Given each new constraint $(\vecb(\ell),\vecj(\ell))$ with weight $w_\ell$, we simply add $(-1)^{b(\ell)_t}w_\ell$ to $\bias(\Psi)_{j(\ell)_t}$ for each $t \in [2]$.

\subsection{Analysis: template distributions}\label{sec:mdcut-template-alg}

To prove \cref{lemma:mdcut-alg-lb,lemma:mdcut-alg-ub}, we define a few useful notions. For a distribution $\CD \in \Delta(\BZ_2^2)$, let $\CD\langle0\rangle=\CD(0,0),\CD\langle1\rangle=\CD(1,0)+\CD(0,1),$ and $\CD\langle2\rangle=\CD(1,1)$, i.e., $\CD\langle t \rangle$ is the probability mass on Hamming weight $t$.

Given an assignment $\vecx \in \BZ_2^n$ to an instance $\Psi$ of $\mtwoand$ with constraints $((\vecb(\ell),\vecj(\ell),w(\ell))_{\ell \in [m]}$, we define a \emph{template distribution} $\CD_\Psi^\vecx \in \Delta(\BZ_2^2)$ as follows: We sample $\ell$ with probability $\frac{w(\ell)}{W_\Psi}$ and output $\vecb(\ell) + \vecx|_{\vecj(\ell)}$. Thus, $\CD_\Psi^\vecx\langle t \rangle$ is the fraction weight of constraints such that $t$ of the variables equal $1$ under assignment $\vecx$.

\begin{example}
Let $\Psi$ consist of $n=2$ variables and $m=3$ constraints $((0,0),(1,2))$, $((0,1),(1,2))$, and $((1,1),(1,2))$ with weights $2$, $1$, and $3$, respectively. (In Boolean notation, these constraints would be written $x_1 \wedge x_2, x_1\wedge \bar{x_2},\bar{x_1}\wedge\bar{x_2}$.) Let $\vecx=(1,1)$. Then $\CD_\Psi^\vecx\langle2\rangle=\frac13,\CD_\Psi^\vecx\langle1\rangle=\frac16,$ and $\CD_\Psi^\vecx\langle0\rangle=\frac12$.
\end{example}

The distribution $\CD^\vecx_\Psi$ succinctly describes several important properties of the assignment $\vecx$, and helps us to bridge between what we \emph{can} measure (bias) and what we \emph{would like} to measure (value). For instance, we define
\begin{equation}\label{eqn:2and-mu}
    \mus(\CD) \eqdef \CD\langle2\rangle - \CD\langle0\rangle.
\end{equation} Roughly, $\mus(\CD_\Psi^\vecx)$ measures how well $\vecx$ performs at assigning $i$ to $\mathsf{sign}(\bias_\Psi(i))$ when $|\bias_\Psi(i)|$ is large:

\begin{fact}\label{lemma:2and-mu}
For all $\vecx \in \BZ_q^n$, $\mus(\CD_\Psi^\vecx) = \frac1{2W} \sum_{i=1}^n (-1)^{x_i} \bias_\Psi(i)$. In particular, $\mus(\CD_\Psi^\vecx) \leq \bias_\Psi$, and moreover, $\mus(\CD_\Psi^\vecx) = \bias_\Psi$ iff for each $i \in [n]$, $\bias_\Psi(i) > 0 \Longrightarrow x_i =1$ and $\bias_\Psi(i) < 0 \Longrightarrow x_i = 0$.
\end{fact}

\cref{lemma:2and-mu} is a special case of \cref{item:d-mu,item:d-bias} of \cref{prop:d}.

Now to a distribution $\CD \in \Delta(\BZ_2^2)$ we associate a \emph{canonical instance} $\Psi^{\CD}$ of $\mtwoand$ on 2 variables, which puts weight $\CD(\vecb)$ on the constraint $((1,2),\vecb)$ for each $\vecb \in \BZ_2^2$. We can thus define the quantity
\begin{equation}\label{eqn:2and-gamma}
    \gamma(\CD) \eqdef \val_{\Psi^{\CD}}(\veczero) = \CD\langle2\rangle.
\end{equation}

$\gamma(\CD_\Psi^\vecx)$ measures $\vecx$'s value:

\begin{fact}\label{lemma:2and-gamma}
For all $\vecx \in \BZ_2^n$, $\gamma(\CD_\Psi^\vecx) = \val_\Psi(\vecx)$.
\end{fact}
\cref{lemma:2and-gamma} is a special case of \cref{item:d-val} of \cref{prop:d} below. Now we have:

\begin{proof}[Proof of \cref{lemma:mdcut-alg-lb}]
Let $\vecx^* \in \BZ_2^n$ be the optimal assignment for $\Psi$. Then
\begin{align*}
    \val_\Psi &= \val_\Psi(\vecx^*) \tag{optimality of $\vecx^*$} \\
    &= \gamma(\CD_\Psi^{\vecx^*}) \tag{\cref{lemma:2and-gamma}} \\
    &= \CD_\Psi^{\vecx^*}\langle2\rangle \tag{definition of $\gamma$} \\
    &= \frac12\left(\CD_\Psi^{\vecx^*}\langle0\rangle+\CD_\Psi^{\vecx^*}\langle1\rangle+\CD_\Psi^{\vecx^*}\langle2\rangle\right) + \frac12\left(\CD_\Psi^{\vecx^*}\langle2\rangle - \CD_\Psi^{\vecx^*}\langle0\rangle\right) \\
    &\leq \frac12(1+\mus(\CD_\Psi^{\vecx^*})) \tag{definition of $\mus$} \\
    &\leq \frac12(1+\bias_\Psi) \tag{\cref{lemma:2and-mu}}.
\end{align*}
\end{proof}

On the other hand, to prove \cref{lemma:mdcut-alg-ub}, for $\CD \in \Delta(\BZ_2^2)$ and $p \in [0,1]$, we define the quantity
\begin{equation}\label{eqn:2and-lambda}
    \lambda(\CD,p)\eqdef \E_{\vecb \sim \Bern_p^2}[\val_{\Psi^{\CD}}(\vecb)] = q^2\CD\langle0\rangle + pq\CD\langle1\rangle + p^2\CD\langle2\rangle
\end{equation}
where $q = 1-p$. In particular, $\lambda(\CD,1) = \gamma(\CD)$. We can also define $\beta(\CD) \eqdef \max_{p \in [0,1]} \lambda(\CD,p)$, in which case $\beta(\CD) \geq \gamma(\CD)$.

\begin{fact}\label{lemma:2and-lambda}
For all $\vecx \in \BZ_2^n$ and $p \in [0,1]$, $\lambda(\CD_\Psi^\vecx,p) = \E_{\veca\sim\Bern_p^n}[\val_\Psi(\veca+\vecx)]$. In particular, $\beta(\CD_\Psi^\vecx) \leq \val_\Psi$.
\end{fact}
\cref{lemma:2and-lambda} is a special case of \cref{item:d-val} of \cref{prop:d} below. We have:

\begin{proof}[Proof of \cref{lemma:mdcut-alg-lb}]
Let $\tilde{\vecx} \in \BZ_2^n$ be the ``majority assignment'', i.e., $\tilde{x}_i = \1_{\bias_\Psi(i)\geq0}$. We have
\begin{align*}
    \val_\Psi &\geq \lambda\left(\CD_\Psi^{\tilde{\vecx}},\frac23\right) \tag{\cref{lemma:2and-lambda}} \\ &=\frac49\CD_\Psi^{\tilde{\vecx}}\langle2\rangle+\frac29\CD_\Psi^{\tilde{\vecx}}\langle1\rangle +\frac19\CD_\Psi^{\tilde{\vecx}}\langle0\rangle \tag{definition of $\lambda$} \\
    &= \frac29\left(\CD_\Psi^{\tilde{\vecx}}\langle0\rangle+\CD_\Psi^{\tilde{\vecx}}\langle1\rangle+\CD_\Psi^{\tilde{\vecx}}\langle2\rangle\right) + \frac29\left(\CD_\Psi^{\tilde{\vecx}}\langle2\rangle-\CD_\Psi^{\tilde{\vecx}}\langle0\rangle\right)+\frac19\CD_\Psi^{\tilde{\vecx}}\langle0\rangle \tag{definition of $\mus$} \\
    &\geq \frac29(1+\mus(\CD_\Psi^{\vecx})) \\
    &=\frac29(1+\bias_\Psi)\tag{\cref{lemma:2and-mu} and definition of $\tilde{\vecx}$}
\end{align*}
\end{proof}

The proof of \cref{lemma:mdcut-alg-lb} contains the inequality \[ \lambda\left(\CD_\Psi^{\tilde \vecx},\frac23\right) \geq \frac29(1+\bias_\Psi). \] Combined with \cref{lemma:mdcut-alg-ub}, this gives \[ \lambda\left(\CD_\Psi^{\tilde \vecx},\frac23\right) \geq \frac49\val_\Psi. \] This yields a simple streaming algorithm for a different problem, namely \emph{outputting} an assignment with expected value at least $\frac49\val_\Psi$, in linear time and space: We simply calculate the majority assignment $\tilde{\vecx}$ and then flip each bit independently with probability $\frac13$.

\begin{remark}\label{rem:cgsv-vs-bhp-2and}
The original proof of \cref{lemma:mdcut-alg-lb} in \cite{CGV20} was substantially more complex than the one presented here (see the proof of Lemma 3.3 in that paper), because it considers the value $p \in [0,1]$ which maximizes the quadratic $\lambda(\CD_{\Psi}^{\tilde{\vecx}},p)$ (which is $\frac12+\frac{\bias_\Psi}{2(1-2\bias_\Psi)}$ in the regime $\bias_\Psi \in [0,\frac13]$). The insight that this is ``overkill'' and setting $p = \frac23$ is sufficient to get $(\frac49-\epsilon)$-approximations is due to our joint work \cite{BHP+22}. This issue will become more prominent when we consider $\kand$ for $k > 2$, since $\lambda$ will have degree $k$ and thus its maximizer over $[0,1]$ has no simple expression; see the discussion in \cref{sec:cgsv-opt} below.
\end{remark}

\section{Proving hardness for $\mdcut$}\label{sec:mdcut-hardness}

In this section, we prove \cref{item:mdcut-hardness} of \cref{thm:mdcut-characterization}, which is a hardness-of-approximation result for $\mdcut$ in the streaming setting. To begin, we give some intuition for the construction.

\subsection{Intuition: What's wrong with $\seqbpd$?}

Our first hope might be to directly reduce from $\seqbpd$ using \cref{cons:seqbpd-to-mcut} (i.e., the reduction we used for $\mcut$), by converting each $\mcut$ constraint $(u,v)$ into the pair of $\mdcut$ constraints $\{(u,v),(v,u)\}$. Could we hope to prove an analogue of \cref{lemma:seqbpd-to-mcut-analysis} in this setting, with a $\frac49$ gap? In the $\no$ case, $\vecz(t)$ is random and so for a fixed assignment $\vecx\in\BZ_2^n$, $\E[\val_{\Psi}(\vecx)]=\frac14$. But in the $\yes$ case, if $\vecx^* \in \BZ_2^n$ is $\Alice$'s input, when $z(t)_\ell = 1$ for $\vece(\ell)=(u,v)$, then $(x^*_u,x^*_v)$ is either $(0,1)$ or $(1,0)$, so exactly one of the constraints $(u,v),(v,u)$ will be satisfied! Thus, we have $\val_{\Psi}(\vecx^*)=\frac12$. Hence \cref{cons:seqbpd-to-mcut} only seems to rule out $\approx\frac12$-approximations to $\mcut$.

We can frame the issue with \cref{cons:seqbpd-to-mcut} in the following way: Its $\yes$ instances have low values because they are too ``symmetric''. In particular, we also have $\val_{\Psi}(1+\vecx^*) = \frac12$. To break this symmetry, we can have $\Alice$ add constraints $(u,v)$ where $x^*_u = 1,x^*_v=0$. These have the effect of biasing towards $\vecx^*$ and away from $1+\vecx^*$. But this increases the value of $\vecx^*$ even in the $\no$ case (because $\frac14$-fraction of each $\Bob_t$'s constraints will be satisfied by $\vecx^*$ in expectation). To compensate, we change the game $\seqbpd$ slightly, so that in the $\no$ case, $\Bob_t$'s constraints are never satisfied by $\vecx^*$; that is, when he adds $\{(u,v),(v,u)\}$, we guarantee that $(x^*_u,x^*_v)\in\{(1,1),(0,0)\}$.

\subsection{A new problem and a new reduction}\label{sec:mdcut-template-hardness}

We carry out the proof of \cref{item:mdcut-hardness} using a close cousin of $\seqbpd$:

\begin{definition}\label{def:seqbpd'}
Let $\alpha \in (0,1)$ and $T,n\in\BN$. $\seqbpd'_{\alpha,T}(n)$ is defined identically to $\seqbpd'_{\alpha,T}(n)$ (\cref{def:seqbpd}), except that in the $\no$ case, we set $\Bob_t$'s vector $\vecz(t)$ to the opposite of its value in the $\yes$ case. That is, $\vecz(t)=\vecone +(M^\fold) \vecx^*$.
\end{definition}

Now we can formally state the reduction:

\begin{construction}[C2S reduction from $\seqbpd'$ to $\mdcut$]\label{cons:seqbpd'-to-mdcut}
$\Alice$'s reduction function, denoted $\R_0$, outputs an instance $\Psi_0$ consisting of $\frac{\alpha T n}4$ uniformly random constraints $(u,v)$ such that $x^*_u = 1,x^*_v = 0$. For each $t \in [T]$, $\Bob_t$'s reduction function $\R_t$ outputs an instance $\Psi_t$ as follows: For each $\vece(t,\ell) = (u,v)$ in $M_t$, $\Bob_t$ adds $(u,v)$ and $(v,u)$ to $\Psi_t$ iff $z(t)_\ell = 1$.
\end{construction}

The hard instances for $\mdcut$ produced by \cref{cons:seqbpd'-to-mdcut} are represented pictorially in \cref{fig:mdcut}.

\begin{figure}
\centering
\begin{subfigure}{0.4\textwidth}
\centering
\begin{tikzpicture}[vertex/.style={fill=black},block/.style={draw=black,fill=white!70!lightgray}, goodedge/.style={line width=1.5pt,draw=black!40!green}, badedge/.style={line width=1.5pt,draw=black!10!red}, greatedge/.style={->,line width=1.5pt,draw=black!20!blue,-{Latex[width=8pt,length=10pt]}}]

% block 0
\draw[block] (0,3) ellipse (0.75 and 3.5);
\draw[vertex] (0,0) circle (3pt);
\draw[vertex] (0,1) circle (3pt);
\draw[vertex] (0,2) circle (3pt);
\draw[vertex] (0,3) circle (3pt);
\draw[vertex] (0,4) circle (3pt);
\draw[vertex] (0,5) circle (3pt);
\draw[vertex] (0,6) circle (3pt);

% block 1
\draw[block] (3,3) ellipse (0.75 and 3.5);
\draw[vertex] (3,0) circle (3pt);
\draw[vertex] (3,1) circle (3pt);
\draw[vertex] (3,2) circle (3pt);
\draw[vertex] (3,3) circle (3pt);
\draw[vertex] (3,4) circle (3pt);
\draw[vertex] (3,5) circle (3pt);
\draw[vertex] (3,6) circle (3pt);

% good edges
\draw[goodedge] (0,0) to (3,4);
\draw[goodedge] (0,3) to (3,6);
\draw[goodedge] (0,4) to (3,4);
\draw[goodedge] (0,2) to (3,1);
\draw[goodedge] (0,5) to (3,1);

% bad edges
\draw[badedge] (0,2) to[bend left] (0,6);
\draw[badedge] (0,1) to[bend left] (0,4);
\draw[badedge] (3,2) to[bend right] (3,5);
\draw[badedge] (3,1) to[bend right] (3,6);
\draw[badedge] (3,0) to[bend right] (3,3);
\end{tikzpicture}
\caption{$\yes$ sample from $\seqbpd'$.}
\label{fig:seqbpd'-yes}
\end{subfigure}
\begin{subfigure}{0.4\textwidth}
\centering
\begin{tikzpicture}[vertex/.style={fill=black},block/.style={draw=black,fill=white!70!lightgray}, goodedge/.style={line width=1.5pt,draw=black!40!green}, badedge/.style={line width=1.5pt,draw=black!10!red}, greatedge/.style={->,line width=1.5pt,draw=black!20!blue,-{Latex[width=8pt,length=10pt]}}]
\draw[block] (0,3) ellipse (0.75 and 3.5);
\draw[vertex] (0,0) circle (3pt);
\draw[vertex] (0,1) circle (3pt);
\draw[vertex] (0,2) circle (3pt);
\draw[vertex] (0,3) circle (3pt);
\draw[vertex] (0,4) circle (3pt);
\draw[vertex] (0,5) circle (3pt);
\draw[vertex] (0,6) circle (3pt);

\draw[block] (3,3) ellipse (0.75 and 3.5);
\draw[vertex] (3,0) circle (3pt);
\draw[vertex] (3,1) circle (3pt);
\draw[vertex] (3,2) circle (3pt);
\draw[vertex] (3,3) circle (3pt);
\draw[vertex] (3,4) circle (3pt);
\draw[vertex] (3,5) circle (3pt);
\draw[vertex] (3,6) circle (3pt);

% good edges
\draw[badedge] (0,0) to (3,4);
\draw[badedge] (0,3) to (3,6);
\draw[badedge] (0,4) to (3,4);
\draw[badedge] (0,2) to (3,1);
\draw[badedge] (0,5) to (3,1);

% bad edges
\draw[goodedge] (0,2) to[bend left] (0,6);
\draw[goodedge] (0,1) to[bend left] (0,4);
\draw[goodedge] (3,2) to[bend right] (3,5);
\draw[goodedge] (3,1) to[bend right] (3,6);
\draw[goodedge] (3,0) to[bend right] (3,3);
\end{tikzpicture}
\caption{$\no$ sample from $\seqbpd'$.}
\label{fig:seqbpd'-no}
\end{subfigure}

\begin{subfigure}{0.4\textwidth}
\centering
\begin{tikzpicture}[vertex/.style={fill=black},block/.style={draw=black,fill=white!70!lightgray}, goodedge/.style={line width=1.5pt,draw=black!40!green}, badedge/.style={line width=1.5pt,draw=black!10!red}, greatedge/.style={->,line width=1.5pt,draw=black!20!blue,-{Latex[width=8pt,length=10pt]}}]

% block 0
\draw[block] (0,3) ellipse (0.75 and 3.5);
\draw[vertex] (0,0) circle (3pt);
\draw[vertex] (0,1) circle (3pt);
\draw[vertex] (0,2) circle (3pt);
\draw[vertex] (0,3) circle (3pt);
\draw[vertex] (0,4) circle (3pt);
\draw[vertex] (0,5) circle (3pt);
\draw[vertex] (0,6) circle (3pt);

% block 1
\draw[block] (3,3) ellipse (0.75 and 3.5);
\draw[vertex] (3,0) circle (3pt);
\draw[vertex] (3,1) circle (3pt);
\draw[vertex] (3,2) circle (3pt);
\draw[vertex] (3,3) circle (3pt);
\draw[vertex] (3,4) circle (3pt);
\draw[vertex] (3,5) circle (3pt);
\draw[vertex] (3,6) circle (3pt);

% crossing edges
\draw[greatedge] (0,1) to (3,0);
\draw[greatedge] (0,2) to (3,3);
\draw[greatedge] (0,6) to (3,4);

% good edges
\draw[goodedge] (0,0) to (3,4);
\draw[goodedge] (0,3) to (3,6);
\draw[goodedge] (0,4) to (3,4);
\draw[goodedge] (0,2) to (3,1);
\draw[goodedge] (0,5) to (3,1);

\end{tikzpicture}
\caption{$\yes$ instance of $\mdcut$.}
\label{fig:mdcut-yes}
\end{subfigure}
\begin{subfigure}{0.4\textwidth}
\centering
\begin{tikzpicture}[vertex/.style={fill=black},block/.style={draw=black,fill=white!70!lightgray}, goodedge/.style={line width=1.5pt,draw=black!40!green}, badedge/.style={line width=1.5pt,draw=black!10!red}, greatedge/.style={->,line width=1.5pt,draw=black!20!blue,-{Latex[width=8pt,length=10pt]}}]
\draw[block] (0,3) ellipse (0.75 and 3.5);
\draw[vertex] (0,0) circle (3pt);
\draw[vertex] (0,1) circle (3pt);
\draw[vertex] (0,2) circle (3pt);
\draw[vertex] (0,3) circle (3pt);
\draw[vertex] (0,4) circle (3pt);
\draw[vertex] (0,5) circle (3pt);
\draw[vertex] (0,6) circle (3pt);

\draw[block] (3,3) ellipse (0.75 and 3.5);
\draw[vertex] (3,0) circle (3pt);
\draw[vertex] (3,1) circle (3pt);
\draw[vertex] (3,2) circle (3pt);
\draw[vertex] (3,3) circle (3pt);
\draw[vertex] (3,4) circle (3pt);
\draw[vertex] (3,5) circle (3pt);
\draw[vertex] (3,6) circle (3pt);

% crossing edges
\draw[greatedge] (0,1) to (3,0);
\draw[greatedge] (0,2) to (3,3);
\draw[greatedge] (0,6) to (3,4);

% bad edges
\draw[goodedge] (0,2) to[bend left] (0,6);
\draw[goodedge] (0,1) to[bend left] (0,4);
\draw[goodedge] (3,2) to[bend right] (3,5);
\draw[goodedge] (3,1) to[bend right] (3,6);
\draw[goodedge] (3,0) to[bend right] (3,3);
\end{tikzpicture}
\caption{$\no$ instance of $\mdcut$.}
\label{fig:mdcut-no}
\end{subfigure}
\caption[]{
\cref{fig:seqbpd'-yes,fig:seqbpd'-no} depict samples from the $\yes$ and $\no$ distributions of $\seqbpd'$, respectively. The $\yes$ distribution for $\seqbpd'$ is the same as for $\seqbpd$ (\cref{fig:seqbpd-yes}); that is, $\vecz(t)$ marks which edges cross the cut. However, the $\no$ distribution of $\seqbpd'$ marks which edges do \emph{not} cross the cut, as opposed to $\seqbpd$'s $\no$ distribution, which marks uniformly random edges (\cref{fig:seqbpd-no}). The graphs in \cref{fig:seqbpd'-yes,fig:seqbpd'-no} again represent the union of the matchings $M_1,\ldots,M_t$; $\vecx^*$ partitions the vertices into ``left'' ($0$) and ``right'' ($1$); and the edges' $z$-values are either ``green'' ($1$) or ``red'' ($0$). In the reduction from $\seqbpd'$ to $\mcut$ (\cref{cons:seqbpd-to-mcut}), $\Alice$ adds edges crossing the cut from left to right (``blue''), and each $\Bob_t$ adds (undirected copies of) edges with $z$-value ``green'' ($1$). We pick parameters so that there are four times as many ``green'' edges as ``blue'' edges. In the $\yes$ case, $\vecx^*$ cuts all ``blue'' edges and half of the (directed) ``green'' edges, so the $\yes$ instances (\cref{fig:mdcut-yes}) have $\mdcut$ value $\approx \frac35$. In the $\no$ case, $\vecx^*$ cuts all ``blue'' edges but none of the ``green'' edges, and we show that these instances (\cref{fig:mdcut-no}) have $\mdcut$ value $\approx \frac4{15}$ (for sufficiently large $T$) in \cref{lemma:seqbpd'-to-mdcut-analysis}.}
\label{fig:mdcut}

% \caption[]{Example hard instances for $\mcut$. \cref{fig:seqbpd-yes,fig:seqbpd-no} depict samples from the $\yes$ and $\no$ distributions of $\seqbpd$, respectively. Recall, in $\seqbpd$, $\Alice$ receives a hidden partition $\vecx^* \in \BZ_2^n$ and each $\Bob_t$ receive a matching $M_t$ along with a vector $\vecz(t)$ annotating $M_t$'s edges. In the $\yes$ case, $\vecz(t)$ marks the edges of $M_t$ which cross the partition $\vecx^*$, while in the $\no$ case, $\vecz(t)$ is uniformly random. The graphs in \cref{fig:seqbpd-yes,fig:seqbpd-no} represent the union of the matchings $M_1,\ldots,M_t$; $\vecx^*$ partitions the vertices into ``left'' ($0$) and ``right'' ($1$); and the edges' $z$-values are either ``green'' ($1$) or ``red'' ($0$). In our reduction from $\seqbpd$ to $\mcut$ (\cref{cons:seqbpd-to-mcut}), $\Alice$ adds no edges, and each $\Bob_t$ adds all edges with $z$-value ``green'' ($1$). In the $\yes$ case, the resulting graph is bipartite (\cref{fig:mcut-yes}) and thus has $\mcut$ value $1$, while in the $\no$ case, the graph is random (\cref{fig:mcut-no}) and has value $\approx \frac12$ with high probability (for sufficiently large $T$).
% }
\end{figure}

Luckily, there is a simple proof by ``symmetry'' that $\seqbpd'$ is also hard:

\begin{lemma}\label{lemma:seqbpd-to-seqbpd'}
Suppose that for some $\alpha \in (0,1)$ and $T,n \in \BN$, there is a protocol for $\seqbpd'_{\alpha,T}(n)$ achieving advantage $\delta$ with communication $s$. Then there is a protocol for $\seqbpd_{\alpha,T}(n)$ achieving advantage $\frac{\delta}2$, also with communication $s$.
\end{lemma}

\begin{proof}
Suppose $\Pi$ is a protocol for $\seqbpd'$ achieving advantage $\delta$. By the triangle inequality, $\Pi$ achieves advantage $\frac{\delta}2$ in distinguishing one of the following pairs of distributions:

\begin{enumerate}
    \item $\yes$ instances of $\seqbpd'$ and $\no$ instances of $\seqbpd$.
    \item $\no$ instances of $\seqbpd'$ and $\no$ instances of $\seqbpd$.
\end{enumerate}

Case (1) immediately gives the desired result, since $\seqbpd'$ and $\seqbpd$ have the same $\yes$ instances. Case (2) also gives the desired result, since by applying the ``reduction'' of adding 1 to all input $\vecz(t)$ vectors, $\no$ instances of $\seqbpd'$ become $\yes$ instances of $\seqbpd$, while $\no$ instances of $\seqbpd'$ remain $\no$ instances of $\seqbpd$.
\end{proof}

To finally prove \cref{item:mdcut-hardness}, it suffices to prove the following lemma (analogous to \cref{lemma:seqbpd-to-mcut-analysis}):

\begin{lemma}\label{lemma:seqbpd'-to-mdcut-analysis}
For all $\alpha \in (0,1)$ and $\epsilon \in (0,\frac12)$, there exist $T, n_0 \in \BN$ such that for every $n \geq n_0$, the following holds. Let $\CY$ and $\CN$ denote the $\yes$ and $\no$ distributions for $\seqbpd'_{\alpha,T}(n)$, and let $(\R_0,\ldots,\R_T)$ be the reduction functions from \cref{cons:seqbpd'-to-mdcut}. Then \[ \Pr_{\Psi \sim (\R_0,\ldots,\R_T) \circ \CY}\left[\val_\Psi \leq \frac35 - \epsilon \right] \leq \exp(-n) \text{ and } \Pr_{\Psi \sim (\R_0,\ldots,\R_T) \circ \CN}\left[\val_\Psi \geq \frac4{15}+\epsilon\right]\leq \exp(-n). \]
\end{lemma}

However, we only give a heuristic proof, omitting concentration bounds and independence arguments. The full proof can be found in e.g. \cite[\S5]{CGV20}.

\begin{proof}[Proof sketch of \cref{lemma:seqbpd'-to-mdcut-analysis}]
Recall our definition of template distributions from the previous section; let's compute the expected template distributions $\CD^{\vecx^*}_\Psi$ in the $\yes$ and $\no$ cases, which we will denote $\CD_Y$ and $\CD_N$, respectively. In expectation, $\Alice$ adds $\frac{\alpha}4Tn$ constraints, and $\Bob_t$ adds $\alpha Tn$ constraints. In both the $\yes$ and $\no$ cases, the constraints introduced by $\Alice$ are always satisfied by $\vecx^*$. In the $\yes$ case, $\Bob_t$ sees edges $(u,v)$ such that $x^*_u\neq x^*_v$, and he adds the constraints $\{(u,v),(v,u)\}$; thus, $\vecx^*$ satisfies both literals in one of the clauses, and neither in the other. Thus, $\CD_Y\langle2\rangle = \frac{1/4+1/2}{1/4+1}=\frac35$ and $\CD_Y\langle0\rangle = \frac{1/2}{1/4+1}=\frac25$; and so $\gamma(\CD_Y)=\frac35$. On the other hand, in the $\no$ case, $\Bob_t$'s edges $(u,v)$ satisfy $x^*_u \neq x^*_v$; thus, $\vecx^*$ satisfies one literal in both of the clauses $(u,v)$ and $(v,u)$. Hence in expectation, $\CD_N\langle2\rangle = \frac{1/4}{1/4+1}=\frac15$ and $\CD_N\langle1\rangle = \frac{1}{1/4+1}=\frac45$. Now \[ \lambda(\CD_N,p) = \frac45p(1-p) + \frac25p^2 = \frac15 p(4-3p), \] so $\lambda(\CD_N,p)$ is maximized at $p=\frac23$, yielding $\beta(\CD_N) = \frac4{15}$.\footnote{Note also that $\mus(\CD_Y) = \mus(\CD_N) = \frac15$; thus, the algorithm presented in the previous section (\cref{sec:mdcut-algorithm}) fails to solve the $(\frac4{15},\frac35)\text{-}\m[\twoand]$ problem.}

To prove the claimed bound on $\yes$-instance values, \cref{lemma:2and-gamma} implies that $\val_\Psi(\vecx^*) \geq \gamma(\CD_\Psi^{\vecx^*}) \approx \gamma(\CD_Y) = \frac35$. The bound on $\no$-instance values is trickier. The key observation is that the distribution $(\R_0,\ldots,\R_T) \circ \CN$ is invariant under permutations of variables. Thus, it suffices to show that for each $\ell \in \{0\}\cup[n]$, a uniformly random solution of Hamming weight $\ell$ has value below $\frac4{15}+\epsilon$; but the expected value of such a solution is precisely $\lambda(\CD_\Psi^{\vecx^*},\frac{\ell}n) \approx \lambda(\CD_N,\frac{\ell}n) \leq \beta(\CD_N) = \frac4{15}$.
\end{proof}

\section{Discussion}\label{sec:mdcut-discussion}

Again, we conclude with some discussion on various aspects of the reductions and algorithms in this chapter.

\subsection{Weaknesses of the reduction (it's $\Alice$'s fault)}\label{sec:mdcut-rand-linear}

Thinking back to our discussion for $\mcut$ (\cref{sec:mdcut-discussion}), the fact that $\Alice$ did not add any constraints in $\mcut$ reduction (\cref{cons:seqbpd-to-mcut}) was crucial in extending $\mcut$'s hardness to the random-ordering and linear-space settings \cite{KKS15,KK19}. For $\mdcut$, the picture is much less rosy, because in the $\mdcut$ reduction (\cref{cons:seqbpd'-to-mdcut}), $\Alice$ has a significant role to play, creating around $\frac15$-fraction of the constraints. Thus, it is not clear at all how to derive randomly ordered instances --- even if each $\Bob_t$ receives a random graph instead of a random matching --- since the distribution of constraints created by $\Alice$ is very different from the distribution of constraints created by each $\Bob_t$, and the constraints are added to the stream in sequence. Nor is it clear how to define an appropriate variant of $\seqibpd$ (which, recall, omitted $\Alice$ entirely!) to effect a linear-space hardness reduction. (However, $(\frac12+\epsilon)$-hardness for $\mcut$ in these settings does imply $(\frac12+\epsilon)$-hardness for $\mdcut$, by the reduction which given a $\mcut$ constraint $(u,v)$ randomly outputs either $(u,v)$ or $(v,u)$ as a $\mdcut$ constraint.)

Indeed, we know from personal communication with Chen, Kol, Paramonov, Saxena, Song, and Yu~\cite{CKP+21} and Chou, Golovnev, Sudan, Velingker, and Velusamy~\cite{CGS+21} that the hard instances produced by \cref{cons:seqbpd'-to-mdcut} are distinguishable by streaming algorithms in the $O(\log n)$-space random-order and $o(n)$-space adversarial-order settings, respectively. In the remainder of this subsection, we roughly sketch both algorithms.

Recall the definition of the bias $\bias_\Psi(i)$ of variable $i$ in an instance $\Psi$ of $\mtwoand$. We can view instances of $\mdcut$ as instances of $\mtwoand$ where every constraint has negation pattern $(0,1)$ (since $\dcut(a,b) = \twoand(a,b+1)$). Then by definition, $\bias_\Psi(i)$ is the difference in total weight of constraints in which $i$ appears on the left vs. on the right. On the other hand, we can also view an instance of $\mdcut$ as a (weighted) graph on $n$ vertices (see \cref{sec:hypergraphs}); under this interpretation, the bias $\bias_\Psi(i)$ of a vertex $i$ is the difference between its out-weight and its in-weight.

Now let's examine the distributions of $\yes$ and $\no$ instances for $\mdcut$ from \cite{CGV20} (produced by \cref{cons:seqbpd'-to-mdcut} from $\seqbpd'$, see \cref{fig:mdcut}). Letting $b \eqdef \frac{\alpha T n}4$, we see that vertices with $x^*=1$ have nonnegative bias ($b$ in expectation) and vertices with $x^*=0$ have nonpositive bias ($-b$ in expectation). Furthermore, in the $\yes$ case, all edges go from vertices with nonnegative bias to those with nonpositive bias, while in the $\no$ case, there is a mixture of nonnegative to nonpositive ($\approx \frac15$ fraction of edges), nonnegative to nonnegative ($\approx \frac25$ fraction), and nonpositive to nonpositive ($\approx \frac25$ fraction).

\paragraph{The random-ordering algorithm.} If we have the ability to randomly sample edges and measure the biases of their endpoints, it will quickly become apparent whether we are seeing $\yes$ or $\no$ instances, and this can be accomplished in the random-ordering setting. Indeed, it is sufficient even to store the first $m'$ edges for some large constant $m' = O(1)$ and measure the biases of all their endpoints. This technique is similar in spirit to the random-ordering algorithms for counting components and calculating minimum spanning trees in \cite{PS18}. Note that while randomly sampling edges is still possible in the adversarial-ordering setting (with e.g. reservoir sampling), there is no clear way to do so while also measuring the biases of their endpoints; indeed, the adversarial-ordering lower bound shows that this is impossible.

\paragraph{The super-$\sqrt n$-space algorithm.} Another strategy to distinguish $\yes$ and $\no$ instances is to randomly sample a subset $V \subseteq [n]$ of the vertices and, during the stream, both measure the bias of every vertex in $V$ and store the induced subgraph on $V$.\footnote{We can store the induced subgraph in $O(|V|)$ space since the instances produced in the reduction have constant max-degree (with high probability).} At the end of the stream, we can simply check for the presence of any edge in the induced subgraph which does not go from a nonnegative-bias vertex to a nonpositive-bias vertex. However, in order for this to succeed, $|V|$ needs to be sufficiently large; picking a random set of $n^{0.51}$ vertices will suffice by the ``birthday paradox'' argument of \cref{sec:mcut-linear-space}.
 
\subsection{Duality, towards dichotomy}\label{sec:towards-dichotomy}

It is quite surprising that using the lower bound $\val_\Psi \geq \lambda(\CD^\vecx_\Psi,\frac23)$, instead of the more general $\val_\Psi \geq \beta(\CD^\vecx_\Psi)$, suffices to prove \cref{lemma:mdcut-alg-lb}. In particular, we can't get a better approximation ratio using the latter inequality, since $\frac49$ is already optimal given \cref{item:mdcut-hardness}.\footnote{However, we certainly can do worse! Guruswami \emph{et al.}'s $(\frac25-\epsilon)$-approximation uses $p=1$, i.e., it \emph{greedily} assigns positively-biased variables to $1$ and negatively-biased variables to $0$. This algorithm is ``overconfident'' and setting $p=\frac23$ instead improves the approximation ratio.} However, one significant reason for studying the quantity $\beta(\CD)$ is that it also arises in the proof of the hardness result (see the end of the proof sketch of \cref{lemma:seqbpd'-to-mdcut-analysis}).

To give further perspective on the quantities $\beta(\CD)$ and $\gamma(\CD)$, we shift gears slightly to the $\bgd\m[\twoand]$ problem for fixed $\beta < \gamma \in [0,1]$, and interpret the algorithm from \cref{sec:mdcut-algorithm} for this problem. Suppose that $(\beta,\gamma)$ satisfies the equation
\begin{equation}\label{eqn:2and-beta-gamma-gap}
    \max_{\CD \in \Delta(\CD_2^2):~ \beta(\CD) \leq \beta} \mus(\CD) < \min_{\CD \in \Delta(\CD_2^2):~ \gamma(\CD) \geq \gamma} \mus(\CD)
\end{equation}
and consider some fixed threshold $\tau$ in between these two values. \cref{lemma:2and-gamma} and \cref{lemma:2and-lambda} imply, respectively, that (1) there exists $\vecx^*\in\BZ_2^n$ such that $\val_\Psi = \gamma(\CD_\Psi^{\vecx^*})$ and (2) for all $\vecx\in\BZ_2^n$, $\val_\Psi \geq \beta(\CD_\Psi^\vecx)$. Thus, \cref{eqn:2and-beta-gamma-gap} implies that measuring $\bias_\Psi = \max_{\vecx\in\BZ_2^n} \mus(\CD_\Psi^\vecx)$ and comparing it to $\tau$ suffices to distinguish the cases $\val_\Psi \leq \beta$ and $\val_\Psi \geq \gamma$. On the other hand, from the proofs of \cref{lemma:mdcut-alg-lb} and \cref{lemma:mdcut-alg-ub} we can extract the inequalities \[ \beta(\CD) \geq \frac29(1+\mus(\CD)) \text{ and } \gamma(\CD) \leq \frac12(1+\mus(\CD)), \] respectively. Thus, whenever $\frac{\beta}{\gamma} > \frac49$, \cref{eqn:2and-beta-gamma-gap} holds, and the $\bgd\m[\twoand]$ problem is tractable!

On the other hand, the lower bound (ruling out $(\frac49+\epsilon)$-approximations) we proved in \cref{sec:mdcut-hardness} was based on constructing $\CD_N,\CD_Y \in \Delta(\BZ_2^2)$ with $\mus(\CD_N)=\mus(\CD_Y)$ and $\frac{\beta(\CD_N)}{\gamma(\CD_Y)}=\frac49$. So there is a kind of duality between the algorithm and the lower bound for $\mdcut$. Indeed, the algorithm works whenever \cref{eqn:2and-beta-gamma-gap} holds --- that is, there is a separation between the $\mus$-values. On the other hand, the lower bound's hard distributions have matching $\mus$-values. A natural question becomes, can we hope to generalize this ``matching $\mus$-value'' criterion to give an algorithm-or-hardness dichotomy for CSPs beyond $\mdcut$ and $\mcut$? In the next chapter, we'll see the answer from \cite{CGSV21-boolean,CGSV21-finite}: \emph{Yes!}, at least for sketching algorithms.

\newcommand{\kq}{(k,q)\text{-}}
\newcommand{\dndy}{(\CD_N,\CD_Y)\text{-}}

\chapter{Prior results: General CSPs}\label{chap:framework-papers}

\epigraph{At the heart of our characterization is a family of algorithms for $\mF$ in the linear sketching streaming setting. We will describe this family soon, but the main idea of our proof is that if no algorithm in this family solves $\bgd\mF$, then we can extract a pair of instances, roughly a family of $\gamma$-satisfiable ``$\yes$'' instances and a family of at most $\beta$-satisfiable ``no'' instances, that certify this inability. We then show how this pair of instances can be exploited as gadgets in a negative result.}{Chou, Golovnev, Sudan, and Velusamy~\cite{CGSV21-finite}}

\newthought{A wide variety of CSPs fall under} several recent generalizations \cite{CGSV21-boolean,CGSV21-finite,CGS+22} of the $\mcut$ and $\mdcut$ results described in the previous two chapters. Our primary goal is to articulate technical statements of these results, since we'll use them later in \cref{chap:ocsps,chap:sym-bool}. We also give some broad-strokes discussions of the techniques involved, building on our work in the previous two chapters.

Specifically, in \cref{sec:cgsv}, we describe the results of Chou, Golovnev, Sudan, and Velusamy~\cite{CGSV21-boolean,CGSV21-finite} on $\sqrt n$-space streaming algorithms, which generalize the results we've already seen for $\mcut$ (\cref{thm:mcut-hardness}, due to \cite{KKS15}) and $\mdcut$ (\cref{thm:mdcut-characterization}, due to \cite{CGV20}). They include a so-called \emph{dichotomy theorem}, which completely characterizes CSP approximability for $\sqrt n$-space \emph{sketching} algorithms (see \cref{sec:streaming}) and builds on our ``template distribution'' analysis for $\mdcut$ in \cref{chap:mdcut}. This dichotomy will later be the basis for the \cite{BHP+22} analysis of $\mbf$ problems for symmetric $f:\BZ_2^k\to\{0,1\}$, described in \cref{chap:sym-bool} below. Next, in \cref{sec:cgsvv}, we present results due to Chou, Golovnev, Sudan, Velingker, and Velusamy~\cite{CGS+22} in the linear-space streaming setting, which generalize the result we've seen for $\mcut$ (discussed in \cref{sec:mcut-linear-space}, due to \cite{KK19}). We will use these in \cref{chap:ocsps} to prove linear-space streaming approximation-resistance results for so-called ``ordering constraint satisfaction problems'' from our joint work \cite{SSV21}.

\section{$\sqrt{n}$-space algorithms and the sketching dichotomy}\label{sec:cgsv}

In the case of $\mtwoand$, we discussed in \cref{sec:towards-dichotomy} how the presence or absence of ``template distributions'' with certain properties imply $\Omega(\sqrt{n})$-space streaming lower bounds (\cref{sec:mdcut-template-hardness}) and $O(\log n)$-space sketching algorithms (\cref{sec:mdcut-template-alg}) for $\bgd\mtwoand$ problems, respectively. Chou, Golovnev, Sudan, and Velusamy~\cite{CGSV21-finite} proved a dichotomy theorem which generalizes this result to every $\mF$ problem (and every $\beta < \gamma \in [0,1]$): Either $\bgd\mF$ has an $O(\polylog n)$-space sketching algorithm, or for every $\epsilon > 0$, sketching algorithms for $\ebgd\mF$ must use at least $\Omega(\sqrt{n})$ space! Importantly, however, the lower bound holds in generality only against sketching algorithms, though we'll discuss important special cases in which it extends to streaming algorithms below.

In \cref{sec:cgsv-bool}, we give precise definitions of template distributions, their properties, and the ensuing upper and lower bounds for the important special case of $\mbf$ problems. These statements are from \cite{CGSV21-boolean}, and we'll require them in \cref{chap:sym-bool} below, where we present results from \cite{BHP+22} which further investigate the case where $f : \BZ_2^k\to\{0,1\}$ is symmetric. Defining templates and their properties for general $\mF$ problems requires significant elaboration which is out of scope for this thesis, but in \cref{sec:cgsv-discussion}, we outline these notions and provide some structural overviews of the techniques used in \cite{CGSV21-boolean,CGSV21-finite}.

\subsection{Technical statements (Boolean case)}\label{sec:cgsv-bool}

Recall that $\Delta(\BZ_2^k)$ denotes the space of all distributions over $\BZ_2^k$. Following the example of $\twoand$ (see \cref{sec:mdcut-template-alg}), for a ``template distribution'' $\CD \in \Delta(\BZ_2^k)$ and a predicate $f : \BZ_2^k \to \{0,1\}$ we define a ``canonical instance'' $\Psi^{\CD}$ of $\mbf$ on $k$ variables putting weight $\CD(\vecb)$ on the constraint $((1,\ldots,k),\vecb)$. Then we define
\begin{equation}\label{eqn:cgsv-bool-lgb}
    \lambda_f(\CD,p) \eqdef \E_{\veca\sim\Bern_p^k}[\val_{\Psi^{\CD}}(\veca)], \quad \gamma_f(\CD) \eqdef \lambda_f(\CD,1), \quad \text{ and } \beta_f(\CD) \eqdef \sup_{p \in [0,1]} \left(\lambda_f(\CD,p)\right).
\end{equation}

Like we did for $\mtwoand$ in \cref{sec:mdcut-template-alg}, we can interpret these template distributions as distributions of negation patterns for a particular assignment to an instance; see the beginning of \cref{sec:template-dists-mbf} below.

We also define a \emph{marginal vector} $\vecmu(\CD) = (\mu(\CD)_1,\ldots,\mu(\CD)_k) \in [-1,1]^k$ whose $i$-th entry $\mu(\CD)_i \eqdef -\E_{\vecb\sim\CD}[(-1)^{a_i}]$.\footnote{Note that in the $\twoand$ case discussed in \cref{sec:mdcut-template-alg}, for $\CD \in \Delta(\BZ_2^2)$ we defined a \emph{scalar} quantity $\mus(\CD) = \CD(1,1)-\CD(0,0)$. According to the definition we just made, however, $\vecmu(\CD) = (\mu(\CD)_1,\mu(\CD)_2)$ where $\mu(\CD)_1 = \CD(1,1)+\CD(1,0)-\CD(0,1)-\CD(0,0)$ and $\mu(\CD)_2 = \CD(1,1)+\CD(0,1)-\CD(1,0)-\CD(0,0)$. Thus, $\mus(\CD) = \frac12(\mu(\CD)_1+\mu(\CD)_2)$. As we'll discuss in \cref{sec:sym-setup}, for $\mtwoand$ it suffices to only consider a ``symmetric'' distributions $\CD$ and a ``scalar'' marginal because $\twoand$ is a \emph{symmetric} predicate.}\footnote{We add an extra negative sign in order to agree with the convention for $\twoand$ we established in \cref{sec:mdcut-template-alg} that ``positively biased vertices want to be assigned $1$.''}

Now for a predicate $f : \BZ_2^k \to \{0,1\}$, we define two sets of marginals \[ K_{f,N}(\beta) \eqdef \{\vecmu(\CD_N): \beta_f(\CD_N) \leq \beta\} \quad \text{ and } K_{f,Y}(\gamma) \eqdef \{\vecmu(\CD_Y): \gamma_f(\CD_Y) \geq \gamma\}. \] The dichotomy theorem of \cite{CGSV21-boolean} states, roughly, that $\bgd\mbf$ is hard iff these two sets intersect: 

\begin{theorem}[Sketching dichotomy, {\cite[Theorem 2.3]{CGSV21-boolean}}]\label{thm:cgsv-bool-dichotomy}
For every $f : \BZ_2^k \to \{0,1\}$ and $\beta < \gamma \in [0,1]$:
\begin{enumerate}[label={\roman*.},ref={\roman*}]
\item If $K_{f,N}(\beta) \cap K_{f,Y}(\gamma) = \emptyset$, then there exists $\tau > 0$ and a linear sketching algorithm for $\bgd\mbf$ using at most $\tau \log n$ space.\label{item:cgsv-bool-alg}
\item If $K_{f,N}(\beta) \cap K_{f,Y}(\gamma) \neq \emptyset$, then for every $\epsilon > 0$, there exists $\tau > 0$ such that every sketching algorithm for $\ebgd\mbf$ uses at least $\tau \sqrt{n}$ space (for sufficiently large $n$).\label{item:cgsv-bool-hardness}
\end{enumerate}
\end{theorem}

Temporarily peeling back the $K_f$ notation, \cref{item:cgsv-bool-hardness} is a statement about ``hard templates'' \`a la \cref{sec:mdcut-template-hardness}: If $\vecmu(\CD_N)=\vecmu(\CD_Y)$ for some $\CD_N,\CD_Y\in\Delta(\BZ_2^k)$, then sketching $(\beta_f(\CD_N)+\epsilon,\gamma_f(\CD_Y)-\epsilon)\text{-}\mbf$ requires $\Omega(\sqrt n)$ space for every $\epsilon > 0$.

Using a standard reduction, \cref{thm:cgsv-bool-dichotomy} implies a dichotomy for \emph{approximation} problems:

\begin{corollary}[{\cite[Proposition 2.10]{CGSV21-boolean}}]\label{cor:cgsv-bool-approx}
For $f : \BZ_2^k \to \{0,1\}$, let \[ \alpha(f) \eqdef \inf_{\CD_N,\CD_Y \in \Delta(\BZ_2^k): ~\vecmu(\CD_N)=\vecmu(\CD_Y)} \left(\frac{\beta_f(\CD_N)}{\gamma_f(\CD_Y)}\right). \] Then:
\begin{enumerate}[label={\roman*.},ref={\roman*}]
    \item For every $\epsilon > 0$, there exists $\tau > 0$ and a $\tau \log n$-space linear sketching algorithm which  $(\alpha(f)-\epsilon)$-approximates $\mbf$.\label{item:cgsv-bool-approx-alg}
    \item For every $\epsilon > 0$, there exists $\tau > 0$ such that every sketching algorithm which $(\alpha(f)+\epsilon)$-approximates $\mbf$ uses at least $\tau \sqrt{n}$ space (for sufficiently large $n$).\label{item:cgsv-bool-approx-hardness}
\end{enumerate}
\end{corollary}

We sketch the proof in order to provide a point of comparison for the simpler sketching algorithms (for certain $\mbf$ problems) from our joint work \cite{BHP+22}, which we present in \cref{sec:thresh-alg} below.

\begin{proof}[Proof sketch]
To prove \cref{item:cgsv-bool-approx-hardness}, for every $\epsilon > 0$, there exist $\CD_N,\CD_Y\in\Delta(\BZ_2^k)$ with $\vecmu(\CD_N)=\vecmu(\CD_Y)$ such that $\beta_f(\CD_N)/\gamma_f(\CD_Y) \leq \alpha(f)+\epsilon/2$. Letting $\beta'=\beta_f(\CD_N)+\epsilon/2$ and $\gamma'=\gamma_f(\CD_N)-\epsilon/2$, by \cref{thm:cgsv-bool-dichotomy}, $(\beta',\gamma')\text{-}\mbf$ requires $\Omega(\sqrt n)$ space. But for sufficiently small $\epsilon > 0$, $\beta'/\gamma' \leq \alpha(f)+\epsilon$ and thus $\mbf$ requires $\Omega(\sqrt n)$ space to $(\alpha(f)+\epsilon)$-approximate.

The proof of \cref{item:cgsv-bool-approx-alg} uses the following observation: If an algorithm for $\bgd\mbf$ is correct on an instance $\Psi$, then outputs of $\yes$ and $\no$ imply $\val_\Psi \geq \beta$ and $\val_\Psi \leq \gamma$, respectively.\footnote{Note that since $\bgd\mbf$ is a promise problem, if $\beta \leq \val_\Psi \leq \gamma$ then the algorithm's output can be arbitrary.} Thus, given an instance $\Psi$, a reasonable strategy for approximating $\val_\Psi$ is to pick some pairs $\{(\beta_s,\gamma_s)\}_{s \in S}$ such that $K_{f,N}(\beta_s) \cap K_{f,Y}(\gamma_s) = \emptyset$ for each $s$; calculate an output $A_s$ for each $(\beta_s,\gamma_s)\text{-}\mbf$ problem using \cref{thm:cgsv-bool-dichotomy};\footnote{The algorithms given by \cref{thm:cgsv-bool-dichotomy} for $(\beta_s,\gamma_s)\text{-}\mbf$ are randomized, and thus each will fail with some constant probability. However, as long as $|S|$ is a constant (given $\epsilon$), we can amplify every algorithm's success probability and take a union bound.} and then output $\beta_{s^*}$, where $s^*$ maximizes $\beta_s$ over $\{s \in S : A_s = \yes\}$. $\beta_{s^*}$ is an underestimate for $\val_\Psi$ by our observation, but how can we guarantee that it is at least $(\alpha(f)-\epsilon) \val_\Psi$?

The idea is to pick some $\delta > 0$, and consider a ``distinguisher grid'' $S = \{(b\delta,g\delta): b,g \in [\lfloor 1/\delta\rfloor], K_{f,N}(b\delta) \cap K_{f,Y}(g\delta)=\emptyset\}$. Then given $s^* = (b^*,g^*)$ maximizing $\beta_s$ over $\{s \in S: A_s = \yes\}$, we set $b' = b+1$ and $g' = \lceil (b+1)/\alpha \rceil$. By construction, $g'/b' > \alpha$; thus by assumption, $K_{f,N}(b'\delta) \cap K_{f,Y}(g'\delta) = \emptyset$, and so $s' = (b',g') \in S$. Now by maximality of $\beta_{s^*}$, we have $A_{(b',g')} = \no$. Hence \[ b^*\delta \leq \val_\Psi \leq g'\delta, \] and we chose $g'$ such that $b^*/g' \approx \alpha$. Setting $\delta = \epsilon \rho(f)/2$ and tweaking the algorithm to output $\max\{\beta_{s^*},\rho(f)\}$ suffices; see the proof of \cite[Proposition 2.10]{CGSV21-boolean} for details.
\end{proof}

While the lower bounds from \cite{CGSV21-boolean} (i.e., \cref{item:cgsv-bool-hardness} of \cref{thm:cgsv-bool-dichotomy}, and by extension, \cref{item:cgsv-bool-approx-hardness} of \cref{cor:cgsv-bool-approx}) only apply generally for \emph{sketching} algorithms,  \cite{CGSV21-boolean} also provides an extension of the lower bound to streaming algorithms when a certain condition holds. In particular, it is a condition on distributions $\CD_N,\CD_Y \in \Delta(\BZ_2^k)$ which is \emph{stronger} than the condition $\vecmu(\CD_N)=\vecmu(\CD_Y)$. We say $\CD\in\BZ_2^k$ is \emph{one-wise independent} (or has \emph{uniform marginals}) if $\vecmu(\CD)=\veczero$, i.e., for all $i \in [k]$, $\E_{\veca\sim\CD}[a_i]=0$. Then $(\CD_N,\CD_Y)$ are a \emph{padded one-wise pair} if there exists $\CD_0,\CD_N',\CD_Y'\in\Delta(\BZ_2^k)$ and $\eta \in [0,1]$ such that $\CD'_N$ and $\CD'_Y$ have uniform marginals, $\CD_N=\eta \CD_0+(1-\eta) \CD_N'$, and $\CD_Y=\eta \CD_0+(1-\eta)\CD_Y'$. Then:

\begin{theorem}[{\cite[Theorem 2.11]{CGSV21-boolean}}]\label{thm:cgsv-streaming-lb}
For every $f : \BZ_2^k \to \{0,1\}$ and padded one-wise pair $\CD_N,\CD_Y \in \Delta(\BZ_2^k)$, there exists $\tau > 0$ such that every \emph{streaming} algorithm for $(\beta_f(\CD_N)+\epsilon,\gamma_f(\CD_Y)-\epsilon)\text{-}\mbf$ requires at least $\tau \sqrt n$ space (for sufficiently large $n$). Thus, for every $\epsilon > 0$, every \emph{streaming} algorithm which $(\beta_f(\CD_N)/\gamma_f(\CD_Y)+\epsilon)$-approximates $\mbf$ requires at least $\tau \sqrt n$ space (for sufficiently large $n$).
\end{theorem}

\begin{example}\label{ex:cgsv-2and}
In \cref{sec:mdcut-template-hardness}, we constructed a pair of hard template distributions for $\mdcut$: $\CD_Y$, which is $(1,1)$ w.p. $\frac35$ and $(0,0)$ w.p. $\frac25$, and $\CD_N$, which is $(1,1)$ w.p. $\frac15$ and $(0,1)$ and $(1,0)$ each w.p. $\frac25$. Recalling that $\mtwoand$ is a generalization of $\mcut$, we can consider this pair $(\CD_Y,\CD_N)$ in light of \cref{thm:cgsv-streaming-lb} for $\mtwoand$. Our calculations in \cref{sec:mdcut-template-hardness} imply that $\beta_{\twoand}(\CD_N)=\frac4{15}$ and $\gamma_{\twoand}(\CD_Y) = \frac35$. Now setting $\eta=\frac15$, $\CD_0$ to be $(1,1)$ w.p. $1$, $\CD'_Y = \Unif_{\{(1,1),(0,0)\}}$, and $\CD'_N=\Unif_{\{(1,0),(0,1)\}}$, we see that $\CD_Y$ and $\CD_N$ are a padded one-wise pair. Thus, \cref{thm:cgsv-streaming-lb} implies that $(\frac4{15}+\epsilon,\frac35-\epsilon)\text{-}\mtwoand$ requires $\Omega(\sqrt n)$ streaming space, and hence that $(\frac49+\epsilon)$-approximating $\mtwoand$ requires $\Omega(\sqrt n)$. This morally recovers the $\mdcut$ lower bound we proved in \cref{sec:mdcut-hardness} (\cref{item:mdcut-hardness} of \cref{thm:mdcut-characterization}), although it is formally weaker (since $\mdcut$ has a smaller set of predicates). Recovering the full bound requires the more general framework of \cite{CGSV21-finite}.

On the other hand, in \cref{sec:mdcut-template-alg}, we showed that if $\frac{\beta}{\gamma} < \frac49$, then $\max_{\CD_N:~\beta_{\twoand}(\CD)} \mus(\CD) < \min_{\CD_Y:~\gamma_{\twoand}(\CD)} \mus(\CD)$, and hence that $K_{\twoand,N}(\beta) \cap K_{\twoand,Y}(\gamma) = \emptyset$. Thus, $\alpha(\twoand)=\frac49$.
\end{example}

\begin{example}\label{ex:one-wise-indep}
An important special case of \cref{thm:cgsv-streaming-lb} is when $\CD_N = \Unif_{\BZ_2^k}$ and $\gamma_f(\CD_Y)=1$. In this case, $\CD_N,\CD_Y$ are one-wise independent (since $\vecmu(\Unif_{\BZ_2^k})=\veczero$ by definition), and thus they form a trivial padded one-wise pair (with $\eta=0$). Furthermore, all $f : \BZ_2^k\to\{0,1\}$ and $p \in [0,1]$, \[ \lambda_f(\CD_N,p) = \E_{\veca\sim\Unif_{\BZ_2^k},\vecb\sim\Bern_p^k}[f(\veca + \vecb)] = \E_{\veca\sim\Unif_{\BZ_2^k}}[f(\veca)] = \rho(f)\] and thus $\beta_f(\CD_N)=\rho(f)$. Hence \cref{thm:cgsv-streaming-lb} states implies that $\mbf$ is \emph{streaming approximation-resistant} (in $\sqrt n$ space), i.e., $\sqrt n$-space streaming algorithms cannot $(\rho(f)+\epsilon)$-approximate $\mbf$ for any $\epsilon > 0$.

If there exists one-wise independent $\CD_Y \in \BZ_2^k$ with $\gamma_f(\CD_Y)=1$, we say that the predicate $f$ \emph{supports one-wise independence}, since the support of any distribution $\CD_Y \in \BZ_2^k$ with $\gamma_f(\CD_Y) =1$ is necessarily contained in the support of $f$. Thus, if $f$ supports one-wise independence, then $\mbf$ is streaming approximation-resistant \cite[Theorem 1.3]{CGSV21-boolean}.
\end{example}

\subsection{Discussion}\label{sec:cgsv-discussion}

In this subsection, we (briefly) discuss various aspects of the \cite{CGSV21-boolean} results presented in the previous subsection, as well as their extension in \cite{CGSV21-finite} to CSPs defined by general predicate families over all alphabets. To begin, we mention that the dichotomy criterion for $\bgd\mbf$ in \cref{thm:cgsv-bool-dichotomy} is ``expressible in the quantified theory of reals'', i.e., it is equivalent to some quantified polynomial system of (in)equalities over real variables; thus, it is known to be decidable, in particular in polynomial-space relative to the size of $f$'s truth table and the bit complexities of $\beta$ and $\gamma$ (see \cite[Theorem 1.1]{CGSV21-boolean} for the formal statement). This is an exciting property of \cref{thm:cgsv-bool-dichotomy} because, in contrast, the dichotomy criterion of Raghavendra~\cite{Rag08} for polynomial-time algorithms (assuming the UGC) is not known to be decidable.

The proof of \cref{item:cgsv-bool-alg} of \cref{thm:cgsv-bool-dichotomy} is quite similar to the proof of \cref{item:mdcut-algorithm} of \cref{thm:mdcut-characterization} for $\mdcut$. Recall the proof for $\bgd\mtwoand$ which we discussed in \cref{sec:towards-dichotomy}: We used a ``separation'' between $\mus$ values (\cref{eqn:2and-beta-gamma-gap}) to get an algorithm for $\bgd\mtwoand$. Now, for arbitrary $\bgd\mbf$ problems, we still harness the disjointness of $K_{f,N}(\beta)$ and $K_{f,Y}(\gamma)$: As in \cite[\S4]{CGSV21-boolean}, we consider ``template distributions'' $\CD_\Psi^\vecx \in \Delta(\BZ_2^k)$ for potential assignments $\vecx$; we invoke the separating hyperplane theorem on the sets $K_{f,N}(\beta)$ and $K_{f,Y}(\gamma)$ (after checking that they are compact and convex), which we think of as sets of template distributions; and then use $1$-norm sketching (i.e., \cref{thm:l1-sketching}) to ``maximize over $\vecx$''.

The hardness results from \cite{CGSV21-boolean} (\cref{item:cgsv-bool-hardness} of \cref{thm:cgsv-bool-dichotomy} and \cref{thm:cgsv-streaming-lb}) also share certain similarities with \cref{item:mdcut-hardness} of \cref{thm:mdcut-characterization}. Indeed, \cite[\S5]{CGSV21-boolean} considers a variant of $\bpd$ (see \cref{def:bpd}) called \emph{randomized mask detection} ($\rmd$). In the $\dndy\rmd_\alpha(n)$ problem, $\Alice$ gets a hidden assignment $\vecx^* \in \BZ_2^n$ and communicates with $\Bob$, who gets a random hypermatching $M \sim \Matchings_{k,\alpha}(n)$ and a vector $\vecz = M^\top \vecx^* + \vecb \in (\BZ_2^k)^{\alpha n}$, where in the $\yes$ case $\vecb \sim \CD_Y^{\alpha n}$ and in the $\no$ case $\vecb \sim \CD_N^{\alpha n}$. That is, for each edge-index $\ell \in [\alpha n]$, if we let $\vecz(\ell) \in \BZ_q^k$ denote the $\ell$-th block of $k$ coordinates in $\vecz$ and $\vece(\ell)$ the $\ell$-th edge in $M$, $\vecz(\ell)$ equals $\vecx^*|_{\vece(\ell)}$ plus a random ``mask'' drawn either from $\CD_Y$ ($\yes$ case) or $\CD_N$ ($\no$ case). The core communication lower bound from \cite{CGSV21-boolean} is an analogue of \cite{GKK+08}'s hardness for $\bpd$ (\cref{thm:bpd-hardness}):

\begin{theorem}[{\cite[Theorem 6.2]{CGSV21-boolean}}]\label{thm:rmd-onewise-hardness}
For every $k \geq 2$ and $\CD \in \Delta(\BZ_2^k)$ such that $\vecmu(\CD)=\veczero$, there exists $\alpha_0 \in (0,1)$ such that for all $\alpha \in (0,\alpha_0), \delta \in (0,1)$, there exists $\tau > 0$ and $n_0 \in \BN$ such that for all $n \geq n_0$, any protocol for $(\Unif_{\BZ_2^k},\CD)\text{-}\rmd_\alpha(n)$ achieving advantage at least $\delta$ requires $\tau \sqrt n$ communication.
\end{theorem}

We briefly describe why the one-wise independence of $\CD$ is important for proving this lower bound. In the generalization of the Fourier-analytic reduction (\cref{lemma:bpd-fourier-reduce}), we have to consider an analogue of $h_\alpha(\ell,n)$, which we'll denote $h_{k,\alpha}(\ell,n)$, which na\"ively measures the probability that there exists $\vecs \in \BZ_2^{k\alpha n}$ such that $M^\top \vecs = \vecv$ for $\vecv \in \BZ_2^n$ of Hamming weight $\|\vecv\|_0=\ell$.\footnote{Note that unlike in our analysis in \cref{lemma:bpd-fourier-reduce}, $M$ is not folded, and $\vecs$ has length $k\alpha n$.} Unfortunately, this na\"ive event is ``too likely'' because it occurs whenever every vertex in $\supp(v)$ is touched by $M$; this has probability roughly $\alpha^\ell$, which is not small enough even for the $\ell=2$ term, which contains a factor of $n$ from \cref{lemma:low-fourier-bound} (see \cref{rem:bpd-low-ell-terms}). Fortunately, one-wise independence actually lets us restrict the set of $\vecs$ vectors we consider. In particular, we can derive the equation for the Fourier coefficients of $\Bob$'s conditional input distribution $\CZ_{A,M}$: \[ \hat{\CZ_{A,M}(\vecs)} = \frac1{2^{k\alpha n}} \E_{\vecx^*\sim \Unif_A, \vecb \sim \CD^{\alpha n}} [(-1)^{-\vecs \cdot (M\vecx^*+\vecb)}] \] (compare to the proof of \cref{lemma:bpd-fourier-reduce}). By independence, we can pull out a factor of $\E_{\vecb \sim \CD^n}[(-1)^{\vecs \cdot \vecb}] = \prod_{\ell=1}^{\alpha n} \E_{\vecb \sim \CD^n}[(-1)^{\vecs(\ell) \cdot \vecb}]$ where $\vecs = (\vecs(1),\ldots,\vecs(\alpha n))$. Suppose $\|\vecs(\ell)\|_0 = 1$ for some $\ell$; WLOG, $\vecs(1)=(1,0,\ldots,0)$, in which case $\hat{\CZ_{A,M}(\vecs)}$ is a multiple of $\E_{\vecb \sim \CD^n}[(-1)^{\vecs(1) \cdot \vecb}] = \E_{\vecb \sim \CD^n}[(-1)^{b_1}] = -\mu(\vecb)_1 = 0$. In other words, all Fourier coefficients of $\CZ_{A,M}$ which are supported on exactly one coordinate in any block vanish. Thus, we can redefine $h_{k,\alpha}(\ell,n)$ as the probability that $\vecs=(\vecs(1),\ldots,\vecs(\alpha n))$ exists satisfying \emph{both} (1) for all $\ell \in [\alpha n]$, $\vecs(\ell) = \veczero$ or $\|\vecs(\ell)\|_0 \geq 2$, \emph{and} (2) $M^\top \vecs = \vecv$. $h_{k,\alpha}(\ell,n)$ then becomes sufficiently small to carry out the proof of the lower bound.

At this point, the \cite{CGSV21-boolean} hardness result ``bifurcates'' into \cref{thm:cgsv-streaming-lb} and \cref{item:cgsv-bool-hardness} of \cref{thm:cgsv-bool-dichotomy}. On one hand, we can define a sequential $(T+1)$-player version $\dndy\seqrmd_{\alpha,T}(n)$ of $\dndy\rmd$ where --- as in $\seqbpd$ vs. $\bpd$ --- there are $T$ players $\Bob_1,\ldots,\Bob_T$, each of whom receive independent $\Bob$ inputs. Suppose we want to apply the hybrid argument we used to reduce $\seqbpd$ to $\bpd$ (\cref{lemma:bpd-to-seqbpd}, see \cref{sec:bpd-to-seqbpd}) to reduce $\seqrmd$ to $\rmd$. As we mentioned in \cref{sec:mcut-hybrid-discussion}, this requires applying the data processing inequality, which holds only when one of the distributions is the uniform distribution $\Unif_{\BZ_2^k}$. By doing so, and then using the same triangle inequality argument we used to prove \cref{lemma:seqbpd-to-seqbpd'}, we get that if $\vecmu(\CD_Y)=\vecmu(\CD_N)=\veczero$ then $\dndy\seqrmd_{\alpha,T}(n)$ also requires $\Omega(\sqrt n)$ space.

Now for a padded one-wise pair $\CD_N,\CD_Y \in \Delta(\BZ_2^k)$ with $\CD_N=\eta \CD_0+(1-\eta)\CD_N'$, $\CD_Y=\eta \CD_0+(1-\eta)\CD_Y'$, and $\vecmu(\CD_N)=\vecmu(\CD_Y)=\veczero$, we can apply the following C2S reduction from $(\CD_N',\CD_Y')\text{-}\seqrmd$ to $\mbf$: Let $m = \alpha T n$. $\Alice$ uniformly samples a list of $k$-hyperedges $\vece(0,1),\ldots,\vece(0,\eta m/(1-\eta))$ on $[n]$, and she creates the subinstance $\Psi_0$ with constraints $(\vece(0,\ell),\vecx^*|_{\vece(0,\ell)} + \vecb(0,\ell))$ where $\vecb(0,\ell) \sim \CD_0$ for $\ell \in [\eta m/(1-\eta)]$. $\Bob_t$, on input $(M_t,\vecz(t))$ where $M_t$ has edges $(\vece(t,1),\ldots,\vece(t,\alpha n))$, creates the subinstance $\Psi_t$ with constraints $(\vece(t,\ell),\vecz(t,\ell))$ for $\ell \in [\alpha n]$. Now since $\vecz(t,\ell) = \vecx^*|_{\vece(t,\ell)} + \vecb(t,\ell)$ where $\vecb(t,\ell) \sim \CD'_Y$ ($\yes$ case) or $\vecb(t,\ell) \sim \CD'_N$ ($\no$ case), the effect is that the template distributions of the $\yes$ and $\no$ instances are $\CD_Y$ ($\yes$ case) or $\CD_N$ ($\no$ case), which roughly proves \cref{thm:cgsv-streaming-lb}.

On the other hand, what happens if $\vecmu(\CD_Y)=\vecmu(\CD) \neq \veczero$? In this case, \cite{CGSV21-boolean} also proves hardness of $\dndy\rmd_{\alpha}(n)$:

\begin{theorem}[{\cite[Theorem 5.3]{CGSV21-boolean}}]\label{thm:rmd-hardness}
For every $k \geq 2$ and $\CD_N,\CD_Y \in \Delta(\BZ_2^k)$ such that $\vecmu(\CD_N)=\vecmu(\CD_Y)$, there exists $\alpha_0 \in (0,1)$ such that for all $\alpha \in (0,\alpha_0), \delta \in (0,1)$, there exists $\tau > 0$ and $n_0 \in \BN$ such that for all $n \geq n_0$, any protocol for $\dndy\rmd_\alpha(n)$ achieving advantage at least $\delta$ requires $\tau \sqrt n$ communication.
\end{theorem}

There are a few ``disadvantages'' to \cref{thm:rmd-hardness} in comparison with its special case \cref{thm:rmd-onewise-hardness}. Firstly, we can no longer use the hybrid argument to get hardness for $\dndy\seqrmd_{\alpha,T}(n)$, since the data processing inequality no longer applies. Thus, we have to settle for proving lower bounds for the \emph{parallel} randomized mask detection problem $\dndy\pllrmd_{\alpha,T}(n)$, which is a $(T+1)$-player game with the following structure: $\Bob_1,\ldots,\Bob_T$ each get an independent $\Bob$ input for $\rmd$ and send a message to a ``referee'' $\Carol$, who has to decide which case they are in; in particular, the $\Bob_t$'s cannot communicate with each other in any way, unlike in the $\seqrmd$ game. By independence, hardness for $\rmd$ (i.e., \cref{thm:rmd-hardness}) extends immediately to hardness for $\pllrmd$, but this ``parallel'' communication game can only rule out \emph{sketching} algorithms.

Moreover, \cref{thm:rmd-hardness} has a significantly more complex proof. The basic outline is the following. Let's think of $\BZ_2^k$ as a \emph{lattice}: it has a partial order, namely entrywise comparison, denoted $\leq$, where we define $0 \leq 1$. Moreover, we can define $\vecu \wedge \vecv, \vecu \vee \vecv \in \BZ_2^k$ as entrywise \textsc{and}'s and \textsc{or}'s for $\vecu,\vecv \in \BZ_2^k$.\footnote{In the context of lattices, $\wedge$ and $\vee$ are typically called the \emph{join} and \emph{meet} operations.} If $\vecu \not\leq \vecv$ and $\vecv \not\leq \vecu$, we say $\vecu$ and $\vecv$ are \emph{incomparable}, denoted $\vecu\parallel \vecv$. Now given a distribution $\CD \in \Delta(\BZ_2^k)$ supported on two incomparable elements $\vecu \parallel \vecv$, we can consider the \emph{polarized} distribution $\CD_{\vecu,\vecv}$ which, letting $\epsilon = \min\{\CD(\vecu),\CD(\vecv)\}$, decreases $\CD(\vecu)$ and $\CD(\vecv)$ by $\epsilon$, and increases $\CD(\vecu \wedge \vecv)$ and $\CD(\vecu\vee\vecv)$ by $\epsilon$. Note that polarization preserves marginals, i.e., $\vecmu(\CD)=\vecmu(\CD_{\vecu,\vecv})$. \cite{CGSV21-boolean} proves two key theorems about this operation:

\begin{enumerate}[label={\roman*.},ref={\roman*}]
    \item $(\CD,\CD_{\vecu,\vecv})\text{-}\rmd_\alpha(n)$ requires $\Omega(\sqrt n)$ communication. This is essentially because we can write $\CD$ and $\CD_{\vecu,\vecv}$ as ``mixtures'' \[ \CD = (1-2\epsilon)\CD_0 + 2\epsilon \; \Unif_{\{\vecu,\vecv\}} \text{ and } \CD_{\vecu,\vecv} = (1-2\epsilon)\CD_0 + 2\epsilon \; \Unif_{\{\vecu\wedge\vecv,\vecu\vee\vecv\}} \] for the same ``base'' distribution $\CD_0 \in \Delta(\BZ_2^k)$, which (roughly) allows us to reduce from $(\Unif_{\{\vecu,\vecv\}},\Unif_{\{\vecu\wedge\vecv,\vecu\vee\vecv\}})\text{-}\rmd_\alpha(n)$. Moreover, on the coordinates where $\vecu$ and $\vecv$ differ, $\Unif_{\{\vecu,\vecv\}}$ and $\Unif_{\{\vecu\wedge\vecv,\vecu\vee\vecv\}}$ have zero marginals; thus, we can (roughly) reduce from the ``restrictions'' to these coordinates and apply \cref{thm:rmd-onewise-hardness}.\label{item:cgsv-polar-hardness}
    \item There is some constant $C \in \BN$ such that after applying at most $C$ polarizations, $\CD$ is no longer polarizable, i.e., its support is a \emph{chain}. Moreover, this final distribution is unique for each starting marginal vector $\vecmu \in [-1,1]^k$; we call it the ``canonical distribution'', denoted $\CD_{\vecmu}$.\label{item:cgsv-polar-path}
\end{enumerate}

Together, \cref{item:cgsv-polar-hardness,item:cgsv-polar-path} suffice to prove \cref{thm:bpd-hardness}. Indeed, given two starting distributions $\CD_N,\CD_Y \in\Delta(\BZ_2^k)$ with matching marginals $\vecmu(\CD_N)=\vecmu(\CD_Y)$, we can use \cref{item:cgsv-polar-path} to repeatedly polarize both $\CD_Y$ and $\CD_N$ to produce a ``path'' of distributions of length at most $2C$ connecting them via their common canonical distribution, such that each adjacent pair of distributions is the result of polarization; \cref{item:cgsv-polar-hardness} then implies $\rmd$-hardness for each such pair; finally, we apply the triangle inequality to conclude $\rmd$-hardness for the path's endpoints, i.e., $\CD_Y$ and $\CD_N$. For more details, see \cite[\S7]{CGSV21-boolean}.

Finally, we mention that \cite{CGSV21-finite} reproves all the algorithmic and hardness results of \cite{CGSV21-boolean} which we've discussed so far in the much more general setting of $\mF$ problems defined by families of predicates over general alphabets $\BZ_q$. We briefly describe some of these results. For a distribution $\CD \in \Delta(\BZ_q^k)$, let $\vecmu(\CD) \in (\Delta(\BZ_q))^k$ denote the vector of ``marginal distributions'' resulting from projecting onto each coordinate. $\CD$ is \emph{one-wise independent} if $\vecmu(\CD) = (\Unif_{\BZ_q})^k$; $\CF$ \emph{weakly supports one-wise independence} if there exists $\CF'\subseteq \CF$ such that $\rho(\CF')=\rho(\CF)$ and for each $f \in \CF'$, there exists a one-wise independent distribution $\CD_f \in \Delta(\BZ_q^k)$ supported on $f^{-1}(1)$. In this case, \cite[Theorem 2.17]{CGSV21-finite} shows that $\mF$ is streaming approximation-resistant in $\sqrt n$ space. More generally, \cite{CGSV21-finite} proves a dichotomy theorem for $\mF$ along the lines of \cref{thm:cgsv-bool-dichotomy}, based on distributions $\CD_N,\CD_Y \in \Delta(\CF \times \BZ_q^k)$ such that projected onto each $f \in \CF$, the marginals $\vecmu$ match. Defining $\lambda,\beta,\gamma$ for these distributions is out of scope for this thesis --- see \cite[\S2.1]{CGSV21-finite} --- but we do mention that the algorithmic result now requires computing the so-called ``$(1,\infty)$-norm'' of an $n \times k$ matrix, which is the $1$-norm of the vector consisting of the largest element in each row, corresponding to greedily assigning each variable to the element in $\BZ_q$ it ``most wants to be'', while the hardness result goes through a more complex version of ``polarization''.

\section{Lower bounds for linear-space streaming}\label{sec:cgsvv}

In this section, we discuss the recent linear-space streaming lower bound of Chou, Golovnev, Sudan, Velingker, and Velusamy~\cite{CGS+22}, which extend the $\mcut$ lower bound of Kapralov and Krachun~\cite{KK19} to a large family of so-called ``wide'' CSPs. We'll begin with some statements of these lower bounds, which we'll need in order to prove lower bounds against ordering CSPs in \cref{chap:ocsps}. (Specifically, we mostly restrict to the single-predicate case $|\CF|=1$; the general formulation is given in \cite{CGS+22}.)

Fix $k,q \in \BN$, and let $\CC \subseteq [q]^k$ denote the subspace of constant vectors (i.e., vectors $(a,\ldots,a)$ for $a \in \BZ_q$). Roughly, a predicate $f : \BZ_q^k \to \{0,1\}$ is ``wide'' if it is support has a large intersection with a coset of $\CC$ in $\BZ_q^k$. To be precise, for $\vecb \in [q]^k$, let
\begin{equation}\label{eqn:omega}
    \omega_\vecb(f) \eqdef \E_{\vecc \sim \Unif_\CC}[f(\vecc+\vecb)],
\end{equation}
and define the \emph{width} of $f$ by \[ \omega(f) \eqdef \max_{\vecb \in \BZ_q^k} \left(\omega_\vecb(f)\right). \]

The lower bounds for wide predicates in \cite{CGS+22} are based on the following communication problem, called \emph{sequential implicit randomized shift detection ($\seqirsd$)}:\footnote{\cite{CGS+22} defines more generally a ``(sequential) implicit randomized \emph{mask} detection'' problem, but proves hardness only when the masks are uniform shifts (i.e., uniform elements of $\CC$).}

\begin{definition}[$\seqirsd$]\label{def:seqirsd}
Let $2\leq q,k \in \BN$, $\alpha \in (0,1)$, and $T, n \in \BN$. Then $\kq\seqirsd_{\alpha,T}(n)$ is the following $T$-player one-way communication problem with players $\Bob_1,\ldots,\Bob_T$:

\begin{itemize}
    \item Sample $\vecx^* \sim \Unif_{\BZ_q^n}$.
    \item Each $\Bob_t$ receives an adjacency matrix $M_t\in\{0,1\}^{k\alpha n \times n}$ sampled from $\Matchings_{k,\alpha}(n)$, and a vector $\vecz(t) \in \BZ_q^{k\alpha n}$ labelling each edge of $M_t$ as follows:
    \begin{itemize}[nosep]
        \item $\yes$ case: $\vecz(t) = M_t \vecx^* + \vecb(t)$, where $\vecb(t) \sim \Unif_C^{\alpha n}$.
        \item $\no$ case: $\vecz(t) \sim (\Unif_{\BZ_q}^k)^{ \alpha n}$.
    \end{itemize}
    \item Each $\Bob_t$ can send a message to $\Bob_{t+1}$, and at the end, $\Bob_T$ must decide whether they are in the $\yes$ or $\no$ case.
\end{itemize}
\end{definition}

To provide some interpretation for this definition, for each player $\Bob_t$ and edge-index $\ell \in [\alpha n]$, let $\vece(t,\ell)$ denote the $\ell$-th edge in $M_t$ and write $\vecz(t) = (\vecz(t,1),\ldots,\vecz(t,\alpha n))$ for $\vecz(t,\ell) \in \BZ_q^k$. In the $\yes$ case, each block $\vecz(t,\ell)$ equals $\vecx^*|_{\vece(t,\ell)}$ plus a random shift (i.e., a random element of $\CC$); in the $\no$ case, each block $\vecz(t,\ell)$ is uniformly random. Note also that this problem is ``implicit'', like the $\seqibpd$ problem we described in \cref{sec:mcut-linear-space}, in the sense that there is no $\Alice$ who knows the hidden assignment.

Moreover, consider the case $k=q=2$. For $t \in [T],\ell\in[\alpha n]$, if $\vece(t,\ell)=(u,v)$, then in the $\yes$ case $\Bob_t$'s $\ell$-th block $\vecz(t,i) = (x^*_u+b(t)_\ell,x^*+b(t)_\ell)$ where $b(t)_\ell \sim \Bern_{\frac12}$; hence, $\vecz(t,\ell)$ is information-theoretically equivalent to the bit $x^*_u+x^*_v$. On the other hand, in the $\no$ case, $\vecz(t,i)$ is simply uniformly random. Thus, in the $k=q=2$ case $\kq\seqirsd_{\alpha,T}(n)$ is equivalent to the $\seqibpd_{\alpha,T}(n)$ problem which we described in \cref{sec:mcut-discussion}, used in \cite{KK19} to prove linear-space hardness of approximating $\mcut$.

The technical core of the lower bounds in \cite{CGS+22} is the following hardness result for $\seqirsd$:

\begin{theorem}[{\cite[Theorem 3.2]{CGS+22}}]\label{thm:seqirsd-hardness}
For every $2 \leq q,k \in \BN$, there exists $\alpha_0 \in (0,1)$ such that for every $\delta \in (0,1)$, $\alpha \in (0,\alpha_0)$, $T \in \BN$, there exist $\tau > 0$ and $n_0 \in \BN$, such that for all $n \geq n_0$, any protocol for $\kq\seqirsd_{\alpha,T}(n)$ achieving advantage at least $\delta$ requires $\tau n$ communication.
\end{theorem}

The following construction and analysis generalize \cref{cons:seqbpd-to-mcut} and \cref{lemma:seqbpd-to-mcut-analysis} for $\mcut$, respectively:

\begin{construction}[C2S reduction from $\seqirsd$ to $\mf$]
\label{cons:seqirsd}
Let $\vecb \in \BZ_q^k$ and $f : \BZ_q^k\to\{0,1\}$. For each $t \in [T]$, $\Bob_t$'s reduction function $\R_t$ outputs an instance $\Psi_t$ as follows: For each $\vece(t,\ell) \in M_t$ and corresponding block $\vecz(t,\ell)\in\BZ_q^k$ of $\vecz(t)$, $\Bob_t$ adds $\vece(t,\ell)$ to $\Psi_t$ iff $\vecz(t,\ell)-\vecb \in \CC$.
\end{construction}

\begin{lemma}\label{lemma:seqirsd-analysis}
For all $f : \BZ_q^k \to \{0,1\}$, $\alpha \in (0,1)$, $\epsilon \in (0,\frac12)$, and $\vecb \in \BZ_q^k$, there exist $T, n_0 \in \BN$ such that for every $n \geq n_0$, the following holds. Let $\CY$ and $\CN$ denote the $\yes$ and $\no$ distributions for $\kq\seqirsd_{\alpha,T}(n)$, and let $(\R_0,\ldots,\R_T)$ be the reduction functions from \cref{cons:seqirsd}. Then \[ \Pr_{\Psi \sim (\R_0,\ldots,\R_T) \circ \CY}\left[\val_\Psi \leq \omega_\vecb(f)-\epsilon \right]\leq\exp(-n) \text{ and } \Pr_{\Psi \sim (\R_0,\ldots,\R_T) \circ \CN}\left[\val_\Psi \geq \rho(f) + \epsilon\right]\leq \exp(-n). \]
\end{lemma}

Together, \cref{thm:seqirsd-hardness,lemma:seqirsd-analysis} give the following corollary (which, without too much extra work, can be generalized to all families of predicates):

\begin{corollary}[{\cite[Theorem 4.3]{CGS+22}}]\label{thm:cgsvv}
For every $f : \BZ_q^k \to \{0,1\}$ and constant $\epsilon > 0$, any streaming algorithm which $(\frac{\omega(f)}{\rho(f)}+\epsilon)$-approximates $\mf$ requires $\Omega(n)$ space. Moreover, for every \emph{family} of predicates $\CF$, streaming algorithms which $(\frac{\omega(\CF)}{\rho(\CF)}+\epsilon)$-approximate $\mF$ require $\Omega(n)$ space, where $\omega(\CF) \eqdef \min_{f\in\CF} \omega(f)$.
\end{corollary}

Noting that by definition $\omega(f) \geq \frac1q$ for every $f : \BZ_q^k \to \{0,1\}$, we have a further corollary which narrows the linear-space streaming approximability of every predicate family $\CF$ to the interval $[\rho(\CF),q\cdot \rho(\CF)]$:

\begin{corollary}[{\cite[Theorem 1.2]{CGS+22}}]\label{cor:cgsvv-2rho}
For every family of predicates $\CF$ over $\BZ_q$ and every $\epsilon > 0$, every streaming algorithm which $(q\cdot\rho(\CF)+\epsilon)$-approximates $\mF$ uses at least $\Omega(n)$ space.
\end{corollary}

Finally, we remark that qualitatively, the main obstacle involved in proving linear-space lower bounds (i.e., \cref{thm:seqirsd-hardness}, or its special case \cref{thm:seqibpd-hardness} for $\mcut$) is in ``improving the low-Fourier weight bounds to $\left(\frac{\zeta c}\ell\right)^{\ell/2}$'' in comparison to \cref{lemma:low-fourier-bound}. (See \cref{rem:bpd-low-ell-terms} for a discussion in the special case of $\seqibpd$ for $\mcut$.) Though \cref{lemma:low-fourier-bound} itself is tight, if we consider the distribution $\CZ_{A_t,M_t}$ of $\Bob_t$'s second input, where $M_t$ is $\Bob_t$'s matching and $A_t$ the set of $\vecx^*$'s consistent with $\Bob_{t-1}$'s output message, and directly apply \cref{lemma:low-fourier-bound}, we will immediately disqualify ourselves from proving linear-space lower bounds. The key observation is that we can do better than directly applying \cref{lemma:low-fourier-bound} because $\CZ_{A_t,M_t}$ is a ``structured'' distribution, in the sense that when we draw $\vecz(t) \sim \CZ_{A_t,M_t}$, each entry of $\vecz(t)$ only tells us about \emph{sums} of entries of $\vecx^*$. The proof ultimately does succeed by showing inductively that the indicators of the sets $A_t$ satisfy certain Fourier weight bounds (with high probability) by applying \cref{lemma:low-fourier-bound} to a carefully defined ``reduced'' version of $\CZ_{A_t,M_t}$. However, exactly stating these bounds, and formulating the right inductive hypothesis with which to prove them, is much more involved than in the $\bpd$ case; see \cite[\S5]{CGS+22} for details.

\part{Contributions}\label{part:contributions}

\newcommand{\cmas}{\Pi_\mas^{\coarsen q}}
\newcommand{\cPi}{\Pi^{\coarsen q}}
\newcommand{\Piq}{(\Pi,\vecb,q)\text{-}}

\chapter{Ordering constraint satisfaction problems}\label{chap:ocsps}

\epigraph{A natural direction would be to pose the $\mas$ [problem] as a CSP. $\mas$ is fairly similar to a CSP, with each vertex being a variable taking values in domain $[n]$ and each directed edge a constraint between two variables. However, the domain, $[n]$, of the CSP is not fixed but grows with input size. We stress here that this is not a superficial distinction but an essential characteristic of the problem.}{Guruswami, H{\aa}stad, Raghavendra, Manokaran, and Charikar~\cite{GHM+11}}

\newthought{Scheduling problems can be modeled} as \emph{ordering constraint satisfaction problems (OCSPs)}, variants of CSPs in which assignments correspond to orderings of $n$ objects, and constraints to allowed orderings for small sets of objects. That is, in the scheduling interpretation, the goal of an OCSP is to find the best ``schedule'' for $n$ ``tasks'' given a list of \emph{precedence} constraints such as ``task $j$ must come between task $i$ and task $k$''. In this chapter, we prove a strong streaming approximation-resistance result for every OCSP from our joint work with Sudan and Velusamy~\cite{SSV21}: For every OCSP, linear-space streaming algorithms cannot perform better than the trivial approximation ratio (see \cref{thm:ocsp-hardness} below). But we begin by formally defining OCSPs and two specific examples, the \emph{maximum acyclic subgraph} and \emph{maximum betweenness} problems.

\section{OCSPs: Definitions, motivations, history}\label{sec:ocsps}

A vector $\vecsigma = (\sigma_1,\ldots,\sigma_n) \in [n]^n$ is a \emph{permutation} if all its elements are distinct, i.e., if $\sigma_i = \sigma_{i'}$ then $i=i'$. Let $\sym_n \subseteq [n]^n$ denote the set of all permutations.\footnote{We use this non-standard ``vector notation'' for permutations to emphasize the analogy with CSP assignments, which come from $[q]^n$.} We interpret a permutation $\vecsigma \in \sym_n$ as an ordering on $n$ objects, labeled $1$ through $n$, which places the $i$-th object in position $\sigma_i$.\footnote{\emph{I.e.}, this is an interpretative convention; the other would be that the $\sigma_i$-th object is in position $i$.}

Let $\veca = (a_1,\ldots,a_k) \in \BZ^k$ be a $k$-tuple of integers. We define a symbol $\ord(\veca) \in \sym_k \cup \{\bot\}$ which captures the ordering of the entries of $\veca$, in the following way: If $\veca$'s entries are not all distinct, then $\ord(\veca) = \bot$; otherwise, $\ord(\veca)$ is the unique permutation $\vecpi \in \sym_k$ such that $a_{\pi^{-1}_1} < \cdots < a_{\pi^{-1}_k}$, where $\vecpi^{-1}$ is the inverse permutation to $\vecpi$. In particular, if $\veca \in \sym_k$ is a permutation then $\ord(\veca) = \veca$.

For a permutation $\vecsigma \in \sym_n$ and a $k$-tuple of distinct indices $\vecj = (j_1,\ldots,j_k) \in [n]^k$, we define the \emph{induced permutation} $\vecsigma|_\vecj \eqdef \ord(\sigma_{j_1},\ldots,\sigma_{j_k})$. Thus, for instance, the permutation $\vecsigma = (5,1,3,2,4) \in \sym_5$ places the third object in the third position, and since $\vecsigma|_{(1,5,3)} = (3,2,1) \in \sym_3$, we see that if we restrict to only the first, third, and fifth objects, the fifth is in the second position.

For $k \leq 2 \in \BN$, an \emph{ordering predicate} $\Pi : \sym_k \to \{0,1\}$ defines the \emph{ordering constraint satisfaction problem} $\mPi$ as follows. A \emph{constraint} on $n \in \BN$ variables is given by a $k$-tuple $\vecj = (j_1,\ldots,j_k)\in[n]^k$ of distinct indices.\footnote{For simplicity, and since we're proving lower bounds, we define only unweighted instances of OCSPs.} An assignment $\vecsigma \in \sym_n$ \emph{satisfies} $C=\vecj$ iff $\Pi(\vecsigma|_\vecj)=1$. An \emph{instance} $\Phi$ of $\mPi$ consists of $m$ constraints $(C_\ell = (\vecj(\ell))_{\ell\in[m]}$, and the \emph{value} of an assignment $\vecsigma \in \sym_n$ on $\Phi$, denoted $\ordval_\Phi(\vecsigma)$, is the (fractional) weight of constraints satisfied by $\vecsigma$, i.e., \[ \ordval_\Phi(\vecsigma) \eqdef \frac1m \sum_{\ell=1}^m \Pi(\vecsigma|_{\vecj(\ell)}). \] The \emph{value} of $\Phi$, denoted $\ordval_\Phi$, is the maximum value of any assignment, i.e., \[ \ordval_\Phi \eqdef \max_{\vecsigma \in \sym_n}\left( \ordval_\Phi(\vecsigma)\right). \] We consider, in the streaming setting, the problem of $\alpha$-approximating $\mPi$, as well as the distinguishing problem $\bgd\mPi$; the setup is the same as in the standard, non-ordering CSP case (see \cref{sec:csps,sec:streaming}). One important difference between OCSPs and (non-ordering) CSPs is that the solution space of an OCSP has super-exponential size $|\sym_n| = n! \geq (n/e)^n$, while a CSP over $\BZ_q$ has exponential solution space size $|\BZ_q^n|=q^n$.

The \emph{maximum acylic subgraph} problem ($\mas$) is the prototypical ordering CSP. $\mas$ is the problem $\mocsp[\Pi_\mas]$ for the predicate $\Pi_\mas : \sym_2\to\{0,1\}$ which is supported on $(1,2)$. Thus, an $\mas$ constraint $(u,v)$ is satisfied by an ordering $\vecsigma \in \sym_n$ iff $\sigma_u < \sigma_v$.\footnote{Like in the standard CSP case, an instance $\Phi$ of $\mas$ corresponds to a directed constraint graph $G(\Psi)$, where each constraint $(u,v)$ corresponds to a directed edge. Any ordering $\vecsigma \in \sym_n$ induces an acyclic subgraph of $G(\Psi)$ consisting of all the \emph{forward} edges with respect to $\vecsigma$, i.e., those such that $\sigma_u < \sigma_v$. Thus, $\val_\Psi$ corresponds to measuring the size of the largest acyclic subgraph in $G(\Psi)$, justifying the name ``maximum acyclic subgraph''.} In the scheduling interpretation, a constraint $(u,v)$ is satisfied by a schedule $\vecsigma \in \sym_n$ iff $u$ is scheduled earlier than $v$ in $\vecsigma$. Karp's classic enumeration of 21 $\NP$-complete problems~\cite{Kar72} includes the problem of, given an instance $\Phi$ of $\mas$ and $\gamma \in [0,1]$, deciding whether $\ordval_\Phi \geq \gamma$.\footnote{Also, in the classical setting, depth-first search can be used to decide whether $\val_\Psi=1$, i.e., to test whether $G(\Psi)$ is \emph{acyclic}. In the streaming setting, however, acylicity testing is known to take $\Theta(n^2)$ space~\cite{CGMV20}.} Several works \cite{New00,AMW15,BK19} have studied the $\NP$-hardness of \emph{approximating} $\mas$; \cite{BK19} shows that $(\frac23+\epsilon)$-approximating $\mas$ is $\NP$-hard for every $\epsilon > 0$.

Another ordering CSP of interest is the \emph{maximum betweenness} problem $\mbtwn \eqdef \mocsp[\Pi_\Btwn]$ where $\Pi_\Btwn : \sym_3 \to \{0,1\}$ is supported on $(1,2,3)$ and $(3,2,1)$. Thus, a $\mbtwn$ constraint $(u,v,w)$ is satisfied by an ordering $\vecsigma \in \sym_n$ iff $\sigma_u < \sigma_v < \sigma_w$ or $\sigma_w < \sigma_v < \sigma_u$. This OCSP was introduced by Opatrny~\cite{Opa79}, who showed that even deciding whether $\ordval_\Phi = 1$ is $\NP$-hard. $\mbtwn$'s $\NP$-hardness of approximation has been studied in \cite{CS98,AMW15}; the latter work shows that $(\frac12+\epsilon)$-approximating $\mbtwn$ is $\NP$-hard for every $\epsilon > 0$.

However, in another analogy with the CSP case, for every $\Pi : \sym_k \to \{0,1\}$, defining $\rho(\Pi) \eqdef \E_{\vecpi\sim\Unif_{\sym_k}}[\Pi(\vecpi)]$, \emph{every} instance of $\mPi$ has value at least $\rho(\Pi)$, and thus every $\mPi$ is trivially $\rho(\Pi)$-approximable. Again, we consider, for various predicates $\Pi$ and classes of algorithms $\CS$, whether $\mPi$ is \emph{approximation-resistant} (for every $\epsilon>0$, no algorithm in $\CS$ can $(\rho(\Pi)+\epsilon)$-approximate $\mPi$) or \emph{nontrivially approximable} (there is some $(\rho(\Pi)+\epsilon)$-approximation). Note that $\rho(\Pi_{\mas}) = \frac12$ and $\rho(\Pi_{\Btwn}) = \frac13$; thus, the results of \cite{AMW15,BK19} are not strong enough to show that it is $\NP$-hard to nontrivially approximate $\mas$ or $\mbtwn$. However, Guruswami, H{\aa}stad, Manokaran, Raghavendra, and Charikar~\cite{GHM+11} showed that it is \emph{unique games}-hard to $(\rho(\Pi)+\epsilon)$-approximate $\mPi$, for every $\Pi : \sym_k \to \{0,1\}$ and $\epsilon > 0$.

The result of our work~\cite{SSV21} is that $\mPi$ is also approximation-resistant in the streaming setting, even to \emph{linear-space} algorithms:

\begin{theorem}[\cite{SSV21}]\label{thm:ocsp-hardness}
For every $k \leq 2 \in \BN$, predicate $\Pi : \sym_k \to \{0,1\}$, and $\epsilon > 0$, there exists $\tau > 0$ such that every streaming algorithm which $(\rho(\Pi)+\epsilon)$-approximates $\mPi$ uses at least $\tau n$ space.
\end{theorem}

The space bound in this theorem is optimal up to logarithmic factors; indeed, just as in the case of CSPs, $(1-\epsilon)$-approximations are possible for every $\epsilon > 0$ in $\tilde{O}(n)$ space (see \cref{rem:sparsifier} above)!

In the next section, we begin by giving some intuition for this theorem, and highlighting surprising similarities with the classical proof of \cite{GHM+11}.

\section{Proof outline for $\mas$}

The following observation is due to Guruswami \emph{et al.}~\cite{GHM+11}. Since $\sym_n \subseteq [n]^n$, we can view $\mas$ as a CSP over the alphabet $[n]$, with predicate $f_\mas : [n]^2 \to \{0,1\}$ given by $f_\mas(b_1,b_2)=\1_{b_1<b_2}$. The hope is then to analyze this predicate using machinery for CSPs. Unfortunately, this predicate does not actually define a CSP, since the alphabet size $n$ is non-constant. We can, however, attempt to salvage this strategy by decreasing the alphabet size to a large constant and ``seeing what happens''. To be precise, for $q \in \BN$, let $\iota_q : \BZ_q \to [q]$ denote the map taking elements of $\BZ_q$ to their representatives in $\{1,\ldots,q\}$. We define the predicate $\cmas : \BZ_q^2 \to \{0,1\}$ by $\cmas(b_1,b_2)=\1_{\iota_q(b_1)<\iota_q(b_2)}$,\footnote{Note that the comparison operator is not \emph{a priori} defined for $\BZ_q$, only for $\BZ$. Under the convention we just picked, $\iota(0) = q$, and in particular, $\iota(0) > \iota(1)$. This choice is for intercompatibility with the notation $[n] = \{1,\ldots,n\}$, and can be safely ignored.} and consider the problem $\mcsp[\cmas]$.

In the previous section, we observed that $\mas$ can be interpreted as a scheduling problem in which a constraint $(u,v)$ is satisfied iff $u$'s position is earlier than $v$'s position. Under this view, $\mcsp[\cmas]$ is a ``batched'' scheduling problem, where the goal is to assign $n$ tasks to $q = O(1)$ batches, and a constraint $(u,v)$ is satisfied iff $u$'s batch is earlier than $v$'s batch. Thus, $\mcsp[\cmas]$ is a \emph{coarser} version of $\mas$, because in $\mas$ we have flexibility in assigning execution orders even \emph{within} the same batch.

To make this precise, for every instance $\Phi$ of $\mas$ and $q \in \BN$, let $\Phi^{\coarsen q}$ denote the instance of $\mcsp[\cmas]$ with the exact same list of constraints. Conversely, given an instance $\Psi$ of $\mcsp[\cmas]$, let $\Psi^\refine$ denote the instance of $\mas$ with the same list of constraints. The operations $\coarsen q$ (\emph{$q$-coarsening}) and $\refine$ (\emph{refinement}) are inverses. Also, for an assignment $\vecx = (x_1,\ldots,x_n) \in\BZ_q^n$ to $\mcsp[\cmas]$ and $b \in \BZ_q$, let $\vecx^{-1}(b) \eqdef \{i\in[n]:x_i=b\}$. (If $\vecx$ is a batched schedule, then $\vecx^{-1}(b)$ is the set of jobs in batch $b$.) Then we have the following:

\begin{claim}\label{claim:mas-refinement}
For every $q \in \BN$ and instance $\Psi$ of $\mcsp[\cmas]$, \[ \val_\Psi \leq \ordval_{\Psi^\refine}. \]
\end{claim}

\begin{proof}
For $\vecx\in\BZ_q^n$, let $s_b = |\vecx^{-1}(b)|$ for each $b \in \BZ_q$. We construct an ordering $\vecx^\refine \in \sym_n$ by assigning $\vecx^{-1}(1)$ the first $s_0$ positions (in some arbitrary order), and then iteratively assigning $\vecx^{-1}(b)$ the next $s_b$ positions for $b \in \{2,\ldots,q\}$. Then if $\iota_q(x_u) < \iota_q(x_v)$, $\vecx^\refine_u < \vecx^\refine_v$ by construction. Thus, $\val_\Psi(\vecx) \leq \ordval_{\Psi^\refine}(\vecx^\refine)$, and \cref{claim:mas-refinement} follows.
\end{proof}

Now how can this notion of coarsening help us prove that $\mas$ is streaming approximation-resistant? Recall the linear-space streaming lower bounds of Chou \emph{et al.}~\cite{CGS+22} which rule out $(\omega(\CF)/\rho(\CF)+\epsilon)$-approximations for $\mF$ (see \cref{thm:cgsvv} and the definitions in \cref{sec:cgsvv} above). Fix a large constant $q \in \BN$. Then the trivial approximation threshold for $\mcsp[\cmas]$ is \[ \rho(\cmas) = \Pr_{(b_1,b_2) \sim \Unif_{\BZ_q^2}}[\iota_q(b_1) < \iota_q(b_2)] = \frac{q(q-1)}{2q^2} \approx \frac12. \] On the other hand, for $\vecb=(1,2) \in \BZ_q^2$, we have \[ \omega_\vecb(\cmas) = \Pr_{c \sim \Unif_{\BZ_q}}[\iota_q(c) < \iota_q(c+1)] = \frac{q-1}q \approx 1, \] since the only way that $\iota_q(c) \not< \iota_q(c+1)$ is if $c=0$. Thus, $\mcsp[\cmas]$ is almost approximation-resistant to linear-space streaming algorithms! Indeed, \cref{thm:seqirsd-hardness,cons:seqirsd,lemma:seqirsd-analysis} together give us a pair of distributions $\CY'$ and $\CN'$ over instances of $\mcsp[\cmas]$ which (1) require $\Omega(n)$ space to distinguish and (2) have values close to $\omega_{\vecb}(\cmas) \approx 1$ and $\rho(\cmas)\approx \frac12$, respectively, with high probability. 

Now to get back to $\mas$, \cref{claim:mas-refinement} shows that for $\Psi \sim \CY'$, $\ordval_{\Psi^\refine} \geq \val_\Psi$, and thus $\ordval_{\Psi^\refine} \approx 1$ with high probability. To show that $\mas$ is approximation-resistant, therefore, it suffices to show that for $\Psi \sim \CN'$, $\ordval_{\Psi^\refine} \approx \frac12$ --- i.e., the inequality in \cref{claim:mas-refinement} is not too loose --- with high probability. To do this, we need to actually look at the structure of $\CN'$. Recall, $\CN'$ is defined by composing \cref{cons:seqirsd} with the $\no$-distribution of $\seqirsd_{T,\alpha}$.\footnote{A natural idea, following the $\no$-case analysis for e.g. $\mcut$ (see \cref{sec:bpd}), is to show using concentration bounds that for every fixed ordering $\vecsigma \in \sym_n$, $\Pr_{\Psi \sim \CN'}[\ordval_{\Psi^\refine}(\vecsigma) > \frac12+\epsilon] \leq \exp(-n)$, and then take a union bound over $\vecsigma$. However, $|\sym_n| = n!$ grows faster than $\exp(n)$, so the union bound fails.} We make the following (informal) claim:

\begin{claim}\label{claim:mas-coarsening}
For fixed $\epsilon > 0$ and sufficiently large choice of $q$, \[ \Pr_{\Psi \sim \CN'}[\ordval_{\Psi^{\refine}} > \val_\Psi + \epsilon] \leq \exp(-n). \]
\end{claim}

\cref{claim:mas-coarsening} is sufficient to complete the proof, since we know that $\val_\Psi \approx \frac12$ with high probability over $\Psi \sim \CN'$. \cref{claim:mas-coarsening} is stated and proven formally for general OCSPs below as \cref{lemma:ocsp-coarsening}, but for now, we give a proof sketch. For any $\vecsigma \in \sym_n$, define an assignment $\vecsigma^{\coarsen q} \in \BZ_q^n$ by $\sigma_i^{\coarsen q} = \lceil q\sigma_i/n \rceil$. In the scheduling interpretation, $\vecsigma^{\coarsen q}$'s first batch contains the first $\approx n/q$ tasks scheduled by $\vecsigma$, the second batch the next $\approx n/q$, etc. Then we have:

\begin{proof}[Proof sketch of \cref{claim:mas-coarsening}]
It suffices to show that with probability $1-\exp(-n)$ over the choice of $\Psi$, for every $\vecsigma \in \sym_n$, $\ordval_{\Psi^{\refine}} < \val_{\Psi}+\epsilon$. A constraint $(u,v)$ is satisfied by $\vecsigma$ (in $\Psi^\refine$) but not by $\vecsigma^{\coarsen q}$ (in $\Psi$) iff $\sigma_u < \sigma_v$ but $\sigma^{\coarsen q}_u = \sigma^{\coarsen q}_v$. Thus, it suffices to upper-bound, for every partition of $[n]$ into subsets of size $\leq q$, the fraction of constraints for which both variables are in the same subset.

Looking at the definition of $\seqirsd$ (\cref{def:seqirsd}) and the reduction which produces $\CN'$ (\cref{cons:seqirsd}), we see that the constraints of a $\mcsp[\cmas]$ instances drawn from $\CN'$ correspond to a random graph in a particular model: we sample a union of random matchings and then subsample each edge independently with probability $q^{-1}$. This graph is, with high probability, a ``small set expander'' in the following sense: For every subset $S \subseteq [n]$ of size at most $q$, at most $O(q^{-2})$ fraction of the edges lie entirely within $S$. (This fact can be proven using concentration inequalities, although the subsampling makes the calculations a bit messy --- just like in the analysis of $\mcut$ (see \cref{sec:bpd}).) This small set expansion implies another property, which we'll call ``balanced partition expansion'': In any partition of $[n]$ into subsets $S_1,\ldots,S_t$ of size at most $q$, at most $O(q^{-1})$ fraction of the edges do not cross between two distinct subsets. Thus, setting $q = \Omega(\epsilon^{-1})$ gives the desired bound.
\end{proof}

We formalize exactly what we mean by small set expansion and balanced partition expansion (including for hypergraphs) below; see \cref{sec:ocsp-coarsening}.

In summary, we can show that $\mas$ is streaming approximation-resistant by (1) defining an appropriate predicate $\cmas$ over $\BZ_q$ which is ``coarser'' than $\mas$ (in the sense of \cref{claim:mas-refinement}), (2) showing the corresponding problem $\mcsp[\cmas]$ is almost approximation-resistant using pre-existing tools for CSPs (i.e., the results of \cite{CGS+22}), and (3) ensuring that the values of the $\no$ instances for $\mcsp[\cmas]$ do not increase ``too much'' when refined into instances of $\mas$ (\cref{claim:mas-coarsening}).\footnote{The same ``coarsening'' construction arises in the classical $\ug$-hardness proof of Guruswami \emph{et al.}~\cite{GHM+11}, but its usage is significantly more sophisticated. In particular, the \cite{GHM+11} proof follows the typical $\ug$-hardness paradigm which constructs so-called \emph{dictatorship tests} using the predicate at hand (in this case, $\mas$); the construction and soundness analysis of these tests is based on Raghavendra's CSP dichotomy theorem~\cite{Rag08} applied to the coarsened CSP $\mcsp[\cmas]$. For an introduction to this approach to classical $\ug$-hardness, see the surveys of Khot~\cite{Kho10} and Trevisan~\cite{Tre12} as well as \cite[\S7 and \S11.7]{OD14}.} In the next subsection, we carry out this proof formally and generalize it to OCSPs.

\section{Linear-space approximation resistance of all OCSPs: Proving \cref{thm:ocsp-hardness}}

The goal of this section is to prove \cref{thm:ocsp-hardness}, which states that every OCSP is approximation-resistant to linear-space streaming algorithms. We begin by generalizing the definitions of ``coarsening'' and ``refinement'' from the previous subsection to all OCSPs. Consider an arbitrary ordering predicate $\Pi : \sym_k \to \{0,1\}$ and an alphabet size $q \in \BN$. We define the coarse predicate $\cPi : \BZ_q^k \to \{0,1\}$ by $\cPi(a_1,\ldots,a_k) = \Pi(\ord(\iota_q(a_1),\ldots,\iota_q(a_k)))$.\footnote{$\cPi(a_1,\ldots,a_k)= 0$ in the case $\ord(\iota_q(a_1),\ldots,\iota_q(a_k))=\bot$, i.e., $a_1,\ldots,a_k$ are not all distinct.} For every instance $\Phi$ of $\mPi$ and $q \in \BN$, we let $\Phi^{\coarsen q}$ denote the instance of $\mcsp[\cPi]$ with the same constraints, and given an instance $\Psi$ of $\mcsp[\cPi]$, we let $\Psi^\refine$ denote the instance of $\mas$ with the same constraints. Then we have the following analogue of \cref{claim:mas-refinement}:

\begin{lemma}\label{lemma:ocsp-refinement}
For every $k,q \leq 2 \in \BN$, $\Pi : \sym_k \to \{0,1\}$, and instance $\Psi$ of $\mcsp[\cPi]$, \[ \val_\Psi \leq \ordval_{\Psi^\refine}. \]
\end{lemma}

\begin{proof}
For $\vecx\in\BZ_q^n$, we construct $\vecx^\refine \in \sym_n$ as in the proof of \cref{claim:mas-refinement}, which has the ``monotonicity'' property that $\iota_q(x_u) < \iota_q(x_v) \implies x^\refine_u < x^\refine_v$. Consider some constraint $\vecj = (j_1,\ldots,j_k)$. If $\vecx$ satisfies $\vecj$ (as a $\mcsp[\cPi]$ constraint), then $\Pi(\ord(\iota_q(x_{j_1}),\ldots,\iota_q(x_{j_k})))=1$. By monotonicity, $\ord(\iota_q(x_{j_1}),\ldots,\iota_q(x_{j_k})) = \vecx^\refine|_\vecj$, and so $\vecx^\refine$ satisfies $\vecj$ (as a $\mPi$ constraint). Thus, $\val_\Psi(\vecx) \leq \ordval_{\Psi^\refine}(\vecx^\refine)$.
\end{proof}

Now, we ``import'' the linear-space hardness results of \cite{CGS+22} we need from \cref{sec:cgsvv}. Specifically, for $\Pi : \sym_k \to \{0,1\}$ and $k \leq q \in \BN$, we define a hard pair of distributions over $\mcsp[\cPi]$ instances which, under refinement ($\refine$), will become hard instances for $\mPi$. Fix $\vecb \in \supp(\Pi)$, and let $\R_1,\ldots,\R_T$ denote the reduction functions from \cref{cons:seqirsd} (using $\vecb$ and $\mcsp[\cPi]$). For $\alpha \in (0,1)$ and $T,n\in\BN$, let $\kq\CY_{\alpha,T}(n)$ and $\kq\CN_{\alpha,T}(n)$ denote the $\yes$ and $\no$ distributions for $\kq\seqirsd_{\alpha,T}(n)$, respectively. Then we let $\Piq\CY'_{\alpha,T}(n)\eqdef (\R_1,\ldots,\R_T) \circ \kq\CY_{\alpha,T}(n)$ and $\Piq\CN'_{\alpha,T}(n)\eqdef (\R_1,\ldots,\R_T) \circ \kq\CN_{\alpha,T}(n)$. These are distributions over $\mcsp[\cPi]$ instances, which are indistinguishable to linear-space algorithms (\cref{thm:cgsvv}), and which have values close to $\omega_\vecb(\cPi)$ and $\rho(\cPi)$ with high probability for large enough $T$ and $n$, respectively (\cref{lemma:seqirsd-analysis}).

Now we have
\begin{equation}\label{eqn:ocsp-omega}
\omega_{\vecb}(\cPi) \geq 1-\frac{k-1}q
\end{equation}
since by definition $\omega_{\vecb}(\cPi) = \E_{\vecc \sim \CC}[\cPi(\vecb+\vecc)]$ (see \cref{eqn:omega}), $\cPi(\ord(\vecb))=1$ by definition, and $\ord(\vecb+\vecc)=\ord(\vecb)$ unless $c \geq q-(k-1)$ (where $\vecc=(c,\ldots,c)$). On the other hand,
\begin{equation}\label{eqn:ocsp-rho}
    \rho(\cPi) = \frac{q!}{(q-k)!q^k}\rho(\Pi) \leq \rho(\Pi),
\end{equation}
since by definition $\rho(\cPi) = \E_{\veca\sim\Unif_{\BZ_q^k}}[\cPi(\veca)]$; $\ord(\veca)\neq\bot$ (i.e., $\veca$'s entries are all distinct) with probability $\frac{q\cdots (q-(k-1))}{q^k} = \frac{q!}{(q-k)!q^k} \leq 1$, and if $\veca$'s entries are all distinct, it satisfies $\cPi$ with probability $\rho(\Pi)$. We also claim the following lemma, which generalizes \cref{claim:mas-coarsening}:

\begin{lemma}\label{lemma:ocsp-coarsening}
For every $\Pi : \sym_k \to \{0,1\}$, $\vecb \in \supp(\Pi)$, and $\epsilon > 0$, there exists $q_0 \in \BN$ and $\alpha_0 > 0$ such that for all $q \geq q_0 \in \BN$ and $\alpha \in (0,\alpha_0)$, there exists $T_0 \in \BN$ such that for all $T \geq T_0 \in \BN$ and $n \in \BN$, \[ \Pr_{\Psi \sim \Piq\CN'_{\alpha,T}(n)}[\ordval_{\Psi^\refine} \geq \val_\Psi + \epsilon] \leq \exp(-n). \]
\end{lemma}

Modulo the proof of \cref{lemma:ocsp-coarsening}, we can now prove \cref{thm:ocsp-hardness}:

\begin{proof}[Proof of \cref{thm:ocsp-hardness}]
Let $\epsilon' = \epsilon/4$. Pick $q_0 \in \BN$ such that \cref{lemma:ocsp-coarsening} holds with error probability $\epsilon'$, and let $q = \max\{q_0,(k-1)/\epsilon'\}$. Now let $\alpha$ be the smaller of the $\alpha_0$'s from \cref{thm:cgsvv,lemma:ocsp-coarsening}, and let $T$ be the larger of the $T_0$'s from \cref{lemma:seqirsd-analysis,lemma:ocsp-coarsening} applied with error probabilities $\epsilon'$. Let $n_0 \in \BN$ be the larger of the $n_0$'s from \cref{thm:cgsvv,lemma:seqirsd-analysis}.

Now fix $\vecb \in \supp(\Pi)$. For $n \geq n_0 \in \BN$, let $\CY^\refine$ and $\CN^\refine$ denote the distributions of $\mPi$ instances $\Psi^\refine$ where $\Psi \sim \Piq\CY'_{\alpha,T}(n)$ and $\Piq\CN'_{\alpha,T}(n)$, respectively. By \cref{thm:cgsvv}, distinguishing $\CY^\refine$ and $\CN^\refine$ requires space $\tau n$ for some $\tau >0$. On the other hand, by \cref{lemma:seqirsd-analysis,lemma:ocsp-refinement,eqn:ocsp-omega}, \[ \Pr_{\Phi \sim \CY^\refine}[\ordval_\Phi \leq 1-\epsilon/2] \leq \exp(-n), \] while by \cref{lemma:seqirsd-analysis,lemma:ocsp-coarsening,eqn:ocsp-rho}, \[ \Pr_{\Phi \sim \CN^\refine}[\ordval_\Phi \geq \rho(\Pi) + \epsilon/2] \leq \exp(-n). \] Thus, $(\rho(\Pi)+\epsilon)$-approximating $\mPi$ requires at least $\tau n$ space, as desired (see \cref{prop:yao}).
\end{proof}

It remains to prove \cref{lemma:ocsp-coarsening}; we do so in the final section.

\section{The coarsening analysis: Proving \cref{lemma:ocsp-coarsening}}\label{sec:ocsp-coarsening}

The goal of this section is to prove \cref{lemma:ocsp-coarsening}. We do so by carrying out the plan based on ``balanced partition expanders'', as described in the proof sketch for \cref{claim:mas-coarsening}.

Given an instance $\Psi$ of $\mcsp[\cPi]$ on $n$ variables and a subset $S \subseteq [n]$, we denote by $N(\Psi,S)$ the number of constraints $\vecj = (j_1,\ldots,j_k)$ in $\Psi$ which ``touch $S$ twice'', i..e, such that $j_i,j_{i'} \in S$ for some $i \neq i'$.

\begin{definition}[Small set expansion (SSE)]
Let $\Psi$ be an instance of $\mcsp[\cPi]$ on $n$ variables and $m$ constraints. For $\gamma,\epsilon\in(0,1)$, $\Psi$ is a \emph{$(\gamma,\epsilon)$-small set expander (SSE)} if for every subset $S \subseteq [n]$ of size at most $\gamma n$, $N(\Psi,S) \leq \epsilon m$.
\end{definition}

\begin{definition}[Balanced partition expansion (BPE)]
Let $\Psi$ be an instance of $\mcsp[\cPi]$ on $n$ variables and $m$ constraints. For $\gamma,\epsilon\in(0,1)$, $\Psi$ is a \emph{$(\gamma,\epsilon)$-balanced partition expander (BPE)} if for every $\vecb \in \BZ_q^n$ where each block $\vecb^{-1}(a)$ has size at most $\gamma n$, \[ \sum_{a \in \BZ_q} N(\Psi,\vecb^{-1}(a)) \leq \epsilon m. \]
\end{definition}

Now we give several lemmas which connect these notions to the $\no$-distribution $\CN'$ of the previous subsection:

\begin{lemma}\label{lemma:n'-sse}
Let $\Pi : \sym_k \to \{0,1\}$, $\vecb \in \supp(\Pi)$, $q \in \BN$, and $\gamma > 0$. There exists $\alpha_0 \in \BN$ such that for every $\alpha \in (0,\alpha_0)$, there exists $T_0 \in \BN$ such that for $T \geq T_0 \in \BN$ and every $n \in \BN$, \[ \Pr_{\Psi \sim \Piq\CN'_{\alpha,T}(n)} \left[\Psi\text{ is not a }(\gamma,8k^2\gamma^2)\text{-SSE} \right] \leq \exp(-n). \]
\end{lemma}

We defer the proof of \cref{lemma:n'-sse} until the end of this subsection, as it involves some somewhat messy concentration bounds.

\begin{lemma}[Good SSEs are good BPEs]\label{lemma:sse-to-bpe}
For every $\gamma,\epsilon\in(0,1)$, if $\Psi$ is a $(\gamma,\epsilon)$-SSE, then it is a $\left(\gamma,3\gamma/\epsilon\right)$-BPE.
\end{lemma}

\begin{proof}
Consider any $\vecb \in \BZ_q^n$ of $[n]$ where each block $\vecb^{-1}(a)$ has size at most $\gamma n$. Firstly, note that if two blocks have sizes $|\vecb^{-1}(a)|,|\vecb^{-1}(a')|$ both smaller than $\frac{\gamma n}2$, we can reassign $\vecb^{-1}(a')$ to $a$, since this only increases the sum $\sum_{a \in \BZ_q} N(\Psi,\vecb^{-1}(a))$ and every block still has size at most $\gamma n$. Thus, we can assume WLOG that $\vecb$ consists of empty blocks, a single block of size at most $\frac{\gamma n}2$, and blocks of size between $\frac{\gamma n}2$ and $\gamma n$. The number of non-empty blocks is at most $\frac{n}{\lfloor \gamma n / 2 \rfloor}+1 \leq 3\gamma$, and each such block has $N(\Psi,\vecb^{-1}(a)) \leq \epsilon m$ by the SSE assumption.
\end{proof}

\begin{lemma}[Refining roughly preserves value in BPEs]\label{lemma:bpe-gap-bound}
Let $\Pi : \sym_k \to \{0,1\}$, $q \in \BN$ and $\epsilon > 0$. If $\Psi$ is a $\mcsp[\cPi]$ instance which is a $(1/q,\epsilon)$-BPE, then \[ \ordval_{\Psi^\refine} \leq \val_\Psi + \epsilon. \]
\end{lemma}

\begin{proof}
Let $n$ and $m$ denote the number of variables and constraints in $\Psi$, respectively. Consider any ordering $\vecsigma \in \sym_n$, and, as in the proof sketch of \cref{claim:mas-coarsening}, let $\vecsigma^{\coarsen q} \in \BZ_q^n$ be defined by $\sigma^{\coarsen q}_i = \lceil q\sigma_i /n\rceil$. It suffices to show that $\ordval_{\Psi^\refine}(\vecsigma) \leq \val_{\Psi}(\vecsigma^{\coarsen q}) + \epsilon$. $\vecsigma^{\coarsen q}$ has the ``monotonicity'' property that for every $u,v \in [n]$, if $\sigma_u < \sigma_v$ then $\sigma^{\coarsen q}_u \leq \sigma^{\coarsen q}_v$, and each block $(\vecsigma^{\coarsen q})^{-1}(a)$ has size at most $\frac{n}q$.

Suppose a constraint $\vecj = (j_1,\ldots,j_k)$ is satisfied by $\vecsigma$ (in $\Psi^{\refine}$). If $\sigma^{\coarsen q}_{j_1},\ldots,\sigma^{\coarsen q}_{j_k}$ are all distinct, then by monotonicity $\ord(\vecsigma^{\coarsen q}|_\vecj) = \vecsigma|_\vecj$ and so $\vecj$ is satisfied by $\vecsigma^{\coarsen q}$ (in $\Psi$). Thus, it suffices to show that at most $\epsilon m$ constraints $\vecj$ have the property that $\sigma^{\coarsen q}_{j_1},\ldots,\sigma^{\coarsen q}_{j_k}$ are not all distinct; this is precisely the BPE property of $\Psi$.
\end{proof}

Modulo the proof of \cref{lemma:n'-sse}, we can now prove \cref{lemma:ocsp-coarsening}:

\begin{proof}[Proof of \cref{lemma:ocsp-coarsening}]
Let $q_0 = \lceil 24k^2/\epsilon \rceil$; consider $q \geq q_0 \in \BN$ and let $\gamma = 1/q$; let $\alpha_0$ be as in \cref{lemma:n'-sse}; for any $\alpha \in (0,\alpha_0)$, let $T_0$ be as in \cref{lemma:n'-sse}; and consider any $T \geq T_0 \in \BN$.

If $\Psi$ is a $(\gamma,8k^2\gamma^2)$-SSE, then by \cref{lemma:sse-to-bpe} it is a $(\gamma,24k^2\gamma)$-BPE, in which case by \cref{lemma:bpe-gap-bound} we have $\ordval_{\Psi^\refine} \leq \val_\Psi + 24k^2\gamma$, and $24k^2\gamma \leq \epsilon$ by assumption. But by \cref{lemma:n'-sse}, $\Psi \sim \Piq\CN'_{\alpha,T}(n)$ is a $(\gamma,8k^2\gamma^2)$-SSE except with probability $\exp(-n)$.
\end{proof}

Finally, we prove \cref{lemma:ocsp-coarsening}.

\begin{proof}[Proof of \cref{lemma:ocsp-coarsening}]
Recall that we sample instances $\Psi \sim \Piq\CN'_{\alpha,T}(n)$ by first sampling from the $\no$ distribution $\kq\CN_{\alpha,T}(n)$ of $\kq\seqirsd_{\alpha,T}(n)$ (see \cref{def:seqirsd}) and then applying the reduction functions $\R_1,\ldots,\R_T$ from \cref{cons:seqirsd} with some fixed base vector $\vecb \in \supp(\Pi)$.

For concreteness, we'll repeat the definitions for this case here. For each $t \in [T]$, we get an instance $\Psi_t$ (produced by $\Bob_t$):

\begin{enumerate}
    \item Sample a matching $M_t \sim \Matchings_{k,\alpha}(n)$ and a vector $\vecz(t) \sim (\CU_{\BZ_q}^k)^{\alpha n}$, which we think of as consisting of $\alpha n$ independently sampled blocks $\vecz(t,\ell) \sim \CU_{\BZ_q}^k$.
    \item Including the $\ell$-th hyperedge $\vece(t,\ell) = (e(t,\ell)_1,\ldots,e(t,\ell)_k)$ of $M_t$ as a constraint in $\Psi_t$ iff $\vecz(t,\ell) - \vecb \in \CC$ where $\CC = \{(c,\ldots,c):c\in\BZ_q\}$.
\end{enumerate}

and then $\Psi = \Psi_1 \cup \cdots \cup \Psi_T$. Since $\vecz(t,\ell)$ is a uniform vector in $\BZ_q^k$ and $\CC$ is a one-dimensional subspace of $\BZ_q^k$, each hyperedge of $M_t$ is included in $\Psi_t$ independently w.p. $q^{-(k-1)}$.

Let $m_t$ denote the number of constraints in $\Psi_t$ for each $t$, and $m = \sum_{t=1}^T m_t$ the number of constraints in $\Pi$. Therefore, each $m_t$ is distributed as the sum of $\alpha n$ independent $\Bern(q^{-(k-1)})$ random variables. Now, consider the event that $m \geq \alpha T n/(2q^{k-1})$. Since $m$ is distributed as the sum of $\alpha T n$ independent $\Bern(q^{-(k-1)})$'s, by the Chernoff bound, this event fails to occur with probability at most $\exp(-\alpha T n/(8q^{k-1}))$, which is $\exp(-n)$ for sufficiently large $T_0$. Thus, it suffices to prove the lemma conditioned on fixed $m_1,\ldots,m_T$ satisfying $m \geq \alpha T n/(2q^{k-1})$. With this conditioning, each sub-instance $\Psi_t$ is the result of a simpler sampling process: the constraints are the hyperedges of a hypermatching drawn from $\Matchings_{k,m/n}(n)$.

Now fix any set $S \subseteq [n]$ of size at most $\gamma n$. (We will later take a union bound over all $S$.) Label the hyperedges of each $M_t$ as $\vece(t,1),\ldots,\vece(t,m_t)$. Consider the collection of $m$ random variables $\{X_{t,\ell}\}_{t\in[T],\ell\in[m_t]}$, each of which is the indicator for the event that two distinct vertices of $\vece(t,\ell)$ are in $S$. By definition, $N(\Psi,S) = \sum_{t=1}^T \sum_{\ell=1}^{m_t} X_{t,\ell}$.

For fixed $t \in [T]$, we first bound $\E[X_{t,\ell} \mid X_{t,1},\ldots,X_{t,\ell-1}]$ for each $\ell\in[m_t]$. Conditioned on $\vece(t,1),\ldots,\vece(t,\ell-1)$ incident to some subset $V_{t,\ell} \subseteq [n]$ of $k(\ell-1)$ vertices, the hyperedge $\vece(t,\ell)$ is uniformly distributed over $[n]\setminus V_{t,\ell}$. It suffices to union-bound, over distinct pairs $\{j_1,j_2\} \in \binom{[k]}2$, the probability that the $j_1$-st and $j_2$-nd vertices of $\vece(t,\ell)$ are in $S$ (conditioned on $X_{t,0},\ldots,X_{t,\ell-1}$). We can sample the $j_1$-st and $j_2$-nd vertices of $\vece(t,\ell)$ first, and then ignore the remaining vertices. Setting $\alpha_0 = 1/(2k)$, we have the upper bound 

\begin{align*}
    \E[X_{t,\ell} \mid X_{t,1},\ldots,X_{t,\ell-1}] & \leq \binom{k}2 \cdot \frac{|S|(|S|-1)}{(n-k(\ell-1))(n-k(\ell-1)-1)}\\
    &\leq \binom{k}2 \cdot \left(\frac{|S|}{n-k(\ell-1)}\right)^2 \\
    &\leq \binom{k}2 \cdot \left(\frac{|S|}{n-km_t}\right)^2 \\
    & \leq 4k^2\gamma^2,
\end{align*} since $m_t \leq \alpha n \leq n/(2k)$.

Now $X_{t,\ell}$ is independent of $X_{t',\ell'}$ if $t \neq t'$ since $M_t$ and $M_{t'}$ are independent. Thus, \cref{lemma:azuma} implies that \[ \Pr_{\Psi \sim \Piq\CN'_{\alpha,T}(n)}\left[N(\Psi,S) \geq 8 k^2 \gamma^2 m \mid m_1,\ldots,m_T \right] \leq \exp\left(-2k^2 \gamma^2 m \right). \] Finally, we use the inequality $m \geq \alpha T n /(2k^2)$, take the union bound over $S \subseteq [n]$, and set $T_0$ large enough to ensure that $2^n \exp(-\gamma^2 \alpha T n) \leq \exp(-n)$.
\end{proof}

\newcommand{\kz}{[k]\cup\{0\}}
\newcommand{\rroot}{\mathsf{root}_\BR}

\newcommand{\mbfSk}{\mbcsp[f_{S,k}]}
\newcommand{\mbTh}{\mbcsp[\Th^t_k]}

\chapter{Symmetric Boolean predicates}\label{chap:sym-bool}

\epigraph{In our algorithm [for $\mkand$], we use the approach of Hast~\cite{Has05}: We first obtain a ``preliminary'' solution $z_1,\ldots,z_n \in \{-1,1\}^n$ and then independently flip the values of $z_i$ using a slightly biased distribution (i.e., we keep the old value of $z_i$ with probability slightly larger than $1/2$).}{Charikar, Makarychev, and Makarychev~\cite{CMM09}}

\newthought{$\mkand$ is the simplest $k$-ary Boolean CSP} which is nontrivially approximable, and is also, in some sense, the easiest to approximate (see \cref{rem:kand-approx} below). In this chapter, we present our joint work with Boyland, Hwang, Prasad, and Velusamy~\cite{BHP+22} which studies several questions regarding the results of \cite{CGSV21-boolean,CGS+22} (\cref{chap:framework-papers} above) for $\mkand$ and other $\mbf$ problems:

\begin{enumerate}
    \item Can the dichotomy theorem in \cite{CGSV21-boolean} (i.e., \cref{thm:cgsv-bool-dichotomy}) be used to find closed-form sketching approximability ratios $\alpha(f)$ for nontrivially approximable problems $\mbf$ beyond $\mtwoand$ (\cref{ex:cgsv-2and})? We note that to the best of our knowledge, in the classical setting Raghavendra's UG-dichotomy theorem~\cite{Rag08} has never been used for this purpose, but we may have more ``hope'' for using \cite{CGSV21-boolean}'s dichotomy since it is at least decidable (see the beginning of \cref{sec:cgsv-discussion} above).
    \item \cite{CGS+22} shows that for every predicate $f : \BZ_2^k \to \{0,1\}$ and $\epsilon > 0$, $(2\rho(f)+\epsilon)$-approximating $\mbf$ with a streaming algorithm requires $\Omega(n)$ space (\cref{cor:cgsvv-2rho}). How tight is this upper bound on the approximation ratio?
    \item Does the streaming lower bound in \cite{CGSV21-boolean} based on padded one-wise pairs (\cref{thm:cgsv-streaming-lb}) suffice to resolve the streaming approximability of $\mbf$ for every predicate $f : \BZ_2^k \to \{0,1\}$?
    \item For every predicate $f : \BZ_2^k \to \{0,1\}$, \cite{CGSV21-boolean} gives an optimal sketching $(\alpha(f)-\epsilon)$-approximation algorithm for $\mbf$ in \cite{CGSV21-boolean}, but this algorithm runs a ``grid'' of $O(1/\epsilon^2)$ distinguishers for $\bgd\mbf$ distinguishing problems in parallel (see \cref{cor:cgsv-bool-approx} and its proof sketch). Can we obtain simpler optimal sketching approximations for interesting predicates?
\end{enumerate}

Specifically, we investigate the sketching (and streaming) approximability of $\mbf$ when $f:\BZ_2^k \to \{0,1\}$ is a \emph{symmetric} predicate; that is, $f$ depends only on the Hamming weight (a.k.a. number of $1$'s) $\|\vecb\|_0$ of its input $\vecb \in \BZ_2^k$. For a set $S \subseteq [k]$, let $f_{S,k} : \BZ_2^k \to \{0,1\}$ denote the symmetric predicate defined by $f_{S,k}(\vecb) = \1_{\|\vecb\|_0 \in S}$. Some well-studied examples of predicates in this class include $\kand = f_{\{k\},k}$ and the \emph{threshold functions} $\Th^t_k = f_{\{t,t+1,\ldots,k\},k}$. As we'll see in \cref{sec:sym-setup} below, we consider symmetric predicates because for such predicates, the \cite{CGSV21-boolean} results (specifically \cref{cor:cgsv-bool-approx,thm:cgsv-streaming-lb}) take on significantly simpler forms.

We use computer assistance for algebraic manipulations in several of the proofs in this chapter; our code is available on the Mathematica Notebook Archive at \url{https://notebookarchive.org/2022-03-a5vpzhg}.

\section{Results}

We begin by presenting the major results of our work \cite{BHP+22}.

\subsection{The sketching approximability of ${\m}\kand$}\label{sec:kand-overview}

Recall from \cref{ex:cgsv-2and} that the $\sqrt n$-space sketching approximability of $\mtwoand$ is $\alpha(\twoand) = \frac49$, and $(\frac49+\epsilon)$-approximations can be ruled out even for $\sqrt n$-space \emph{streaming} algorithms using the padded one-wise pair criterion (\cref{thm:cgsv-streaming-lb}). We build on this result by obtaining closed-form expressions for the $\sqrt n$-space sketching approximation ratio $\alpha(\kand)$ for every $k$. For odd $k \geq 3$, define the constant
\begin{equation}\label{eqn:alpha'_k}
    \alpha'_k \eqdef \left(\frac{(k-1)(k+1)}{4k^2}\right)^{(k-1)/2} = 2^{-(k-1)} \cdot \left(1-\frac1{k^2}\right)^{(k-1)/2}.
\end{equation}

Then in \cref{sec:kand-analysis}, we prove the following:

\begin{theorem}\label{thm:kand-approximability}
For odd $k \geq 3$, $\alpha(\kand) = \alpha'_k$, and for even $k \geq 2$, $\alpha(\kand) = 2\alpha'_{k+1}$.
\end{theorem}

For instance, $\alpha(\threeand) = \alpha'_3 = \frac29$. Since $\rho(\kand) = 2^{-k}$, \cref{thm:kand-approximability} also has the following important corollary:

\begin{corollary}\label{cor:kand-asympt}
$\lim_{k \to \infty} \frac{\alpha(\kand)}{2\rho(\kand)} = 1$.
\end{corollary}

Recall that Chou, Golovnev, Sudan, Velusamy, and Velingker~\cite{CGS+22} show that any predicate $f$ cannot be $(2\rho(f)+\epsilon)$-approximated even by linear-space streaming algorithms (see \cref{cor:cgsvv-2rho}). On the other hand, in \cref{sec:thresh-alg-overview} below, we describe simple $O(\log n)$-space sketching algorithms for $\mkand$ achieving the optimal ratio from \cite{CGSV21-boolean}. Thus, as $k \to \infty$, these algorithms achieve an asymptotically optimal approximation ratio even among linear-space streaming algorithms!

\subsection{The sketching approximability of other symmetric predicates}

We also analyze the sketching approximability of a number of other symmetric Boolean predicates. For instance, we show that:

\begin{theorem}\label{thm:k-1-k-approximability}
For even $k \geq 2$, $\alpha(\Th^{k-1}_k) = \frac{k}2\alpha'_{k-1}$.
\end{theorem}

We prove \cref{thm:k-1-k-approximability} in \cref{sec:k-1-k-analysis} using techniques similar to our proof of \cref{thm:kand-approximability}. We also provide partial results for $f_{\{(k+1)/2\},k}$, including closed forms for small $k$ and an asymptotic analysis of $\alpha(f_{\{(k+1)/2\},k})$:

\begin{theorem}[Informal version of \cref{thm:k+1/2-approximability}]\label{thm:k+1/2-approximability-informal}
For odd $k \in \{3,\ldots,51\}$, there is an explicit expression for $\alpha(f_{\{(k+1)/2\},k})$ as a function of $k$.
\end{theorem}

\begin{theorem}\label{thm:k+1/2-asymptotic-lb}
$\lim_{\text{odd } k \to \infty} \frac{\alpha\left(f_{\{(k+1)/2\},k}\right)}{\rho\left(f_{\{(k+1)/2\},k}\right)}=1$.
\end{theorem}

We prove \cref{thm:k+1/2-approximability-informal,thm:k+1/2-asymptotic-lb} in \cref{sec:k+1/2-analysis}. Finally, in \cref{sec:other-analysis}, we explicitly resolve fifteen other cases (e.g., $f_{\{2,3\},3}$ and $f_{\{4\},5}$) not covered by \cref{thm:kand-approximability,thm:k-1-k-approximability,thm:k+1/2-approximability-informal}.

\subsection{Simple approximation algorithms for threshold functions}\label{sec:thresh-alg-overview}

Recall from \cref{sec:mdcut-algorithm} that \cite{CGV20} gives optimal $(\frac49-\epsilon)$-approximation sketching algorithms for $\mtwoand$ based on measuring a quantity $\bias_\Psi \in [0,1]$ of the input instance $\Psi$ (see \cref{eqn:2and-bias}) using 1-norm sketching algorithms \cite{Ind06,KNW10}. In \cref{sec:thresh-alg}, we extend the definition of bias to arbitrary CSPs and give simple optimal bias-based approximation algorithms for threshold predicates:

\begin{theorem}\label{thm:thresh-bias-alg}
Let $k \in \BN, i \leq k$. Then for every $\epsilon > 0$, there exists a piecewise linear function $\gamma : [-1,1]\to[0,1]$ and a constant $\epsilon'>0$ such that the following is a sketching $(\alpha(\Th^t_k)-\epsilon)$-approximation for $\mbcsp[\Th^t_k]$: On input $\Psi$, compute an estimate $\hat{b}$ for $\bias_\Psi$ up to a multiplicative $(1\pm \epsilon')$ error and output $\gamma(\hat{b})$.
\end{theorem}

Our construction generalizes the algorithm in \cite{CGV20} for $\twoand$ to all threshold predicates, and is also a simplification, since the \cite{CGV20} algorithm computes a more complicated function of $\hat{b}$; see \cref{rem:cgsv-vs-bhp-2and}. 

For all CSPs whose approximability we resolve in this chapter, we apply an analytical technique which we term the ``max-min method;'' see the discussion in \cref{sec:max-min} below. For such CSPs, our algorithm can be extended to solve the problem of outputting an approximately optimal \emph{assignment}, instead of just an estimate for the value, following the example for $\mtwoand$ we discussed at the end of \cref{sec:mdcut-template-alg}. Indeed, for this problem, we give a simple randomized streaming algorithm using $O(n)$ space and time:

\begin{theorem}[Informal version of \cref{thm:thresh-bias-output-alg}]\label{thm:thresh-bias-alg-classical}
Let $\Th^t_k$ be a threshold predicate for which the max-min method applies, such as $\kand$, or $\Th^{k-1}_k$ (for even $k$). Then there exists a constant $p^* \in [0,1]$ such that following algorithm, on input $\Psi$, outputs an assignment with expected value at least $\alpha(\Th^t_k) \val_\Psi$: Assign variable $i$ to $1$ if $\bias_\Psi(i) \geq 0$ and $0$ otherwise, and then flip each variable's assignment independently with probability $p^*$.
\end{theorem}

In particular, it is not \emph{a priori} implied by \cite{CGSV21-boolean,CGSV21-finite} that setting $p^*$ to be a fixed constant is sufficient (and this was not noticed by \cite{CGV20} in the $\twoand$ case); we view this as an important contribution to the general understanding of sketching approximability of CSPs. Also, our algorithm can potentially be derandomized using universal hash families, as in Biswas and Raman's recent derandomization \cite{BR21} of the $\mtwoand$ algorithm in \cite{CGV20}.

\subsection{Sketching vs. streaming approximability}

\cref{thm:kand-approximability} implies that $\alpha(\threeand) = \frac29$, and thus for every $\epsilon > 0$, $\mthreeand$ can be $(\frac29-\epsilon)$-approximated by $O(\log n)$-space linear sketching algorithms, but not $(\frac29+\epsilon)$-approximated by $\sqrt n$-space sketching algorithms. We prove that the padded one-wise pair criterion of Chou, Golovnev, Sudan, and Velusamy~\cite{CGSV21-boolean} (\cref{thm:cgsv-streaming-lb}) is not sufficient to completely resolve the \emph{streaming} approximability of \m[$3$AND], i.e., to show that $\sqrt n$-space \emph{streaming} algorithms cannot $(\frac29+\epsilon)$-approximate $\mthreeand$ for every $\epsilon > 0$; however, it does show that $\sqrt n$-space streaming algorithms cannot $0.2362$-approximate $\mthreeand$. We state these results formally in \cref{sec:cgsv-streaming-failure-3and-overview} below. Separately, \cref{thm:k-1-k-approximability} implies that $\alpha(\Th^3_4) = \frac49$, and the padded one-wise pair criterion \emph{can} be used to show that $(\frac49+\epsilon)$-approximating $\mbcsp[\Th^3_4]$ requires $\Omega(\sqrt n)$ space in the streaming setting (see \cref{obs:th34-streaming-lb} below).

\subsection*{Related work}

The classical approximability of ${\m}\kand$ has been the subject of intense study, both in terms of algorithms \cite{GW95,FG95,Zwi98,Tre98-alg,TSSW00,Has04,Has05,CMM09} and hardness-of-approximation \cite{Has01,Tre98-hardness,ST98,ST00,EH08,ST09}, given its intimate connections to $k$-bit PCPs. Charikar, Makarychev, and Makarychev~\cite{CMM09} constructed an $\Omega(k 2^{-k})$-approximation to ${\m}\kand$, while  Samorodnitsky and Trevisan~\cite{ST09} showed that $k2^{-(k-1)}$-approximations and $(k+1)2^{-k}$-approximations are $\NP$- and UG-hard, respectively.

Interestingly, recalling that $\alpha(\kand) \to 2\rho(\kand) = 2^{-(k-1)}$ as $k \to \infty$, in the large-$k$ limit our simple randomized algorithm (given in \cref{thm:thresh-bias-alg-classical}) matches the performance of Trevisan's~\cite{Tre98-alg} parallelizable LP-based algorithm for $\kand$, which (to the best of our knowledge) was the first work on the general $\kand$ problem! The subsequent works \cite{Has04,Has05,CMM09} superseding \cite{Tre98-alg} use more complex techniques involving semidefinite programming, but are structurally similar to our algorithm in \cref{thm:thresh-bias-alg-classical}: They all involve ``guessing'' an assignment $\vecx \in \BZ_2^n$ and then perturbing each bit with constant probability.

\begin{remark}\label{rem:kand-approx}
Trevisan~\cite[Theorem 18]{Tre98-alg} observes that for every predicate $f:\BZ_2^k\to\{0,1\}$, $\alpha(f)/\rho(f) \leq \alpha(\kand)/\rho(\kand)$ (in the classical setting, but the proof carries over easily to the sketching setting). Thus, $\alpha(f)$ is ``easiest to approximate'' among all Boolean functions, relative to the threshold of nontrivial approximability. Intuitively, it holds because $\kand$ is the most ``informative'' predicate: It exactly specifies what values its variables should be assigned to. More precisely, given any predicate $f : \BZ_2^k \to \{0,1\}$, let $t = |\supp(f)|$. Given any instance $\Psi$ of $\mbf$, we can create an instance $\Psi'$ of $\mbcsp[\kand]$ by replacing each constraint $C$ in $\Psi$ with $t$ constraints in $\Psi'$ corresponding to $C$'s $t$ satisfying assignments; that is, $C = (\vecb,\vecj,w)$ becomes $C_1=(\vecb+\veca(1)+\vecone,\vecj,w),\ldots,C_t=(\vecb+\veca(t)+\vecone,\vecj,w)$ where $\supp(f) = \{\veca(1),\ldots,\veca(t)\}$. Every assignment $\vecx\in\BZ_2^n$ satisfies either one or zero of the constraints $\{C_1,\ldots,C_t\}$, corresponding to whether it satisfies or fails to satisfy $C$, respectively. Thus, $\val_\Psi(\vecx) = t \val_{\Psi'}(\vecx)$ for every $\vecx \in \BZ_2^n$. The inequality then follows from the fact that $\rho(f)/\rho(\kand) = t$.
\end{remark}

Classical approximability for various classes of symmetric predicates has been studied in \cite{CHIS12,ABM12,GL17}.

\section{Setup for the symmetric case}\label{sec:sym-setup}

We begin by showing how the \cite{CGSV21-boolean} results (specifically, \cref{cor:cgsv-bool-approx,thm:cgsv-streaming-lb}) are significantly simpler to instantiate when the predicates are symmetric (as observed by \cite{CGSV21-boolean} in the $\twoand$ case). Let $\lambda_S(\CD,p),\gamma_S(\CD_Y),\beta_S(\CD_N)$ denote the $\lambda,\gamma,\beta$ functions from \cite{CGSV21-boolean}, respectively, for a symmetric predicate $f_{S,k}$ (see \cref{eqn:cgsv-bool-lgb}). We will show in \cref{sec:lambda-gamma-formulas} below that $\lambda_S(\CD,p)$ is in general a multivariate polynomial in $p$ and $\CD\langle 0 \rangle, \ldots, \CD\langle k \rangle$, which is degree-$k$ and linear in $\CD\langle t \rangle$.

A distribution $\CD \in \Delta(\BZ_2^k)$ is \emph{symmetric} if strings of equal Hamming weight are equiprobable, i.e., $\|\veca\|_0=\|\veca'\|_0 \implies \CD(\veca) = \CD(\veca')$. Let $\Delta_k \subseteq \Delta(\BZ_2^k)$ denote the space of symmetric distributions over $\BZ_2^k$. For a distribution $\CD \in \Delta(\BZ_2^k)$ and $i \in \kz$, let $\CD\langle i \rangle \eqdef \sum_{\|\veca\|_0=i} \CD(i)$ denote the total mass on strings of Hamming weight $i$. We can view symmetric distributions $\CD \in \Delta_k$ as distributions over $\kz$ which take value $i$ with probability $\CD\langle i \rangle$. There is also a natural projection of $\Delta(\BZ_2^k)$ onto $\Delta_k$ given by the \emph{symmetrization} operation: For a distribution $\CD \in \Delta(\BZ_2^k)$, we let $\Sym(\CD) \in \Delta_k$ denote the unique \emph{symmetric} distribution such that $\Sym(\CD)\langle i \rangle = \CD\langle i \rangle$ for all $i \in \kz$. (In other words, symmetrization redistributes probability in $\CD$ over all strings of equal Hamming weight.) Finally, if $\CD \in \Delta_k$ is symmetric, then we define $\mu(\CD) = \E_{\veca\sim\CD}[(-1)^{a_1+1}]$; $\vecmu(\CD)$ is then the constant vector $(\mu(\CD),\ldots,\mu(\CD))$.

The following proposition states that to use the tools from \cite{CGSV21-boolean} (i.e., \cref{thm:cgsv-bool-dichotomy,cor:cgsv-bool-approx,thm:cgsv-streaming-lb}) for symmetric predicates $f : \BZ_2^k\to\{0,1\}$, it suffices to examine only symmetric distributions:

\begin{proposition}\label{prop:bool-sym}
\begin{enumerate}[label={\roman*.},ref={\roman*}]
    \item For every symmetric predicate $f_{S,k} : \BZ_2^k \to \{0,1\}$, $\CD \in \Delta(\BZ_2^k)$, and $p \in [0,1]$, $\lambda_S(\CD,p) = \lambda_S(\Sym(\CD),p)$.
    \item For all $\CD_N, \CD_Y \in \Delta(\BZ_2^k)$ with matching marginals (i.e., $\vecmu(\CD_N)=\vecmu(\CD_Y)$), $\mu(\Sym(\CD_N))=\mu(\Sym(\CD_Y))$.
    \item For all padded one-wise pairs $\CD_N,\CD_Y\in\Delta(\BZ_2^k)$, $\Sym(\CD_N)$ and $\Sym(\CD_Y)$ are also a padded one-wise pair.
\end{enumerate}
\end{proposition}

\begin{proof}
Omitted (follows immediately from definitions and linearity of expectation).
\end{proof}

In particular, together with \cref{cor:cgsv-bool-approx} we have:

\begin{corollary}\label{cor:sym-bool-alpha}
For every symmetric predicate $f : \BZ_2^k \to \{0,1\}$, \[ \alpha(f) = \inf_{\CD_N,\CD_Y \in \Delta_k: ~\mu(\CD_N)=\mu(\CD_Y)} \left(\frac{\beta_f(\CD_N)}{\gamma_f(\CD_Y)}\right). \] Then:
\begin{enumerate}[label={\roman*.},ref={\roman*}]
    \item For every $\epsilon > 0$, there exists $\tau > 0$ and a $\tau \log n$-space linear sketching algorithm which  $(\alpha(f)-\epsilon)$-approximates $\mbf$.
    \item For every $\epsilon > 0$, there exists $\tau > 0$ such that every sketching algorithm which $(\alpha(f)+\epsilon)$-approximates $\mbf$ uses at least $\tau \sqrt{n}$ space (for sufficiently large $n$).
\end{enumerate}
\end{corollary}

Our focus on symmetric predicates $f$ is motivated by the simpler form of \cref{cor:sym-bool-alpha}, in comparison with \cref{cor:cgsv-bool-approx} for general predicates. Since we need to consider only symmetric distributions in the infimum, $\CD_Y$ and $\CD_N$ are each parameterized by $k+1$ variables (as opposed to $2^k$ variables), and there is a single linear equality constraint (as opposed to $k$ constraints).

Next, we give an explicit formula for $\mu(\CD)$ for a symmetric distribution $\CD \in \Delta_k$. For $i \in [k]$, let $\epsilon_{i,k} \eqdef -1+\frac{2i}k$.

\begin{lemma}\label{lemma:sym-bool-mu}
For any $\CD \in \Delta_k$, \[ \mu(\CD) = \sum_{i=0}^k \epsilon_{i,k} \,\CD\langle i \rangle. \]
\end{lemma}

\begin{proof}
Recall that we defined $\epsilon_{i,k} = -1+\frac{2i}k$. By definition, $\mu(\CD) = \E_{\vecb\sim\CD}[b_1]$. We use linearity of expectation; the contribution of weight-$i$ vectors to $\mu(\CD)$ is $\CD\langle i \rangle \cdot \frac1k (i \cdot 1 + (k-i) \cdot (-1)) = \epsilon_{i,k} \,\CD\langle i \rangle$.
\end{proof}

\begin{example}\label{ex:one-wise-indep-sym}
In \cref{ex:one-wise-indep}, we showed that if a predicate $f : \BZ_2^k \to \{0,1\}$ supports one-wise independence, i.e., there exists $\CD_Y \in \Delta(\BZ_2^k)$ supported on $\supp(f)$ such that $\vecmu(\CD_Y)=\veczero$, then $\mbf$ is approximation-resistant for $\sqrt n$-space streaming algorithms. \cite[Lemma 2.14]{CGSV21-boolean} shows that if $f = f_{S,k}$ is symmetric, this condition is also \emph{necessary} for $\sqrt n$-space streaming approximation-resistance. By \cref{prop:bool-sym}, this condition is equivalent to the existence of a \emph{symmetric} distribution $\CD_Y \in \Delta_k$ supported on $S$ (in the sense that $\CD_Y\langle s \rangle = 0$ for all $s \not\in S$) such that $\mu(\CD_Y) = 0$.

Now if $k$ is even and $k/2 \in S$, then the distribution $\CD_Y \in \Delta_k$ with $\CD_Y \langle k/2\rangle = 1$ has $\mu(\CD_Y) = 0$ (by \cref{lemma:sym-bool-mu}) and is supported on $S$; thus, $\mbfSk$ is streaming approximation-resistant. Moreover, if $S$ contains elements $s \leq k/2$ and $t \geq k/2$, we can let $\delta = \frac{\epsilon_{t,k}}{\epsilon_{t,k}-\epsilon_{s,k}}$ (note that $\epsilon_{t,k}> 0$ and $\epsilon_{s,k} < 0$), and let $\CD_Y \in \Delta_k$ be defined by $\CD_Y\langle s \rangle = \delta$ and $\CD_Y\langle t \rangle = 1-\delta$. Then again, $\mu(\CD_Y) = 0$ and $\CD_Y$ is supported on $S$, so $\mbfSk$ is streaming approximation-resistant.
\end{example}

Given \cref{ex:one-wise-indep-sym}, in the remainder of this chapter we focus on the case where all elements of $S$ are either larger than or smaller than $k/2$. Note also that if $S' = \{k-s : s \in S\}$, every instance of $\mbfSk$ can be viewed as an instance of $\mbcsp[f_{S',k}]$ with the same value, since for any constraint $C=(\vecb,\vecj,w)$ and assignment $\vecx \in \BZ_2^n$, we have $f_{S,k}(\vecb+ \vecx\vert_\vecj)$ = $f_{S',k}(\vecb+ (\vecx+\vecone)\vert_\vecj)$. Thus, we further narrow our focus to the case where every element of $S$ is larger than $k/2$.

\section{Techniques}\label{sec:techniques}

\subsection{Formulations of the optimization problem}\label{sec:cgsv-opt}

In order to show that $\alpha(\twoand) = \frac49$, Chou, Golovnev, Sudan, and Velusamy~\cite[Example 1]{CGSV21-boolean} use the following reformulation of the optimization problem from \cref{cor:cgsv-bool-approx}. For a symmetric predicate $f_{S,k}$ and $\mu \in [-1,1]$, let
\begin{equation}\label{eqn:beta_Sk-gamma_Sk-def}
    \beta_{S,k}(\mu) = \inf_{\CD_N \in \Delta_k:~\mu(\CD_N) = \mu} \beta_S(\CD_N) \text{ and } \gamma_{S,k}(\mu) = \sup_{\CD_Y \in \Delta_k:~\mu(\CD_Y) = \mu} \gamma_S(\CD_Y);
\end{equation}
then
\begin{equation}\label{eqn:alpha-optimize-over-mu}
    \alpha(f_{S,k}) = \inf_{\mu \in [-1,1]} \left(\frac{\beta_{S,k}(\mu)}{\gamma_{S,k}(\mu)}\right).
\end{equation}

The optimization problem on the right-hand side of \cref{eqn:alpha-optimize-over-mu} appears simpler than that of \cref{cor:cgsv-bool-approx} because it is univariate, but there is a hidden difficulty: Finding an explicit solution requires giving explicit formulas for $\beta_{S,k}(\mu)$ and $\gamma_{S,k}(\mu)$. In the case of $\twoand = f_{\{2\},2}$, Chou, Golovnev, Sudan, and Velusamy~\cite{CGSV21-boolean} first show that $\gamma_{\{2\},2}(\mu)$ is a linear function in $\mu$. Then, to find $\beta_{S,k}(\mu)$, they maximize the quadratic $\lambda_{\{2\}}(\CD_N,p)$ over $p \in [0,1]$ to find $\beta_{\{2\}}(\CD_N)$ (see \cref{ex:cgsv-2and}), and then optimize over $\CD_N$ such that $\mu(\CD_N)=\mu$ to find $\beta_{\{2\},2}(\mu)$.

While we'll see in \cref{sec:lambda-gamma-formulas} below that $\gamma_{S,k}(\mu)$ is piecewise linear in $\mu$ for all symmetric predicates $f_{S,k}$, we do not know how to find closed forms for $\beta_{S,k}(\mu)$ even for $\threeand$ (though $\gamma_{S,k}(\mu)$ is in general a piecewise linear function of $\mu$, see \cref{lemma:sym-bool-gamma} below). Thus, in this work we introduce a different formulation of the optimization problem:

\begin{equation}\label{eqn:alpha-optimize-over-dn}
    \alpha(f_{S,k}) =  \inf_{\CD_N \in \Delta_k} \left(\frac{\beta_S(\CD_N)}{\gamma_{S,k}(\mu(\CD_N))}\right).
\end{equation}

We view optimizing directly over $\CD_N \in \Delta_k$ as an important conceptual switch. In particular, our formulation emphasizes the calculation of $\beta_{S}(\CD_N)$ as the centrally difficult feature, yet we can still take advantage of the relative simplicity of calculating $\gamma_{S,k}(\mu)$.

\subsection{Our contribution: The max-min method}\label{sec:max-min}

\emph{A priori}, solving the optimization problem on the right-hand side of \cref{eqn:alpha-optimize-over-dn} still requires calculating $\beta_S(\CD_N)$, which involves maximizing a degree-$k$ polynomial. To get around this difficulty, we have made a key discovery, which was not noticed by Chou, Golovnev, Sudan, and Velusamy~\cite{CGSV21-boolean} even in the $\twoand$ case (see \cref{rem:cgsv-vs-bhp-2and}). Let $\CD_N^*$ minimize the right-hand side of \cref{eqn:alpha-optimize-over-dn}, and $p^*$ maximize $\lambda_S(\CD_N^*,\cdot)$. After substituting $ \beta_S(\CD) = \sup_{p \in [0,1]} \lambda_S (\CD,p)$ in \cref{eqn:alpha-optimize-over-dn}, and applying the max-min inequality, we get
\begin{equation}
\begin{aligned}
    \alpha(f_{S,k})  = \inf_{\CD_N \in \Delta_k}\sup_{p\in[0,1]} \left(\frac{\lambda_S(\CD_N,p)}{\gamma_{S,k}(\mu(\CD_N))}\right)
    &\geq \sup_{p\in[0,1]}  \inf_{\CD_N \in \Delta_k} \left(\frac{\lambda_S(\CD_N,p)}{\gamma_{S,k}(\mu(\CD_N))}\right)
    \\
    & \ge \inf_{\CD_N \in \Delta_k} \left(\frac{\lambda_S(\CD_N,p^*)}{\gamma_{S,k}(\mu(\CD_N))}\right)\, .\label{eqn:max-min}
\end{aligned}
\end{equation}

Given $p^*$, the right-hand side of \cref{eqn:max-min} is relatively easy to calculate, being a ratio of a linear and piecewise linear function of $\CD_N$. Our discovery is that, in a wide variety of cases, the quantity on the right-hand side of \cref{eqn:max-min} \emph{equals} $\alpha(f_{S,k})$; that is, $(\CD_N^*,p^*)$ is a \emph{saddle point} of $\frac{\lambda_S(\CD_N,p)}{\gamma_{S,k}(\mu(\CD_N))}$.\footnote{This term comes from the optimization literature; such points are also said to satisfy the ``strong max-min property'' (see, e.g., \cite[pp. 115, 238]{BV04}). The saddle-point property is guaranteed by von Neumann's minimax theorem for functions which are concave and convex in the first and second arguments, respectively, but this theorem and the generalizations we are aware of do not apply even to $\threeand$.}

This yields a novel technique, which we call the ``max-min method'', for finding a closed form for $\alpha(f_{S,k})$. First, we guess $\CD_N^*$ and $p^*$, and then, we show analytically that $\frac{\lambda_S(\CD_N,p)}{\gamma_{S,k}(\mu(\CD_N))}$ has a saddle point at $(\CD_N^*,p^*)$ and that $\lambda_S(\CD_N,p)$ is maximized at $p^*$. These imply that $\frac{\lambda_S(\CD_N^*,p^*)}{\gamma_{S,k}(\mu(\CD_N^*))}$ is a lower and upper bound on $\alpha(f_{S,k})$, respectively. For instance, in \cref{sec:kand-analysis}, in order to give a closed form for $\alpha(\kand)$ for odd $k$ (i.e., the odd case of \cref{thm:kand-approximability}), we guess $\CD_N^*\langle (k+1)/2 \rangle=1$ and $p^* = \frac{k+1}{2k}$ (by using Mathematica for small cases), and then check the saddle-point and maximization conditions in two separate lemmas (\cref{lemma:kand-lb,lemma:kand-ub}, respectively). Then, we show that $\alpha(\kand) = \alpha'_k$ by analyzing the right hand side of the appropriate instantiation of \cref{eqn:max-min}. We use similar techniques for $\kand$ for even $k$ (also \cref{thm:kand-approximability}) and for various other cases in \cref{sec:other-analysis,sec:k-1-k-analysis,sec:k+1/2-analysis}.

In all of these cases, the $\CD_N^*$ we construct is supported on at most two distinct Hamming weights, which is the property which makes finding $\CD_N^*$ tractable (using computer assistance). However, this technique is not a ``silver bullet'': it is not the case that the sketching approximability of every symmetric Boolean CSP can be exactly calculated by finding the optimal $\CD_N^*$ supported on two elements and using the max-min method. Indeed, (as mentioned in \cref{sec:other-analysis}) we verify using computer assistance that this is not the case for $f_{\{3\},4}$.

Finally, we remark that the saddle-point property is precisely what defines the value $p^*$ required for our simple classical algorithm for outputting approximately optimal assignments for $\mbTh$ where $f_{S,k} = \Th^t_k$ is a threshold function (see \cref{thm:thresh-bias-output-alg}).

To actually carry out the max-min method, we rely on the following simple inequality for optimizing ratios of linear functions:

\begin{proposition}\label{prop:lin-opt}
Let $f:\BR^n \to \BR$ be defined by the equation $f(\vecx) = \frac{\veca \cdot \vecx}{\vecb\cdot \vecx}$ for some $\veca,\vecb \in \BR_{\geq 0}^n$. For every $\vecy(1),\ldots,\vecy(r) \in \BR_{\geq 0}^n$, and every $\vecx = \sum_{i=1}^r \alpha_i \vecy(i)$ with each $x_i \geq 0$, we have $ f(\vecx) \geq \min_i f(\vecy(i))$. In particular, taking $r = n$ and $\vecy(1),\ldots,\vecy(n)$ as the standard basis for $\BR^n$, for every $\vecx \in \BR_{\geq 0}^n$, we have $f(\vecx) \geq \min_i \frac{a_i}{b_i}$.
\end{proposition}

\begin{proof}
Firstly, we show that it suffices WLOG to take the special case where $r=n$ and $\vecy(1),\ldots,\vecy(n)$ is the standard basis for $\BR^n$. Indeed, assume the special case and note that for a general case, we can let $\veca'=(\veca\cdot\vecy(1),\ldots,\veca\cdot\vecy(r))$, $\vecb'=(\vecb\cdot\vecy(1),\ldots,\vecb\cdot\vecy(r))$, $\vecx'=(x_1,\ldots,x_r)$, and let $\vecy'(1),\ldots,\vecy'(r)$ be the standard basis for $\BR^r$. Then $\vecx' = \sum_{i=1}^r \alpha_i \vecy'(i)$ and \[ f(\vecx) = \frac{\sum_{i=1}^r (\veca \cdot \vecy(i)) \alpha_i}{\sum_{i=1}^r (\vecb \cdot \vecy(i)) \alpha_i} = \frac{\veca' \cdot \vecx'}{\vecb' \cdot \vecx'} \geq \min_{i \in [r]} \frac{\veca' \cdot \vecy'(i)}{\vecb' \cdot \vecy'(i)} = \min_{i\in[r]} \frac{\veca \cdot \vecy(i)}{\vecb \cdot \vecy(i)}. \]

Now we prove the special case: Assume $r=n$ and $\vecy(1),\ldots,\vecy(n)$ is the standard basis for $\BR^n$. We have $f(\vecy(i)) = \frac{a_i}{b_i}$. Assume WLOG that $f(\vecy(1)) = \min \{f(\vecy(i)): i \in [n]\}$, i.e., $\frac{a_1}{b_1} \leq \frac{a_i}{b_i}$ for all $i \in [n]$. Then $a_i \geq \frac{a_1b_i}{b_1}$ for all $i \in [n]$, so \[ \veca \cdot \vecx \geq \sum_{i=1}^n \frac{a_1b_i}{b_1} \alpha_i = \frac{a_1}{b_1} (\vecb \cdot \vecx). \] Hence \[ f(\vecx) = \frac{\veca \cdot \vecx}{\vecb \cdot \vecx} \geq \frac{a_1}{b_1} = f(\vecy(1)), \] as desired.
\end{proof}

\subsection{Streaming lower bounds}\label{sec:cgsv-streaming-failure-3and-overview}

Given this setup, we also can state our results on \cite{CGSV21-boolean}'s streaming lower bounds' applicability (or lack thereof) to $\mthreeand$:

\begin{theorem}\label{thm:cgsv-streaming-failure-3and}
There is no infinite sequence $(\CD_Y^{(1)},\CD_N^{(1)}),(\CD_Y^{(2)},\CD_N^{(2)}),\ldots$ of padded one-wise pairs on $\Delta_3$ such that \[ \lim_{t \to \infty} \frac{\beta_{\{3\}}(\CD_N^{(t)})}{\gamma_{\{3\}}(\CD_Y^{(t)})} = \frac29. \]
\end{theorem}

Yet we still can achieve decent bounds using padded one-wise pairs:

\begin{observation}\label{obs:cgsv-streaming-3and-lb}
The padded one-wise pair $\CD_N=(0,0.45,0.45,0.1),\CD_Y=(0.45,0,0,0.55)$ (discovered by numerical search) \emph{does} prove a streaming approximability upper bound of $\approx .2362$ for $\threeand$, which is still quite close to $\alpha(\threeand)=\frac29$.
\end{observation}

\cref{thm:cgsv-streaming-failure-3and} is proven formally in \cref{sec:cgsv-streaming-failure-3and}; here is a proof outline:

\begin{proof}[Proof outline]
As discussed in \cref{sec:max-min}, since $k=3$ is odd, to prove \cref{thm:kand-approximability} we show, using the max-min method, that $\CD_N^* = (0,0,1,0)$ minimizes $\frac{\beta_{\{3\}}(\cdot)}{\gamma_{\{3\},3}(\mu(\cdot))}$. We can show that the corresponding $\gamma_{\{3\},3}$ value is achieved by $\CD_Y^* = (\frac13,0,0,\frac23)$. In particular, $(\CD_N^*,\CD_Y^*)$ are not a padded one-wise pair.

We can show that the minimizer of $\gamma_{\{3\}}$ for a particular $\mu$ is in general unique. Hence, it suffices to furthermore show that $\CD_N^*$ is the \emph{unique} minimizer of $\frac{\beta_{\{3\}}(\cdot)}{\gamma_{\{3\},3}(\mu(\cdot))}$. For this purpose, the max-min method is not sufficient because $\frac{\lambda_{\{3\}}(\cdot,p^*)}{\gamma_{\{3\},3}(\mu(\cdot))}$ is not uniquely minimized at $\CD_N^*$ (where we chose $p^* = \frac23$). Intuitively, this is because $p^*$ is not a good enough estimate for the maximizer of $\lambda_{\{3\}}(\CD_N,\cdot)$. To remedy this, we observe that $\lambda_{\{3\}}((1,0,0,0),\cdot),\lambda_{\{3\}}((0,1,0,0),\cdot)$, $\lambda_{\{3\}}((0,0,1,0),\cdot)$ and $\lambda_{\{3\}}((0,0,0,1),\cdot)$ are minimized at $0,\frac13,\frac23$, and $1$, respectively. Hence, we instead lower-bound $\lambda_{\{3\}}(\CD_N,\cdot)$ by evaluating at $\frac13 \CD_N\langle 1 \rangle + \frac23 \CD_N\langle 2 \rangle + \CD_N\langle 3 \rangle$, which does suffice to prove the uniqueness of $\CD_N^*$. The theorem then follows from continuity arguments.
\end{proof}

\section{Explicit formulas for $\lambda_S$ and $\gamma_{S,k}$}\label{sec:lambda-gamma-formulas}

In this section, we state and prove explicit formula for $\lambda_S(\CD,p)$ and $\gamma_{S,k}(\mu)$ which will be useful in later sections.

\begin{lemma}\label{lemma:sym-bool-lambda}
For any $\CD \in \Delta_k$ and $p \in [0,1]$, we have \[ \lambda_S(\CD,p) = \sum_{s \in S} \sum_{i=0}^k \left(\sum_{j=\max\{0,s-(k-i)\}}^{\min\{i,s\}} {i \choose j} {k-i \choose s-j} q^{s+i-2j} p^{k-s-i+2j} \right) \CD\langle i \rangle \] where $q \eqdef 1-p$.
\end{lemma}

\begin{proof}
By linearity of expectation and symmetry, it suffices to fix $s$ and $i$ and calculate, given a fixed string $\veca = (a_1,\ldots,a_k)$ of Hamming weight $i$ and a random string $\vecb = (b_1,\ldots,b_k) \sim \Bern_p^k$, the probability of the event $\|\veca + \vecb\|_0 = s$.

Let $A = \supp(\veca) = \{t \in [k] : a_t = 1\}$ and similarly $B = \supp(\vecb)$. We have $|A| = i$ and \[ s = \|\veca + \vecb\|_0 = |A \cap B| + |([k] \setminus A) \cap ([k] \setminus B)|. \] Let $j = |A \cap B|$, and consider cases based on $j$.

Given fixed $j$, we must have $|A \cap B| = j$ and $|([k] \setminus A) \cap ([k] \setminus B)| = s-j$. Thus if $j$ satisfies $j \leq i, s-j \leq k-i,j\geq0,j\leq s$, we have $\binom{i}{j}$ choices for $A \cap B$ and $\binom{k-i}{s-j}$ choices for $([k]\setminus A) \cap ([k] \setminus B)$; together, these completely determine $B$. Moreover $\|\vecb\|_0 = |B| = |B \cap A| + |B \cap ([k] \setminus A)| = j + (k-i)-(s-j) = k - s - i + 2j$, yielding the desired formula.
\end{proof}

\begin{lemma}\label{lemma:sym-bool-gamma}
Let $S \subseteq [k]$, and let $s$ be its smallest element and $t$ its largest element (they need not be distinct). Then for $\mu \in [-1,1]$, \[\gamma_{S,k}(\mu) = \begin{cases}\frac{1+\mu}{1+\epsilon_{s,k}} & \mu \in [-1,\epsilon_{s,k}) \\
1 & \mu \in [\epsilon_{s,k},\epsilon_{t,k}] \\ \frac{1-\mu}{1-\epsilon_{t,k}} & \mu \in (\epsilon_{t,k},1] \end{cases} \] (which also equals $\min\left\{\frac{1+\mu}{1+\epsilon_{s,k}}, 1, \frac{1-\mu}{1-\epsilon_{t,k}}\right\}$).
\end{lemma}

\begin{proof}
For $\mu \in [-1,1]$, in (\cref{eqn:beta_Sk-gamma_Sk-def}) we defined \[ \gamma_{S,k}(\mu) = \sup_{\CD_Y \in \Delta_k : \mu(\CD_Y) = \mu} \gamma_S(\CD_Y), \] where by \cref{eqn:cgsv-bool-lgb}, $\gamma_S(\CD_Y) = \sum_{i \in S} \CD_Y\langle i \rangle$. For $\CD_Y \in \Delta_k$, let $\supp(\CD_Y) = \{i \in [k]:\CD_Y\langle i \rangle > 0\}$. We handle cases based on $\mu$.

\subparagraph*{Case 1: $\mu \in [-1,\epsilon_{s,k}]$.} Our strategy is to reduce to the case $\supp(\CD_Y) \subseteq \{0,s\}$ while preserving the marginal $\mu$ and (non-strictly) increasing the value of $\gamma_S$.

Consider the following operation on a distribution $\CD_Y \in \Delta_k$: For $u < v < w \in [k]$, increase $\CD_Y\langle u \rangle$ by $\CD_Y\langle v \rangle \,\frac{w-v}{w-u}$, increase $\CD_Y\langle w \rangle$ by $\CD_Y\langle v \rangle\, \frac{v-u}{w-u}$, and set $\CD_Y\langle v \rangle$ to zero. Note that this results in a new distribution with the same marginal, since \[ \CD_Y\langle v \rangle \frac{w-v}{w-u} \epsilon_{u,k} + \CD_Y\langle v \rangle \frac{v-u}{w-u} \epsilon_{w,k} = \CD_Y\langle v \rangle \,\epsilon_{v,k}. \] Given an initial distribution $\CD_Y$, we can apply this operation to zero out $\CD_Y\langle v \rangle$ for $v \in \{1,\ldots,s-1\}$ by redistributing to $\CD_Y\langle 0 \rangle$ and $\CD_Y\langle s \rangle$, preserving the marginal and only increasing the value of $\gamma_S$ (since $v \not\in S$ while $s \in S$). Similarly, we can redistribute $\CD_Y\langle v \rangle$ to $\CD_Y\langle t \rangle$ and $\CD_Y\langle k \rangle$ when $v \in \{t+1,\ldots,k-1\}$, and to $\CD_Y\langle s \rangle$ and $\CD_Y\langle t \rangle$ when $v \in \{s+1,\ldots,t-1\}$. Thus, we need only consider the case $\supp(\CD) \subseteq \{0,s,t,k\}$. We assume for simplicity that $0,s,t,k$ are distinct.

By definition of $\epsilon_{i,k}$ we have \[ \mu(\CD) = -\CD_Y\langle 0 \rangle + \CD_Y\langle s \rangle\left(-1+\frac{2s}k\right) + \CD_Y\langle t \rangle\left(-1+\frac{2t}k\right) + \CD_Y\langle k \rangle \leq -1 + \frac{2s}k \] (by assumption for this case). Substituting $\CD_Y\langle s \rangle = 1-\CD_Y\langle 0 \rangle-\CD_Y\langle t \rangle-\CD_Y\langle k \rangle$ and multiplying through by $\frac{k}2$, we have \[ k\CD_Y\langle k \rangle-s\CD_Y\langle 0 \rangle-s\CD_Y\langle t \rangle-s\CD_Y\langle k \rangle+t\CD_Y\langle t \rangle \leq 0; \] defining $\delta = \CD_Y\langle t \rangle (\frac{t}s-1)+\CD_Y\langle k \rangle(\frac{k}s-1)$, we can rearrange to get $\CD_Y\langle 0 \rangle \geq \delta$. Then given $\CD_Y$, we can zero out $\CD_Y\langle t \rangle$ and $\CD_Y\langle k \rangle$, decrease $\CD_Y\langle 0 \rangle$ by $\delta$, and correspondingly increase $\CD_Y\langle s \rangle$ by $\CD_Y\langle t \rangle+\CD_Y\langle k \rangle+\delta$. This preserves the marginal since \[ (\delta + \CD_Y\langle t \rangle + \CD_Y\langle k \rangle) \,\epsilon_{s,k} = -\delta + \CD_Y\langle t \rangle \,\epsilon_{t,k} + \CD_Y\langle k \rangle \] and can only increase $\gamma_S$.

Thus, it suffices to only consider the case $\supp(\CD_Y) \subseteq \{0,s\}$. This uniquely determines $\CD_Y$ (because $\mu$ is fixed); we have $\CD_Y\langle 0 \rangle = \frac{\epsilon_{s,k}-\mu}{\epsilon_{s,k}+1}$ and $\CD_Y\langle s \rangle = \frac{1+\mu}{\epsilon_{s,k}+1}$, yielding the desired value of $\gamma_S$.

\subparagraph*{Case 2: $\mu \in [\epsilon_{s,k},\epsilon_{t,k}]$.} We simply construct $\CD_Y$ with $\CD_Y\langle s \rangle = \frac{\epsilon_{t,k}-\mu}{\epsilon_{s,k}-\epsilon_{t,k}}$ and $\CD_Y\langle t \rangle = \frac{\mu-\epsilon_{s,k}}{\epsilon_{s,k}-\epsilon_{t,k}}$; we have $\mu(\CD_Y) = \mu$ and $\gamma_S(\CD_Y) = 1$.

\subparagraph*{Case 3: $\mu \in [\epsilon_{t,k},1]$.} Following the symmetric logic to Case 1, we consider $\CD_Y$ supported on $\{t,k\}$ and set $\CD_Y\langle t \rangle = \frac{1-\mu}{1-\epsilon_{t,k}}$ and $\CD_Y\langle k \rangle = \frac{\mu-\epsilon_{t,k}}{1-\epsilon_{t,k}}$, yielding $\mu(\CD_Y) = \mu$ and $\gamma_S(\CD_Y) = \CD_Y\langle t \rangle$.
\end{proof}

\section{Sketching approximation ratio analyses}

\subsection{$\mkand$}\label{sec:kand-analysis}

In this section, we prove \cref{thm:kand-approximability} (on the sketching approximability of $\mkand$). Recall that in \cref{eqn:alpha'_k}, we defined \[ \alpha'_k = \left(\frac{(k-1)(k+1)}{4k^2}\right)^{(k-1)/2}. \] \cref{thm:kand-approximability} follows immediately from the following two lemmas:

\begin{lemma}\label{lemma:kand-ub}
For all odd $k \geq 3$, $\alpha(\kand) \leq \alpha'_k$. For all even $k \geq 2$, $\alpha(\kand) \leq 2\alpha'_{k+1}$.
\end{lemma}

\begin{lemma}\label{lemma:kand-lb}
For all odd $k \geq 3$, $\alpha(\kand) \geq \alpha'_k$. For all even $k \geq 2$, $\alpha(\kand) \geq 2\alpha'_{k+1}$.
\end{lemma}

To begin, we give explicit formulas for $\gamma_{\{k\},k}(\mu(\CD))$ and $\lambda_{\{k\}}(\CD,p)$. Note that the smallest element of $\{k\}$ is $k$, and $\epsilon_{k,k} = 1$. Thus, for $\CD \in \Delta_k$, we have by \cref{lemma:sym-bool-gamma,lemma:sym-bool-mu} that
\begin{equation}\label{eqn:kand-gamma}
    \gamma_{\{k\},k}(\mu(\CD)) = \frac{1+\sum_{i=0}^k (-1+\frac{2i}k)\,\CD\langle i \rangle}{2} = \sum_{i=0}^k \frac{i}k \,\CD\langle i \rangle.
\end{equation}
Similarly, we can apply \cref{lemma:sym-bool-lambda} with $s = k$; for each $i \in \kz$, $\max\{0,s-(k-i)\}=\min\{i,k\} = i$, so we need only consider $j=i$, and then $\binom{i}{j} = \binom{k-i}{s-j} = 1$. Thus, for $q = 1-p$, we have
\begin{equation}\label{eqn:kand-lambda}
    \lambda_{\{k\}}(\CD,p)  = \sum_{i=0}^k q^{k-i}p^i\,\CD\langle i \rangle
\end{equation}

Now, we prove \cref{lemma:kand-ub} directly:

\begin{proof}[Proof of \cref{lemma:kand-ub}]
Consider the case where $k$ is odd. Define $\CD_N^*$ by $\CD_N^*\langle (k+1)/2 \rangle=1$ and let $p^* = \frac12 + \frac1{2k}$. Since \[ \alpha(\kand) \leq \frac{\beta_{\{k\}}(\CD_N^*)}{\gamma_{\{k\},k}(\mu(\CD_N^*))} \text{ and } \beta_{\{k\}}(\CD_N) = \sup_{p \in [0,1]} \lambda_{\{k\}}(\CD_N^*,p), \] by \cref{eqn:alpha-optimize-over-dn,eqn:cgsv-bool-lgb}, respectively, it suffices to check that $p^*$ maximizes $\lambda_{\{k\}}(\CD_N^*,\cdot)$ and \[ \frac{\lambda_{\{k\}}(\CD_N^*,p^*)}{\gamma_{\{k\},k}(\mu(\CD_N^*))} = \alpha'_k. \] Indeed, by \cref{eqn:kand-lambda}, \[ \lambda_{\{k\}}(\CD_N^*,p) = (1-p)^{(k-1)/2} p^{(k+1)/2}. \] To show $p^*$ maximizes $\lambda_{\{k\}}(\CD_N^*,\cdot)$, we calculate its derivative: \[ \frac{d}{dp}\left[(1-p)^{(k-1)/2} p^{(k+1)/2}\right] = - (1-p)^{(k-3)/2}p^{(k-1)/2}\left(kp-\frac{k+1}2\right), \] which has zeros only at $0,1,$ and $p^*$. Thus, $\lambda_{\{k\}}(\CD_N^*,\cdot)$ has critical points only at $0,1,$ and $p^*$, and it is maximized at $p^*$ since it vanishes at $0$ and $1$. Finally, by \cref{eqn:kand-gamma,eqn:kand-lambda} and the definition of $\alpha'_k$, \[ \frac{\lambda_{\{k\}}(\CD_N^*,p^*)}{\gamma_{\{k\},k}(\mu(\CD_N^*))} = \frac{\left(\frac12-\frac1{2k}\right)^{(k-1)/2} \left(\frac12+\frac1{2k}\right)^{(k+1)/2}}{\frac12 \left(1+\frac1k\right)} = \alpha'_k, \] as desired.

Similarly, consider the case where $k$ is even; here, we define $\CD_N^*$ by $\CD_N^*\langle k/2 \rangle = \frac{\left(\frac{k}2+1\right)^2}{\left(\frac{k}2\right)^2+\left(\frac{k}2+1\right)^2}$ and $\CD_N^*\langle \frac{k}2+1 \rangle = \frac{\left(\frac{k}2\right)^2}{\left(\frac{k}2\right)^2+\left(\frac{k}2+1\right)^2}$, and set $p^* = \frac12+\frac1{2(k+1)}$. Using \cref{eqn:kand-lambda} to calculate the derivative of $\lambda_{\{k\}}(\CD_N^*,\cdot)$ yields
\begin{multline*}
    \frac{d}{dp}\left[\frac{\left(\frac{k}2+1\right)^2}{\left(\frac{k}2\right)^2+\left(\frac{k}2+1\right)^2} (1-p)^{k/2}p^{k/2} + \frac{\left(\frac{k}2\right)^2}{\left(\frac{k}2\right)^2+\left(\frac{k}2+1\right)^2}(1-p)^{k/2-1}p^{k/2+1}\right] \\
    = -\frac{k}{2+2k+2k^2}(1-p)^{k/2-2} p^{k/2-1} \left(\frac{k}2+1-2p\right)\left((k+1)p-\left(\frac{k}2+1\right)\right),
\end{multline*} so $\lambda_{\{k\}}(\CD_N^*,\cdot)$ has critical points at $0,1,\frac12+\frac{k}4$. and $p^*$; $p^*$ is the only critical point in the interval $[0,1]$ for which $\lambda_{\{k\}}(\CD_N^*,\cdot)$ is positive, and hence is its maximum. Finally, it can be verified algebraically using \cref{eqn:kand-gamma,eqn:kand-lambda} that $\frac{\lambda_{\{k\}}(\CD_N^*,p^*)}{\gamma_{\{k\},k}(\mu(\CD_N^*))} = 2\alpha'_{k+1}$, as desired.
\end{proof}

We prove \cref{lemma:kand-lb} using the max-min method.

\begin{proof}[Proof of \cref{lemma:kand-lb}]
First, suppose $k \geq 3$ is odd. Set $p^* = \frac12+\frac1{2k} = \frac{k+1}{2k}$. We want to show that
\begin{align*}
    \alpha'_k &\leq \inf_{\CD_N \in \Delta_k} \frac{\lambda_{\{k\}}(\CD_N,p^*)}{\gamma_{\{k\},k}(\mu(\CD_N))} \tag{max-min inequality, i.e., \cref{eqn:max-min}} \\
    &= \inf_{\CD_N \in \Delta_k} \frac{\sum_{i=0}^k (1-p^*)^{k-i}(p^*)^i\,\CD_N\langle i \rangle}{\sum_{i=0}^k \frac{i}k \,\CD_N\langle i \rangle} \tag{\cref{eqn:kand-gamma,eqn:kand-lambda}}.
\end{align*}

By \cref{prop:lin-opt}, it suffices to check that \[\forall i \in \kz,\quad (1-p^*)^{k-i}(p^*)^i \geq \alpha'_k \cdot \frac{i}k. \] By definition of $\alpha'_k$, we have that $\alpha'_k  = (1-p^*)^{(k-1)/2} (p^*)^{(k-1)/2}$. Defining $r = \frac{p^*}{1-p^*} = \frac{k+1}{k-1}$ (so that $p^* = r(1-p^*)$), factoring out $(1-p^*)^k$, and simplifying, we can rewrite our desired inequality as
\begin{equation}\label{eqn:kand-odd-lb-goal}
    \forall i \in \kz, \quad \frac12 (k-1) r^{i-\frac{k-1}2} \geq i.
\end{equation}
When $i = (k+1)/2$ or $(k-1)/2$, we have equality in \cref{eqn:kand-odd-lb-goal}. We extend to the other values of $i$ by induction. Indeed, when $i \geq (k+1)/2$, then ``$i$ satisfies \cref{eqn:kand-odd-lb-goal}'' implies ``$i+1$ satisfies \cref{eqn:kand-odd-lb-goal}'' because $r i \geq i + 1$, and when $i \leq (k-1)/2$, then ``$i$ satisfies \cref{eqn:kand-odd-lb-goal}'' implies ``$i-1$ satisfies \cref{eqn:kand-odd-lb-goal}'' because $\frac1r i \geq i - 1$.

Similarly, in the case where $k \geq 2$ is even, we set $p^* = \frac12+\frac1{2(k+1)}$ and $r = \frac{p^*}{1-p^*} = \frac{k+2}{k}$. In this case, for $i \in \kz$ the following analogue of \cref{eqn:kand-odd-lb-goal} can be derived: \[ \forall i \in \kz, \quad \frac12 k r^{i-\frac{k}2} \geq i, \] and these inequalities follow from the same inductive argument.
\end{proof}

\subsection{$\Th^{k-1}_k$ for even $k$}\label{sec:k-1-k-analysis}

In this subsection, we prove \cref{thm:k-1-k-approximability} (on the sketching approximability of $\Th^{k-1}_k$ for even $k \geq 2$). It is necessary and sufficient to prove the following two lemmas:

\begin{lemma}\label{lemma:k-1-k-ub}
For all even $k \geq 2$, $\alpha(\Th^{k-1}_k) \leq \frac{k}{2} \alpha_{k-1}'$.
\end{lemma}

\begin{lemma}\label{lemma:k-1-k-lb}
For all even $k \geq 2$, $\alpha(\Th^{k-1}_k) \geq \frac{k}{2} \alpha_{k-1}'$.
\end{lemma}

Firstly, we give explicit formulas for $\gamma_{\{k-1,k\},k}$ and $\lambda_{\{k-1,k\}}$. We have $\Th_k^{k-1} = f_{\{k-1,k\},k}$, and $\epsilon_{k-1,k} = -1+\frac{2(k-1)}{k} = 1-\frac2k$. Thus, \cref{lemma:sym-bool-gamma,lemma:sym-bool-mu} give
\begin{equation}
    \gamma_{\{k-1,k\},k}(\mu(\CD)) = \min \left\{\frac{1 + \sum_{i=0}^k (-1+\frac{2i}k) \,\CD\langle i \rangle}{2-\frac2k} ,1\right\} = \min\left\{\sum_{i=0}^k \frac{i}{k-1}\,\CD\langle i \rangle,1\right\}.\label{eqn:k-1-k-gamma}
\end{equation}
Next, we calculate $\lambda_{\{k-1,k\}}(\CD,p)$ with \cref{lemma:sym-bool-lambda}. Let $q = 1-p$, and let us examine the coefficient on $\CD\langle i \rangle$. $s=k$ contributes $q^{k-i}p^k$. In the case $i \leq k-1$, $s=k-1$ contributes $(k-i)q^{k-i-1}p^{i+1}$ for $j = i$, and in the case $i \geq 1$, $s=k-1$ contributes $iq^{k-i+1}p^{i-1}$ for $j = i-1$. Thus, altogether we can write
\begin{equation}
    \lambda_{\{k-1,k\}} (\CD,p) = \sum_{i=0}^k q^{k-i-1} p^{i-1} \left((k-i)p^2+ pq + iq^2\right) \,\CD\langle i \rangle.\label{eqn:k-1-k-lambda}
\end{equation}

Now, we prove \cref{lemma:k-1-k-ub,lemma:k-1-k-lb}.

\begin{proof}[Proof of \cref{lemma:k-1-k-ub}]
As in the proof of \cref{lemma:kand-ub}, it suffices to construct $\CD_N^*$ and $p^*$ such that $p^*$ maximizes $\lambda_{\{k-1,k\}}(\CD_N^*,\cdot)$ and $\frac{\lambda_{\{k-1,k\}}(\CD_N^*,p^*)}{\gamma_{\{k-1,k\},k}(\mu(\CD_N^*))} = \frac{k}2 \alpha'_{k-1}$.

We again let $p^* = \frac12+\frac1{2(k-1)}$, but define $\CD_N^*$ by $\CD_N^*\langle k/2 \rangle=\frac{\left(\frac{k}2\right)^2}{\left(\frac{k}2\right)^2+\left(\frac{k}2-1\right)^2}$ and $\CD_N^*\langle \frac{k}2+1 \rangle=\frac{\left(\frac{k}2-1\right)^2}{\left(\frac{k}2\right)^2+\left(\frac{k}2-1\right)^2}$. By \cref{eqn:k-1-k-lambda}, the derivative of $\lambda_{\{k-1,k\}}(\CD_N^*,\cdot)$ is now \begin{multline*}
    \frac{d}{dp}\Bigg[\frac{\left(\frac{k}2\right)^2}{\left(\frac{k}2\right)^2+\left(\frac{k}2-1\right)^2} (1-p)^{k/2-1} p^{k/2-1} \left(\frac{k}2p^2+pq+\frac{k}2q^2\right) + \\ \frac{\left(\frac{k}2-1\right)^2}{\left(\frac{k}2\right)^2+\left(\frac{k}2-1\right)^2}(1-p)^{k/2-2} p^{k/2} \left(\left(\frac{k}2-1\right)p^2+pq+\left(\frac{k}2+1\right)q^2\right)\Bigg] \\
    = -\frac1{8(k^2-2k+2)} (1-p)^{k/2-3}p^{k/2-2}(-k+(2(k-1)p) \xi(p),
\end{multline*} where $\xi(p)$ is the cubic \[ \xi(p) = -8k(k-1)p^3 + 2(k^3+k^2+6k-12)p^2 - 2(k^3-4)p + k^2(k-2). \] Thus, $\lambda_{\{k-1,k\}}$'s critical points on the interval $[0,1]$ are $0,1,p^*$ and any roots of $\xi$ in this interval. We claim that $\xi$ has no additional roots in the interval $(0,1)$. This can be verified directly by calculating roots for $k = 2,4$, so assume WLOG $k \geq 6$.

Suppose $\xi(p) = 0$ for some $p \in (0,1)$, and let $x = \frac1p - 1 \in (0,\infty)$. Then $p = \frac1{1+x}$; plugging this in for $p$ and multiplying through by $(x+1)^3$ gives the new cubic
\begin{equation}\label{eqn:for-descartes}
    (k^3-2k^2)x^3+(k^3-6k^2+8)x^2+(k^3-4k^2+12k-8)x+(k^3-8k^2+20k-16) = 0
\end{equation}
whose coefficients are cubic in $k$. It can be verified by calculating the roots of each coefficient of $x$ in \cref{eqn:for-descartes} that all coefficients are positive for $k \geq 6$. Thus, \cref{eqn:for-descartes} cannot have roots for positive $x$, a contradiction. Hence $\lambda_{\{k-1,k\}}(\CD_N^*,\cdot)$ is maximized at $p^*$. Finally, it can be verified that $\frac{\lambda_{\{k-1,k\}}(\CD_N^*,p^*)}{\gamma_{\{k-1,k\},k}(\mu(\CD_N^*))} = \frac{k}2\alpha'_{k-1}$, as desired.
\end{proof}

\begin{proof}[Proof of \cref{lemma:k-1-k-lb}]
Define $p^* = \frac12+\frac1{2(k-1)}$. Following the proof of \cref{lemma:kand-lb} and using the lower bound $\gamma_{\{k-1,k\},k}(\mu(\CD_N)) \leq \sum_{i=0}^k \frac{i}{k-1} \,\CD_N\langle i \rangle$, it suffices to show that \[ \frac{k}2 \alpha'_{k-1} \leq \inf_{\CD_N \in \Delta_k} \frac{\sum_{i=0}^k (1-p^*)^{k-i-1} (p^*)^{i-1} ((k-i)(p^*)^2+p^*(1-p^*)+i(1-p^*)^2) \,\CD_N\langle i \rangle}{\sum_{i=0}^k \frac{i}{k-1}\,\CD_N\langle i \rangle} \] for which by \cref{prop:lin-opt}, it in turn suffices to prove that for each $i \in \kz$, \[ \frac{k}2 \alpha'_{k-1} \frac{i}{k-1} \leq (1-p^*)^{k-i-1} (p^*)^{i-1} ((k-i)(p^*)^2+p^*(1-p^*)+i(1-p^*)^2). \] We again observe that $\alpha'_{k-1} = (1-p^*)^{k/2-1} (p^*)^{k/2-1}$, define $r = \frac{p^*}{1-p^*} = \frac{k}{k-2}$, and factor out $(1-p^*)^{k-1}$, which simplifies our desired inequality to
\begin{equation}\label{eqn:k-1-k-lb-goal}
    \frac{1}{2} r^{i-\frac{k}2-1} \cdot \frac{k-2}{k-1} \left(i + r + (k-i) r^2\right) \geq i.
\end{equation} for each $i \in \kz$. Again, we assume $k \geq 6$ WLOG; the bases cases $i = k/2-1,k/2$ can be verified directly, and we proceed by induction. If \cref{eqn:k-1-k-lb-goal} holds for $i$, and we seek to prove it for $i+1$, it suffices to cross-multiply and instead prove the inequality \[ r(i+1+r+(k-(i+1))r^2)i \geq (i+1) (i+r+(k-i)r^2), \] which simplifies to \[ (k-2i)(k-1)(k^2-4i-4) \leq 0, \] which holds whenever $ k/2\leq i \leq (k^2-4)/4$ (and $(k^2-4)/4\geq k$ for all $k \geq 6$). The other direction (where $i \leq k/2-1$ and we induct downwards) is similar.
\end{proof}

\begin{observation}\label{obs:th34-streaming-lb}
For $\Th^3_4$ the optimal $\CD_N^* = (0,0,\frac45,\frac15,0)$ does participate in a padded one-wise pair with $\CD_Y^* = (\frac4{15},0,0,\frac{11}{15},0)$ (given by $\CD_0 = (0,0,0,1,0)$, $\tau = \frac15$, $\CD'_N = (0,0,1,0,0)$, and $\CD'_Y=(\frac4{15},0,0,\frac{8}{15},0)$) so we can rule out \emph{streaming} $(\frac49+\epsilon)$-approximations to $\mbcsp[\Th^3_4]$ in $o(\sqrt n)$ space.
\end{observation}

\subsection{$f_{\{(k+1)/2\},k}$ for (small) odd $k$}\label{sec:k+1/2-analysis}

In this subsection, we prove bounds on the sketching approximability of $f_{\{(k+1)/2\},k}$ for odd $k\in \{3,\dots,51\}$. Define $\CD_{0,k} \in \Delta_k$ by $\CD_{0,k}\langle 0 \rangle=\frac{k-1}{2k}$ and $\CD_{0,k}\langle k \rangle = \frac{k+1}{2k}$. We prove the following two lemmas:

\begin{lemma}\label{lemma:k+1/2-ub}
For all odd $k \geq 3$, $\alpha(f_{\{(k+1)/2\},k}) \leq \lambda_{\{(k+1)/2\}}(\CD_{0,k},p'_k)$, where $p'_k \eqdef \frac{3 k - k^2 + \sqrt{4 k + k^2 - 2 k^3 + k^4}}{4k}$.
\end{lemma}

\begin{lemma}\label{lemma:k+1/2-lb}
The following holds for all odd $k \in \{3,\ldots,51\}$. For all $p \in [0,1]$, the expression $ \frac{\lambda_{\{(k+1)/2\}}(\cdot,p)}{\gamma_{\{(k+1)/2\},k}(\mu(\cdot))}$ is minimized at $\CD_{0,k}$.
\end{lemma}

We begin by writing an explicit formula for $\lambda_{\{(k+1)/2\}}$. \cref{lemma:sym-bool-lambda} gives \[ \lambda_{\{(k+1)/2\}}(\CD,p) = \sum_{i=0}^k \left(\sum_{j=\max\{0,i-(k-1)/2\}}^{\min\{i,(k+1)/2\}} \binom{i}j \binom{k}{(k+1)/2-j} (1-p)^{(k+1)/2+i-2j} p^{(k-1)/2-i+2j}\right)\,\CD\langle i \rangle. \] For $i \leq (k-1)/2$, the sum over $j$ goes from $0$ to $i$, and for $i \geq (k+1)/2$, it goes from $i-(k-1)/2$ to $(k+1)/2$. Thus, plugging in $\CD_{0,k}$, we get:
\begin{equation}
    \lambda_{\{(k+1)/2\}}(\CD_{0,k},p) = \binom{k}{(k+1)/2} \left(\frac{k-1}{2k}(1-p)^{(k+1)/2}p^{(k-1)/2} + \frac{k+1}{2k}(1-p)^{(k-1)/2}p^{(k+1)/2}\right).\label{eqn:k+1/2-d0_k}
\end{equation}

By \cref{lemma:sym-bool-gamma,lemma:sym-bool-mu}, $\gamma_{\{(k+1)/2\},k}(\mu(\CD_{0,k})) = \gamma_{\{(k+1)/2\},k}(\frac1k) = 1$. Thus, \cref{lemma:k+1/2-ub,lemma:k+1/2-lb} together imply the following theorem:

\begin{theorem}\label{thm:k+1/2-approximability}
For odd $k \in \{3,\ldots,51\}$, \[ \alpha(f_{\{(k+1)/2\},k}) = \binom{k}{(k+1)/2} \left(\frac{k-1}{2k}(1-p'_k)^{(k+1)/2}(p'_k)^{(k-1)/2} + \frac{k+1}{2k}(1-p'_k)^{(k-1)/2}(p'_k)^{(k+1)/2}\right), \] where $p'_k = \frac{3 k - k^2 + \sqrt{4 k + k^2 - 2 k^3 + k^4}}{4k}$ as in \cref{lemma:k+1/2-ub}.
\end{theorem}

Recall that $\rho(f_{(k+1)/2,k}) = \binom{k}{(k+1)/2} 2^{-k}$. Although we currently lack a lower bound on $\alpha(f_{\{(k+1)/2\},k})$ for large odd $k$, the upper bound from \cref{lemma:k+1/2-ub} suffices to prove \cref{thm:k+1/2-asymptotic-lb}, i.e., it can be verified that \[ \lim_{k \text{ odd} \to \infty} \frac{\binom{k}{(k+1)/2} \left( \frac{k-1}{2k}(1-p'_k)^{(k+1)/2}(p'_k)^{(k-1)/2} + \frac{k+1}{2k}(1-p'_k)^{(k-1)/2}(p'_k)^{(k+1)/2}\right)}{\rho(f_{\{(k+1)/2\},k})} = 1. \]

We remark that for $f_{\{(k+1)/2\},k}$, our lower bound (\cref{lemma:k+1/2-lb}) is \emph{stronger} than what we were able to prove for $\kand$ (\cref{lemma:kand-lb}) and $\Th^{k-1}_k$ (\cref{lemma:k-1-k-lb}) because the inequality holds regardless of $p$. This is fortunate for us, as the optimal $p^*$ from \cref{lemma:k+1/2-ub} is rather messy.\footnote{The analogous statement is false for e.g. $\threeand$, where we had $\CD_N^* = (0,0,1,0)$, but at $p=\frac34$, \[ \frac{\lambda_{\{3\}}((0,\frac12,\frac12,0),\frac34)}{\gamma_{\{3\},3}(\mu(0,\frac12,\frac12,0))} = \frac3{16} \leq \frac{27}{128} =\frac{\lambda_{\{3\}}((0,0,1,0),\frac34)}{\gamma_{\{3\},3}(\mu(0,0,1,0))}. \]} It remains to prove \cref{lemma:k+1/2-ub,lemma:k+1/2-lb}.

\begin{proof}[Proof of \cref{lemma:k+1/2-ub}]
Taking the derivative with respect to $p$ of \cref{eqn:k+1/2-d0_k} yields \[ \frac{d}{dp}\left[\lambda_{\{(k+1)/2\}} (\CD_{0,k},p)\right] = - \frac1{4k} \binom{k}{(k+1)/2} (pq)^{(k-3)/2} (4kp^2 + (2k^2-6k)p + (-k^2+2k-1)), \] where $q=1-p$. Thus, $\lambda_{\{(k+1)/2\}} (\CD_{0,k},\cdot)$ has critical points at $p=0,1,p'_k,$ and $\frac{3 k - k^2 - \sqrt{4 k + k^2 - 2 k^3 + k^4}}{4k}$. This last value is nonpositive for all $k \geq 0$ (since $(3k-k^2)^2-(4 k + k^2 - 2 k^3 + k^4)=-4k(k-1)^2$).
\end{proof}

The proof of our lower bound (\cref{lemma:k+1/2-lb}) is slightly different than those of our earlier lower bounds (i.e., \cref{lemma:kand-lb,lemma:k-1-k-lb}) in the following sense. For $i \in \kz$, let $\CD_i \in \Delta_k$ be defined by $\CD_i\langle i \rangle=1$. For $\kand$ (\cref{lemma:kand-lb}), we used the fact that $\frac{\lambda_{\{k\}}(\cdot,p^*)}{\gamma_{\{k\},k}(\mu(\cdot))}$ is a ratio of linear functions, and thus using \cref{prop:lin-opt}, it is sufficient to verify the lower bound at $\CD_0,\ldots,\CD_k$. For $\Th_k^{k-1}$ (\cref{lemma:k-1-k-lb}), $\frac{\lambda_{\{k-1,k\}}(\cdot,p^*)}{\gamma_{\{k-1,k\},k}(\mu(\cdot))}$ is \emph{not} a ratio of linear functions, because the denominator $\gamma_{\{k-1,k\},k}(\mu(\CD)) = \min\{\sum_{i=0}^k \frac{i}{k-1} \,\CD\langle i \rangle, 1\}$ is not linear over $\Delta_k$. However, we managed to carry out the proof by upper-bounding the denominator with the linear function $\gamma'(\CD) = \sum_{i=0}^k \frac{i}{k-1} \CD\langle i \rangle$, and then invoking \cref{prop:lin-opt} (again, to show that it suffices to verify the lower bound at $\CD_0,\ldots,\CD_k$).

For $f_{\{(k+1)/2\},k}$, we show that it suffices to verify the lower bound on a larger (but still finite) set of distributions.

\begin{proof}[Proof of \cref{lemma:k+1/2-lb}]
Recalling that $\epsilon_{(k+1)/2,k} = \frac1k$, let $\Delta^+_k = \{\CD \in \Delta_k : \mu(\CD) \leq \frac1k\}$ and $\Delta^-_k = \{\CD \in \Delta_k : \mu(\CD) \geq \frac1k\}$. Note that $\Delta_k^+ \cup \Delta_k^- = \Delta_k$, and restricted to either $\Delta_k^+$ or $\Delta_k^-$, $\gamma_{\{(k+1)/2\},k}(\mu(\cdot))$ is linear and thus we can apply \cref{prop:lin-opt} to $\frac{\lambda_{\{k-1,k\}}(\cdot,p^*)}{\gamma_{\{k-1,k\},k}(\mu(\cdot))}$.

Let $\CD_{i,j} \in \Delta_k$, for $i < (k+1)/2, j > (k+1)/2$, be defined by $\CD_{i,j}\langle i \rangle=\frac{2j-(k+1)}{2(j-i)}$ and $\CD_{i,j}\langle j \rangle=\frac{(k+1)-2i}{2(j-i)}$. Note that $\mu(\CD_{i,j}) = \frac1k$ for each $i,j$. We claim that $\{\CD_i\}_{i \leq (k+1)/2} \cup \{\CD_{i,j}\}$ are the extreme points of $\Delta_k^+$, or more precisely, that every distribution $\CD \in \Delta_k^+$ can be represented as a convex combination of these distributions. Indeed, this follows constructively from the procedure which, given a distribution $\CD$, subtracts from each $\CD\langle i \rangle$ for $i < (k+1)/2$ (adding to the coefficient of the corresponding $\CD_i$) until the marginal of the (renormalized) distribution is $\frac1k$, and then subtracts from pairs $\CD\langle i \rangle,\CD\langle j \rangle$ with $i < (k+1)/2$ and $j > (k+1)/2$, adding it to the coefficient of the appropriate $\CD_{i,j}$) until $\CD$ vanishes (i.e., $\CD\langle i \rangle$ is zero for all $i \in \kz$). Similarly, every distribution $\CD \in \Delta_k^-$ can be represented as a convex combination of the distributions $\{\CD_i\}_{i \geq (k+1)/2} \cup \{\CD_{i,j}\}$. Thus, by \cref{prop:lin-opt}, it is sufficient to verify that \[\frac{\lambda_{\{(k+1)/2\}}(\CD,p)}{\gamma_{\{(k+1)/2\},k}(\mu(\CD))} \geq \frac{\lambda_{\{(k+1)/2\}}(\CD_N^*,p)}{\gamma_{\{(k+1)/2\},k}(\mu(\CD_N^*))} \] for each $\CD \in \{\CD_i\}\cup\{\CD_{i,j}\}$. Treating $p$ as a variable, for each odd $k\in\{3,\ldots,51\}$ we produce a list of $O(k^2)$ degree-$k$ polynomial inequalities in $p$ which we verify using Mathematica.
\end{proof}

\subsection{Other symmetric predicates}\label{sec:other-analysis}

In \cref{table:other-sym-funcs} below, we list four more symmetric Boolean predicates (beyond $\kand$, $\Th^{k-1}_k$, and $f_{\{(k+1)/2\},k}$) whose sketching approximability we have analytically resolved using the ``max-min method''. These values were calculated using two subroutines in the Mathematica code, \texttt{estimateAlpha} --- which numerically or symbolically estimates the $\CD_N$, with a given support, which minimizes $\alpha$ --- and \texttt{testMinMax} --- which, given a particular $\CD_N$, calculates $p^*$ for that $\CD_N$ and checks analytically whether lower-bounding by evaluating $\lambda_S$ at $p^*$ proves that $\CD_N$ is minimal.

\begin{table}[h]
    \centering
    \begin{tabular}{|c|c|c|c|}
        \hline
         $S$ & $k$ & $\alpha$ & $\CD_N^*$ \\ \hline
         $\{2,3\}$ & $3$ & $\frac12+\frac{\sqrt3}{18} \approx 0.5962$ & $(0,\frac12,0,\frac12)$ \\ 
         $\{4,5\}$ & $5$ & $8\rroot(P_1) \approx 0.2831$ & $(0,0,1-\rroot(P_2),\rroot(P_2),0,0)$ \\
         $\{4\}$ & $5$ & $8\rroot(P_3) \approx 0.2394$ & $(0,0,1-\rroot(P_4),\rroot(P_4),0,0)$ \\
         $\{3,4,5\}$ & $5$ & $\frac12+\frac{3\sqrt5}{125} \approx 0.5537$ & $(0,\frac12,0,0,0,\frac12)$\\
         \hline
    \end{tabular}
    \caption{Symmetric predicates for which we have analytically calculated exact $\alpha$ values using the ``max-min method''. For a polynomial $P : \BR \to \BR$ with a \emph{unique} positive real root, let $\rroot(p)$ denote that root, and define the polynomials $P_1(z)=-72+4890z-108999z^2+800000z^3, P_2(z)=-908+5021z-9001z^2+5158z^3$, $P_3(z) = -60+5745z-183426z^2+1953125z^3$, $P_4(z) = -344+1770z-3102z^2+1811z^3$. (We note that in the $f_{\{4\},5}$ and $f_{\{4,5\},5}$ calculations, we were required to check equality of roots numerically (to high precision) instead of analytically).}
    \label{table:other-sym-funcs}
\end{table}

We remark that two of the cases in \cref{table:other-sym-funcs} (as well as $\kand$), the optimal $\CD_N$ is rational and supported on two coordinates. However, in the other two cases in \cref{table:other-sym-funcs}, the optimal $\CD_N$ involves roots of a cubic.

In \cref{sec:k+1/2-analysis}, we showed that $\CD_N^*$ defined by $\CD_N^*\langle 0 \rangle=\frac{k-1}{2k}$ and $\CD_N^*\langle k \rangle=\frac{k+1}{2k}$ is optimal for $f_{\{(k+1)/2\},k}$ for odd $k \in \{3,\ldots,51\}$. Using the same $\CD_N^*$, we are also able to resolve 11 other cases in which $S$ is ``close to'' $\{(k+1)/2\}$; for instance, $S=\{5,6\},\{5,6,7\},\{5,7\}$ for $k=9$. (We have omitted the values of $\alpha$ and $\CD_N$ because they are defined using the roots of polynomials of degree up to 8.)

In all previously-mentioned cases, the condition ``$\CD_N^*$ has support size $2$'' was helpful, as it makes the optimization problem over $\CD_N^*$ essentially univariate; however, we have confirmed analytically in two other cases ($S=\{3\},k=4$ and $S=\{3,5\},k=5$) that ``max-min method on distributions with support size two'' does not suffice for tight bounds on $\alpha$ (see \texttt{testDistsWithSupportSize2} in the Mathematica code). However, using the max-min method with $\CD_N$ supported on two levels still achieves decent (but not tight) bounds on $\alpha$. For $S = \{3\},k=4$, using $\CD_N = (\frac14,0,0,0,\frac34)$, we get the bounds $\alpha(f_{\{3\},4}) \in [0.3209,0.3295]$ (the difference being $2.67\%$). For $S = \{3,5\},k=5$, using $\CD_N = (\frac14,0,0,0,\frac34,0)$, we get $\alpha(f_{\{3,5\},5}) \in [0.3416,0.3635]$ (the difference being $6.42\%$).

Finally, we have also analyzed cases where we get numerical solutions which are very close to tight, but we lack analytical solutions because they likely involve roots of high-degree polynomials. For instance, in the case $S = \{4,5,6\}, k=6$,  setting $\CD_N = (0,0,0,0.930013,0,0,0.069987)$  gives $\alpha(f_{\{4,5,6\},6}) \in [0.44409972,0.44409973]$, differing only by 0.000003\%. (We conjecture here that $\alpha=\frac49$.) For $S=\{6,7,8\},k=8$, using
$\CD_N=(0,0,0,0,0.699501,0.300499)$, we get the bounds $\alpha(f_{\{6,7,8\},8}) \in [0.20848,0.20854]$ (the difference being $0.02\%$).\footnote{Interestingly, in this latter case, we get bounds differing by $2.12\%$ using $\CD_N=(0,0,0,0,\frac9{13},\frac4{13},0,0,0)$ in an attempt to continue the pattern from $f_{\{7,8\},8}$ and $f_{\{8\},8}$ (where we set $\CD_N^* = (0,0,0,0,\frac{16}{25},\frac9{25},0,0,0)$ and $(0,0,0,0,\frac{25}{41},\frac{16}{41},0,0,0)$ in \cref{sec:k-1-k-analysis} and \cref{sec:kand-analysis}, respectively).}

\section{Simple sketching algorithms for threshold predicates}\label{sec:thresh-alg}

Let $f_{S,k} = \Th^t_k$ be a threshold predicate (so that $S = \{t,\ldots,k\}$). The main goal of this section is to prove \cref{thm:thresh-bias-alg}, giving a simple ``bias-based'' sketching algorithm for $\mbTh$. Following our definition of the bias of a variable for $\mtwoand$ instances in \cref{sec:mdcut-algorithm}, given an instance $\Psi$ of $\mbTh$, for $i \in [n]$, let $\bias_\Psi(i)$ denote the total weight of clauses in which $x_i$ appears positively minus the weight of those in which it appears negatively; that is, if $\Psi$ consists of clauses $((\vecb(\ell),\vecj(\ell),w(\ell))_{\ell\in[m]}$, then \[ \bias_\Psi(i) \eqdef \sum_{\ell\in[m]:~j(\ell)_t = i} (-1)^{b(\ell)_t} w_\ell, \] and $\bias_\Psi \eqdef \frac1{kW_\Psi} \sum_{i=1}^n |\bias_\Psi(i)|$. $\bias_\Psi$ is measurable using $1$-norm sketching algorithms (i.e., \cref{thm:l1-sketching} due to \cite{Ind06,KNW10}, as used also in \cite{GVV17,CGV20,CGSV21-boolean}):

\begin{corollary}\label{cor:bias-sketching}
For every predicate $f : \BZ_2^k \to \{0,1\}$ and every $\epsilon>0$, there exists an $O(\log n/\epsilon^2)$-space randomized sketching algorithm for the following problem: The input is an instance $\Psi$ of $\mbTh$ (given as a stream of constraints), and the goal is to estimate $\bias_\Psi$ up to a multiplicative factor of $1\pm\epsilon$.
\end{corollary}

\begin{proof}
Invoke the $1$-norm sketching algorithm from \cref{thm:l1-sketching} as follows: On each input constraint $(\vecb,\vecj,\ell)$ with weight $w$, insert the updates $(j_1,wb_1),\ldots,(j_k,wb_k)$ into the stream (and normalize appropriately).
\end{proof}

Now recall the definitions of $\beta_{S,k}(\mu)$ and $\gamma_{S,k}(\mu)$ from \cref{eqn:alpha-optimize-over-mu}. Our simple algorithm for $\mbTh$ relies on the following two lemmas, which we prove below:

\begin{lemma}\label{lemma:thresh-value-ub}
$\val_\Psi \leq \gamma_{S,k}(\bias_\Psi)$.
\end{lemma}

\begin{lemma}\label{lemma:thresh-value-lb}
$\val_\Psi \geq \beta_{S,k}(\bias_\Psi)$.
\end{lemma}

Together, these two lemmas imply that outputting $\alpha(\Th^t_k)\cdot \gamma_{S,k}(\bias_\Psi)$ gives an $\alpha(\Th^t_k)$-approximation to $\mbTh$, since $\alpha(\Th^t_k) = \inf_{\mu \in [-1,1]} \frac{\beta_{S,k}(\mu)}{\gamma_{S,k}(\mu)}$ (\cref{eqn:alpha-optimize-over-mu}). We can implement this as a small-space sketching algorithm (losing an arbitrarily small additive constant $\epsilon > 0$ in the approximation ratio) because $\gamma_{S,k}(\cdot)$ is piecewise linear by \cref{lemma:sym-bool-gamma}:

\begin{proof}[Proof of \cref{thm:thresh-bias-alg}]
To get an $(\alpha-\epsilon)$-approximation to $\val_\Psi$, let $\delta > 0$ be small enough such that $\frac{1-\delta}{1+\delta}\alpha(\Th^t_k) \geq \alpha(\Th^t_k)-\epsilon$. We claim that calculating an estimate $\hat{b}$ for $\bias_\Psi$ (using \cref{cor:bias-sketching}) up to a multiplicative $\delta$ factor and outputting $\hat{v} = \alpha(\Th^t_k)\gamma_{S,k}(\frac{\hat{b}}{1+\delta})$ is sufficient. Note that $\gamma_{S,k}$ is monotone by \cref{lemma:sym-bool-gamma} because $\Th^t_k$ is a threshold function.

Indeed, suppose $\hat{b} \in [(1-\delta)\bias_\Psi,(1+\delta)\bias_\Psi]$; then $\frac{\hat{b}}{1+\delta} \in [\frac{1-\delta}{1+\delta}\bias_\Psi,\bias_\Psi]$. Now we observe
\begin{align*}
    \gamma_{S,k}\left(\frac{\hat{b}}{1+\delta}\right) &\geq \gamma_{S,k}\left(\frac{1-\delta}{1+\delta} \bias_\Psi\right) \tag{monotonicity of $\gamma_{S,k}$} \\ 
    &=\min\left\{\frac{1+\frac{1-\delta}{1+\delta}\bias_\Psi}{1+\epsilon_{s,k}},1\right\} \tag{\cref{lemma:sym-bool-gamma}} \\
    &\geq \frac{1-\delta}{1+\delta} \min\left\{\frac{1+\bias_\Psi}{1+\epsilon_{s,k}},1\right\} \tag{$\delta>0$} \\
    &= \frac{1-\delta}{1+\delta} \gamma_{S,k}(\bias_\Psi) \tag{\cref{lemma:sym-bool-gamma}}.
\end{align*}
Then we conclude
\begin{align*}
    (\alpha(\Th^t_k)-\epsilon)\val_\Psi &\leq (\alpha(\Th^t_k)-\epsilon)\gamma_{S,k}(\bias_\Psi) \tag{\cref{lemma:thresh-value-ub}} \\
    &\leq \alpha(\Th^t_k) \cdot \frac{1-\delta}{1+\delta}\gamma_{S,k}(\bias_\Psi) \tag{assumption on $\delta$} \\
    &\leq \hat{v} \tag{our observation} \\
    &\leq \alpha(\Th^t_k)\gamma_{S,k}(\bias_\Psi) \tag{monotonicity of $\gamma_{S,k}$} \\
    &\leq \beta_{S,k}(\bias_\Psi) \tag{\cref{eqn:alpha-optimize-over-mu}} \\
    &\leq \val_\Psi \tag{\cref{lemma:thresh-value-lb}} ,
\end{align*} as desired.
\end{proof}

\subsection{Proving the lemmas}\label{sec:template-dists-mbf}

To prove \cref{lemma:thresh-value-lb,lemma:thresh-value-ub}, we generalize the definition of template distributions for $\mtwoand$ from \cref{sec:mdcut-template-alg} to all $\mbf$ problems. For a predicate $f : \BZ_2^k \to \{0,1\}$, an instance $\Psi$ of $\mbf$ on $n$ variables, and an assignment $\vecx \in \BZ_2^n$, let $\CD^\vecx_\Psi \in \Delta(\BZ_2^k)$ denote the following \emph{template distribution}: If $\Psi$'s constraints are $(\vecb(\ell),\vecj(\ell),w(\ell))_{\ell \in [m]}$, we sample $\ell$ with probability $\frac{w(\ell)}{W_\Psi}$ and output $\vecb(\ell) + \vecx|_{\vecj(\ell)}$. That is, $\CD_\Psi^\vecx(\veca)$ is the fractional weight of constraints $\ell$ such that plugging in the assignment $\vecx$ and negating according to $\vecb(\ell)$ results in $\veca$. Note that $\CD_\Psi^\vecx$ is not necessarily symmetric; however, we can still define a scalar ``marginal'' $\mus(\CD) \eqdef \mu(\Sym(\CD))$. The important properties of $\CD_\Psi^\vecx$ are summarized in the following proposition:

\begin{proposition}\label{prop:d}
Let $\Psi$ be an instance of $\mbTh$. Then:
\begin{enumerate}[label={\roman*.},ref={\roman*}]
\item For any $p \in [0,1]$ and $\vecx \in \BZ_2^n$, $\E_{\veca \sim \Bern_p^n}[\val_\Psi(\veca+\vecx)] = \lambda_S(\CD^\vecx_\Psi,p)$.\label{item:d-val}
\item For every $\vecx \in \BZ_2^n$, $\mus(\CD^\vecx_\Psi) = \frac1{kW_\Psi} \sum_{i=1}^n (-1)^{x_i+1} \bias_\Psi(i) \leq \bias_\Psi$.\label{item:d-mu}
\item If for every $i$, $\bias_\Psi(i) > 0 \implies x_i = 1$ and $\bias_\Psi(i) < 0 \implies x_i = 0$, then $\mus(\CD^\vecx_\Psi) = \bias_\Psi$.\label{item:d-bias}
\end{enumerate}
\end{proposition}

\begin{proof}
Omitted (follows immediately from definitions and linearity of expectation).
\end{proof}

Now, we are equipped to prove the lemmas:

\begin{proof}[Proof of \cref{lemma:thresh-value-ub}]
Let $\vecx \in \BZ_2^n$ denote the optimal assignment for $\Psi$. Then
\begin{align*}
    \val_\Psi &= \val_\Psi(\vecx) \tag{def. of $\vecx$} \\
    &= \lambda_S(\CD_\Psi^\vecx, 1) \tag{\cref{item:d-val} of \cref{prop:d} with $p=1$} \\
    &= \gamma_S(\CD_\Psi^\vecx) \tag{\cref{eqn:cgsv-bool-lgb}} \\
    &\leq \gamma_{S,k}(\mus(\CD_\Psi^\vecx)) \tag{\cref{eqn:beta_Sk-gamma_Sk-def,prop:bool-sym}} \\
    &\leq \gamma_{S,k}(\bias(\Psi)) \tag{\cref{item:d-mu} of \cref{prop:d} and monotonicity of $\gamma_{S,k}$} \\
\end{align*}
as desired.
\end{proof}

\begin{proof}[Proof of \cref{lemma:thresh-value-lb}]
Let $\vecx \in \BZ_2^n$ denote the assignment assigning $x_i$ to $1$ if $\bias_\Psi(i) \geq 0$ and $0$ otherwise. Now
\begin{align*}
    \val_\Psi &\geq \sup_{p \in [0,1]} \left(\E_{\veca \sim \Bern_p^n}[\val_{\Psi}(\veca+\vecx)]\right) \tag{probabilistic method} \\
    &= \sup_{p \in [0,1]} (\lambda_S(\CD_\Psi^\vecx,p)) \tag{\cref{item:d-val} of \cref{prop:d}} \\
    & \geq \beta_S(\CD_\Psi^\vecx) \tag{\cref{eqn:cgsv-bool-lgb}} \\
    &\geq \beta_{S,k}(\mus(\CD_\Psi^\vecx)) \tag{\cref{eqn:beta_Sk-gamma_Sk-def,prop:bool-sym}} \\
    &= \beta_{S,k}(\bias(\Psi)) \tag{\cref{item:d-bias} of \cref{prop:d}}
\end{align*}
as desired.
\end{proof}

\subsection{The classical algorithm}

Finally, we state another consequence of \cref{lemma:thresh-value-ub} --- a simple randomized, $O(n)$-time-and-space streaming algorithm for \emph{outputting} approximately-optimal assignments when the max-min method applies.

\begin{theorem}\label{thm:thresh-bias-output-alg}
Let $\Th^t_k$ be a threshold predicate and $p^* \in [0,1]$ be such that the max-min method applies, i.e., \[ \alpha(\Th^t_k) = \inf_{\CD_N \in \Delta_k} \left(\frac{\lambda_S(\CD_N,p^*)}{\gamma_{S,k}(\mu(\CD_N))}\right). \] Then the following algorithm, on input $\Psi$, outputs an assignment with expected value at least $\alpha(\Th^t_k) \cdot \val_\Psi$: Assign every variable to $1$ if $\bias_\Psi(i) \geq 0$, and $0$ otherwise, and then flip each variable's assignment independently with probability $p^*$.
\end{theorem}

\begin{proof}[Proof of \cref{thm:thresh-bias-output-alg}]
Let $p^*$ be as in the theorem statement, and define $\vecx$ as in the proof of \cref{lemma:thresh-value-lb}. We output the assignment $\vecx + \veca$ for $\veca \sim \Bern_{p^*}^n$, and our goal is to show that its expected value is at least $\alpha(\Th^t_k) \val_\Psi$.

Applying the max-min method to $\Sym(\CD_\Psi^\vecx)$ and using \cref{prop:bool-sym}, we have:
\begin{equation}\label{eqn:max-min-for-alg}
    \lambda_S(\CD_\Psi^\vecx,p^*) \geq \alpha(\Th^t_k) \gamma_{S,k}(\mus(\CD_\Psi^\vecx)).
\end{equation}
Thus our expected output value is
\begin{align*}
    \E_{\veca\sim\Bern_{p^*}}[\val_\Psi(\vecx + \veca)] &= \lambda_S(\CD_\Psi^\vecx,p^*) \tag{\cref{item:d-val} of \cref{prop:d}} \\
    &\geq \alpha(\Th^t_k)\gamma_{S,k}(\mus(\CD_\Psi^\vecx)) \tag{\cref{eqn:max-min-for-alg}} \\
    &= \alpha(\Th^t_k) \gamma_{S,k}(\bias_\Psi) \tag{\cref{item:d-bias} of \cref{prop:d}} \\
    &\geq \alpha(\Th^t_k) \val_\Psi \tag{\cref{lemma:thresh-value-ub}},
\end{align*}
as desired.
\end{proof}

\section{\cite{CGSV21-boolean} is incomplete: Streaming lower bounds for $\mthreeand$?}\label{sec:cgsv-streaming-failure-3and}

In this section, we prove \cref{thm:cgsv-streaming-failure-3and}. We begin with a few lemmas.

\begin{lemma}\label{lemma:3and-unique-minimum}
For $\CD \in \Delta_3$, the expression \[  \frac{\lambda_{\{3\}}(\CD,\frac13\CD\langle 1 \rangle+\frac23\langle2\rangle+\CD\langle3\rangle)}{\gamma_{\{3\},3}(\mu(\CD))} \] is minimized uniquely at $\CD = (0,0,1,0)$, with value $\frac29$.
\end{lemma}

\begin{proof}
Letting $p = \frac13\CD\langle 1 \rangle+\frac23\CD\langle2\rangle+\CD\langle3\rangle$ and $q = 1-p$, by \cref{lemma:sym-bool-lambda,lemma:sym-bool-gamma,lemma:sym-bool-mu} the expression expands to \[ \frac{\CD\langle 0 \rangle \, p^3 + \CD\langle 1 \rangle \, p^2(1-p)+\CD\langle3\rangle \, p(1-p)^2 + \CD\langle3\rangle \, (1-p)^3}{\frac12(1-\CD\langle 0 \rangle -\frac13\CD\langle 1 \rangle+\frac13\CD\langle2\rangle+\CD\langle3\rangle)}. \] The expression's minimum, and its uniqueness, are confirmed analytically in the Mathematica code.
\end{proof}

\begin{lemma}\label{lemma:3and-top-lemma}
Let $X$ be a compact topological space, $Y \subseteq X$ a closed subspace, $Z$ a topological space, and $f : X \to Z$ a continuous map. Let $x^* \in X, z^* \in Z$ be such that $f^{-1}(z^*) = \{x^*\}$. Let $\{x_i\}_{i \in \BN}$ be a sequence of points in $Y$ such that $\{f(x_i)\}_{i \in \BN}$ converges to $z^*$. Then $x^* \in Y$.
\end{lemma}

\begin{proof}
By compactness of $X$, there is a subsequence $\{x_{j_i}\}_{i \in \BN}$ which converges to a limit $\tilde{x}$. By closure, $\tilde{x} \in Y$. By continuity, $f(\tilde x) = z^*$, so $\tilde x = x^*$.
\end{proof}

Finally, we have:

\begin{proof}[Proof of \cref{thm:cgsv-streaming-failure-3and}]
By \cref{lemma:3and-unique-minimum}, $\frac{\beta_{\{3\}}(\CD_N)}{\gamma_{\{3\},3}(\mu(\CD_N))}$ is minimized \emph{uniquely} at $\CD_N^* = (0,0,1,0)$. By \cref{lemma:sym-bool-mu} we have $\mu(\CD_N^*) = \frac13$, and by inspection from the proof of \cref{lemma:sym-bool-gamma} below, $\gamma_{\{3\}}(\CD_Y)$ with $\mu(\CD_Y)=\frac13$ is uniquely minimized by $\CD_Y^*=(\frac13,0,0,\frac23)$.

Finally, we rule out the possibility of an infinite sequence of padded one-wise pairs which achieve ratios arbitrarily close to $\frac29$ using topological properties. View a distribution $\CD \in \Delta_3$ as the vector $(\CD\langle 0 \rangle,\CD\langle 1 \rangle,\CD\langle2\rangle,\CD\langle3\rangle) \in \BR^4$. Let $D \subset \BR^4$ denote the set of such distributions. Let $M\subset D \times D \subset\BR^8$ denote the subset of pairs of distributions with matching marginals, and let $M' \subset M$ denote the subset of pairs with uniform marginals and $P \subset M$ the subset of padded one-wise pairs. $D$, $M$, $M'$, and $P$ are compact (under the Euclidean topology); indeed, $D$, $M$, and $M'$ are bounded and defined by a finite collection of linear equalities and strict inequalities, and letting $M' \subset M$ denote the subset of pairs of distributions with matching \emph{uniform} marginals, $P$ is the image of the compact set $[0,1] \times D \times M' \subset \BR^{13}$ under the continuous map $\tau \times \CD_0 \times (\CD'_Y,\CD'_N) \mapsto (\tau \CD_0 + (1-\tau) \CD'_Y,\tau \CD_0+(1-\tau)\CD'_N)$. Hence, $P$ is closed.

Now the function \[ \alpha : M \to \BR \cup \{\infty\}: (\CD_N,\CD_Y) \mapsto \frac{\beta_{\{3\}}(\CD_N)}{\gamma_{\{3\}}(\CD_Y)} \] is continuous, since a ratio of continuous functions is continuous, and $\beta_{\{3\}}$ is a single-variable supremum of a continuous function (i.e., $\lambda_S$) over a compact interval, which is in general continuous in the remaining variables. Thus, if there were a sequence of padded one-wise pairs $\{(\CD_N^{(i)},\CD_Y^{(i)}) \in P\}_{i \in \BN}$ such that $\alpha(\CD_N^{(i)},\CD_Y^{(i)})$ converges to $\frac29$ as $i \to \infty$, since $M$ is compact and $P$ is closed, \cref{lemma:3and-top-lemma,lemma:3and-unique-minimum} imply that $(\CD_N^*,\CD_Y^*) \in P$, a contradiction.
\end{proof}

\chapter{Conclusions and future directions}\label{chap:conclusions}

\newcommand{\ltf}{\mathrm{Pot}}
\newcommand{\GLST}{\mathrm{GLST}}

\newthought{Open directions for further progress abound} in the quest to understand the approximability of CSPs in the streaming setting. In this chapter, we collect several such questions.

\section{Unique games-hardness vs. streaming hardness}

We begin by giving a concrete account of a problem suggested by Chou \emph{et al.}~\cite{CGSV21-boolean}, who write:

\begin{quote}
    ``[There are] some strange gaps in our knowledge. For instance, it would be natural to suspect that (conditional) inapproximability in the polynomial time setting should also lead to (unconditional) inapproximability in the streaming setting. But we don't have a formal theorem proving this.''
\end{quote}

That is, the question is the following: If $\bgd\mF$ is classically UG-hard according to the dichotomy theorem of Raghavendra~\cite{Rag08}, then is it also necessarily hard for $\sqrt n$-space streaming algorithms? An interesting special case is when Raghavendra's dichotomy theorem~\cite{Rag08} implies that $\mF$ is \emph{approximation resistant} (under the UGC). Must it also be streaming approximation-resistant according to \cite{CGSV21-boolean,CGSV21-finite}? Specifically, \cite{CGSV21-boolean} considers the construction by Potechin~\cite{Pot19} of a balanced linear threshold predicate $f_{\ltf} : \BZ_2^k \to \{0,1\}$ such that $\mbcsp[f_{\ltf}]$ is UG-approximation resistant according to \cite{Rag08}; they observe that since $f_{\ltf}$ is a balanced linear threshold function, it cannot support one-wise independence (cf. \cref{ex:one-wise-indep}), so new tools are needed to prove the streaming approximation-resistance of $\mbcsp[f_{\ltf}]$.

Here, we give an explicit formulation of this open problem. For a distribution $\CD \in \Delta(\BZ_2^k)$, let $\vecmu^{(2)}(\CD) \in [-1,1]^{k \times k}$ denote $\CD$'s matrix of \emph{second-order marginals}, which is the symmetric matrix with entries $\mu^{(2)}(\CD)_{i,j} = \E_{\vecb \sim \CD}[(-1)^{b_i+b_j}]$ for $i,j \in [k]$. Now \cite{Pot19} phrases the \cite{Rag08} criterion for approximation-resistance in the following way:

\begin{definition}\label{def:perfect-integrality-gap}
A Boolean predicate $f : \BZ_2^k \to \{0,1\}$ \emph{supports a perfect integrality gap instance for the standard SDP} if there exists $\CD_N \in \Delta(\BZ_2^k)$ and $\CD_{Y,\vecb} \in \Delta(\BZ_2^k)$ for each $\vecb \in \supp(\CD_N)$ such that:

\begin{enumerate}[label={\roman*.},ref={\roman*}]
    \item For all $\veca \in \BZ_2^k$, $\val_{\Psi(\CD_N)}(\veca) = \rho(f)$.\label{item:sdp-no-dist}
    \item For all $\vecb \in \supp(\CD_N)$, $\val_{\Psi(\CD_{Y,\vecb})}(\vecb)=1$.\label{item:sdp-yes-dists}
    \item For all $\vecb,\vecb' \in \supp(\CD_N)$, $\vecmu(\CD_{Y,\vecb})=\vecmu(\CD_{Y,\vecb'})$ and $\vecmu^{(2)}(\CD_{Y,\vecb})=\vecmu^{(2)}(\CD_{Y,\vecb'})$.
\end{enumerate}
\end{definition}

In comparison with the \cite{CGSV21-boolean} dichotomy theorem for streaming (\cref{thm:cgsv-bool-dichotomy}), note that \cref{item:sdp-no-dist} of \cref{def:perfect-integrality-gap} is \emph{stronger} than the condition $\beta_f(\CD_N)=\rho(f)$ since $\Psi(\CD_N)$'s value is bounded on \emph{every} assignment, not just random symmetric assignments. On the other hand, there is no single $\yes$ distribution $\CD_Y$; instead, we have a \emph{collection} of $\yes$ distributions with identical (first-order) marginals and second-order marginals.

\begin{example}
Guruswami, Levin, Sudan, and Trevisan~\cite{GLST98} studied the predicate $f_{\GLST} : \BZ_2^4 \to \{0,1\}$ defined by $f_{\GLST}(b_1,b_2,b_3,b_4) = 1 + b_2 + b_1b_3 + (1-b_1)b_4$. As Potechin~\cite{Pot19} observes, $f_{\GLST}$ supports a perfect integrality gap instance. Indeed, take $\CD_N = \Unif_{\{(0,0,0,0),(0,1,0,0)\}}$, $\CD_{Y,(0,0,0,0)}=\Unif_{\{(0,0,0,1),(0,1,1,0),(1,0,1,0),(1,1,0,1)\}}$, and $\CD_{Y,(0,1,0,0)} = \Unif_{\{(1,0,0,1),(0,1,0,1),(0,0,1,0),(1,1,1,0)\}}$; in particular, $\vecmu(\CD_{Y,(0,0,0,0)})=\vecmu(\CD_{Y,(0,1,0,0)})=\veczero$ and $\vecmu^{(2)}(\CD_{Y,(0,0,0,0)})$ and $\vecmu^{(2)}(\CD_{Y,(0,1,0,0)})$ are both identically zero except for the $(3,4)$/$(4,3)$ entries, which are $-1$. We observe that $f_{\GLST}$ also supports one-wise independence, e.g., it supports the one-wise independent $\CD = \Unif_{\{(0,1,0,0),(1,0,1,1)\}}$.
\end{example}

\begin{openproblem}\label{prob:ug-vs-strm}
Show that if $f : \BZ_2^k \to \{0,1\}$ supports a perfect integrality gap instance for the standard SDP, then $\mbf$ is streaming approximation-resistant in $\sqrt n$ space, i.e., for every $\epsilon > 0$, every streaming algorithm which $(\rho(f)+\epsilon)$-approximates $\mbf$ uses at least $\Omega(\sqrt n)$ space.
\end{openproblem}

One interesting aspect of \cref{prob:ug-vs-strm} is that the perfect integrality gap assumption for $f$ does not seem to suggest any particular $\yes$ distribution to use to prove a streaming lower bound.

\section{Random-ordering and linear-space streaming}\label{sec:conc-rand-lspace}

Recall from our discussion in \cref{sec:mcut-discussion} that \cite{KKS15} and \cite{KK19} proved $\sqrt n$-space random-ordering and linear-space adversarial-ordering lower bounds against $\mcut$, respectively. As discussed in \cref{sec:cgsvv}, \cite{CGS+22} subsequently extended the linear-space lower bounds to so-called ``wide'' families of predicates. However, to the best of our knowledge there is no evidence that either of these lower bounds doesn't extend to all families which are $\sqrt n$-space, adversarial-ordering streaming approximation-resistant according to \cite{CGSV21-finite}, i.e., all families which ``weakly support one-wise independence'' (see the end of \cref{sec:cgsv-discussion} above). An ambitious goal along these lines would be to prove lower bounds against algorithms which can use \emph{both} linear space and random input order:

\begin{openproblem}\label{prob:better-lbs}
Show that if $\CF$ is a family of predicates weakly supporting one-wise independence, then for every $\epsilon > 0$, $(\rho(\CF)+\epsilon)$-approximating $\mF$ with a random-order streaming algorithm requires $\Omega(n)$ space.
\end{openproblem}

We view \cref{prob:better-lbs} as likely being true but requiring substantially more ``in-the-weeds'' combinatorial analyses than were required even in \cite{KK19,CGS+22}.

For families which don't have these one-wise independence properties, on the other hand, the situation appears to be much murkier. Even in the simplest case, that of $\mdcut$, we discussed in \cref{sec:mdcut-rand-linear} the fact that our hard instances for $\sqrt n$-space adversarial-order streaming break down both for $O(\log n)$-space random-ordering streaming and $o(n)$-space adversarial-order streaming. Indeed, the only lower bounds in these settings for $\mdcut$ we are aware of come from the trivial reduction from $\mcut$, which implies only hardness of $(\frac12+\epsilon)$-approximation. We conjecture, therefore, that the ideas of \cite{CKP+21,CGS+21} used for distinguishing the hard instances from \cite{CGV20} extend to general approximation algorithms for $\mdcut$ in these settings:

\begin{openproblem}
Show that for every $\epsilon > 0$, $\mdcut$ can be $(\frac12-\epsilon)$-approximated by:
\begin{enumerate}[label={\roman*.},ref={\roman*}]
    \item $O(\polylog n)$-space random-order streaming algorithms.
    \item $o(n)$-space adversarial-order streaming algorithms.
\end{enumerate}
\end{openproblem}

\section{Streaming (vs. sketching) lower bounds}

In our joint work \cite{BHP+22}, we show that \cite{CGSV21-boolean}'s streaming lower bounds, in particular the padded one-wise pair criterion (\cref{thm:cgsv-streaming-lb}), cannot rule out streaming approximations which beat the best sketching approximations in $\sqrt n$ space. That is, while $\sqrt n$-space sketching algorithms cannot $(\frac29+\epsilon)$-approximate $\mdcut$ for every $\epsilon > 0$, the techniques in \cite{CGSV21-boolean} cannot strengthen this statement to hold for all streaming algorithms (\cref{thm:cgsv-streaming-failure-3and}). However, as we mentioned in \cref{obs:cgsv-streaming-3and-lb}, \cite{CGSV21-boolean}'s techniques do show that $\sqrt n$-space streaming algorithms cannot $\approx 0.2362$-approximate $\mdcut$. Thus, we believe the following conjecture is quite plausible:

\begin{openproblem}\label{prob:3and}
Show that for every $\epsilon > 0$, streaming algorithms which $(\frac29+\epsilon)$-approximate $\mthreeand$ use at least $\Omega(\sqrt n)$ space.
\end{openproblem}

However, we are not currently aware of any candidate communication game which could be used to prove \cref{prob:3and}.

Another potential direction for $\sqrt n$-space streaming is to try to extend to the setting $k=2,q>2$. (\cite[Proposition 2.12]{CGSV21-boolean} shows that in the case $k=q=2$, distributions with matching marginals are always padded one-wise pairs.) Indeed, we construct the following example in the case $q=3,k=2$:

\begin{example}
We'll construct two distributions $\CD_Y,\CD_N \in \Delta(\BZ_3^2)$ corresponding to distributions in the $k=2,q=3$ case (as discussed at the end of \cref{sec:cgsv-discussion}). Let \[ \CD_Y = \left\{(0,1) \text{ w.p. } \frac14, (1,1) \text{ w.p. } \frac14, (2,2) \text{ w.p. } \frac12\right\} \] and \[ \CD_N = \left\{(0,2) \text{ w.p. } \frac14, (1,2) \text{ w.p. } \frac14, (2,1) \text{ w.p. } \frac12\right\}. \] Then both $\CD_Y$ and $\CD_N$ have first-coordinate marginals $\{0 \text{ w.p. } \frac14, 1 \text{ w.p. } \frac14, 2 \text{ w.p. } \frac12\}$ and second-coordinate marginals $\{1 \text{ w.p. } \frac12, 2 \text{ w.p. }\frac12\}$. Thus, their marginals match, but are not uniform; moreover, $\CD_Y$ and $\CD_N$ have disjoint supports, and so they are not a padded one-wise pair.
\end{example}

\section{Multi-pass lower bounds}

The recent multi-pass cycle-counting lower bounds of Assadi and N~\cite{AN21} imply roughly that for every fixed $\epsilon > 0$ and $p \in \BN$, $(1-3\epsilon,1-2\epsilon)\text{-}\mcut$ requires $\Omega(n^{1-O(p\epsilon)})$ space for $p$-pass streaming algorithms (see \S5.3 in their paper).

\begin{openproblem}\label{prob:multi-pass}
Prove that for every $\epsilon > 0$ and $p \in \BN$, there exists $\delta > 0$ such that every $p$-pass streaming algorithm which $(\frac12+\epsilon)$-approximates $\mcut$ uses $\Omega(n^\delta)$ space. (Even better, $\Omega(n)$ space, and in the random-ordering setting!)
\end{openproblem}

Such a bound could not follow from reductions from cycle-counting \`a la \cite{AKSY20,AN21}, which have the following general structure: The $\yes$ and $\no$ instances are unions of vertex-disjoint cycles and paths, and among the cycles, more are even in the $\yes$ instances and more are odd in the $\no$ instances. $\no$ instances have lower value than $\yes$ instances because every cut fails to cut at least one edge in every odd cycle. However, since odd cycles have length at least $3$, the $\no$ instances have value at least $\frac23$ and thus cannot be used to prove approximation-resistance. 

Instead, if it is indeed possible, we expect \cref{prob:multi-pass} to be resolved using a version of the $\seqbpd$ game (see \cref{def:seqbpd}) which has $p$ ``rounds'' in which each $\Bob_t$ gets to see his input again (at the end of each round, $\Bob_T$ gets to send a message to $\Bob_1$). Note that this ``$p$-round, $T$-player'' $\seqbpd$ game is \emph{not} the same as (standard) $\seqbpd$ with $pT$ players; e.g., in Round $1$ $\Bob_1$ can store some ``hash'' of his input and then verify this hash in Round $2$. Thus, $\Bob_1$'s input in Round $2$ is not ``independent'' of the previous inputs, unlike in single-round $\seqbpd$. We would therefore have to carry out a $\seqbpd$ hardness analysis \`a la \cite{KK19} in the setting where, starting in Round $2$, the input $(M_t,\vecz(t))$ for each $\Bob_t$'s is drawn from a ``typical'', but not uniform, distribution. Proving hardness in this setting seems far beyond current lower-bound techniques.

\vspace{1in}
\hrulefill\hspace{0.2cm} \decoone \hspace{0.2cm} \hrulefill

\setstretch{\dnormalspacing}

% the back matter
\backmatter

% \clearpage
% \bibliography{references}

@article{ABM12,
  title = {On Quadratic Threshold {{CSPs}}},
  author = {Austrin, Per and Benabbas, Siavosh and Magen, Avner},
  date = {2012-11},
  journaltitle = {Discrete Mathematics and Theoretical Computer Science},
  volume = {14},
  number = {2},
  pages = {205--228},
  note = {Conference version in LATIN 2010}
}

@inproceedings{AKSY20,
  title = {Multi-{{Pass Graph Streaming Lower Bounds}} for {{Cycle Counting}}, {{MAX-CUT}}, {{Matching Size}}, and {{Other Problems}}},
  booktitle = {2020 {{IEEE}} 61st {{Annual Symposium}} on {{Foundations}} of {{Computer Science}}},
  author = {Assadi, Sepehr and Kol, Gillat and Saxena, Raghuvansh R. and Yu, Huacheng},
  date = {2020-11},
  pages = {354--364},
  doi = {10.1109/FOCS46700.2020.00041},
  eventdate = {2020-11-16/2020-11-19},
  eventtitle = {{{FOCS}} 2020},
  venue = {virtual}
}

@article{ALM+98,
  title = {Proof Verification and the Hardness of Approximation Problems},
  author = {Arora, Sanjeev and Lund, Carsten and Motwani, Rajeev and Sudan, Madhu and Szegedy, Mario},
  date = {1998-05-01},
  journaltitle = {Journal of the ACM},
  shortjournal = {J. ACM},
  volume = {45},
  number = {3},
  pages = {501--555},
  issn = {0004-5411},
  doi = {10.1145/278298.278306},
  note = {Conference version in FOCS 1992}
}

@article{AM09,
  title = {Approximation {{Resistant Predicates}} from {{Pairwise Independence}}},
  author = {Austrin, Per and Mossel, Elchanan},
  date = {2009-06-01},
  journaltitle = {Computational Complexity},
  shortjournal = {Comput. Complex.},
  volume = {18},
  number = {2},
  pages = {249--271},
  publisher = {{Birkh\"auser-Verlag}},
  doi = {10.1007/s00037-009-0272-6},
  note = {Conference version in CCC 2008}
}

@article{AMS99,
  title = {The {{Space Complexity}} of {{Approximating}} the {{Frequency Moments}}},
  author = {Alon, Noga and Matias, Yossi and Szegedy, Mario},
  date = {1999-02-01},
  journaltitle = {Journal of Computer and System Sciences},
  shortjournal = {Journal of Computer and System Sciences},
  volume = {58},
  number = {1},
  pages = {137--147},
  publisher = {{Elsevier}},
  doi = {10.1006/jcss.1997.1545},
  note = {Conference version in STOC 1996}
}

@article{AMW15,
  title = {On the {{NP-hardness}} of Approximating Ordering-Constraint Satisfaction Problems},
  author = {Austrin, Per and Manokaran, Rajsekar and Wenner, Cenny},
  date = {2015},
  journaltitle = {Theory of Computing},
  shortjournal = {Theory Comput.},
  volume = {11},
  pages = {257--283},
  doi = {10.4086/toc.2015.v011a010},
  note = {Conference version in APPROX 2013}
}

@inproceedings{AN21,
  title = {Graph Streaming Lower Bounds for Parameter Estimation and Property Testing via a Streaming {{XOR}} Lemma},
  booktitle = {Proceedings of the 53rd {{Annual ACM SIGACT Symposium}} on {{Theory}} of {{Computing}}},
  author = {Assadi, Sepehr and N, Vishvajeet},
  date = {2021-06-15},
  pages = {612--625},
  publisher = {{Association for Computing Machinery}},
  doi = {10.1145/3406325.3451110},
  eventdate = {2021-06-21/2021-06-25},
  eventtitle = {{{STOC}} 2021},
  venue = {virtual}
}

@inproceedings{AOW15,
  title = {How to {{Refute}} a {{Random CSP}}},
  booktitle = {2015 {{IEEE}} 56th {{Annual Symposium}} on {{Foundations}} of {{Computer Science}}},
  author = {Allen, Sarah R. and O'Donnell, Ryan and Witmer, David},
  date = {2015-10},
  issn = {0272-5428},
  doi = {10.1109/FOCS.2015.48},
  eventtitle = {{{FOCS}} 2015}
}

@article{AS98,
  title = {Probabilistic Checking of Proofs: A New Characterization of {{NP}}},
  author = {Arora, Sanjeev and Safra, Shmuel},
  date = {1998-01-01},
  journaltitle = {Journal of the ACM},
  shortjournal = {J. ACM},
  volume = {45},
  number = {1},
  pages = {70--122},
  doi = {10.1145/273865.273901},
  note = {Conference version in FOCS 1992}
}

@inproceedings{Aus07,
  title = {Balanced {{Max}} 2-{{SAT Might Not}} Be the {{Hardest}}},
  booktitle = {Proceedings of the 39th {{Annual ACM Symposium}} on {{Theory}} of Computing},
  author = {Austrin, Per},
  date = {2007-06-11},
  pages = {189--197},
  publisher = {{Association for Computing Machinery}},
  doi = {10.1145/1250790.1250818},
  eventdate = {2007-06-11/2007-06-13},
  eventtitle = {{{STOC}} 2007},
  venue = {San Diego, CA, USA}
}

@article{Aus10,
  title = {Towards {{Sharp Inapproximability}} for {{Any}} 2-{{CSP}}},
  author = {Austrin, Per},
  date = {2010-01},
  journaltitle = {SIAM Journal on Computing},
  shortjournal = {SIAM J. Comput.},
  volume = {39},
  number = {6},
  pages = {2430--2463},
  publisher = {{Society for Industrial and Applied Mathematics}},
  doi = {10.1137/070711670},
  note = {Conference version in FOCS 2007}
}

@inproceedings{BC17,
  title = {Towards {{Tighter Space Bounds}} for {{Counting Triangles}} and {{Other Substructures}} in {{Graph Streams}}},
  booktitle = {34th {{Symposium}} on {{Theoretical Aspects}} of {{Computer Science}}},
  author = {Bera, Suman K. and Chakrabarti, Amit},
  editor = {Vollmer, Heribert and Vall\'ee, Brigitte},
  date = {2017},
  series = {{{LIPIcs}}},
  volume = {66},
  pages = {11:1--11:14},
  publisher = {{Schloss Dagstuhl \textemdash{} Leibniz-Zentrum f\"ur Informatik}},
  location = {{Dagstuhl, Germany}},
  doi = {10.4230/LIPIcs.STACS.2017.11},
  eventdate = {2017-03-08/2017-03-11},
  eventtitle = {{{STACS}} 2017},
  venue = {Hanover, Germany}
}

@inproceedings{BCK+18,
  title = {Matrix {{Norms}} in {{Data Streams}}: {{Faster}}, {{Multi-Pass}} and {{Row-Order}}},
  booktitle = {Proceedings of the 35th {{International Conference}} on {{Machine Learning}}},
  author = {Braverman, Vladimir and Chestnut, Stephen and Krauthgamer, Robert and Li, Yi and Woodruff, David and Yang, Lin},
  date = {2018-07-03},
  series = {{{PMLR}}},
  volume = {80},
  pages = {649--658},
  eventdate = {2018-04-30/2018-05-03},
  eventtitle = {{{ICLR}} 2018},
  venue = {Vancouver, BC, Canada}
}

@article{BGJR88,
  title = {An {{Application}} of {{Combinatorial Optimization}} to {{Statistical Physics}} and {{Circuit Layout Design}}},
  author = {Barahona, Francisco and Gr\"otschel, Martin and J\"unger, Michael and Reinelt, Gerhard},
  date = {1988},
  journaltitle = {Operations Research},
  volume = {36},
  number = {3},
  pages = {493--513},
  publisher = {{INFORMS}},
  doi = {10.1287/opre.36.3.493}
}

@misc{BHP+22,
  title = {Sketching Approximations for (Some) Symmetric {{Boolean CSPs}}: Closed-Form Ratios and Simple Algorithms},
  author = {Boyland, Joanna and Hwang, Michael and Prasad, Tarun and Singer, Noah and Velusamy, Santhoshini},
  date = {2022-02},
  eprint = {2112.06319},
  eprinttype = {arxiv},
  primaryclass = {cs.DS},
  archiveprefix = {arXiv}
}

@inproceedings{BK19,
  title = {{{UG-Hardness}} to {{NP-Hardness}} by {{Losing Half}}},
  booktitle = {34th {{Computational Complexity Conference}}},
  author = {Bhangale, Amey and Khot, Subhash},
  date = {2019},
  series = {{{LIPIcs}}},
  volume = {137},
  publisher = {{Schloss Dagstuhl \textemdash{} Leibniz-Zentrum f\"ur Informatik}},
  doi = {10.4230/LIPIcs.CCC.2019.3},
  eventdate = {2019-08-18/2019-08-20},
  eventtitle = {{{CCC}} 2019},
  venue = {New Brunswick, New Jersey, USA}
}

@inproceedings{BKKS20,
  title = {Schatten {{Norms}} in {{Matrix Streams}}: {{Hello Sparsity}}, {{Goodbye Dimension}}},
  booktitle = {Proceedings of the 37th {{International Conference}} on {{Machine Learning}}},
  author = {Braverman, Vladimir and Krauthgamer, Robert and Krishnan, Aditya and Sinoff, Roi},
  date = {2020-11-21},
  series = {{{PMLR}}},
  volume = {119},
  pages = {1100--1110},
  eventdate = {2020-04-26/2020-05-01},
  eventtitle = {{{ICLR}} 2020}
}

@inproceedings{BKPS98,
  booktitle = {Proceedings of the 30th {{Annual ACM Symposium}} on {{Theory}} of Computing},
  author = {Beame, Paul and Karp, Richard and Pitassi, Toniann and Saks, Michael},
  date = {1998-05-23},
  pages = {561--571},
  doi = {10.1145/276698.276870},
  eventdate = {1998-05-24/1998-05-26},
  eventtitle = {{{STOC}} 1998},
  venue = {Dallas, TX, USA},
  title = {On the Complexity of Unsatisfiability Proofs for Random {\(k\)}-{{CNF}} Formulas}
}

@inproceedings{BM16,
  title = {Noisy {{Tensor Completion}} via the {{Sum-of-Squares Hierarchy}}},
  booktitle = {Proceedings of the 29th {{Annual Conference}} on {{Learning Theory}}},
  author = {Barak, Boaz and Moitra, Ankur},
  date = {2016},
  series = {{{PMLR}}},
  volume = {49},
  pages = {417--445},
  eventdate = {2016-06-23/2016-06-26},
  eventtitle = {{{COLT}} 2016},
  venue = {New York, NY, USA}
}

@inproceedings{BR21,
  booktitle = {Computing and Combinatorics},
  author = {Biswas, Arindam and Raman, Venkatesh},
  date = {2021},
  series = {{{LNCS}}},
  volume = {13025},
  pages = {124--136},
  publisher = {{Springer, Cham}},
  eventdate = {2021-10-24/2021-10-26},
  eventtitle = {{{COCOON}} 2021},
  venue = {Tainan, Taiwan},
  title = {Sublinear-{{Space Approximation Algorithms}} for {{Max}} {\(r\)}-{{SAT}}}
}

@inproceedings{Bul17,
  title = {A {{Dichotomy Theorem}} for {{Nonuniform CSPs}}},
  booktitle = {2017 {{IEEE}} 58th {{Annual Symposium}} on {{Foundations}} of {{Computer Science}}},
  author = {Bulatov, Andrei A.},
  date = {2017-10},
  pages = {319--330},
  doi = {10.1109/FOCS.2017.37},
  eventdate = {2017-10-15/2017-10-17},
  eventtitle = {{{FOCS}} 2017},
  venue = {Berkeley, CA, USA}
}

@book{BV04,
  title = {Convex {{Optimization}}},
  author = {Boyd, Stephen and Vandenberghe, Lieven},
  date = {2004-03},
  edition = {1},
  publisher = {{Cambridge University Press}},
  doi = {10.1017/CBO9780511804441},
  pagetotal = {727}
}

@article{CGL07,
  author = {Coja-Oghlan, Amin and Goerdt, Andreas and Lanka, Andr\'e},
  date = {2007-01},
  journaltitle = {Combinatorics, Probability and Computing},
  volume = {16},
  number = {1},
  pages = {5--28},
  publisher = {{Cambridge University Press}},
  doi = {10.1017/S096354830600784X},
  title = {Strong {{Refutation Heuristics}} for {{Random}} {\(k\)}-{{SAT}}},
  note = {Conference version in APPROX 2004}
}

@inproceedings{CGMV20,
  title = {Vertex Ordering Problems in Directed Graph Streams},
  booktitle = {Proceedings of the 31st {{Annual ACM-SIAM Symposium}} on {{Discrete Algorithms}}},
  author = {Chakrabarti, Amit and Ghosh, Prantar and McGregor, Andrew and Vorotnikova, Sofya},
  date = {2020-01},
  pages = {1786--1802},
  publisher = {{Society for Industrial and Applied Mathematics}},
  doi = {10.5555/3381089.3381198},
  eventdate = {2020-01-05/2020-01-09},
  eventtitle = {{{SODA}} 2020},
  venue = {Salt Lake City, Utah, USA}
}

@letter{CGS+21,
  type = {E-mail},
  author = {Chou, Chi-Ning and Golovnev, Alexander and Sudan, Madhu and Velingker, Ameya and Velusamy, Santhoshini},
  date = {2021},
  note = {Personal communication}
}

@inproceedings{CGS+22,
  title = {Linear {{Space Streaming Lower Bounds}} for {{Approximating CSPs}}},
  booktitle = {Proceedings of the 54th {{Annual ACM Symposium}} on {{Theory}} of {{Computing}}},
  author = {Chou, Chi-Ning and Golovnev, Alexander and Sudan, Madhu and Velingker, Ameya and Velusamy, Santhoshini},
  date = {2022},
  eventdate = {2022-06-20/2022-06-24},
  eventtitle = {{{STOC}} 2022},
  venue = {Rome, Italy},
  note = {To appear}
}

@misc{CGSV21-boolean,
  title = {Approximability of All {{Boolean CSPs}} with Linear Sketches},
  author = {Chou, Chi-Ning and Golovnev, Alexander and Sudan, Madhu and Velusamy, Santhoshini},
  date = {2021-02-24},
  eprint = {2102.12351v7},
  eprinttype = {arxiv},
  primaryclass = {cs.CC},
  archiveprefix = {arXiv}
}

@inproceedings{CGSV21-finite,
  title = {Approximability of All Finite {{CSPs}} with Linear Sketches},
  booktitle = {Proceedings of the 62nd {{Annual IEEE Symposium}} on {{Foundations}} of {{Computer Science}}},
  author = {Chou, Chi-Ning and Golovnev, Alexander and Sudan, Madhu and Velusamy, Santhoshini},
  date = {2021},
  publisher = {{IEEE Computer Society}},
  doi = {10.1109/FOCS52979.2021.00117},
  eventdate = {2022-02-07/2022-02-10},
  eventtitle = {{{FOCS}} 2021},
  venue = {Denver, CO, USA}
}

@inproceedings{CGV20,
  booktitle = {2020 {{IEEE}} 61st {{Annual Symposium}} on {{Foundations}} of {{Computer Science}}},
  author = {Chou, Chi-Ning and Golovnev, Alexander and Velusamy, Santhoshini},
  date = {2020-11},
  pages = {330--341},
  publisher = {{IEEE Computer Society}},
  doi = {10.1109/FOCS46700.2020.00039},
  eventdate = {2020-11-16/2020-11-19},
  eventtitle = {{{FOCS}} 2020},
  venue = {virtual},
  title = {Optimal {{Streaming Approximations}} for All {{Boolean Max-2CSPs}} and {{Max-}}{\(k\)}{{SAT}}}
}

@misc{Cha20,
  type = {Lecture notes (draft)},
  title = {Data {{Stream Algorithms}}},
  author = {Chakrabarti, Amit},
  date = {2020-07-02},
  url = {https://www.cs.dartmouth.edu/~ac/Teach/CS35-Spring20/Notes/lecnotes.pdf},
  howpublished = {Lecture notes (draft)}
}

@article{CHIS12,
  title = {Approximating {{Linear Threshold Predicates}}},
  author = {Cheraghchi, Mahdi and H\aa stad, Johan and Isaksson, Marcus and Svensson, Ola},
  date = {2012-03-01},
  journaltitle = {ACM Transactions on Computation Theory},
  shortjournal = {ACM Trans. Comput. Theory},
  volume = {4},
  number = {1},
  pages = {2:1--2:31},
  publisher = {{Assocaition for Computing Machinery}},
  doi = {10.1145/2141938.2141940},
  note = {Conference version in APPROX 2010}
}

@letter{CKP+21,
  type = {E-mail},
  author = {Chen, Lijie and Kol, Gillat and Paramonov, Dmitry and Saxena, Raghuvansh and Song, Zhao and Yu, Huacheng},
  date = {2021-03-16}
}

@book{CKS01,
  title = {Complexity {{Classifications}} of {{Boolean Constraint Satisfaction Problems}}},
  author = {Creignou, Nadia and Khanna, Sanjeev and Sudan, Madhu},
  date = {2001},
  series = {Discrete {{Mathematics}} and {{Applications}}},
  publisher = {{Society for Industrial and Applied Mathematics}},
  doi = {10.1137/1.9780898718546}
}

@article{CMM09,
  title = {Near-Optimal Algorithms for Maximum Constraint Satisfaction Problems},
  author = {Charikar, Moses and Makarychev, Konstantin and Makarychev, Yury},
  date = {2009-07},
  journaltitle = {ACM Transactions on Algorithms},
  shortjournal = {ACM Trans. Algorithms},
  volume = {5},
  number = {3},
  pages = {1--14},
  publisher = {{Association for Computing Machinery}},
  doi = {10.1145/1541885.1541893},
  note = {Conference version in SODA 2007}
}

@article{CS98,
  title = {A Geometric Approach to Betweenness},
  author = {Chor, Benny and Sudan, Madhu},
  date = {1998},
  journaltitle = {SIAM Journal on Discrete Mathematics},
  shortjournal = {SIAM J. Discrete Math.},
  volume = {11},
  number = {4},
  pages = {511--523},
  publisher = {{Society for Industrial and Applied Mathematics}},
  doi = {10.1137/S0895480195296221},
  note = {Conference version in Algorithms, ESA 1995}
}

@article{EH08,
  title = {More Efficient Queries in {{PCPs}} for {{NP}} and Improved Approximation Hardness of Maximum {{CSP}}},
  author = {Engebretsen, Lars and Holmerin, Jonas},
  date = {2008-12},
  journaltitle = {Random Structures and Algorithms},
  shortjournal = {Random Struct. Alg.},
  volume = {33},
  number = {4},
  pages = {497--514},
  publisher = {{Wiley}},
  doi = {10.1002/rsa.20226},
  langid = {english},
  note = {Conference version in STACS 2005}
}

@inproceedings{Fei02,
  title = {Relations between Average Case Complexity and Approximation Complexity},
  booktitle = {Proceedings of the 34th {{Annual ACM Symposium}} on {{Theory}} of {{Computing}}},
  author = {Feige, Uriel},
  date = {2002-05-19},
  pages = {534--543},
  publisher = {{Association for Computing Machinery}},
  location = {{New York, NY, USA}},
  doi = {10.1145/509907.509985},
  eventtitle = {{{STOC}} 2002}
}

@article{Fei98,
  author = {Feige, Uriel},
  date = {1998-07},
  journaltitle = {Journal of the ACM},
  shortjournal = {J. ACM},
  volume = {45},
  number = {4},
  pages = {634--652},
  doi = {10.1145/285055.285059},
  title = {A Threshold of {\(\ln n\)} for Approximating Set Cover},
  note = {Conference version in STOC 1996}
}

@inproceedings{FG95,
  title = {Approximating the Value of Two Prover Proof Systems, with Applications to {{MAX 2SAT}} and {{MAX DICUT}}},
  booktitle = {Proceedings of the 3rd {{Israel Symposium}} on the {{Theory}} of {{Computing}} and {{Systems}}},
  author = {Feige, Uriel and Goemans, Michel X.},
  date = {1995},
  pages = {182--189},
  publisher = {{IEEE Computer Society}},
  doi = {10.1109/ISTCS.1995.377033},
  eventdate = {1995-01-04/1995-01-06},
  eventtitle = {{{ISTCS}} 2003}
}

@article{FGK05,
  author = {Friedman, Joel and Goerdt, Andreas and Krivelevich, Michael},
  date = {2005-01},
  journaltitle = {SIAM Journal on Computing},
  shortjournal = {SIAM J. Comput.},
  volume = {35},
  number = {2},
  pages = {408--430},
  publisher = {{Society for Industrial and Applied Mathematics}},
  doi = {10.1137/S009753970444096X},
  title = {Recognizing {{More Unsatisfiable Random}} {\(k\)}-{{SAT Instances Efficiently}}},
  note = {Conference version in STACS 2001}
}

@article{FK98,
  title = {Zero {{Knowledge}} and the {{Chromatic Number}}},
  author = {Feige, Uriel and Kilian, Joe},
  date = {1998-10-01},
  journaltitle = {Journal of Computer and System Sciences},
  shortjournal = {Journal of Computer and System Sciences},
  volume = {57},
  number = {2},
  pages = {187--199},
  doi = {10.1006/jcss.1998.1587},
  url = {https://www.sciencedirect.com/science/article/pii/S0022000098915872},
  langid = {english},
  note = {Conference version in CCC 1996}
}

@article{FO07,
  title = {Easily Refutable Subformulas of Large Random {{3CNF}} Formulas},
  author = {Feige, Uriel and Ofek, Eran},
  date = {2007-02-09},
  journaltitle = {Theory of Computing},
  volume = {3},
  pages = {25--43},
  doi = {10.4086/toc.2007.v003a002}
}

@article{FV98,
  title = {The {{Computational Structure}} of {{Monotone Monadic SNP}} and {{Constraint Satisfaction}}: {{A Study}} through {{Datalog}} and {{Group Theory}}},
  author = {Feder, Tom\'as and Vardi, Moshe Y.},
  date = {1998-01},
  journaltitle = {SIAM Journal on Computing},
  shortjournal = {SIAM J. Comput.},
  volume = {28},
  number = {1},
  pages = {57--104},
  publisher = {{Society for Industrial and Applied Mathematics}},
  doi = {10.1137/S0097539794266766}
}

@article{GHM+11,
  title = {Beating the {{Random Ordering}} Is {{Hard}}: {{Every}} Ordering {{CSP}} Is Approximation Resistant},
  author = {Guruswami, Venkatesan and H\aa stad, Johan and Manokaran, Rajsekar and Raghavendra, Prasad and Charikar, Moses},
  date = {2011},
  journaltitle = {SIAM Journal on Computing},
  shortjournal = {SIAM J. Comput.},
  volume = {40},
  number = {3},
  pages = {878--914},
  publisher = {{Society for Industrial and Applied Mathematics}},
  doi = {10.1137/090756144},
  note = {Conference version in FOCS 2008}
}

@article{GKK+08,
  title = {Exponential Separation for One-Way Quantum Communication Complexity, with Applications to Cryptography},
  author = {Gavinsky, Dmitry and Kempe, Julia and Kerenidis, Iordanis and Raz, Ran and de Wolf, Ronald},
  options = {useprefix=true},
  date = {2008-12},
  journaltitle = {SIAM Journal on Computing},
  shortjournal = {SIAM J. Comput.},
  volume = {38},
  number = {5},
  pages = {1695--1708},
  publisher = {{Society for Industrial and Applied Mathematics}},
  issn = {0097-5397},
  doi = {10.1137/070706550},
  note = {Conference version in STOC 2007}
}

@inproceedings{GKM22,
  title = {Algorithms and {{Certificates}} for {{Boolean CSP Refutation}}: ``{{Smoothed}} Is No Harder than {{Random}}''},
  booktitle = {Proceedings of the 54th {{Annual ACM Symposium}} on {{Theory}} of {{Computing}}},
  author = {Guruswami, Venkatesan and Kothari, Pravesh K. and Manohar, Peter},
  date = {2022},
  eventdate = {2022-06-20/2022-06-24},
  eventtitle = {{{STOC}} 2022},
  venue = {Rome, Italy}
}

@article{GL17,
  title = {Towards a {{Characterization}} of {{Approximation Resistance}} for {{Symmetric CSPs}}},
  author = {Guruswami, Venkatesan and Lee, Euiwoong},
  date = {2017-06-13},
  journaltitle = {Theory of Computing},
  volume = {13},
  number = {1},
  pages = {1--24},
  doi = {10.4086/toc.2017.v013a003},
  note = {Conference version in APPROX 2015}
}

@inproceedings{GLST98,
  title = {A Tight Characterization of {{NP}} with 3 Query {{PCPs}}},
  booktitle = {Proceedings of the 39th {{Annual Symposium}} on {{Foundations}} of {{Computer Science}}},
  author = {Guruswami, V. and Lewin, D. and Sudan, M. and Trevisan, L.},
  date = {1998-11},
  pages = {8--17},
  doi = {10.1109/SFCS.1998.743424},
  eventdate = {1998-11-08/1998-11-11},
  eventtitle = {{{FOCS}} 1998},
  venue = {Palo Alto, CA, USA}
}

@article{GO16,
  title = {Superlinear {{Lower Bounds}} for {{Multipass Graph Processing}}},
  author = {Guruswami, Venkatesan and Onak, Krzysztof},
  date = {2016-03},
  journaltitle = {Algorithmica},
  volume = {76},
  number = {3},
  pages = {654--683},
  publisher = {{Springer}},
  doi = {10.1007/s00453-016-0138-7},
  note = {Conference version in CCC 2013}
}

@incollection{Gol11,
  title = {Three {{XOR-Lemmas}} \textemdash{} {{An Exposition}}},
  booktitle = {Studies in {{Complexity}} and {{Cryptography}}: {{Miscellanea}} on the {{Interplay}} between {{Randomness}} and {{Computation}}},
  author = {Goldreich, Oded},
  editor = {Goldreich, Oded},
  date = {2011},
  series = {Lecture {{Notes}} in {{Computer Science}}},
  pages = {248--272},
  publisher = {{Springer}},
  location = {{Berlin, Heidelberg}},
  doi = {10.1007/978-3-642-22670-0_22},
  isbn = {978-3-642-22670-0}
}

@incollection{GS21,
  title = {Random-{{Order Models}}},
  booktitle = {Beyond the {{Worst-Case Analysis}} of {{Algorithms}}},
  author = {Gupta, Anupam and Singla, Sahil},
  editor = {Roughgarden, Tim},
  date = {2021-01-14},
  publisher = {{Cambridge University Press}},
  location = {{234-258}},
  url = {10.1017/9781108637435.015},
  isbn = {978-1-108-49431-1}
}

@inproceedings{GT19,
  title = {Streaming {{Hardness}} of {{Unique Games}}},
  booktitle = {Approximation, {{Randomization}}, and {{Combinatorial Optimization}}. {{Algorithms}} and {{Techniques}}},
  author = {Guruswami, Venkatesan and Tao, Runzhou},
  editor = {Achlioptas, Dimitris and V\'egh, L\'aszl\'o A.},
  date = {2019-09},
  series = {{{LIPIcs}}},
  volume = {145},
  pages = {5:1--5:12},
  publisher = {{Schloss Dagstuhl \textemdash{} Leibniz-Zentrum f\"ur Informatik}},
  doi = {10.4230/LIPIcs.APPROX-RANDOM.2019.5},
  eventdate = {2019-09-20/2019-09-22},
  eventtitle = {{{APPROX}} 2019},
  venue = {Cambridge, MA, USA}
}

@inproceedings{GVV17,
  title = {Streaming {{Complexity}} of {{Approximating Max 2CSP}} and {{Max Acyclic Subgraph}}},
  booktitle = {Approximation, Randomization, and Combinatorial Optimization. {{Algorithms}} and Techniques},
  author = {Guruswami, Venkatesan and Velingker, Ameya and Velusamy, Santhoshini},
  editor = {Jansen, Klaus and Rolim, Jos\'e D. P. and Williamson, David and Vempala, Santosh S.},
  date = {2017-08},
  series = {{{LIPIcs}}},
  volume = {81},
  pages = {8:1--8:19},
  publisher = {{Schloss Dagstuhl \textemdash{} Leibniz-Zentrum f\"ur Informatik}},
  doi = {10.4230/LIPIcs.APPROX-RANDOM.2017.8},
  eventdate = {2017-08-16/2017-08-18},
  eventtitle = {{{APPROX}} 2017},
  venue = {Berkeley, CA, USA}
}

@article{GW95,
  title = {Improved Approximation Algorithms for Maximum Cut and Satisfiability Problems Using Semidefinite Programming},
  author = {Goemans, Michel X. and Williamson, David P.},
  date = {1995-11},
  journaltitle = {Journal of the ACM},
  shortjournal = {J. ACM},
  volume = {42},
  number = {6},
  pages = {1115--1145},
  issn = {0004-5411, 1557-735X},
  doi = {10.1145/227683.227684},
  note = {Conference version in STOC 1994}
}

@article{Has01,
  title = {Some Optimal Inapproximability Results},
  author = {H\aa stad, Johan},
  date = {2001},
  journaltitle = {Journal of the ACM},
  shortjournal = {J. ACM},
  volume = {48},
  number = {4},
  pages = {798--859},
  publisher = {{Association for Computing Machinery}},
  doi = {10.1145/502090.502098}
}

@inproceedings{Has04,
  booktitle = {Approximation, {{Randomization}}, and {{Combinatorial Optimization}}. {{Algorithms}} and {{Techniques}}},
  author = {Hast, Gustav},
  editor = {Jansen, Klaus and Khanna, Sanjeev and Rolim, Jos\'e D. P. and Ron, Dana},
  date = {2004},
  series = {{{LNCS}}},
  volume = {3122},
  pages = {151--162},
  publisher = {{Springer}},
  doi = {10.1007/978-3-540-27821-4_14},
  eventdate = {2004-08-22/2004-08-24},
  eventtitle = {{{APPROX}} 2004},
  venue = {Cambridge, MA, USA},
  title = {Approximating {{Max}} {\(k\)}{{CSP Using Random Restrictions}}}
}

@inproceedings{Has05,
  booktitle = {Automata, {{Languages}} and {{Programming}}},
  author = {Hast, Gustav},
  editor = {Caires, Lu\'is and Italiano, Giuseppe F. and Monteiro, Lu\'is and Palamidessi, Catuscia and Yung, Moti},
  date = {2005},
  series = {{{LNCS}}},
  volume = {3580},
  pages = {956--968},
  publisher = {{Springer}},
  doi = {10.1007/11523468_77},
  eventdate = {2005-07-11/2005-07-15},
  eventtitle = {{{ICALP}} 2005},
  title = {Approximating {{Max}} {\(k\)}{{CSP}} -- {{Outperforming}} a {{Random Assignment}} with {{Almost}} a {{Linear Factor}}}
}

@article{Has99,
  author = {H\aa stad, Johan},
  date = {1999-03-01},
  journaltitle = {Acta Mathematica},
  shortjournal = {Acta Math.},
  volume = {182},
  number = {1},
  pages = {105--142},
  issn = {1871-2509},
  doi = {10.1007/BF02392825},
  title = {Clique Is Hard to Approximate within {\(n^{1-\epsilon}\)}},
  note = {Conference version in FOCS 1996}
}

@misc{IMNO11,
  title = {Open Problems in Data Streams, Property Testing, and Related Topics},
  author = {Indyk, Piotr and McGregor, Andrew and Newman, Ilan and Onak, Krzysztof},
  date = {2011-06-14},
  url = {https://sublinear.info/files/bertinoro2011_kanpur2009.pdf},
  note = {Compiled from IITK Workshop on Algorithms for Processing Massive Data Sets (2009) and Bertinoro Workshop on Sublinear Algorithms (2011).}
}

@article{Ind06,
  title = {Stable Distributions, Pseudorandom Generators, Embeddings, and Data Stream Computation},
  author = {Indyk, Piotr},
  date = {2006-05},
  journaltitle = {Journal of the ACM},
  shortjournal = {J. ACM},
  volume = {53},
  number = {3},
  pages = {307--323},
  doi = {10.1145/1147954.1147955},
  issue_date = {May 2006},
  note = {Conference version in FOCS 2000}
}

@incollection{Kar72,
  title = {Reducibility among {{Combinatorial Problems}}},
  booktitle = {The {{IBM Research Symposia Series}}},
  author = {Karp, Richard M.},
  editor = {Miller, R.E. and Thatcher, J.W. and Bohlinger, J.D.},
  date = {1972},
  pages = {85--103},
  publisher = {{Springer}},
  doi = {10.1007/978-1-4684-2001-2_9}
}

@inproceedings{Kho02,
  title = {On the Power of Unique 2-Prover 1-Round Games},
  booktitle = {Proceedings of the 34th {{Annual ACM Symposium}} on {{Theory}} of {{Computing}}},
  author = {Khot, Subhash},
  date = {2002-05},
  pages = {767--775},
  publisher = {{Association for Computing Machinery}},
  doi = {10.1145/509907.510017},
  eventdate = {2002-05-19/2002-05-21},
  eventtitle = {{{STOC}} 2002},
  venue = {Qu\'ebec, Canada}
}

@inproceedings{Kho10,
  title = {On the {{Unique Games Conjecture}} ({{Invited Survey}})},
  booktitle = {Proceedings of the 2010 {{IEEE}} 25th {{Annual Conference}} on {{Computational Complexity}}},
  author = {Khot, Subhash},
  date = {2010-06-09},
  pages = {99--121},
  publisher = {{IEEE Computer Society}},
  doi = {10.1109/CCC.2010.19},
  eventdate = {2010-06-09/2010-06-11},
  eventtitle = {{{CCC}} 2010},
  venue = {Cambridge, MA, USA}
}

@inproceedings{KK15,
  title = {Sketching Cuts in Graphs and Hypergraphs},
  booktitle = {Proceedings of the 6th {{Annual Conference}} on {{Innovations}} in {{Theoretical Computer Science}}},
  author = {Kogan, Dmitry and Krauthgamer, Robert},
  date = {2015},
  pages = {367--376},
  publisher = {{Association for Computing Machinery}},
  doi = {10.1145/2688073.2688093},
  eventdate = {2015-01-11/2015-01-13},
  eventtitle = {{{ITCS}} 2015},
  venue = {Rehovot, Israel}
}

@inproceedings{KK19,
  title = {An Optimal Space Lower Bound for Approximating {{MAX-CUT}}},
  booktitle = {Proceedings of the 51st {{Annual ACM SIGACT Symposium}} on {{Theory}} of {{Computing}}},
  author = {Kapralov, Michael and Krachun, Dmitry},
  date = {2019-06},
  pages = {277--288},
  publisher = {{Association for Computing Machinery}},
  doi = {10.1145/3313276.3316364},
  eventdate = {2019-06-23/2019-06-26},
  eventtitle = {{{STOC}} 2019},
  venue = {Phoenix, AZ, USA}
}

@article{KKMO07,
  title = {Optimal {{Inapproximability Results}} for {{MAX}}-{{CUT}} and {{Other}} 2-{{Variable CSPs}}?},
  author = {Khot, Subhash and Kindler, Guy and Mossel, Elchanan and O'Donnell, Ryan},
  date = {2007-05-14},
  journaltitle = {SIAM Journal on Computing},
  volume = {37},
  number = {1},
  publisher = {{Society for Industrial and Applied Mathematics}},
  doi = {10.1137/S0097539705447372},
  note = {Conference version in FOCS 2004}
}

@inproceedings{KKP18,
  title = {The {{Sketching Complexity}} of {{Graph}} and {{Hypergraph Counting}}},
  booktitle = {2018 {{IEEE}} 59th {{Annual Symposium}} on {{Foundations}} of {{Computer Science}}},
  author = {Kallaugher, John and Kapralov, Michael and Price, Eric},
  date = {2018-10},
  pages = {556--567},
  doi = {10.1109/FOCS.2018.00059},
  eventdate = {2018-10-07/2018-10-09},
  eventtitle = {{{FOCS}} 2018},
  venue = {Paris, France}
}

@inproceedings{KKS15,
  title = {Streaming Lower Bounds for Approximating {{MAX-CUT}}},
  booktitle = {Proceedings of the 26th {{Annual ACM-SIAM Symposium}} on {{Discrete Algorithms}}},
  author = {Kapralov, Michael and Khanna, Sanjeev and Sudan, Madhu},
  date = {2015-01},
  pages = {1263--1282},
  publisher = {{Society for Industrial and Applied Mathematics}},
  doi = {10.1137/1.9781611973730.84},
  eventdate = {2015-01-04/2015-01-06},
  eventtitle = {{{SODA}} 2015},
  venue = {San Diego, California, USA}
}

@inproceedings{KKSV17,
  booktitle = {Proceedings of the 28th {{Annual ACM-SIAM Symposium}} on {{Discrete Algorithms}}},
  author = {Kapralov, Michael and Khanna, Sanjeev and Sudan, Madhu and Velingker, Ameya},
  date = {2017-01},
  pages = {1703--1722},
  publisher = {{Society for Industrial and Applied Mathematics}},
  doi = {10.5555/3039686.3039798},
  eventdate = {2017-01-16/2017-01-19},
  eventtitle = {{{SODA}} 2017},
  venue = {Barcelona, Spain},
  title = {{\((1 + \omega(1))\)}-Approximation to {{MAX-CUT}} Requires Linear Space}
}

@inproceedings{KMOW17,
  title = {Sum of Squares Lower Bounds for Refuting Any {{CSP}}},
  booktitle = {Proceedings of the 49th {{Annual ACM SIGACT Symposium}} on {{Theory}} of {{Computing}}},
  author = {Kothari, Pravesh K. and Mori, Ryuhei and O'Donnell, Ryan and Witmer, David},
  date = {2017-06-19},
  pages = {132--145},
  publisher = {{Association for Computing Machinery}},
  doi = {10.1145/3055399.3055485},
  eventtitle = {{{STOC}} 2017},
  note = {To appear}
}

@inproceedings{KNW10,
  title = {On the {{Exact Space Complexity}} of {{Sketching}} and {{Streaming Small Norms}}},
  booktitle = {Proceedings of the 2010 {{Annual ACM-SIAM Symposium}} on {{Discrete Algorithms}}},
  author = {Kane, Daniel M. and Nelson, Jelani and Woodruff, David P.},
  date = {2010},
  pages = {1161--1178},
  publisher = {{Society for Industrial and Applied Mathematics}},
  doi = {10.1137/1.9781611973075.93},
  eventdate = {2010-01-17/2010-01-19},
  eventtitle = {{{SODA}} 2010},
  venue = {Austin, TX, USA}
}

@article{KO09,
  title = {{{SDP Gaps}} and {{UGC-hardness}} for {{Max-Cut-Gain}}},
  author = {Khot, Subhash and O'Donnell, Ryan},
  date = {2009-05-29},
  journaltitle = {Theory of Computing},
  volume = {5},
  number = {4},
  pages = {83--117},
  doi = {10.4086/toc.2009.v005a004},
  note = {Conference version in STOC 2006}
}

@inproceedings{KP20,
  title = {Separations and Equivalences between Turnstile Streaming and Linear Sketching},
  booktitle = {Proceedings of the 52nd {{Annual ACM SIGACT Symposium}} on {{Theory}} of {{Computing}}},
  author = {Kallaugher, John and Price, Eric},
  date = {2020},
  pages = {1223--1236},
  publisher = {{Association for Computing Machinery}},
  location = {{New York, NY, USA}},
  url = {http://doi.org/10.1145/3357713.3384278},
  eventdate = {2020-06-22/2020-06-26},
  eventtitle = {{{STOC}} 2020},
  venue = {virtual}
}

@article{KR08,
  author = {Khot, Subhash and Regev, Oded},
  date = {2008-05-01},
  journaltitle = {Journal of Computer and System Sciences},
  shortjournal = {Journal of Computer and System Sciences},
  series = {Computational {{Complexity}} 2003},
  volume = {74},
  number = {3},
  pages = {335--349},
  doi = {10.1016/j.jcss.2007.06.019},
  langid = {english},
  title = {Vertex Cover Might Be Hard to Approximate to within {\(2-\epsilon\)}},
  note = {Conference version in CCC 2003}
}

@article{KSTW01,
  title = {The {{Approximability}} of {{Constraint Satisfaction Problems}}},
  author = {Khanna, Sanjeev and Sudan, Madhu and Trevisan, Luca and Williamson, David P.},
  date = {2001-01-01},
  journaltitle = {SIAM Journal on Computing},
  shortjournal = {SIAM J. Comput.},
  volume = {30},
  number = {6},
  pages = {1863--1920},
  publisher = {{Society for Industrial and Applied Mathematics}},
  issn = {0097-5397},
  doi = {10.1137/S0097539799349948},
  url = {https://epubs.siam.org/doi/10.1137/S0097539799349948},
  urldate = {2021-10-29},
  note = {Conference versions in STOC 1997 and CCC 1997}
}

@inproceedings{LLZ02,
  title = {Improved {{Rounding Techniques}} for the {{MAX}} 2-{{SAT}} and {{MAX DI-CUT Problems}}},
  booktitle = {Integer {{Programming}} and {{Combinatorial Optimization}}},
  author = {Lewin, Michael and Livnat, Dror and Zwick, Uri},
  editor = {Cook, William J. and Schulz, Andreas S.},
  date = {2002},
  pages = {67--82},
  doi = {10.1007/3-540-47867-1_6},
  eventtitle = {{{IPCO}} 2002}
}

@article{McG14,
  title = {Graph Stream Algorithms: A Survey},
  author = {McGregor, Andrew},
  date = {2014-05-13},
  journaltitle = {ACM SIGMOD Record},
  shortjournal = {SIGMOD Rec.},
  volume = {43},
  number = {1},
  pages = {9--20},
  doi = {10.1145/2627692.2627694}
}

@inproceedings{MM01,
  title = {0.863-{{Approximation Algorithm}} for {{MAX DICUT}}},
  booktitle = {Approximation, {{Randomization}}, and {{Combinatorial Optimization}}: {{Algorithms}} and {{Techniques}}},
  author = {Matuura, Shiro and Matsui, Tomomi},
  editor = {Goemans, Michel and Jansen, Klaus and Rolim, Jos\'e D. P. and Trevisan, Luca},
  date = {2001},
  series = {{{LNCS}}},
  volume = {2129},
  pages = {138--146},
  publisher = {{Springer}},
  location = {{Berlin, Heidelberg}},
  doi = {10.1007/3-540-44666-4_17},
  eventdate = {2001-08-18/2001-08-20},
  eventtitle = {{{APPROX}} 2001},
  venue = {Berkeley, CA, USA}
}

@incollection{MM17,
  title = {Approximation {{Algorithms}} for {{CSPs}}},
  booktitle = {The {{Constraint Satisfaction Problem}}: {{Complexity}} and {{Approximability}}},
  author = {Makarychev, Konstantin and Makarychev, Yury},
  editor = {Krokhin, Andrei and Zivny, Stanislav},
  date = {2017},
  series = {Dagstuhl {{Follow-Ups}}},
  volume = {7},
  pages = {287--325},
  publisher = {{Schloss Dagstuhl \textemdash{} Leibniz-Zentrum f\"ur Informatik}},
  location = {{Dagstuhl, Germany}},
  doi = {10.4230/DFU.Vol7.15301.287}
}

@book{Mut05,
  title = {Data {{Streams}}: {{Algorithms}} and {{Applications}}},
  author = {Muthukrishnan, S.},
  date = {2005-09-26},
  series = {Foundations and {{Trends}} in {{Theoretical Computer Science}}},
  volume = {1},
  number = {2},
  publisher = {{Now Publishers, Inc.}},
  doi = {10.1561/0400000002},
  langid = {english},
  pagetotal = {119}
}

@thesis{New00,
  type = {mathesis},
  title = {Approximating the Maximum Acyclic Subgraph},
  author = {Newman, Alantha},
  date = {2000},
  institution = {{Massachusetts Institute of Technology}}
}

@book{OD14,
  title = {Analysis of {{Boolean Functions}}},
  author = {O'Donnell, Ryan},
  date = {2014-06-01},
  edition = {1st edition},
  publisher = {{Cambridge University Press}},
  location = {{New York, NY}},
  isbn = {978-1-107-03832-5},
  langid = {english},
  pagetotal = {444},
  shorthand = {O'D14}
}

@article{Opa79,
  title = {Total {{Ordering Problem}}},
  author = {Opatrny, Jaroslav},
  date = {1979-02},
  journaltitle = {SIAM Journal on Computing},
  shortjournal = {SIAM J. Comput.},
  volume = {8},
  number = {1},
  pages = {111--114},
  publisher = {{Society for Industrial and Applied Mathematics}},
  doi = {10.1137/0208008}
}

@inproceedings{Pot19,
  title = {On the Approximation Resistance of Balanced Linear Threshold Functions},
  booktitle = {Proceedings of the 51st {{Annual ACM SIGACT Symposium}} on {{Theory}} of {{Computing}}},
  author = {Potechin, Aaron},
  date = {2019-06-23},
  pages = {430--441},
  publisher = {{Association for Computing Machinery}},
  doi = {10.1145/3313276.3316374},
  eventdate = {2019-06-23/2019-06-26},
  eventtitle = {{{STOC}} 2019},
  venue = {Phoenix, AZ, USA}
}

@inproceedings{PS18,
  title = {Estimating {{Graph Parameters}} from {{Random Order Streams}}},
  booktitle = {Proceedings of the 29th {{Annual ACM-SIAM Symposium}} on {{Discrete Algorithms}}},
  author = {Peng, Pan and Sohler, Christian},
  date = {2018-01},
  publisher = {{Society for Industrial and Applied Mathematics}},
  doi = {10.1137/1.9781611975031.157},
  eventdate = {2018-01-07/2018-01-10},
  eventtitle = {{{SODA}} 2018},
  venue = {New Orleans, LA, USA}
}

@inproceedings{Rag08,
  title = {Optimal Algorithms and Inapproximability Results for Every {{CSP}}?},
  booktitle = {Proceedings of the 40th {{Annual ACM Symposium}} on {{Theory}} of {{Computing}}},
  author = {Raghavendra, Prasad},
  date = {2008},
  pages = {245--254},
  doi = {10.1145/1374376.1374414},
  eventdate = {2008-05-17/2008-05-20},
  eventtitle = {{{STOC}} 2008},
  venue = {Victoria, BC, Canada}
}

@inproceedings{RRS17,
  title = {Strongly Refuting Random {{CSPs}} below the Spectral Threshold},
  booktitle = {Proceedings of the 49th {{Annual ACM SIGACT Symposium}} on {{Theory}} of {{Computing}}},
  author = {Raghavendra, Prasad and Rao, Satish and Schramm, Tselil},
  date = {2017-06-19},
  pages = {121--131},
  doi = {10.1145/3055399.3055417},
  eventdate = {2017-06-19/2017-06-23},
  eventtitle = {{{STOC}} 2017},
  venue = {Montreal, Quebec, CA}
}

@inproceedings{Sch78,
  title = {The Complexity of Satisfiability Problems},
  booktitle = {Proceedings of the 10th {{Annual ACM Symposium}} on {{Theory}} of Computing},
  author = {Schaefer, Thomas J.},
  date = {1978-05-01},
  pages = {216--226},
  publisher = {{Association for Computing Machinery}},
  doi = {10.1145/800133.804350},
  eventdate = {1978-05-01/1978-05-03},
  eventtitle = {{{STOC}} 1978},
  venue = {San Diego, CA, USA}
}

@inproceedings{SSV21,
  title = {Streaming Approximation Resistance of Every Ordering {{CSP}}},
  booktitle = {Approximation, {{Randomization}}, and {{Combinatorial Optimization}}. {{Algorithms}} and {{Techniques}}},
  author = {Singer, Noah and Sudan, Madhu and Velusamy, Santhoshini},
  editor = {Wootters, Mary and Sanit\`a, Laura},
  date = {2021-09},
  series = {{{LIPIcs}}},
  volume = {207},
  pages = {17:1-17:19},
  publisher = {{Schloss Dagstuhl \textemdash{} Leibniz-Zentrum f\"ur Informatik}},
  doi = {10.4230/LIPIcs.APPROX/RANDOM.2021.17},
  eventdate = {2021-08-16/2021-08-18},
  eventtitle = {{{APPROX}} 2021}
}

@inproceedings{ST00,
  title = {A {{PCP}} Characterization of {{NP}} with Optimal Amortized Query Complexity},
  booktitle = {Proceedings of the 32nd {{Annual ACM Symposium}} on {{Theory}} of {{Computing}}},
  author = {Samorodnitsky, Alex and Trevisan, Luca},
  date = {2000},
  pages = {191--199},
  publisher = {{Association for Computing Machinery}},
  location = {{Portland, Oregon, United States}},
  doi = {10.1145/335305.335329},
  eventdate = {2000-05-21/2000-05-23},
  eventtitle = {{{STOC}} 2000},
  venue = {Portland, OR, USA}
}

@article{ST09,
  title = {Gowers {{Uniformity}}, {{Influence}} of {{Variables}}, and {{PCPs}}},
  author = {Samorodnitsky, Alex and Trevisan, Luca},
  date = {2009-01},
  journaltitle = {SIAM Journal on Computing},
  shortjournal = {SIAM J. Comput.},
  volume = {39},
  number = {1},
  pages = {323--360},
  publisher = {{Society for Industrial and Applied Mathematics}},
  doi = {10.1137/070681612},
  note = {Conference version in STOC 2006}
}

@inproceedings{ST98,
  title = {Probabilistically Checkable Proofs with Low Amortized Query Complexity},
  booktitle = {Proceedings of the 39th {{Annual Symposium}} on {{Foundations}} of {{Computer Science}}},
  author = {Sudan, Madhu and Trevisan, Luca},
  date = {1998},
  pages = {18--27},
  publisher = {{IEEE Computer Society}},
  doi = {10.1109/SFCS.1998.743425},
  eventdate = {1998-11-08/1998-11-11},
  eventtitle = {{{SFCS}} 1998},
  venue = {Palo Alto, CA, USA}
}

@misc{Tre04,
  title = {Inapproximability of {{Combinatorial Optimization Problems}}},
  author = {Trevisan, Luca},
  date = {2004-09-23},
  eprint = {cs/0409043},
  eprinttype = {arxiv},
  archiveprefix = {arXiv}
}

@article{Tre12,
  title = {On {{Khot}}'s Unique Games Conjecture},
  author = {Trevisan, Luca},
  date = {2012-01},
  journaltitle = {Bulletin of the American Mathematical Society},
  shortjournal = {Bull. Amer. Math. Soc.},
  volume = {49},
  number = {1},
  pages = {91--111},
  doi = {10.1090/S0273-0979-2011-01361-1}
}

@article{Tre98-alg,
  title = {Parallel {{Approximation Algorithms}} by {{Positive Linear Programming}}},
  author = {Trevisan, Luca},
  date = {1998-05},
  journaltitle = {Algorithmica},
  shortjournal = {Algorithmica},
  volume = {21},
  number = {1},
  pages = {72--88},
  doi = {10.1007/PL00009209}
}

@inproceedings{Tre98-hardness,
  title = {Recycling Queries in {{PCPs}} and in Linearity Tests},
  booktitle = {Proceedings of the 30th {{Annual ACM Symposium}} on {{Theory}} of {{Computing}}},
  author = {Trevisan, Luca},
  date = {1998},
  pages = {299--308},
  publisher = {{Association for Computing Machinery}},
  doi = {10.1145/276698.276769},
  eventdate = {1998-05-24/1998-05-26},
  eventtitle = {{{STOC}} 1998},
  venue = {Dallas, Texas, USA}
}

@article{TSSW00,
  title = {Gadgets, {{Approximation}}, and {{Linear Programming}}},
  author = {Trevisan, Luca and Sorkin, Gregory B. and Sudan, Madhu and Williamson, David P.},
  date = {2000-01-01},
  journaltitle = {SIAM Journal on Computing},
  shortjournal = {SIAM J. Comput.},
  volume = {29},
  number = {6},
  pages = {2074--2097},
  publisher = {{Society for Industrial and Applied Mathematics}},
  doi = {10.1137/S0097539797328847},
  note = {Conference version in FOCS 1996}
}

@inproceedings{Yao77,
  title = {Probabilistic Computations: {{Toward}} a Unified Measure of Complexity},
  booktitle = {Proceedings of the 18th {{Annual Symposium}} on {{Foundations}} of {{Computer Science}}},
  author = {Yao, Andrew Chi-Chih},
  date = {1977-09-30},
  pages = {222--227},
  publisher = {{IEEE Computer Society}},
  doi = {10.1109/SFCS.1977.24},
  eventdate = {1977-10-31/1977-11-02},
  eventtitle = {{{SFCS}} 1977},
  venue = {Providence, RI, USA}
}

@article{Zhu20,
  title = {A {{Proof}} of the {{CSP Dichotomy Conjecture}}},
  author = {Zhuk, Dmitriy},
  date = {2020-08-26},
  journaltitle = {Journal of the ACM},
  shortjournal = {J. ACM},
  volume = {67},
  number = {5},
  pages = {30:1--30:78},
  doi = {10.1145/3402029},
  note = {Conference version in FOCS 2017}
}

@misc{Zwi00,
  title = {Analyzing the {{MAX}} 2-{{SAT}} and {{MAX DI-CUT Approximation Algorithms}} of {{Feige}} and {{Goemans}}},
  author = {Zwick, Uri},
  date = {2000-03-02},
  url = {https://www.cs.tau.ac.il/~zwick/papers/max2sat-maxdicut.ps.gz}
}

@inproceedings{Zwi98,
  title = {Approximation Algorithms for Constraint Satisfaction Problems Involving at Most Three Variables per Constraint},
  booktitle = {Proceedings of the 9th {{Annual ACM-SIAM Symposium}} on {{Discrete Algorithms}}},
  author = {Zwick, Uri},
  date = {1998},
  pages = {201--210},
  publisher = {{Association for Computing Machinery}},
  doi = {10.5555/314613.314701},
  eventdate = {1998-01-25/1998-01-27},
  eventtitle = {{{SODA}} 1998},
  venue = {San Francisco, CA, USA}
}
% \addcontentsline{toc}{chapter}{References}
% \bibliographystyle{apalike2}

% \include{endmatter/colophon}

\end{document}